\documentclass[12pt, reqno]{amsart}

\usepackage[all=normal, paragraphs=tight, bibbreaks=tight, floats=tight,
mathdisplays=tight,bibnotes=tight]{savetrees}

\usepackage{placeins}
\usepackage{titletoc}
\usepackage{chngcntr}
\usepackage[thm]{macros_proof}
\usepackage{thm-restate}
\usepackage{pdflscape}
\usepackage{afterpage}
\usepackage{times}

\usepackage{tabularx}
\usepackage{bold-extra}
\usepackage{pifont}%
\usepackage{appendix}
\usepackage{inconsolata}

\usepackage{clipboard}
\newclipboard{output-draft}

\usepackage[natbib=true, style=authoryear, maxcitenames=2, uniquename=false,
maxbibnames=9, uniquelist=false, backref=true,firstinits=false]
{biblatex}
\DefineBibliographyStrings{english}{%
  backrefpage = {pp.},%
  backrefpages = {pp.},%
}

\newif\ifrepeatthm
\repeatthmfalse

\newgeometry{margin=1.25in}
\newcommand\numberthis{\addtocounter{equation}{1}\tag{\theequation}}
\hypersetup{final}
\crefname{enumi}{}{}
\crefname{rmksq}{Remark}{Remarks}
\Crefname{rmksq}{Remark}{Remarks}
\crefname{rmk}{Remark}{Remarks}
\Crefname{rmk}{Remark}{Remarks}

\newcommand{\naive}{\textsc{naive}}
\newcommand{\oracle}{\textsc{oracle}}

\newcommand{\indepnpmle}{\textsc{independent-npmle}}
\newcommand{\meanrank}{\textsc{mean rank}}
\newcommand{\topprob}{\textsc{top-{\small 20} probability}}
\newcommand{\incarceration}{\textsc{incarceration}}
\newcommand{\pooled}{\textsc{pooled}}

\newcommand{\raw}{\mathrm{raw}}
\newcommand{\RF}{R_{\mathrm{F}}}
\newcommand{\RB}{R_{\mathrm{B}}}
\newcommand{\A}{\mathbf{A}}
\newcommand{\absauto}[1]{\left| #1 \right|}
\newcommand{\normauto}[1]{\left\lVert #1 \right\rVert}

\newcommand{\sigl}{\sigma_\ell}
\newcommand{\sigu}{\sigma_u}
\newcommand{\fs}[1]{f_{#1, \boldsymbol{\cdot}}}
\newcommand{\barh}{\bar h}

\newcommand{\invphi}{\varphi_{+}}

\newcommand{\hyperparams}{\sigl, \sigu, s_{0\ell}, s_{0u}}
\newcommand{\Hyperparams}{\mathcal{H}}

\newcommand{\leh}{\lesssim_{\hyperparams}}
\newcommand{\leH}{\lesssim_{\mathcal{H}}}
\newcommand{\hatm}{\hat m}
\newcommand{\hats}{\hat s}
\newcommand{\hateta}{\hat\eta}
\newcommand{\regret}{\mathrm{BayesRegret}_n}
\newcommand{\reg}{\mathrm{MSERegret}_n}

\newcommand{\sub}{\mathrm{Sub}}
\newcommand{\PE}{\mathbf{E}}
\newcommand{\comp}{\mathrm{C}}
\newcommand{\EB}{\mathrm{EB}}
\newcommand{\bY}{\mathbf{Y}}
\newcommand{\btheta}{\bm{\theta}}

\newcommand{\utilmax}{\textsc{utility maximization by selection}}
\newcommand{\topm}{\textsc{top-{\small $m$} selection}}
\newcommand{\utilmaxreg}{\mathrm{UMRegret}_n}
\newcommand{\topmreg}{\mathrm{TopRegret}_{n}^{(m)}}
\newcommand{\close}{\textsc{close}}
\newcommand{\npmle}{\textsc{npmle}}
\newcommand{\closenpmle}{\close-\npmle}
\newcommand{\closegauss}{\close-\textsc{gauss}}
\newcommand{\bdelta}{\bm{\delta}}
\newcommand{\Dinfty}{\norm{\hateta - \eta_0}_{\infty}} %

\newcommand{\indepgauss}{\textsc{independent-gauss}}

\title{Empirical Bayes When Estimation Precision Predicts Parameters}

\author{Jiafeng Chen \\ Stanford University and SIEPR \\ 
\href{mailto:jiafeng@stanford.edu}
{jiafeng@stanford.edu}}
\date{\today. This paper is based on the second chapter of my Ph.D. thesis at Harvard
 University. It was previously titled ``Gaussian Heteroskedastic Empirical Bayes without
 Independence.'' I thank my doctoral advisors, Isaiah Andrews, Elie Tamer, Jesse Shapiro,
 and Edward Glaeser, for their guidance and generous support. I thank Patrick Kline for
 discussing this paper at ASSA 2025. For comments and discussion, I thank two anonymous
 referees, Harvey Barnhard, Raj Chetty, Dominic Coey, Aureo de Paula, Bryan Graham,
 Jiaying Gu, Aditya Guntuboyina, Nathaniel Hendren, Keisuke Hirano, Peter Hull, Kenneth
 Hung, Lawrence Katz, Patrick Kline, Scott Duke Kominers, Soonwoo Kwon, Lihua Lei, Andrew
 Lo, Michael Luca, Anna Mikusheva, Joris Pinkse, Mikkel Plagborg-M\o {}ller, Azeem
 Shaikh, Suproteem Sarkar, Ashesh Rambachan, David Ritzwoller, Brad Ross, Jonathan Roth,
 Neil Shephard, Rahul Singh, Asher Spector, Harald Uhlig, Winnie van Dijk, Davide
 Viviano, Christopher Walker, Chris Walters,  and workshop and seminar participants at
 Brown, Harvard, Penn State, Philadelphia Fed, Rutgers, Princeton, Stanford, the
 University of Chicago, Berkeley, UCLA, Yale, USC, Paris Econometrics Seminar, California
 Econometrics Conference, UCSD, Cornell, University of Bonn, LMU Munich, the University
 of Mannheim, ASSA 2025, Amazon, and the Chamberlain Seminar. \Cref {asec:max_gauss} of
 this paper supersedes the preprint arXiv:2303.08653. An \textsf {R} implementation of
\close{} is found at \url{https://github.com/jiafengkevinchen/close}.  Replication files
are found at \url{https://github.com/jiafengkevinchen/close-replication}. I am responsible
 for any and all errors. }

\makeatletter

\let\newtitle\@title

\def\paragraph{\@startsection{paragraph}{4}%
  \z@\z@{-\fontdimen2\font}%
  {\normalfont\bfseries}}
\makeatother
\newcommand\DoToC{%
  \startcontents
  \printcontents{}{0}{\textbf{Contents}\vskip1em\hrule\vskip1em}
  \vskip1em\hrule\vskip5pt
}
\numberwithin{equation}{section}

\usepackage{etoolbox}

\makeatletter
\patchcmd{\@settitle}{\uppercasenonmath\@title}{\large}{}{}
\patchcmd{\@setauthors}{\MakeUppercase}{\vspace{-1em}\normalsize}{}{}
\patchcmd{\section}{\scshape}{\bfseries}{}{}
\makeatother

\newcolumntype{Y}{>{\centering\arraybackslash}X}

\begin{document}

{\onehalfspacing
\begin{abstract}
Gaussian empirical Bayes methods usually maintain a \emph{precision independence}
assumption: The unknown parameters of interest are independent from the known standard errors of the
estimates.
This assumption is often theoretically questionable and empirically rejected. This paper
proposes to model the conditional distribution of the parameter given the standard errors
as a flexibly parametrized location-scale family of distributions, leading to a family of
methods that we call \close. The \close{} framework unifies and generalizes several
proposals under precision dependence. We argue that the most flexible member of
the \close{} family is a minimalist and computationally efficient default for accounting
for precision dependence. We analyze this method and show that it is competitive in terms
of the regret of subsequent decisions rules. 
 Empirically, using \close{} leads to sizable gains for selecting high-mobility Census
 tracts.

\vspace{1em}

\end{abstract}

\maketitle
}
\newpage
\onehalfspacing

\section{Introduction}

Applied economists often use empirical Bayes methods to shrink noisy parameter estimates,
in hopes of accounting for the imprecision in the estimates and improving subsequent
decisions. Many such settings\footnote{Empirical Bayes methods are applicable whenever
many parameters for heterogeneous populations are estimated in tandem. These settings
include value-added modeling \citep
{angrist2017leveraging,mountjoy2021returns,chandra2016health,doyle2017evaluating,hull2018estimating,einav2022producing,abaluck2021mortality},
place-based effects \citep
{chyn2021neighborhoods,finkelstein2021place,chetty2018opportunity,chetty2018impacts,diamond2021standard,baum2019microgeography,aloni2023one},
discrimination \citep
{kline2022systemic,kline2023discrimination,rambachan2021identifying,egan2022harry,arnold2022measuring,montiel2021empirical},
meta-analysis \citep
{azevedo2020b,meager2022aggregating,andrews2019identification,elliott2022detecting,wernerfelt2022estimating,dellavigna2022rcts,abadie2023estimating},
and correlated random effects in panel data \citep
{chamberlain1984panel,arellano2009robust,bonhomme2020much,bonhomme2015grouped,liu2020forecasting,giacomini2023robust,bonhommedenis}.
}  can be described by a heteroskedastic Gaussian sequence model with known variances.
That is, researchers obtain statistical estimates $Y_i$ and accompanying standard errors
$\sigma_i$ for parameters $\theta_i$ associated with units $i=1,\ldots, n$. Motivated by
the central limit theorem, we model $Y_i$ as unbiased Gaussian signals on $\theta_i$ with
known variances $\sigma_i^2$:
 \[Y_i
 \mid \theta_i,
\sigma_i \sim \Norm(\theta_i, \sigma_i^2) \quad i=1,\ldots, n. \numberthis 
\label{eq:gaussian_heteroskedastic_location}\]  Loosely speaking, empirical Bayes methods
improve decisions---e.g., estimating $\theta_i$ or identifying units with high
 $\theta_i$---by pooling strength across the many estimates $(Y_i, \sigma_i)_{i=1}^n$ and
 accounting for differing levels of noise $\sigma_i$.

Commonly used empirical Bayes methods often assume \emph{precision independence}---that
the known standard errors $\sigma_i$ do not predict the underlying parameters $\theta_i$
(i.e., $\sigma_i \indep \theta_i$). However, precision independence is economically
questionable and empirically rejected in many contexts. %
 Inappropriately imposing it can harm empirical Bayes decisions, possibly even making them
 underperform decisions without shrinkage. Motivated by these concerns, this paper
 introduces and analyzes empirical Bayes methods that allow for precision dependence.

To be concrete, our empirical application \citep{bergman2019creating} uses empirical Bayes
methods to shrink raw economic mobility estimates $(Y_i, \sigma_i)$ of low-income
children, curated by \citet{chetty2018opportunity}. Here, $\theta_i$ represents true
unobserved economic mobility of low-income children from Census tract $i$. In this
context, precision independence assumes that the standard errors of these estimates do not
predict true economic mobility. However, more upwardly mobile Census tracts tend to have
noisier estimates, in part because they contain fewer low-income households. Consequently,
the standard errors $\sigma_i$ and true mobility $\theta_i$ are positively correlated.

In this context, imposing precision independence can be costly for decision-making. 
\citet{bergman2019creating} select high-mobility Census tracts by choosing those with high
 empirical Bayes posterior means (i.e., shrinkage estimates). Under precision
 independence, empirical Bayes methods shrink all estimates to their \emph
 {unconditional} mean (i.e., $\E[\theta_i]$) and shrink noisier estimates more
 aggressively. If $\theta_i$ and $\sigma_i$ are positively correlated, such shrinkage
 tends to systematically
 \emph{underestimate} true mobility of high-$\sigma_i$ tracts. This can harm subsequent
  selection decisions, if we wish to target high-mobility---hence disproportionately
  high-$\sigma_i$---tracts.\footnote{For a few
  measures of economic mobility where precision independence is severely violated, we
  find that screening on conventional estimates selects \emph{less} economically mobile
  tracts, on average, than screening on the unshrunk estimates. Fortunately, for the
  measure of economic mobility (mean income rank pooling over all demographic groups
  whose parents are at the 25\th {} percentile of household income) used in \citet
  {bergman2019creating}, the violation of precision independence is sufficiently mild, so
  that screening on these empirical Bayes shrinkage estimates still outperforms screening
  on the raw estimates. %
} In contrast, screening on shrinkage estimates computed by our methods selects
substantially more mobile tracts.

To introduce empirical Bayes methods, let us return to the Gaussian model 
\eqref{eq:gaussian_heteroskedastic_location}. 
Under this setup, empirical Bayes methods are rationalized as approximations of unknown
optimal decisions. Assume that $(\theta_i, \sigma_i)$ are drawn randomly from some
distribution. Then the optimal, infeasible decisions take the form of Bayes decision
rules for an \emph{oracle Bayesian}, whose prior is the unknown distribution of $
(\theta_i,
\sigma_i)$. Empirical Bayes methods emulate these {oracle decisions} by estimating
 the {oracle's prior} from the data. For instance,
  \emph{shrinkage estimation}, discussed so far, corresponds to using the estimated
   posterior means of $\theta_i$ given $(Y_i, \sigma_i)$ as a decision rule for predicting
   $\theta_1,\ldots,\theta_n$. Under this backdrop, precision independence further
   simplifies the problem of estimating the oracle's prior, but introduces poor
   performance when it fails to hold. 

This paper has two contributions. First, we propose a flexible but tractable framework for
modeling precision dependence that nests various proposals in the literature. Our methods
are then natural estimation strategies under this framework. \Cref {sec:model} models
$\theta_i \mid \sigma_i$ as a conditional location-scale family,\footnote
{A location-scale family with shape $G$, indexed by location $m$ and scale $s$, is a set
of distributions with cumulative distribution functions (CDFs) $F_ {m,s} (t) = G
\pr{\frac{t-m}{s}}$ as $m$ and $s$ vary. For instance, the family $\Norm(m,s^2)$ is
 location-scale with shape $G(t) = \Phi(t)$, for $\Phi$ the standard Gaussian CDF.}
 controlled by $\sigma_i$-dependent {location} hyperparameter $m_0 (\sigma) = \E
 [\theta \mid \sigma]$ and scale hyperparameter $s_0^2(\sigma) = \var(\theta
\mid \sigma)$. Under this assumption, different values of $\sigma_i$ translate, compress,
 or dilate the distribution $\theta_i \mid \sigma_i$, but the underlying {shape} $G_0$ of
 this distribution  is constant over $\sigma_i$. This model subsumes precision
 independence as the special case where the location and scale parameters are constant
 functions of $\sigma_i$. 

This model naturally gives rise to a family of \underline{c}onditional \underline{lo}cation-\underline{s}cale 
\underline{e}mpirical Bayes methods---which we call \close---by  estimating the
 hyperparameters $(m_0 (\sigma), s_0(\sigma),G_0)$. The \close{} framework also makes
 estimating these objects highly tractable. The location and scale hyperparameters $m_0
 (\cdot), s_0(\cdot)$ can be written as conditional moments of $Y \mid \sigma$, reducing
 their estimation to learning conditional expectation functions. Subsequently, given $
 (m_0(\cdot), s_0(\cdot))$, it is possible to normalize the data $(Y_i, \sigma_i)$ so as
 to remove precision dependence. After normalization, one could then apply conventional
 empirical Bayes methods to estimate the remaining
 hyperparameter $G_0$.

 The \close{} framework unifies and generalizes several proposals in the literature \citep
 [among others,][]
 {kline2023discrimination,weinstein2018group,george2017mortality,ignatiadis2019covariate}.
 These proposals can be viewed as specific modeling and estimation choices for $ (m_0,
 s_0, G_0)$. Various subsets of these proposals emphasize a nonparametric perspective for
 modeling and estimating various components of $(m_0, s_0, G_0)$; thus, a natural way to
 generalize is to adopt a nonparametric perspective for all of them. In particular, we
 advocate for using nonparametric regression to estimate $(m_0(\cdot), s_0 (\cdot))$ and for
 using \emph{nonparametric maximum likelihood} (\npmle) to estimate $G_0$
 \citep{kiefer1956consistency,jiang2009general,koenker2014convex}. We refer to this
 variant as \closenpmle. We view \closenpmle{} as a flexible, minimalist, and
 computationally efficient default, in the absence of substantive knowledge that motivates
 further restrictions on $ (m_0, s_0, G_0)$.

The second contribution of the paper is a theoretical analysis of \closenpmle{} in
\cref{sec:regret}. Our main result (\cref{cor:maintext,thm:minimaxlower}) establishes
that, under the
\close{} assumptions, \closenpmle {}
emulates the oracle Bayesian as well as possible in terms of squared error loss. Specifically, we establish upper and lower bounds for the squared
error \emph{Bayes regret} for \closenpmle. These upper and lower bounds match up to
logarithmic factors in the number of observations, indicating that \closenpmle {} attains
a regret rate that is approximately minimax optimal. These results extend existing regret
guarantees for \npmle-based empirical Bayes to account for precision dependence \citep
{soloff2021multivariate,jiang2020general,jiang2009general,saha2020nonparametric}. The key
technical difficulty is accounting for estimation error in $m_0$ and $s_0$, which feed
into
\npmle{} estimation.

We enrich our main result in two additional ways. First, to assess robustness
of \closenpmle{} to the \close{} assumption, we study a population version of \closenpmle
{} under misspecification of the location-scale model. \Cref{thm:worstcaserisk} finds
that its worst-case risk---under arbitrarily different shapes of $\theta_i
\mid \sigma_i$ as a function of $\sigma_i$---is within a bounded multiple of the risk of a
 minimax procedure. Second, we also extend our guarantee for squared error regret to the
 Bayes regret for two ranking-related decision problems, including the problem of
 selecting high-mobility tracts in \citet{bergman2019creating}. \Cref
 {thm:mserelevance} shows that the Bayes regret in squared error dominates the Bayes
 regret for these other decision problems. Coupled with
\cref{cor:maintext}, this implies that \closenpmle{} has good performance for these
 ranking-related problems as well.

To illustrate our method, \cref{sec:empirical} applies \close{} to two empirical
exercises \citep{chetty2018opportunity,bergman2019creating}. The
first exercise is a simulation calibrated to the Opportunity Atlas, the
dataset published by \citet{chetty2018opportunity}. For all 15
measures of economic mobility that we consider,
\closenpmle{} improves over all alternative methods and captures over 90\% of possible
 mean-squared error (MSE) gains relative to no shrinkage, whereas conventional empirical
 Bayes methods capture only 70\% on average and as little as 50\% for some.

The second exercise evaluates the out-of-sample performance of various procedures for 
selecting high-mobility Census tracts \citep{bergman2019creating}, using an out-of-sample
validation procedure based on the coupled bootstrap that we introduce
\citep{oliveira2021unbiased}. \citet {bergman2019creating} use empirical Bayes procedures
to select high-mobility Census tracts in Seattle. In an exercise that mimics theirs, we
find that \closenpmle{} selects more economically mobile tracts than conventional methods.
Conventional methods, on the other hand, frequently select less mobile tracts than
screening based on the noisy estimates directly. The improvements of
\closenpmle{} over the standard method are on median 2.6 times the \emph{value of basic
empirical Bayes}---that is, the improvements the standard method delivers over screening
on the raw estimates $Y_i$ directly. Therefore, for this application, if one finds using
the standard empirical Bayes method a worthwhile methodological investment, then the
additional gain of using \close{} is likewise meaningful.

\section{Model and proposed method}
\label{sec:model}

\subsection{Empirical Bayes assumptions} We observe estimates $Y_i$ and
their standard errors $\sigma_i$ for parameters $\theta_i$, over populations $i \in
\br{1,\ldots,n}$. We maintain two assumptions that are standard in the empirical Bayes
literature
\citep{gilraine2020new,jiang2020general,soloff2021multivariate,gu2023invidious,guwalters_eb,eb_hole}.

First, we assume throughout that the estimates are conditionally Gaussian with known
variances equal to $\sigma_i^2$  and are 
independent across $i$ \eqref{eq:gaussian_heteroskedastic_location}. The Gaussian model \eqref
{eq:gaussian_heteroskedastic_location} is heuristically motivated by a central limit
theorem applied to the underlying micro-data. This assumption is not without loss: We
ignore the fact that the central limit theorem is only an approximation and treat the
Normality as exact.  \Copy{seremark}{As a concrete example \citep[cf. Example 2 in][]
{eb_hole}, suppose
$\theta_i =
 \E_ {Q_i}
 [Y_ {ij}]$ is the population mean of some variable $Y_{ij}
\sim Q_i$ drawn from population $Q_i$. A natural estimator $Y_i$ of $\theta_i$ is the
sample mean of $Y_{i1},\ldots, Y_{in_i}$. A natural estimate for the variance of $Y_i$ is
$\sigma_i^2 = n_i^{-2} \sum_ {j=1}^ {n_i} (Y_ {ij}-Y_i)^2$. By standard arguments, as $n_i
\to \infty$, $\smash{\sigma_i^{-1}(Y_i -
\theta_i) \dto \Norm(0,1)}$. This heuristically motivates 
\eqref{eq:gaussian_heteroskedastic_location} by replacing \smash{``$\dto$''} with
``$\sim$.''\footnote{Note
too that $Y_i - \theta_i = O_P(n_i^{-1/2})$ and $\sigma_{i} - n_i^{-1/2}\var_
{Q_i} (Y_{ij}) = O_P(n_i^{-1})$, and so the estimation error in $\sigma_i$ is negligible
compared to the estimation error in $Y_i$, thereby heuristically justifying treating the
estimated standard error $\sigma_i$ as the true variance of $Y_i$.}}

Second, we assume that $ (\theta_i, \sigma_i)$ are random and sampled i.i.d.  from some
distribution. Since empirical Bayes methods estimate the distribution of $
(\theta_i,\sigma_i)$, it is natural to think of $(\theta_i, \sigma_i)$ as random. For
minor technical reasons, throughout, we condition on $\sigma_ {1:n} = (\sigma_1,\ldots,
\sigma_n)$ and treat them as fixed. Thus, we think of $\theta_i$ as drawn
independently but not necessarily identically:
\[
\theta_{i} \mid \sigma_{i} \overset{\mathrm{i.n.i.d.}}{\sim} G_{(i)}. \numberthis
\label{eq:eb_sampling}
\]
Let $P_0 \equiv (G_ {
(1)},\ldots, G_ {(n)})$ denote the conditional distribution $\theta_
{1:n} \mid \sigma_{1:n}$.

\Copy{covariates}{Throughout, we focus on a setting without additional covariates $X_i$,
returning to accommodating for covariates in the empirical application
(\cref{sec:empirical}). Our methods generalize immediately to settings with
covariates $X_i$---as long as $Y_i \mid X_i, \theta_i, \sigma_i
 \sim \Norm (\theta_i,
 \sigma_i^2)$---by treating $X_i$ symmetrically as $\sigma_i$. We focus on $\sigma_i$ since
  it is always present in heteroskedastic empirical Bayes settings, and it enters the
  likelihood of $Y_i$ unlike other covariates.} Likewise, for simplicity, we focus on a
 setting where $(Y_i,
\theta_i, \sigma_i)$ are independently distributed: We briefly discuss dependence across
$i$ in \cref{rmk:independence}.

Under these assumptions, empirical Bayes methods are desirable for decision-making: They
approximate optimal but infeasible decision rules. To see this,  consider a decision
problem with loss function $L(\bdelta, \theta_{1:n})$, which evaluates an action $\bdelta$
at a vector of parameters $\theta_{1:n}$. The optimal decision---in terms of expected
loss $\E_{P_0}[L(\cdot, \theta_{1:n}) \mid \sigma_{1:n}]$ over $(Y_i,\theta_i) \mid
\sigma_i$---chooses actions that minimize the posterior expected loss under prior $P_0$: \[
\bdelta^\star(Y_{1:n}, \sigma_{1:n}; P_0) \in \argmin_{\bdelta} \E_{P_0}[L(\bdelta,
\theta_{1:n})
\mid Y_{1:n}, \sigma_{1:n}]. \numberthis \label{eq:oracle_bayes}
\]
For this reason, we refer to $\bdelta^\star$ as the oracle Bayes decision rule, and
think of it as the Bayes decision rule for an oracle whose prior is $P_0$. $\bdelta^\star$
is infeasible since we do not know
$P_0$. To remedy, empirical Bayes methods seek to approximate the oracle Bayes rule
$\bdelta^\star$%
. Naturally, one recipe is to plug an estimate $\hat P$ for $P_0$ into
\eqref{eq:oracle_bayes}:\footnote{To
emphasize the distinction between the true
expectation with respect to the data-generating process
\eqref{eq:eb_sampling} and a posterior mean taken with respect to some possibly estimated
measure $\hat P$, we shall use $\E$ to refer to the former and $\PE$ to refer to the
latter. Subscripts typically make the distinction clear as well.
} \[
\bdelta_{\EB} (Y_{1:n}, \sigma_{1:n}; \hat P) \in \argmin_{\bdelta} \PE_{\hat P}[L
(\bdelta,
\theta_{1:n})
\mid Y_{1:n}, \sigma_{1:n}]. \numberthis \label{eq:empirical_bayes_rule}
\]
For the decision problem where $L(\bdelta, \theta_{1:n}) = \frac{1}{n}\sum_
{i=1}^n (\delta_i - \theta_i)^2$ is mean-squared error, 
\eqref{eq:empirical_bayes_rule} generates empirical Bayes posterior means $\PE_{\hat P}
 [\theta_i \mid Y_i, \sigma_i]$, often referred to as shrinkage estimates 
 \citep{james1992estimation,efron1973stein}.

 To simplify the estimation of $P_0$, popular empirical
Bayes methods often assume \emph{precision independence}: $\theta_i \indep \sigma_i$, or,
equivalently, $G_{(1)} =
\cdots = G_{(n)}$ in \eqref{eq:eb_sampling} and equal to some distribution $G_{(0)}$. For
 instance, the standard parametric empirical Bayes method models $G_{(i)}$ as i.i.d.
 Gaussian, $G_{ (0)} \sim
\Norm(m_0, s_0^2)$ \citep{morris1983parametric}. State-of-the-art empirical Bayes methods
 relax the parametric assumptions on $G_{(0)}$ and estimate $G_{(0)}$ with
\emph{nonparametric maximum likelihood}, or \npmle{}
\citep{jiang2020general,gilraine2020new,soloff2021multivariate}. Henceforth, we refer to
these methods as \indepgauss{} and \indepnpmle{}, respectively. The
``\textsc{independent}'' here emphasizes precision independence.

\subsection{Precision independence and its violation}

Despite its convenience, precision independence may be economically implausible; imposing
it may cause empirical Bayes methods to underperform. We illustrate this with an
application to the Opportunity Atlas \citep{chetty2018opportunity}. There, one published
measure of economic mobility $\theta_i$ of tract $i$ defines it as the probability that a
Black individual becomes relatively high-income (i.e., having family income in the top 20
percentiles nationally) after growing up relatively poor in tract $i$ (i.e., with parents
at the 25\th {} percentile nationally).

Intuitively, Census tracts with more low-income Black households should have {more
precise} estimates of $\theta_i$, simply because there is a larger sample size to estimate
$\theta_i$. However, it is likely that these tracts are also on average poorer and are
thus less economically mobile. Thus, these Census tracts should have smaller $\sigma_i$
but also lower $\theta_i$, meaning that $(\sigma_i, \theta_i)$ are positively correlated.

\begin{figure}[!htb]
  \centering
  \includegraphics[width=0.9\textwidth]{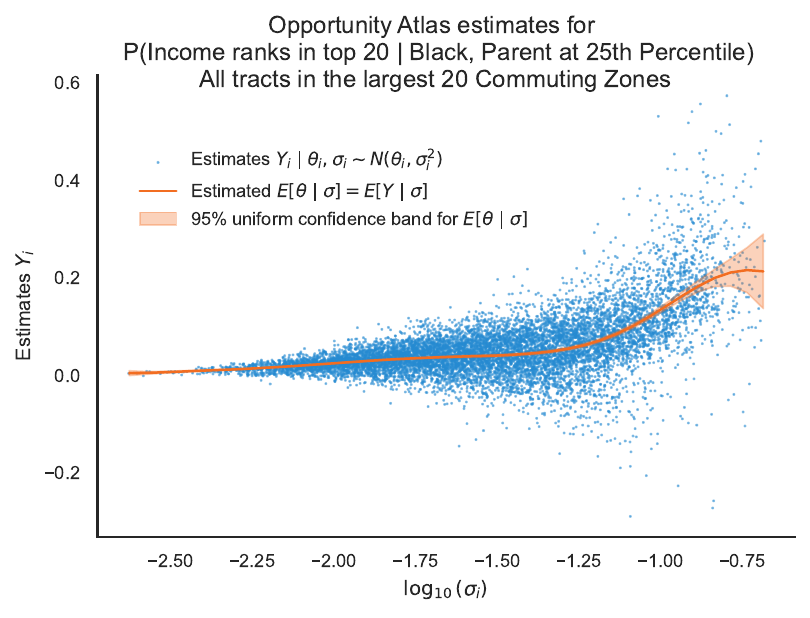}

  \begin{proof}[Notes]
  All tracts within the largest 20 Commuting Zones are shown. Due to the regression
   specification in \citet{chetty2018opportunity}, point estimates of $\theta_i \in [0,1]$
   do not always lie within $[0,1]$. The orange line plots nonparametric regression
   estimates of the conditional mean $\E[Y \mid
  \sigma] = \E[\theta \mid \sigma] \equiv m_0(\sigma)$, estimated via local linear
   regression implemented by \citet{calonico2019nprobust}. The orange shading shows a
   95\% uniform confidence band, constructed by the max-$t$ confidence set over 50
   equally spaced evaluation points. See \cref{sec:nuisance_estimation} for details on
   estimating conditional moments of $\theta_i$ given $\sigma_i$.
  \end{proof}

  \caption{Scatter plot of $Y_i$ against $\log_{10}(\sigma_i)$ in
  \citet{chetty2018opportunity}}

  \label{fig:kfr_top20_black_raw_data}
\end{figure}

As this economic intuition predicts, precision independence is readily rejected for this
measure of economic mobility. \Cref {fig:kfr_top20_black_raw_data} plots the estimates
$Y_i$ against their standard errors, overlaying an estimate of the conditional mean
function $m_0(\sigma_i)
\equiv
\E[\theta_i \mid \sigma_i] =
\E[Y_i \mid \sigma_i]$. If $\theta_i$ were independent of $\sigma_i$, then the
 true conditional mean function $m_0(\sigma_i)$ should be constant. \Cref
 {fig:kfr_top20_black_raw_data} shows the contrary---tracts with more imprecisely
 estimated $Y_i$ indeed tend to have higher $\theta_i$.

\begin{figure}[!htb]
  \centering
  \includegraphics[width=0.9\textwidth]
  {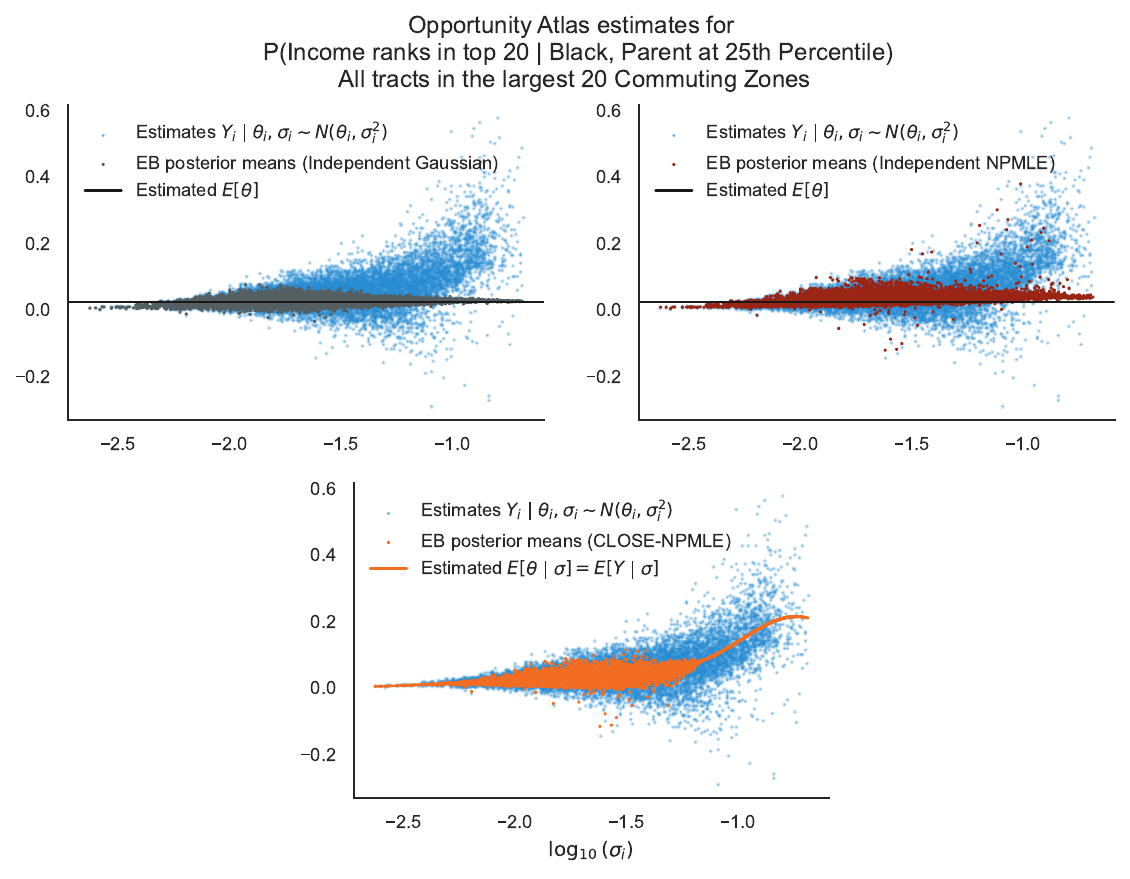}

  \begin{proof}[Notes] The top left panel shows posterior mean estimates with \indepgauss.
   The top right panel shows the same with \indepnpmle. The bottom panel displays
   posterior mean estimates from our preferred procedure, \closenpmle. In the top panels,
   the estimates for the unconditional mean and variance of $\theta_i$ are weighted by
   the precision $1/\sigma_i^2$, following \citet{bergman2019creating}.
  \end{proof}

  \caption{Posterior mean estimates under precision independence}
  \label{fig:kfr_top20_black_pooled_raw_grand_mean_shrink}
\end{figure}

What happens if we apply empirical Bayes methods that assume precision independence here?
\Cref{fig:kfr_top20_black_pooled_raw_grand_mean_shrink} overlays empirical Bayes posterior
means on the scatterplot. In the top left panel,
\indepgauss{} shrinks $Y_i$ towards a common estimated mean $\hat m_0$, depicted
as the black line. When $\sigma_i$ and $\theta_i$ are positively correlated, estimated
posterior means under \indepgauss{} {systematically undershoot}
$\theta_i$ for tracts with imprecise estimates. Similarly, the top right panel of
\cref{fig:kfr_top20_black_pooled_raw_grand_mean_shrink} shows that \indepnpmle{} suffers
from the same undershooting. In contrast, the bottom panel of
\cref{fig:kfr_top20_black_pooled_raw_grand_mean_shrink} previews our preferred procedure,
\closenpmle, which shrinks towards the conditional mean $\E[\theta_i \mid
\sigma_i]$, thus avoiding the undershooting.

\Copy{msevsranking}{Nonetheless, posterior means from \indepgauss{} or \indepnpmle{} may
still be better predictors, on average, for $\theta_i$ in mean-squared error than the
 noisy $Y_i$ \citep{james1992estimation,efron1975data}. However, the undershooting for
 large $\sigma_i$ is particularly problematic if one hopes instead to \emph{select}
 high-mobility Census tracts based on these posterior means, as do \citet
 {bergman2019creating}.} On average, high-mobility tracts are exactly those with high
 $\sigma_i$. Underestimating mobility for these tracts thus leads to suboptimal selections
 that may even underperform screening directly based on $Y_i$ \citep[see, also, ][]
 {mehta2019measuring}.

\begin{figure}[htb]
  \centering
  \includegraphics[width=0.9\textwidth]{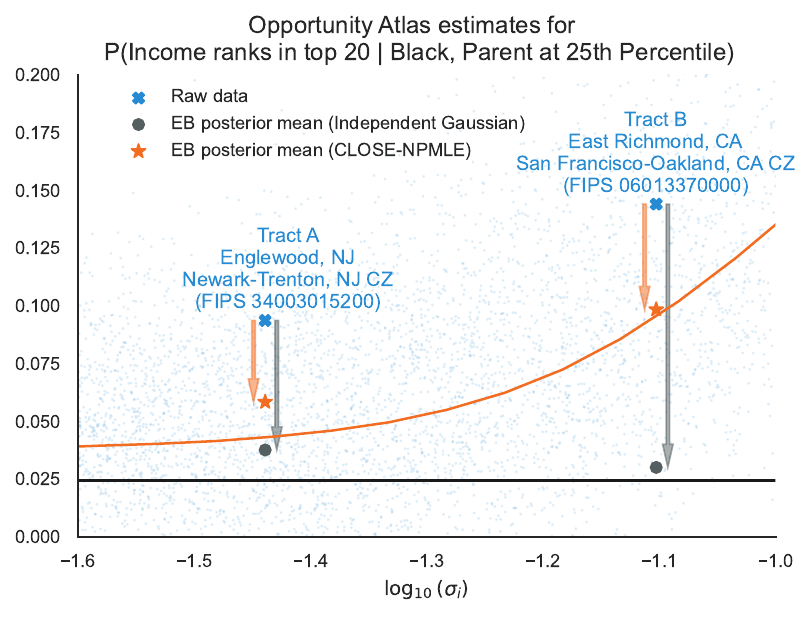}
  
  \begin{proof}[Notes] This plot shows a subregion of \cref
   {fig:kfr_top20_black_raw_data}  and highlights two Census tracts. The two tracts are
   those with $\log_{10}(\sigma_i) < -1.1$ with the highest $Y_i$, for which the selection
   decisions from \indepgauss{} and \closenpmle{} disagree. Like
   \citet{bergman2019creating}, the selection decisions aim to select 1/3 of Census tracts
   so as to maximize the average $\theta_i$ selected, by screening for the top 1/3 of
   empirical Bayes posterior means (formally, see \cref {ex:topm}).
  \end{proof}

  \caption{Ranking decisions for two Census tracts}
  \label{fig:example_shrink_ranking}

\end{figure}

To see this, \cref{fig:example_shrink_ranking} zooms into a subregion of
\cref{fig:kfr_top20_black_raw_data} and highlights two Census tracts, one in Englewood,
 NJ, and one in Richmond, CA---referring to them by tracts $A$ and $B$, respectively.
 Demographically, tract $A$ is 77\% nonwhite according to the 2010 Census, and tract $B$
 is 57\% nonwhite, contributing to different $\sigma_i$'s. Tract $A$ has a  lower raw
 estimate $Y_i$ than tract $B$ ($Y_A < Y_B$); and tracts with similar $\sigma_i$ to tract
 $A$, on average, also have lower estimates than those similar to tract $B$ (i.e., $\hatm
 (\sigma_A) < \hatm (\sigma_B)$). Either gap between the two tracts is
 substantial.\footnote{Both $Y_B-Y_A$ and $\hatm(\sigma_B) - \hatm(\sigma_A)$ are about
 five percentage points. For reference, an estimate of the unconditional standard
 deviation of $\theta_i$ is 3.7 percentage points.} These observations are compelling
 evidence in favor of $\theta_B > \theta_ {A}$: If one would like to select a Census
 tract to recommend, then, between $A$ and $B$, one is probably better off recommending
 tract $B$.

However, \indepgauss{} shrinks both to an estimate of the unconditional mean, which
results in a higher posterior mean estimate for tract $A$. In doing so---fooled by an
excessively low shrinkage target for tract $B$---\indepgauss{} recommends tract $A$ over
$B$ instead. In contrast, our preferred method (\closenpmle) computes posterior means that
preserve the more plausible ordering of the two tracts. We do so by modeling the
conditional distribution of $\theta_i
\mid \sigma_i$ more flexibly, which we turn to now.

\subsection{Conditional location-scale modeling of precision dependence}
\label{sub:close}
We propose the following \emph{conditional location-scale model} as a  relaxation: For a
distribution $G_0$ normalized to have zero mean and unit variance, $\theta_i$ has the
following
representation
\begin{align*}
\theta_i = m_0(\sigma_i) + s_0(\sigma_i) \tau_i &\quad\text{ where }\quad \tau_i \mid \sigma_i \iid
G_0 \quad \eta_0(\cdot) \equiv (m_0(\cdot), s_0(\cdot)) .
\numberthis 
\label{eq:location_scale}
\end{align*}
\eqref{eq:location_scale} states that the conditional distribution $\theta \mid
\sigma$ depends on $\sigma$ via $m_0(\sigma)$ and $s_0(\sigma)$. The function $m_0(\cdot)$
translates the \emph{location} of the distribution and the function $s_0(\cdot)$ controls
the \emph{scaling}. The underlying \emph{shape} of the distribution is governed by $\tau_i
\sim G_0$ and
is restricted by \eqref{eq:location_scale} to be invariant across different $\sigma_i$
values. By the normalization of $G_0$, we can think of $m_0 (\cdot)$ as the
conditional mean of $\theta_i \mid \sigma_i$ and $s_0^2(\cdot)$ as the conditional
variance.

Applying the empirical Bayes recipe \eqref{eq:empirical_bayes_rule} amounts to
estimating the
unknown hyperparameters $(\eta_0, G_0)$. Estimating $\eta_0 = (m_0(\cdot), s_0(\cdot))$ is
 straightforward, as $\eta_0$ can be written as conditional moments
of $Y_i \mid \sigma_i$: \[m_0 (\sigma) = \E[\theta
\mid \sigma] = \E[Y \mid
\sigma]  \quad\text{ and } \quad s_0^2 (\sigma) = \var (\theta \mid \sigma) =
\var(Y\mid \sigma) -
\sigma^2. \numberthis \label{eq:m_s_def}
\]
Estimating $\eta_0$ thus reduces to estimating conditional expectation functions.%

Estimating $G_0$ is more complicated. We do so by normalizing away the precision
dependence. Consider transforming $(Y_i, \sigma_i)$ into $(Z_i, \nu_i)$, defined by
$Z_i \equiv \frac{Y_i - m_0(\sigma_i)}{s_0(\sigma_i)}$ and $\nu_i
\equiv \frac{\sigma_i}{s_0(\sigma_i)}$. Note that 
\eqref{eq:location_scale} implies that
\[
Z_i \mid \tau_i, \nu_i^2 \sim \Norm(\tau_i, \nu_i^2) \quad \tau_i \mid \sigma_i, \nu_i
\iid G_0.
\numberthis
\label{eq:location_scale_tau_form} 
\]
\eqref{eq:location_scale_tau_form} makes clear that, first, the transformed triplet $(Z_i,
 \tau_i, \nu_i)$ obeys an analogue of the Gaussian model
\eqref{eq:gaussian_heteroskedastic_location}, where $Z_i$ is a noisy Gaussian
signal on $\tau_i$ with variance $\nu_i^2$. Second, precision independence holds in
\eqref{eq:location_scale_tau_form}, since $\tau_i
\mid \nu_i \iid G_0$.

This observation motivates the following strategy: First, estimate $m_0$ and $s_0$ with
$\hatm(\cdot)$ and $\hats(\cdot)$ so as to transform $ (Y_i, \sigma_i)$
into $ (\hat Z_i,
\hat \nu_i)$: \[
\hat Z_i = \frac{Y_i - \hatm(\sigma_i)}{\hats(\sigma_i)} \quad \text {and } \quad
\hat\nu_i = \frac{\sigma_i}
 {\hats(\sigma_i)}. \numberthis \label{eq:transformed_data}\] 
Second, apply empirical Bayes methods that assume precision independence on $(\hat Z_i,
\hat\nu_i)$ to
estimate $G_0$. This leads to a family of empirical Bayes strategies that we refer to as
conditional location-scale empirical Bayes, or \close:%

\begin{center}
\begin{minipage}{0.95\textwidth}

\begin{enumerate}[label=$\boxed{\textbf{\close--\textsc{step} {\small\arabic*}}}$,wide]
  \item \label{item:close1} Estimate $m_0(\sigma), s_0^2(\sigma)$ according to
  \eqref{eq:m_s_def}.  %

  \item \label{item:close2} With the estimates $\hateta = (\hatm, \hats)$, transform the
  data according to
  \eqref{eq:transformed_data}. Apply empirical Bayes methods under precision independence
  to
  estimate $G_0$ with some $\hat G_n$ on the transformed data $
   (\hat Z_i, \hat \nu_i)$.

  \item  \label{item:close3} Having estimated $(\hateta, \hat G_n)$ and hence having
  obtained $\hat P$, we then form empirical Bayes decision rules following \eqref
  {eq:empirical_bayes_rule}.
\end{enumerate}

\end{minipage}
\end{center}

This framework produces a family of empirical Bayes strategies,  since
\cref{item:close1,item:close2} can take different forms that practitioners can plug and
 play. \Copy{covariatesrec}{When there are additional covariates $X_i$ (independent of
 the noise $ \frac{Y_i -
\theta_i}{\sigma_i}$), researchers can choose instead to model $m_0(\sigma_i, X_i)$ and
 $s_0(\sigma_i, X_i)$ that include these covariates, and estimate $G_0$ after normalizing
 by $m_0(\sigma_i, X_i)$ and $s_0(\sigma_i, X_i)$.}

This paper focuses on a particular implementation which we call \closenpmle. It uses
nonparametric regression for \cref{item:close1} and \npmle{} for 
\cref{item:close2}. We recommend this method as a flexible default and primarily analyze
it in \cref{sec:regret}. We conclude this section with several
self-contained discussions on implementations of
these two steps, the rationale for \eqref{eq:location_scale} and
\closenpmle, and other miscellaneous issues.

\subsection{Discussions}
\label{sub:model_discussions}

\subsubsection{Implementation}
\label{sub:implementation}

For \cref{item:close1}, one can exploit \eqref{eq:m_s_def} by plugging in estimates of
conditional expectation functions. For $\hat \E[\cdot
\mid \sigma]$ an estimator of conditional means, we may let $\hatm (\sigma) = \hat\E [Y
\mid
\sigma]$ and $\hats^2(\sigma) = \hat\E[(Y - \hatm(\sigma))^2 \mid \sigma] - \sigma^2$. The
 estimator $\hat\E[\cdot \mid \sigma]$ itself may be nonparametric or based on
 judiciously chosen parametric models \citep[see][ for suggestions of the latter]
 {eb_hole}. \Copy
 {supportcomment}{The estimation of $\eta_0$ should also impose known support
 restrictions on $\eta_0$. For instance, the conditional variance estimate $\hat s_0$
 should be nonnegative (see \cref{rmk:practical}), and the conditional mean estimate
 should be within the support of $\theta_i$.} Our subsequent theoretical results simply
 assume that the estimators for $m_0(\cdot), s_0(\cdot)$ are well-behaved and are
 uniformly accurate.

For \cref{item:close2}, one could again model $G_0$ nonparametrically or parametrically.
As a flexible, performant, and minimalist default in the absence of stronger views on the
shape $G_0$, we focus on using \npmle{} to estimate $G_0$ \citep
{koenker2019comment}. Formally, the
\npmle{} $\hat G_n$ maximizes
the log-likelihood of $\hat Z_i$, whose marginal distribution is the convolution $G_0
\star
\Norm(0,
 \hat\nu_i^2)$: %
 For $\varphi(\cdot)$ the Gaussian probability density function and $\mathcal P(\R)$ the
 set of all distributions supported on $\R$, we maximize
\[
\hat G_n \in \argmax_{G \in \mathcal{P}(\R)} \frac{1}{n} \sum_{i=1}^n \log \int_{-\infty}^\infty \varphi\pr{\frac{\hat Z_i - \tau}{\hat \nu_i}} \frac{1}{\hat \nu_i} \, G(d\tau).
\numberthis \label{eq:npmle}
\]
In practice, we approximate $\mathcal P(\R)$ with finitely-supported distributions on a
grid in order to compute \eqref{eq:npmle}
\citep{koenker2014convex}.\footnote{\citet{koenker2017rebayes} provide an efficient
software implementation for \eqref{eq:npmle}, which we use throughout.

In terms of grid choice, theoretically, the only downside of a finer grid is computational
burden. Ideally, adjacent grid points should have a sufficiently small and economically
insignificant gap between them. In our empirical exercises, since the distribution $G_0$
of $\tau_i$ have zero mean and unit variance, we find that a fine grid within $[-6, 6]$
(e.g., 400 equally spaced grid points), with a coarse grid on $[\min_i \hat Z_i,
\max_i \hat Z_i] \setminus [-6, 6]$ (e.g., 100 equally spaced grid points), performs well.
Our subsequent theory accommodates an approximate maximizer of the likelihood, and thus
accommodates the discretization (\cref{as:npmle}). 
} 

\Copy{closegauss}{On the other hand, a default \emph{parametric} model for $G_0$ is to
 simply assume that $G_0 \sim \Norm(0,1)$, which we refer to as \closegauss. This
 approach amounts to using 
\indepgauss{} on the transformed estimates $(Z_i, \nu_i)$, with knowledge that
 the prior $G_0$ has zero mean and unit variance.} Under this model, the oracle Bayes
 posterior means are: \[
\delta_{\text{\closegauss}}^*(Y_i, \sigma_i) = \frac{\sigma_i^2}{s_0^2(\sigma_i) +
\sigma_i^2}
m_0
(\sigma_i) + \frac{s_0^2(\sigma_i)}{s_0^2(\sigma_i)
+ \sigma_i^2} Y_i. \numberthis \label{eq:gaussian_cond_b}
\]
Despite being rationalized under the assumption $\theta_i \mid\sigma_i
\sim \Norm(m_0(\sigma_i), s_0^2(\sigma_i))$, this oracle \eqref{eq:gaussian_cond_b}
enjoys strong robustness properties\footnote{\Cref{thm:worstcaserisk} shows that oracle versions of \closenpmle
 {} satisfy analogous but weaker robustness properties when the location-scale model
 fails. } even without the location-scale model 
\eqref{eq:location_scale} and the assumption that $G_0 \sim \Norm(0,1)$.
First,
\eqref{eq:gaussian_cond_b} is the optimal linear-in-$Y$ decision rule for estimating
$\theta_i$ in
 squared error 
\citep{weinstein2018group}; second,
\eqref{eq:gaussian_cond_b} is minimax in the sense that it minimizes the worst-case
 mean squared error over choices of $G_{(1)}, \ldots, G_{(n)}$ among all decision rules
 (see \cref{lemma:optimal_bayes_linear,lemma:minimax_close_gauss} for formal statements,
 respectively).  This method performs almost as well as \closenpmle{} in our empirical
 exercises.

\subsubsection{The location-scale assumption and \closenpmle}
\label{subsub:rationale}

\begin{table}[tb]
  \caption{Various existing methods fit into the \close{} framework}
  \label{tab:closetable}
  \centering
\scriptsize
  \begin{tabularx}{\textwidth}{m{0.24\textwidth}XX}
\toprule
 & Step 1 & Step 2 \\ \midrule
\citet{weinstein2018group} & Partition-based nonparametric estimator for $m_0, s_0^2$ &
$G_0 \sim \Norm
(0,1)$ \\ 
\citet{george2017mortality} & Parametric models for $m_0, s_0^2$ & $G_0 \sim \Norm(0,1)$
\\
\citet{chamberlain1984panel} & Parametric models for $m_0,
s_0^2$ & $G_0 \sim \Norm(0,1)$ \\ 
\citet{efron2016empirical} & Constant $m_0, s_0^2$ & $G_0$ nonparametric (log-spline sieves)
\\
\citet{kline2023discrimination} & $m_0 = c_1 s(\cdot)$, $s_0 = c_2 s(\cdot)$ for
 parametric $s(\cdot)$ & $G_0$ nonparametric (log-spline sieves) \\
\indepnpmle & Constant $m_0, s_0^2$ & $G_0$ nonparametric \\
\indepgauss & Constant $m_0, s_0^2$ & $G_0 \sim \Norm(0,1)$ \\ 
\citet{ignatiadis2019covariate} & Nonparametric $m_0$, constant $s_0^2$ & $G_0 \sim \Norm
(0,1)$ \\
\citet{jiang2010empirical} & Constant $m_0$, $s_0(\sigma) = \sigma$ 
(see \cref
 {rmk:alts_to_close}) & $G_0$ nonparametric
\\\bottomrule
  \end{tabularx}
\end{table}

We argue that the location-scale assumption provides a unifying framework for a number of
existing methods, and \closenpmle{} is a natural generalization of these methods within
this framework. We also briefly speculate how to generalize beyond \closenpmle.

Several existing methods can be thought of as implementations of
\close{} by making different choices in \cref{item:close1} and \cref{item:close2}. \Cref 
{tab:closetable} summarizes how these methods fit into the \close{} framework. Among these
methods, some choose nonparametric models for \cref{item:close1} and some choose
nonparametric models for \cref{item:close2}. For instance, \citet{weinstein2018group}
propose
\closegauss, with a partition-based nonparametric estimator for $m_0, s_0^2$. 
\citet{kline2023discrimination} consider a scale family $\theta_i = s_0(\sigma_i; \beta)
\tau_i$ for some
$\tau_i \mid \sigma_i \iid G_0$; they model $s_0(\sigma_i; \beta)$ parametrically, but
model $G_0$
flexibly using a log-spline sieve \citep{efron2016empirical}. 
\citet{george2017mortality} propose a fully Bayesian model whose components feature
parametric choices for $m_0, s_0$ with $G_0 \sim \Norm(0,1)$.

While the right modeling approach likely depends on the particular empirical context,
various subsets of these proposals emphasize being flexible in at least one of the two
steps. Thus, absent substantive knowledge that motivates more restrictive assumptions, a
natural default that unifies these approaches is to be flexible in both steps. Among
nonparametric methods,
\closenpmle{} may be particularly attractive due to its minimalism: The \npmle{} is free of
tuning parameters \citep{koenker2019comment}, and tuning parameter choices for
nonparametric regression are relatively well-understood
\citep{calonico2019nprobust,armstrong2018optimal}. That said, at a high level, when
precision
dependence is an issue, any approach that models and estimates $m_0, s_0, G_0$ well is
likely to perform well.

\Copy{transforms}{While \closenpmle{} naturally generalizes the existing methods in 
\cref{tab:closetable}, one might consider methods that do not impose
\eqref{eq:location_scale} and are even more flexible. These methods are potentially more
theoretically and computationally cumbersome: For instance, we can show that these
flexible methods can no longer transform $Y_i$ into some $Z_i = h(Y_i, \sigma_i)$ so as to
exploit precision independence on the transformed model $Z_i
\mid \tau (\theta_i,
\sigma_i), \sigma_i$.\footnote{This is because transforms that preserve
 linear exponential family structure are necessarily affine. Exponential family structure
is important for empirical Bayes because Tweedie's formula holds \citep
 {efron2011tweedie,efron2022exponential}. For an affine transform, the only way for $Z_i
 = a(\sigma_i) + b(\sigma_i) Y_i$ to satisfy precision independence is if
\eqref{eq:location_scale} holds. See \cref{lemma:transform} for a precise statement.}
In this sense, these methods must depart substantially from those that impose precision
independence.}

A natural approach is to estimate \npmle{} locally around $\sigma$ values, and we consider
these approaches important venues of future work. One might consider discretizing observed
$\sigma_i$ values into bins and apply
\indepnpmle{} within each bin.\footnote{Our Monte Carlo exercise in \cref{sec:empirical}
uses a similar approach to construct a Monte Carlo data-generating process. Thus, the
oracle performance in the Monte Carlo is the best-case scenario for the performance of
this procedure. There, we find \closenpmle{} performs well relative to the oracle and thus
to this procedure (\cref{fig:mse_table}).} A smoother---but more computationally
intensive---alternative is to estimate the posterior at some given $\sigma$ by considering
only observations with $\sigma_i \in [\sigma-h,
\sigma+h]$ and again use
\indepnpmle{} for these observations. For these methods, the number of bins and bandwidth
$h$ are tuning parameters. While we anticipate ad hoc choices of tuning parameters to
perform well, a proper theoretical analysis likely needs to link tuning choices to
smoothness in the conditional distribution $\sigma \mapsto f_{Y\mid \sigma}(\cdot \mid
\sigma)$ with respect to certain distributional distances. The corresponding regularity
conditions thus seem more complex than smoothness conditions for conditional expectations
required by \closenpmle.

\subsubsection{Additional remarks}

\begin{rmksq}[Negative $\hats^2$ estimates]
\label{rmk:practical}
Analogue estimators for $s_0^2(\sigma_i) =
\var(Y_i \mid \sigma_i) - \sigma_i^2$ may take negative values.\footnote{The negative
estimated variance phenomenon is in part caused by estimation noise in $\var(Y_i \mid
\sigma_i)$. However, in our empirical application, there is some evidence that
observations with large estimated $\sigma_i$'s are underdispersed for the measures of
economic mobility in the Opportunity Atlas (see \cref{asub:variance_right_tail}).
\citet{armstrong2022robust} propose a Bayesian estimator for the conditional variance. }
In our experience, truncating $\hats$ at zero does not seem to cause bad performance when
computing posterior means. Nevertheless, in \cref{sec:nuisance_estimation} and the
software implementation, we propose a heuristic but data-driven truncation rule that
produces strictly positive $\hats$, borrowing from a statistics literature on estimating
non-centrality parameters for non-central $\chi^2$ distributions
\citep{kubokawa1993estimation}.
\end{rmksq}

\begin{rmksq}[Other transformations]
\label{rmk:alts_to_close} We summarize and compare \close{} to two methodological
 alternatives, deferring a detailed discussion on these and on several others to
 \cref{sub:alt_methods}. First, \citet{jiang2010empirical} propose applying \npmle{} on
 the $t$-ratio $Z_i = Y_i / \sigma_i \sim \Norm(\theta_i/\sigma_i, 1)$; similar approaches
 are used in \citet{efron2016empirical,kline2022systemic}. For estimating $\theta_i$, one
 then uses $\hat\theta_i =
 \sigma_i \cdot \PE_{\hat G_n}[\theta_i/\sigma_i \mid Z_i]$. Interpreting $\hat\theta_i$
  as an estimated posterior mean $\E_{P_0}[\theta_i \mid Y_i, \sigma_i]$ requires that
  $\theta_i / \sigma_i \indep \sigma_i$---meaning that \eqref{eq:location_scale} holds
  with $s_0(\sigma_i) = \sigma_i$ and constant $m_0(\cdot)$. Thus this $t$-ratio approach
  can be viewed as a particular instance of \close, if we wish to imbue it with an
  empirical Bayesian interpretation.

Second, when $Y_i$ and $\theta_i$ are sample and population means of binary outcomes, the
estimated variance of $Y_i$ is mechanically correlated with $\theta_i$: $\sigma_i^2 = 
\frac{Y_i (1-Y_i)}
{n_i}.$ A variance-stabilizing transform, e.g. $Z_i = \arcsin\sqrt{Y_i}$
\citep{10.1214/07-AOAS138}, results in \emph{approximately} Gaussian $Z_i \sim \Norm(\arcsin{
\sqrt{\theta_i}}, \frac{1}{4n_i})$ without the mechanical dependence. However, it is still
 possible that $n_i$ predicts $\theta_i$, and when that happens, proper modeling of
 $\theta_i \mid n_i$---e.g., via an analogue of \eqref{eq:location_scale}---can continue
 to improve performance.
\end{rmksq}

\section{Theoretical results}
\label{sec:regret}

As a review, we observe $(Y_i, \sigma_i)_{i=1}^n$, where $(\theta_i,\sigma_i)$ satisfies
\eqref{eq:location_scale} and $(Y_i, \theta_i, \sigma_i)$ obeys
\eqref{eq:gaussian_heteroskedastic_location}. The procedure
\closenpmle{} transforms the data $(Y_i, \sigma_i)$ into $(\hat Z_i, \hat \nu_i)$, with
 estimated conditional moments $\hateta = (\hatm, \hats)$ for $\eta_0 = (m_0, s_0)$ in
 \cref{item:close1}. It then estimates $G_0$ via \npmle{}
\eqref{eq:npmle} on $(\hat Z_i, \hat \nu_i)_{i=1}^n$. This section introduces a few
 statistical guarantees on the performance of \closenpmle{} in terms of \emph{regret}. To
 unify presentation, we first review decision theory primitives and introduce regret. 

Let $\bdelta(Y_{1:n}, \sigma_{1:n})$ be
a \emph{decision rule} mapping the data $(Y_{1:n},
\sigma_{1:n})$ to \emph{actions}. Recall that $L(\bdelta,
\theta_{1:n})$ denotes a \emph{loss function} mapping actions and parameters to a scalar.
Let $\RB (\bdelta; P_0) = \E_{P_0}[L(\bdelta,
\theta_{1:n}) \mid
\sigma_{1:n}]$ be the \emph{Bayes risk} of $\bdelta$ under $P_0$. The oracle Bayes
decision rule $\bdelta^\star$ \eqref{eq:oracle_bayes} is optimal in the sense that it
minimizes $\RB$. Thus, a natural performance measure for the empirical Bayesian
\eqref{eq:empirical_bayes_rule} is the gap between the Bayes risks of $\bdelta_ {\EB}$ and
$\bdelta^\star$. We refer to this quantity as \emph{Bayes regret}:
\begin{align*}
\regret(\bdelta_{\EB}) &= \E_{P_0}
[L
(\bdelta_{\EB}, \theta_{1:n}) - L(\bdelta^\star, \theta_{1:n}) \mid \sigma_{1:n}],
\numberthis \label{eq:regret_def}
\end{align*}
where the right-hand side integrates over the randomness in $\theta_{1:n}, Y_{1:n}$, and,
 by extension, $\hat P$. %
If an empirical Bayes method achieves low Bayes regret, then it successfully imitates the
decisions of the oracle Bayesian, and its decisions are thus approximately
optimal. Our results show that Bayes regret for \closenpmle{} vanishes quickly as  a
function of $n$.

\begin{rmksq}[Fixed vs. random $\theta$]
\label{rmk:james-stein}
Our results consider asymptotic optimality, in terms of \eqref{eq:regret_def}, of the empirical Bayes
decision rule when $\theta_i \mid \sigma_i$ is randomly sampled from $P_0$, following a
recent literature on nonparametric empirical Bayes
\citep{jiang2020general,soloff2021multivariate}. A separate literature considers instead
the frequentist risk $\RF(\theta_{1:n}; \sigma_{1:n}) \equiv
\E\bk{
 L(\bdelta, \theta_{1:n}) \mid \theta_{1:n}, \sigma_{1:n} }$ under fixed $(\theta_{1:n},
 \sigma_{1:n})$ \citep{robbins1956}. For instance,
 \citet
  {james1992estimation,bock1975minimax,10.1214/07-AOAS138,weinstein2018group} consider
  shrinkage estimators that dominate $\delta_i = Y_i$ uniformly for all configurations of
  $\theta_{1:n}$. \citet{xie2012sure,kwon2021optimal} consider choosing decision rules
  within a restricted class that minimize an unbiased estimate of $\RF$. In
  particular, \citet{xie2012sure} can be thought of as implementing \indepgauss{} with
  different ways of estimating the hyperparameters in
  $\theta_i \mid\sigma_i \iid \Norm\pr{m_0, s_0^2}$, and \citet{weinstein2018group} can be
  thought of as implementing \closegauss.

 While these guarantees for $\RF$ are preserved even if we further average the frequentist
 risk over $\theta_{1:n} \mid
 \sigma_{1:n} \sim P_0$, they are distinct from upper bounding
 \eqref{eq:regret_def}.\footnote{For instance, the oracle Bayes rule for mean-squared error may not
  dominate $\delta_i = Y_i$ in $\RF$ uniformly for all $\theta_{1:n}$. Conversely,
  decisions that merely dominate $\delta_i = Y_i$ may still be quite far from the oracle
  Bayes rule.}  In particular, they may leave much on the table if $\RB$ is
  targeted. Moreover, these guarantees in $\RF$ are typically restricted to MSE. Our
  example in
 \cref{fig:example_shrink_ranking} shows that reasonable decisions for MSE may not be
  reasonable for subsequent selection decisions. As a simple example, \citet
  {bock1975minimax} considers spherical shrinkage rules of the form $\delta_{i} = c\pr
  {\sum_j Y_j^2} Y_i$ for some function $c(\cdot)$. However, despite dominating
  no-shrinkage in MSE, $\delta_i$ does not change the ranking of different units, and
  hence does not improve on ranks over $Y_i$.
\end{rmksq}

In what follows, we use the symbol $C$ to denote a generic positive and finite constant
which does not depend on $n$. We use the symbol $C_{x}$ to denote a generic positive and
finite constant that depends only on $x$, some parameter(s) that describe the problem.
Occurrences of the same symbol $C, C_x$ may not refer to the same constants. Since all
expectation or probability statements are with respect to the conditional distribution
$P_0$ of $\theta_{1:n} \mid \sigma_{1:n}$, going forward, we treat $\sigma_ {1:n}$ as
fixed and simply write $\E[\cdot], \P(\cdot)$ to denote the expectation and probability
over $\theta_{1:n} \mid \sigma_{1:n} \sim P_0$; we may omit the subscript $P_0$ or the
conditioning on $\sigma_{1:n}$.

\subsection{Regret rate in squared error}

Our main result concerns the canonical statistical problem of
estimating the parameters $\theta_{1:n}$ under MSE.

\begin{probsq}[Squared-error estimation of $\theta_{1:n}$]
\label{ex:mse}

The action $\bdelta =(\delta_1,\ldots,
  \delta_n)$ collects estimates $\delta_i$ for $\theta_i$, evaluated with MSE:
  $
  L(\bdelta, \theta_{1:n}) =
  \frac{1}{n}
  \sum_{i=1}^n (\delta_i - \theta_i)^2. $ The oracle Bayes decision rule  $\bdelta^\star
   = (\theta_1^*,\ldots, \theta_n^*)$ here is the posterior mean under $P_0$, where
   $\theta_i^*  \equiv
   \E_{P_0}\bk{\theta_i \mid Y_i, \sigma_i} $. The empirical Bayesian counterpart is $
   \hat\theta_{i, \hat P} = \PE_{\hat P}[\theta_i \mid Y_i, \sigma_i].$  %
\end{probsq}

\Copy{mseregretdef}{
For \cref{ex:mse}, define $\reg$ as the excess loss of the empirical Bayes posterior means
relative to that of the oracle Bayes posterior means: 
\begin{align*}
\reg(G, \eta) &\equiv \frac{1}{n} \sum_{i=1}^n (\hat\theta_{i, G, \eta} - \theta_i)^2 -
\frac{1}{n} \sum_{i=1}^n \pr{\theta_i^* - \theta_i}^2,
\end{align*}
where $\theta_i^*$ are the oracle posterior means and $\hat\theta_ {i,G,\eta}$ are the
posterior means under a prior parametrized by $(G, \eta)$. %
}
The corresponding Bayes regret \eqref{eq:regret_def} for \closenpmle{} in this decision
problem is then the $P_0$-expectation of $\reg$:
\[
\regret = \E\bk{ \reg(\hat G_n, \hateta)} = \E_{P_0}\bk{ \frac{1}{n} \sum_{i=1}^n
  (\theta_i^* - \hat\theta_{i, \hat G_n, \hateta})^2 \numberthis
  \label{eq:mse_regret_and_mse}
}.
\] Equation \eqref{eq:mse_regret_and_mse} additionally notes that expected $\reg$ is equal
 to the expected mean-squared difference between the empirical Bayesian posterior means
 $\hat\theta_{i, \hat G_n, \hateta}$ and their oracle counterparts $\theta_i^*$.
 Our subsequent results (\cref{cor:maintext,thm:minimaxlower}) state upper and lower
 bounds for $\regret$, over a class of data generating processes $\mathcal P_0\ni P_0$. We
 now introduce and discuss the assumptions on $\mathcal P_0$.

\subsubsection{Assumptions for regret upper bound}

We first assume that $\hat G_n$ is an approximate maximizer of the log-likelihood on the
transformed data $(\hat Z_i, \hat\nu_i)$ satisfying some support restrictions. This is
not  restrictive, as the actual maximizers of the log-likelihood function
satisfy it (Proposition 4, \citet{soloff2021multivariate}). This assumption also
accommodates for the fact that the \npmle{} is approximated by a discrete distribution on
a grid.
\begin{restatable}{as}{asnpmle}
\label{as:npmle}
Let $\psi_i(Z_i,
\hateta, G) \equiv \log\pr{\int_{-\infty}^\infty
\varphi
\pr{\frac{\hat Z_i -\tau}{\hat \nu_i}} G(d\tau)}$ be the objective function in
\eqref{eq:npmle}, ignoring the factor $1/\hat\nu_i$ that does not involve $G$. We
assume that $\hat G_n$ satisfies \[
\frac{1}{n} \sum_{i=1}^n \psi_i(Z_i, \hateta, \hat G_n) \ge \sup_{H \in \mathcal P(\R)}
\frac{1}{n} \sum_{i=1}^n \psi_i(Z_i, \hateta, H) - \kappa_n
\numberthis \label{eq:approx_mle}
\]
for tolerance $\kappa_n = \frac{2}{n} \log ({\frac{n}{
\sqrt{2\pi} e}})$. Moreover, we
require that $\hat G_n$ has support points within $[\min_i\hat Z_i,
\max_i \hat
Z_i]$. To ensure that $\kappa_n$ is positive, we assume that $n \ge 7 = \lceil \sqrt{2\pi}
e \rceil$.\footnote{The constants $\kappa_n \rateeq \frac{1}{n}\log(n)$ also feature in
\citet{jiang2020general} to ensure that the fitted likelihood is bounded away from zero.
The particular constants in $\kappa_n$ simplify expressions and are not
material to the result.}
\end{restatable}

We now state further assumptions on $\mathcal P_0$ beyond
\eqref{eq:location_scale}. First, we assume that $G_0$ is sufficiently thin-tailed such
 that its moments grow slowly.\footnote{An equivalent statement to \cref
 {as:moments} is that there exists $a_1, a_2 > 0$ such that $\P_{G_0}(|\tau| > t) \le
 a_1\exp\pr{-a_2 t^\alpha}$ for all $t > 0$. Note that when $\alpha = 2$, $G_0$ is
 subgaussian, and when $\alpha = 1$, $G_0$ is subexponential
\citep[see the definitions in][]{vershynin2018high}. \Cref{as:moments} is slightly stronger than requiring that
 all moments exist for $G_0$, and weaker than requiring $G_0$ to have a moment-generating
 function. Similar tail assumptions feature in the theoretical literature on empirical
 Bayes \citep{soloff2021multivariate,jiang2009general,jiang2020general}. } \Copy{alpha}
 {The thickness
  of its tail is parametrized by $\alpha \in (0,2]$, which subsequently affects the
  log factors in \cref{cor:maintext}.}

\begin{restatable}{as}{moments}
\label{as:moments}

The distribution $G_0$ has zero mean, unit variance, and admits simultaneous moment
control: For some $\alpha \in (0,2]$ and $A_0 > 0$ such
that for all $p > 0$, $\pr{\E_{\tau \sim G_0}[|\tau|^p]}^{1/p} \le A_0 p^
{1/\alpha}.
$
\end{restatable}

Next, \cref{as:variance_bounds} imposes that members of $\mathcal P_0$ have various
variance parameters uniformly bounded away from zero and $\infty$. This is a standard
assumption in the literature, maintained likewise by \citet{jiang2020general} and
\citet{soloff2021multivariate}.

\begin{restatable}{as}{variancebounds}
\label{as:variance_bounds} 
The variances $(\sigma_{1:n}, s_0)$ admit lower and upper
 bounds: There are positive reals $\hyperparams >0$ such that, for all $i$ and all
 $\sigma \in (\sigl, \sigu)$, $\sigl < \sigma_i < \sigu$ and $s_{0\ell} < s_0(\sigma) < s_
 {0u}$.
\end{restatable}

Lastly, we require that $m_0(\cdot)$ and $s_0(\cdot)$ satisfy some smoothness
restrictions. We also require that $\hatm(\cdot)$ and $\hats(\cdot)$ satisfy some
corresponding regularity conditions. Let $C_{A_1}^p([\sigl,\sigu])$ denote the H\"older
class of order $p \ge 1$ with maximal H\"older norm $A_1 > 0$ supported on $
[\sigl,\sigu]$ \citep[Section 2.7.1,][]{vaart1996weak}.

\begin{restatable}{as}{holder}
\label{as:holder}  Assume that
\begin{enumerate}
  \item  The true conditional moments are H\"older-smooth: $m_0, s_0 \in C_{A_1}^p([\sigl,\sigu])$.
\end{enumerate}

Additionally, let $\beta_0 > 0$ be a constant.   Assume
that the  estimators for $m_{0}$ and $s_0$, $\hateta = (\hatm, \hats)$, satisfy:
\begin{enumerate}[resume]
    \item For all
    sufficiently large $C_{1,\H} > 0$ and all $n$, \[
\P\pr{\norm{
  \hateta - \eta_0
}_\infty > C_{1,\H} n^{-
\frac{p}
    {2p+1}} (\log n)^{\beta_0}} < \frac{1}{n^2}
    \] where $\norm{\eta}_\infty \equiv \max(\norm{m}_\infty, \norm{s}_\infty)$ for $\eta
     =(m,s)$.
    \item  $\hateta$ takes values in $\mathcal V$ almost surely: $\P
    \pr{\hatm\in
    \mathcal V, \hats \in \mathcal V} = 1$, where $\mathcal{V}$ is a set of
functions supported on $[\sigl, \sigu]$ that (i) is uniformly bounded $\sup_{f \in
\mathcal V}
\norm{f}_\infty \le C_{A_1}$ and (ii) admits the metric entropy bound $\log N (\epsilon,
\mathcal V, \norm{\cdot}_\infty) \le C_{A_1,p, \sigl,\sigu} (1/\epsilon)^ {1/p}$.
    \item The conditional variance estimator respects the conditional variance bounds in
    \cref{as:variance_bounds}: $\P\pr{\frac{s_{0\ell}}{2} < \hat s < 2s_{0u}} = 1$.
\end{enumerate}
\end{restatable}

\Cref{as:holder} is a H\"older smoothness assumption on the conditional moments $m_0$ and
$s_0$, which is a standard regularity condition for nonparametric regression. Moreover, it
is also a high-level assumption on the quality of the estimation procedure for $(\hatm,
\hats)$. It expects that $\hatm$ and $\hats$ are accurate in $\norm{\cdot}_\infty$, belong
to a class with manageable metric entropy, and obey the bounds for $s_{0}$.\footnote{
  \label{fn:a4}\Cref{as:holder}(2) is slightly stronger than an
 estimation rate
 requirement $\norm{\hat\eta - \eta_0}_\infty = O_P\pr {n^{-p/(2p+1)}(\log n)^
 {\beta_0}}$, in the sense that the probability of large deviations are additionally
 controlled. Local polynomial smoothing estimators can attain the desired estimation rate
 of $n^{-p/(2p+1)}(\log n)^{\beta_0}$ in $\norm{\cdot}_\infty$
 \citep{tsybakov2008introduction,stone1980optimal}. Since the data is assumed to be thin-tailed in
 \cref{as:moments}, such estimators also attain the stronger requirement in
 \cref{as:holder}(2).

  For \cref{as:holder}(3), if the estimators $\hatm$ and $\hats$ are $p$-H\"older smooth
  almost surely,
  we can simply take $\mathcal V = C_{A_1'}^p([\sigl,\sigu])$ for some potentially
  different $A_1'$. This can be achieved in practice by, say, projecting estimated
  parameters $\tilde \eta$ to $C_{A_1}( [\sigl, \sigu])$ in $\norm
  {\cdot}_\infty$. 

  Finally,
 \cref{as:holder}(4) also expects the conditional moment estimates $\hateta$ to respect
 the boundedness constraints for $s_0$. This is mainly so that our results are easier to
 state.
 
 We show in \cref{sec:nuisance_estimation} that a local linear regression estimator (with
 $\hats$ suitably truncated) satisfies weaker conditions than \cref{as:holder}(2)--(4)
 that are nonetheless sufficient for the conclusion of \cref{cor:maintext}. 
}

\Cref{as:moments,as:holder,as:variance_bounds} specify a class of distributions $\mathcal
P_0$ and estimators $\hateta = (\hatm(\cdot), \hats(\cdot))$ regulated by a set of
hyperparameters $\Hyperparams = (\sigl,
\sigu, s_\ell, s_u, A_0, A_1, \alpha, \beta_0, p).$ Our subsequent theoretical results are
 uniform over $\mathcal P_0$ for a fixed $\H$. %

\subsubsection{MSE regret results}

Our main result is a non-asymptotic upper bound for \eqref{eq:mse_regret_and_mse}: The MSE
regret of 
\closenpmle{}
converges to zero no slower than 
$n^{-\frac{2p}{2p+1}}(\log n)^{C}$.

\begin{restatable}{theorem}{cormaintext}
\label{cor:maintext}
Under \cref{as:holder,as:moments,as:variance_bounds,as:npmle}, there exists a
constant $C_{0, \Hyperparams} > 0$ such that the following upper bound holds:\[
\regret = \E\bk{
    \reg(\hat G_n, \hateta) } \le C_{0, \Hyperparams} n^{-\frac{2p}{2p+1}} (\log n)^{\frac{2+\alpha}{\alpha} + 3 + 2\beta_0}.
    \numberthis \label{eq:regret_rate_final}
\]
\end{restatable}

Second, we give a corresponding minimax lower bound on the regret, which shows that
\cref{cor:maintext} cannot be improved by more than logarithmic factors.

\begin{restatable}{theorem}{thmminimaxlower}
\label{thm:minimaxlower}

Fix a set of valid hyperparameters $\Hyperparams$. Let $\mathcal P (\Hyperparams, \sigma_{1:n})$ be the set of
distributions $P_0$ on support points $\sigma_{1:n}$ which satisfy
\eqref{eq:location_scale}
and  \cref{as:moments,as:holder,as:variance_bounds} corresponding to
$\Hyperparams$.\footnote{This result additionally takes the supremum over the
support points $\sigma_{1:n}$. This is because the nonparametric regression problem would
be ``too easy'' for certain configurations of $\sigma_{1:n}$. For instance, when
$\sigma_{1:n}$ only takes $m \ll n$ unique values, nonparametric regression is possible at
rate $\sqrt{m/n}$. For the proof, it suffices to consider $\sigma_{1:n}$ being equally
spaced in $
[\sigl, \sigu]$. } For a given $P_0$, let
$\theta_i^* = \E_{P_0}[\theta_i \mid Y_i,
\sigma_i]$ denote the oracle posterior means. Then there exists a constant $c_
 {\Hyperparams} > 0$ such that
\[
\inf_{\hat\theta_{1:n}} \sup_{\substack{\sigma_{1:n} \in (\sigl, \sigu)\\ P_0 \in \mathcal
P(\mathcal
H, \sigma_{1:n})}} \E_{P_0} \bk{
    \frac{1}{n} \sum_{i=1}^n (\hat\theta_{i} - \theta_i)^2 - \frac{1}{n} \sum_{i=1}^n 
    (\theta_{i}^* - \theta_i)^2
} \ge c_\Hyperparams n^{-\frac{2p}{2p+1}},
\]
where the infimum is taken over all (possibly randomized) estimators of $\theta_{1:n}$.
\end{restatable}

\Cref{cor:maintext} continues a recent statistics literature on empirical Bayes methods
via \npmle, by characterizing the effect of an estimated first-step parameter $\hat\eta$.
Our theory hews closely to---and extends---the results in \citet{jiang2020general} and
\citet{soloff2021multivariate}, which themselves extend earlier results in the
homoskedastic setting \citep{jiang2009general,saha2020nonparametric}. In particular,
\citet{soloff2021multivariate} show that the MSE regret rate is of the form $C (\log n)^
{\beta} \frac{1}{n}$ under precision independence and assumptions similar to ours. In this
context, we show that first-step estimation error degrades this regret rate gracefully,
and we link the corresponding regret rate to the smoothness of $\eta_0$. The proof of
\cref{cor:maintext} is deferred to the Online Appendix, but its main ideas are outlined in
\cref{asec:proof_main}.

\Cref{thm:minimaxlower} shows that the rate \eqref{eq:regret_rate_final} is optimal up to
 logarithmic factors. These logarithmic factors partly reflect inefficiencies in the proof
 of \cref{cor:maintext}, but in any case the gap is not large. We prove
 \cref{thm:minimaxlower} by showing that any good posterior mean estimate $\hat\theta_i$
 implies a good estimate $\hatm(\sigma_i)$ for $m_0$ for some particular choice of $G_0,
 \sigma_
 {1:n}, s_0^2(\cdot)$. Minimax lower bounds for estimation of $m_0$ \citep
 {tsybakov2008introduction} then imply lower bounds for estimation of the oracle
 posterior means $\theta_i^*$
\citep[see][ for a similar argument in a related setting] {ignatiadis2019covariate}.

We additionally note that these regret upper bounds readily extend to the case where
covariates are present and the location-scale assumption \eqref{eq:location_scale} is
specified with respect to the additional covariates $X_i$:
\[\theta_i \mid \sigma_i, X_i \sim G_0\pr{\frac{\cdot - m_0(X_i, \sigma_i)}{s_0(X_i,
\sigma_i)}},
\numberthis \label{eq:location-scale-covariates}
\] under smoothness assumptions on $(m_0, s_0, \hatm, \hats)$ analogous to
\cref{as:holder}. The resulting convergence rate would reflect the
 dimensionality of the covariates, and the term $n^ {- \frac{2p}{2p+1}}$ would be
 replaced with $n^ {-
\frac{2p}{2p+1+d}}$, where $d$ is the dimension of $X$.

Taken together, \cref{cor:maintext,thm:minimaxlower} are statistical optimality guarantees
for \closenpmle{} in terms of \cref{ex:mse}. That is, the worst-case MSE performance gap
of
\closenpmle{} relative to the oracle contracts at the best possible rate, meaning that 
\closenpmle{} mimics the oracle as well
as possible.

\subsection{Robustness to the location-scale assumption \eqref{eq:location_scale}}

We prove \cref{cor:maintext,thm:minimaxlower} imposing the location-scale model
\eqref{eq:location_scale}.  This is an optimistic assessment of the performance of
\closenpmle. While \eqref{eq:location_scale} nests precision independence, it may still be
misspecified. This subsection explores the worst-case behavior of \closenpmle{} without
\eqref{eq:location_scale}. 

We  do so by considering an idealized version of \closenpmle. So long as $\theta_i
\mid \sigma_i$ has two moments, $\eta_0(\cdot) = (m_0(\cdot), s_0(\cdot))$ are well-defined
as conditional moments. We will assume that $m_0, s_0$ are known. Without
\eqref{eq:location_scale}, $G_0$ is ill-defined, but we assume that we obtain some
pseudo-true value $G_0^*$ that has zero mean and unit variance.  Thus, for estimating $\tau_i
= \frac{\theta_i - m_0(\sigma_i)}{s_0(\sigma_i)}$, whose distribution is $\tau_i \mid
\sigma_i \sim G_i$, this idealized procedure uses some misspecified prior $G_0^* \neq
G_i$, where $G_0^*$ agrees with $G_i$ in the first two moments. The worst-case performance
of the procedure that uses $G_0^*$ depends on how far posterior means under $G_0^*$
differs from posterior means under $G_i$. 

\Copy{closegaussconstant}{We show in \cref{asec:max_gauss} that this difference is bounded
 uniformly for all $G_0^*$ satisfying an additional tail assumption. This result implies
 that the maximum risk of this procedure is at most a constant multiple of the minimax
 risk; here, the minimaxity is defined with respect to a game between an analyst and an
 adversary, where the analyst knows $m_0, s_0$ and hopes to estimate $\theta_{1:n}$, and
 the adversary chooses the shape of the distribution $\tau_i
\mid \sigma_i$. In this game, the oracle version of \closegauss{} \eqref
 {eq:gaussian_cond_b} is a minimax procedure (\cref{lemma:minimax_close_gauss}). }

Specifically, let $\mathcal P(m_0, s_0)$ denote
the set of distributions of $\theta_{1:n} \mid \sigma_{1:n}$ where $\E[\theta_i \mid
\sigma_i] = m_0(\sigma_i)$ and $\var(\theta_i \mid \sigma_i) = s_0^2(\sigma_i)$. Let
 \[\mathcal G_0(\lambda, \epsilon) \equiv \br{G_0^*: \E_{G_0^*}[\tau]=0, \var_{G_0^*}
 (\tau) = 1, G_0^* (-z) \vee (1-G_0^*(z))
\le \lambda z^{-2-\epsilon} \text{ for all $z>0$}} \]be the set of mean-zero, variance-one
 distributions satisfying an additional tail condition indexed by $\lambda > 0, \epsilon >
 0$.\footnote{By Markov's inequality, this condition is satisfied if $G_0^*$ has its $
 (2+\epsilon)$\th {} moment bounded by $\lambda$. A previous version of this paper stated
 \cref{thm:worstcaserisk} without this additional tail condition, regrettably due to a
 technical error that is corrected in this version. See \cref{asec:max_gauss}.}
\begin{restatable}{theorem}{worstcaserisk}
\label{thm:worstcaserisk}
  Under the preceding setup and \eqref{eq:eb_sampling}, but not \eqref{eq:location_scale},
  let $\hat \theta_ {i, G_0^*, \eta_0}$ denote the posterior mean for $\theta_i$ under a
  prior $G_0^*$
  for $\tau$. Let $\bar\rho =
  \max_i s_0^2 (\sigma_i) / \sigma_i^2 < \infty$ be the maximal conditional
   signal-to-noise ratio. Then, for some $0 < C_{\bar\rho, \lambda,
  \epsilon}
  <
  \infty$ that solely depends on $\bar\rho, \lambda, \epsilon$, \[
  \frac{\sup_{G_0^* \in \mathcal G_0(\lambda, \epsilon)}
\sup_{P_0 \in \mathcal P(m_0, s_0)} \E_{P_0}\bk{\frac{1}{n} \sum_{i=1}^n (\hat\theta_{i,
G_0^*, \eta_0} - \theta_i)^2}}{\inf_{\hat\theta_{1:n}} \sup_{P_0 \in \mathcal P(m_0, s_0)} \E_{P_0} \bk{\frac{1}{n}
\sum_{i=1}^n (\hat\theta_i - \theta_i)^2}} \le C_
 {\bar\rho, \lambda, \epsilon}, \numberthis 
\label{eq:multiple_of}
  \]
  where the infimum in the denominator is over all (possibly randomized)
  estimators of $\theta_i$ given $(Y_i, \sigma_i)_{i=1}^n$ and $\eta_0(\cdot)$.
\end{restatable}

\Cref{thm:worstcaserisk} shows that the worst-case behavior of an idealized version of
\closenpmle{} comes within a factor of the minimax risk. Thus, \closenpmle{} is not
arbitrarily unreasonable, even under misspecification. We caution that
\eqref{eq:multiple_of} is a fairly weak guarantee, in that the decision rule that simply
outputs the prior conditional mean ($\delta_i = m_0 (\sigma_i)$) also satisfies it.
Nevertheless, even so,
\eqref{eq:multiple_of}
{does not} hold for an idealized version of \indepgauss.\footnote{\Copy
{indepgaussconstant} {That is, it does not hold for the implementation of \indepgauss{}
that plugs in known unconditional moments $m_0 =
\frac{1}{n} \sum_{i=1}^n m_0 (\sigma_i)$ and $s_0^2 = \frac{1}{n} \sum_{i=1}^n (m_0
(\sigma_i) - m_0)^2 + s_0^2 (\sigma_i)$. 
To wit, take $s_0
(\sigma_i) \approx 0$.
Then, the minimax risk as a function of $(s_0 (\cdot), m_0(\cdot))$ is approximately zero,
but $m_0(\cdot)$ can be chosen such that the risk of \indepgauss{} is bounded away from
zero.
See \cref{lemma:indepgauss_not_bounded} for a formal statement.}
}

\subsection{Other decision objectives and relation to squared-error loss}
\label{sub:other_decisions}

So far, our regret guarantees are only about estimation in MSE (\cref{ex:mse}). We now
turn to two decision problems that involve ranking or selection and show similar
guarantees for \closenpmle{} in terms of regret for these decision problems. These
decision problems are likely more economically relevant for, e.g., replacing low
value-added teachers, recommending high-mobility tracts, or treatment choice
\citep{gilraine2020new,bergman2019creating,manski2004statistical,stoye2009minimax,kitagawa2018should,athey2021policy}.

\begin{probsq}[\utilmax]
\label{ex:utilmax}
  Suppose $\bdelta = (\delta_1,\ldots,\delta_n)$ consists of binary selection decisions
  $\delta_i \in \br{0,1}$. For each population, selecting that population has net benefit
  $\theta_i$. The decision maker wishes to maximize utility (i.e., negative loss): $
  -L(\bdelta, \theta_{1:n}) =
  \frac{1}{n} \sum_{i=1}^n \delta_i \theta_i. $ The oracle Bayes rule selects all whose
  posterior mean net benefit $\theta_i$ is nonnegative: $
\delta_i^\star = \one\pr{\theta_{i, P_0}^* \ge 0}.
  $
  One natural empirical Bayes decision rule replaces $\theta_{i, P_0}^*$ with $\theta_{i,
  \hat P}^*$, following \eqref{eq:empirical_bayes_rule}.
\end{probsq}

\begin{probsq}[\topm]
\label{ex:topm}

 Similar to
\utilmax, suppose $\bdelta$ consists of binary selection decisions, with the additional
constraint that exactly $m$ populations are chosen: $\sum_i \delta_i = m$. The decision
maker's utility is the average $\theta_i$ of the
selected set: \[ -L(\bdelta, \theta_{1:n}) = \frac{1}{m} \sum_{i=1}^n \delta_i \theta_i.
\numberthis
\label{eq:topm}
\]
The oracle Bayesian selects the populations corresponding to the $m$ largest posterior
means
$\theta_{i, P_0}^*$: $
\delta_i^\star = \one\pr{
  \theta_{i, P_0}^* \text{ is among the top-$m$ of $\theta_{1:n, P_0}^*$}
}.
$
Again, the empirical Bayes recipe \eqref{eq:empirical_bayes_rule} replaces $P_0$ with the
estimate $\hat P$.
\end{probsq}

\begin{rmksq}
\label{rmk:mover}
The utility function \eqref{eq:topm} rationalizes the widespread practice of screening
based on empirical Bayes posterior means
\citep{gilraine2020new,chetty2014measuring,kane2008estimating,hanushek2011economic,bergman2019creating}.
In \citet{bergman2019creating}, for instance,  where housing voucher holders are
incentivized to move to Census tracts selected according to economic mobility,
\eqref{eq:topm} represents the expected economic mobility of a mover were they to move
randomly to one of the selected tracts. Our theoretical results can accommodate slightly
less restrictive mover behavior (\cref{rmk:nonuniform}).%
\end{rmksq}

The oracle Bayes decision rules $\bdelta^\star$ in \cref {ex:utilmax,ex:topm}
depend solely on the vector of oracle Bayes posterior means $\theta_ {1:n}^*$. %
Therefore, for these problems, the natural empirical Bayes decision rules simply replace
oracle Bayes posterior means ($\theta_i^*$) with empirical Bayes ones ($\hat\theta_i$). It
stands to reason that as $\hat\theta_i$ is close to $\theta_i^*$ in squared error, even
when $\hat\theta_i$ implies the wrong selection decision, this decision is not too costly
for the empirical Bayesian. We formalize this intuition in the following theorem, showing that if $\hat\theta_i$ are close to $\theta_i^*$ in MSE,
then decisions plugging in $\hat\theta_i$ are also close to their oracle counterparts in
terms of Bayes risk.

To specialize, let $\utilmaxreg$ denote $\regret$ for the loss function in 
\cref{ex:utilmax} and let $\topmreg$
denote $\regret$ for \cref{ex:topm}. 

\begin{restatable}{theorem}{mserelevance}
\label{thm:mserelevance} Suppose  \eqref{eq:eb_sampling} holds but \eqref
 {eq:location_scale} does not necessarily hold. %
Let $\hat\delta_i$ be the plug-in decisions with any vector of estimates $\hat \theta_i$.
Then,
\begin{enumerate}
  \item For \utilmax,
  \[
\E[\utilmaxreg(\hat\bdelta)] \le\pr { \E\bk{\frac{1}{n} \sum_{i=1}^n (\hat\theta_i - \theta_i^*)^2}}^
{1/2} . \numberthis \label{eq:utilmax_bound}
  \]
  \item For \topm, \[
\E[\topmreg(\hat\bdelta)] \le 2\sqrt{\frac{n}{m}} \pr{\E\bk{\frac{1}{n} \sum_{i=1}^n (\hat\theta_i -
\theta_i^*)^2}}^{1/2}. \numberthis \label{eq:topm_regret_bound}
  \]
\end{enumerate}
\end{restatable}

\Cref{thm:mserelevance} shows a sense in which \cref{ex:utilmax,ex:topm} are easier than
\cref{ex:mse}: The regret of the latter dominates those of the former. As a result, if we
 use \closenpmle{} under \eqref{eq:location_scale}, our convergence rates from 
\cref{cor:maintext} also
upper bound regret rates for these two decision problems. In particular, for $m/n \to c
\in (0,1)$, both regret rates
 \eqref{eq:utilmax_bound} and \eqref{eq:topm_regret_bound} are of the form $n^{-p/(2p+1)}
 (\log n)^{C} = o(1)$ under \cref{cor:maintext}. Thus, the performance  of the empirical
 Bayes decision rule approximates that of the oracle at least as fast as
 $O(n^{-p/(2p+1)})$, up to log factors.

\begin{rmksq}[Tightness of \cref{thm:mserelevance}]

 We suspect that the actual performance of
 \closenpmle{} for \cref{ex:utilmax,ex:topm} may be better than predicted by
 \cref{thm:mserelevance}. The proof of \cref{thm:mserelevance} exploits the fact that when
 the empirical Bayesian makes a selection mistake, the size of the mistake is not large if
 the square-error regret is low. It does not exploit the fact that if squared error regret
 is low, then the empirical Bayesian may be unlikely to make mistakes in the first place.%
 \footnote{Upper and lower bounds are derived in related but distinct settings by
 \citet{audibert2007fast,bonvini2023minimax}; some upper bounds, under possibly stronger
 assumptions, appear better than implied by \cref{thm:mserelevance}. 
We speculate that the bound for $\utilmax$ can be tightened by verifying a margin
 condition, using Proposition 2 in \citet{bonvini2023minimax}.
 Relatedly,
 \citet{liang2000empirical} shows upper and lower bounds for \cref{ex:utilmax} of the form
 $O ((\log n)^{1.5} /n)$ in a homoskedastic setting, assuming the oracle posterior means
 fall on both sides of zero. 
} Nevertheless, \cref{thm:mserelevance} is
 competitive with recent results. For instance, in nonparametric settings, the rate in
 \cref{thm:mserelevance} is more favorable than the upper bound derived in
 \citet{coeyhung}, who also study \cref{ex:topm}.%
\end{rmksq}

\subsection{Validating performance by coupled bootstrap}
\label{sub:coupled_bootstrap}

We close this section with a procedure that provides unbiased estimates of the loss of
\emph{arbitrary} decision rules for \cref{ex:mse,ex:topm,ex:utilmax}. Practitioners can
use this procedure to evaluate the gain of \closenpmle{} relative to other
alternatives---we do so extensively in \cref{sec:empirical}. The validity of this
validation depends only on the Gaussianity
\eqref{eq:gaussian_heteroskedastic_location}---without assuming $(\theta_i,
\sigma_i)$ are random nor assuming the location-scale model \eqref{eq:location_scale}.

For some $\omega > 0$ and an independent Gaussian noise $W_i \sim
\Norm(0,1)$, consider adding to $Y_i$ and subtracting from $Y_i$ some scaled version
of $W_i$: \[
Y_{i}^{(1)} = Y_i + \sqrt{\omega} \sigma_i W_i \quad Y_{i}^{(2)} = Y_i - \frac{1}{
\sqrt{\omega}}
\sigma_i W_i.
\]
\citet{oliveira2021unbiased} call $(Y_{i}^{(1)}, Y_{i}^{(2)})$ the \emph{coupled
bootstrap} draws. Observe that the two draws are conditionally independent under
\eqref{eq:gaussian_heteroskedastic_location}: \[
\colvecb{2}{Y_{i}^{(1)}}{Y_{i}^{(2)}} \mid \theta_i, \sigma_i^2 \sim \Norm\pr{
  \colvecb{2}{\theta_{i}}{\theta_i}, \begin{bmatrix}
   (1 + \omega) \sigma_i^2 &  0\\ 0 & (1+\omega^{-1}) \sigma_i^2
  \end{bmatrix}
}. \numberthis \label{eq:coupled_bootstrap}
\]
The conditional independence allows us to use $Y_{i}^{(2)}$ as an out-of-sample validation
for decision rules computed based on $Y_i^{(1)}$. We denote their variances by $\sigma_{i,
(1)}^2$ and $\sigma_{i, (2)}^2$.

The coupled bootstrap can be thought of as approximating sample-splitting the micro-data
without needing access. We could imagine splitting the micro-data into training and
testing sets, and think of $Y_i^{(1)}$ as training-set estimates  and $Y_{i}^{ (2)}$ as
testing-set estimates. We might compute decisions based on $Y_i^ {(1)}$ and evaluate them
honestly with fresh data $Y_i^{(2)}$. The coupled bootstrap precisely emulates this
sample-splitting procedure.\footnote{To see this, suppose $Y_{i} =
\frac{1}{n_i}
\sum_{j=1}^ {n_i} Y_ {ij}$ is a sample mean of i.i.d. micro-data $Y_{ij}: j = 1,\ldots,
 n_i$. Suppose we split $Y_{ij}$ into two sets, with proportions $\frac{1} {\omega + 1}$
 and $\frac{\omega}{\omega + 1}$, respectively. Let $Y_i^{(1)}$ and $Y_i^{(2)}$ be the
 sample means on each respective set. Then the central limit theorem motivates that,
 approximately, \eqref{eq:coupled_bootstrap} holds for $Y_i^{ (1)}$ and $Y_i^{(2)}$. For
 instance, coupled bootstrap with a value of $\omega = 1/9$ is statistically equivalent
 to splitting the micro-data with a 90-10 train-test split.}

The following proposition formalizes how to use coupled bootstrap to provide unbiased
estimators for the loss of a generic decision rule.\footnote{\citet
{oliveira2021unbiased} state the unbiased estimation result for the mean-squared error
estimation problem. They connect the coupled bootstrap estimator to Stein's unbiased risk
estimate. Our calculation for other loss functions extends their unbiased estimation
result. \Cref{prop:unbiased} can also be easily generalized to other loss functions that
admit unbiased estimators \citep[Effectively, the loss is a function of a Gaussian
location $\theta_i$. For unbiased estimation of functions of Gaussian parameters, see
Table A1 in][]{voinov2012unbiased}.}

\begin{table}[htb]
  \caption{Unbiased estimators for loss of decision rules and associated conditional
  variance expressions (\cref{prop:unbiased})}
  \label{tab:unbiased_estimation}
  \centering

\scriptsize
\begin{tabularx}{\textwidth}{m{0.2\textwidth}YY}
\toprule
Problem & Unbiased estimator of loss, $T\pr{Y_{1:n}^{(2)}, \bdelta}$ & $\var\pr{T\pr{Y_
{1:n}^{(2)}, \bdelta} \mid \mathcal F}$
\\ \midrule
\Cref{ex:mse} & $\frac{1}{n} \sum_{i=1}^n \pr{Y^{(2)}_i - \delta_i(Y_{1:n}^{(1)})}^2 -
\sigma_{i, (2)}^2$ & $\frac{1}{n^2}\sum_{i=1}^n \var\pr{(Y_i^{(2)} - \delta_i(Y_{1:n}^{
(1)}))^2
\mid
\mathcal F}$ \\
\Cref{ex:utilmax} & $-\frac{1}{n}\sum_{i=1}^n \delta_i(Y_{1:n}^{(1)}) Y^{(2)}_i $ & $
\frac{1}{n^2}
\sum_{i=1}^n \delta_i(Y_{1:n}^{(1)}) \sigma_{i, (2)}^2$ \\
\Cref{ex:topm} & $-\frac{1}{m}\sum_{i=1}^n \delta_i(Y_{1:n}^{(1)}) Y^{(2)}_i $ & $
\frac{1}{m^2}
\sum_{i=1}^n \delta_i(Y_{1:n}^{(1)}) \sigma_{i, (2)}^2$  \\ \bottomrule
\end{tabularx}
\end{table}

\begin{restatable}{prop}{unbiased}
\label{prop:unbiased}
Suppose $(Y_i,\sigma_i)$ obey \eqref{eq:gaussian_heteroskedastic_location}. Fix some
$\omega > 0$ and let $Y_{1:n}^{(1)}, Y_{1:n}^{(2)}$ be the coupled bootstrap draws. For
some decision problem, let $\bdelta(Y_{1:n}^{(1)})$ be some decision rule using only data
$\pr{Y_ {i}^{(1)},
\sigma_{i, (1)}^2}_{i=1}^n$. Let $\mathcal F = \pr{
\theta_{1:n}, Y_{1:n}^{(1)}, \sigma_{1:n, (1)}, \sigma_{1:n, (2)}}$, for 
\cref{ex:mse,ex:utilmax,ex:topm}, the estimators $T(Y_{1:n}^{(2)}, \bdelta)$ displayed
in \cref{tab:unbiased_estimation} are unbiased for the corresponding loss: \[
\E \bk{T(Y_{1:n}^{(2)}, \bdelta(Y_{1:n}^{(1)})) \mid \mathcal F } = L\pr{\bdelta(Y_{1:n}^{
(1)}), \theta_{1:n}}.
\]
Moreover, their conditional variances are equal to those  displayed in
\cref{tab:unbiased_estimation}.
\end{restatable}

\Cref{prop:unbiased} allows for an out-of-sample evaluation of decision rules, as well as
uncertainty quantification around the estimate of loss, solely imposing the Gaussian
model. %
This is a useful property in practice for comparing different empirical Bayes methods,
especially if one is worried about the misspecification of \eqref{eq:location_scale} or
if one is unwilling to evaluate risk integrating over random $\theta_i$.

\section{Empirical illustration}
\label{sec:empirical}

How does \closenpmle{} perform in the field? We now consider two empirical exercises
related to \citet{chetty2018opportunity} and \citet{bergman2019creating}. Using
 Census micro-data, \citet {chetty2018opportunity} estimate a suite of
tract-level children's outcomes in adulthood and publish an ``Opportunity Atlas'' of the
estimates and the corresponding
standard errors.\footnote{
\label{fn:correlation}Like prior work that uses this data
\citep [see, e.g., footnote 28 in] [] {andrews2021inference}, we do not have access to the
variance-covariance matrix of these estimates. Correlations across estimates are due to
small proportion of movers between tracts and are anticipated to be small. } Taking these
estimates, \citet{bergman2019creating} conducted a program
called {Creating Moves to Opportunity}. 
\citet{bergman2019creating} provided assistance to treated low-income individuals to move
 to Census tracts with estimated posterior means in the top third. We view \citet
 {bergman2019creating}'s objectives as \topm, for $m$ equal to one third of the number of
 tracts in Seattle and King County, WA.

The Opportunity Atlas published by \citet{chetty2018opportunity}  also includes
tract-level covariates, a complication that we have so far abstracted away from. In the
ensuing empirical exercises, following \citet{bergman2019creating}, the estimates are
residualized against the covariates as a preprocessing
step  \citep{fay1979estimates}.\footnote{\label{fn:covariate_additive}Alternatively, \cref{asub:covariate_additive}
shows that flexibly modeling $\E[\theta_i \mid \sigma_i, X_i] = m_0(\sigma_i, X_i)$ and
$\var(\theta_i \mid \sigma_i, X_i) = s_0^2(\sigma_i, X_i)$, as in
\eqref{eq:location-scale-covariates}, induces substantial additional benefits, relative to
simply projecting out the covariates linearly. Here, including $\sigma_i$ in the modeling
remains important---modeling $\theta_i \mid X_i$ flexibly does not fully capture these
benefits.} We now let $\tilde Y_i$ denote the raw Opportunity
Atlas estimates for a
pre-residualized parameter $\vartheta_i$ and let $(Y_i, \theta_i)$ be their residualized
counterparts against a vector of tract-level covariates $X_i$, with regression coefficient
$\beta$.\footnote{Precisely speaking, let $X_i$ be a vector of tract-level covariates. Let
$(\tilde Y_i, \sigma_i)$ be the raw Opportunity Atlas estimates of a parameter
$\vartheta_i$. Let $\beta$ be some vector of coefficients, typically estimated by weighted
least-squares of $Y_i$ on $X_i$. Let $Y_i = \tilde Y_i - X_i'\beta$ and $\theta_i =
\vartheta_i - X_i'\beta$ be the residuals. Since $\beta$ is precisely estimated, we ignore
its estimation noise. Then, the residualized objects $(Y_i, \theta_i)$ obey the Gaussian
sequence model $Y_i
\mid \theta_i, \sigma_i
\sim \Norm(\theta_i,
\sigma_i^2).$  
}  We can apply the empirical Bayes procedures in this paper to
 $ (Y_i,
\sigma_i^2)$ and obtain an estimated posterior for $\theta_i$. This estimated posterior for
the residualized parameter $\theta_i$ then implies an estimated posterior for the original
parameter $\vartheta_i = \theta_i + X_i' \beta$, by adding back the fitted values
$X_i'\beta$. When there are no covariates,
$\vartheta_i = \theta_i$ and $Y_i = \tilde Y_i$.

 The covariates we use are included in the publicly available data from
 \citet{chetty2018opportunity} and cross-referenced with their labels in
 \cref{fn:covariates}. They include: poverty rate in 2010, share of Black individuals in
 2010, mean household income in 2000, log wage growth for high school graduates, fraction
 with college or post-graduate degrees in 2010, mean parent family income rank, mean
 parent family income rank for Black individuals, number of all and Black children under
 18 with parents whose household income is below median in 2000 (in both levels and logs).

We consider 15 measures of economic mobility $\vartheta_i$. Each $\vartheta_i$ is the
population mean of \emph{some} outcome for individuals of \emph{some} demographic
subgroup growing up in tract $i$, whose parents are at the 25\th{} income
percentile.\footnote{\label{fn:alpha}\Copy{fnalpha}{Since all measures of economic
mobility have bounded support, as
either percentile ranks or percentage rates, \cref{as:moments} is automatically satisfied
for $\theta_i$ with $\alpha = 2$, at least when there are no covariates.}} We will
consider three types of outcomes:
\begin{enumerate*}[label=(\roman*)]
  \item percentile rank of adult income (\meanrank),
  \item an indicator for whether the individual has incomes in the top 20
  percentiles (\topprob), and
  \item an indicator for whether the individual is incarcerated (\incarceration)
\end{enumerate*}
for the following five demographic subgroups: all individuals (\pooled), white
individuals, white men, Black individuals, and Black men. Under these shorthands, the
outcome in \cref{sec:model} is
\topprob{} (Black), while \citet{bergman2019creating} consider \meanrank{}
\pooled.

The remainder of this section compares several methods on two exercises. In the first
exercise, a calibrated simulation, we compare MSE performance of various methods to
that of the oracle posterior. The second exercise is an empirical application to a
scale-up of the exercise in \citet{bergman2019creating}. It uses the coupled bootstrap
(\cref{sub:coupled_bootstrap}) to evaluate whether
\closenpmle{} selects more economically mobile tracts than alternatives.

\subsection{Calibrated simulation}
\label{sub:calibrated}
\Copy{simdgp}{ We draw from a data-generating process estimated from the data.
This data-generating process does not impose the location-scale assumption.
On the
 data $(Y_i, \sigma_i)$, we estimate $\hat m(\cdot), \hat s^2(\cdot)$ via local linear
 regression. We
 then transform to obtain $\hat Z_i = \frac{Y_i - \hatm(\sigma_i)}{\hats(\sigma_i)}$ and
 $\hat\nu_i = \frac{\sigma_i}{\hats(\sigma_i)}$. We partition $\sigma_i$ into vingtiles.
 For the data $(\hat Z_i, \hat\nu_i)$ whose $\sigma_i$ falls in a given vingtile $v
\in \br{1,2,3,4,5}$,
we estimate a vingtile-specific $\hat G_{n,v}$ via \npmle. We then normalize this
estimated \npmle{} to have mean zero and variance one, by affinely transforming the
estimated distribution. Finally, to generate synthetic data, for a $\sigma_i$
corresponding to the $v(\sigma_i)$\th{} vingtile, we draw $\tau_i^* \mid
\sigma_i \sim \hat G_{n,v(\sigma_i)}^{\text{normalized}}$, and set $\theta_i^* = \tau_i^*\hats
(\sigma_i) + \hatm(\sigma_i)$, $Y_i^* \mid \theta_i^*, \sigma_i \sim \Norm(\theta^*_i,
\sigma_i^2)$ and $\tilde Y_i^* = Y_i^* + X_i'\beta$. Additional details for the sampling
 process and simulation setup are documented in \cref{asub:sim_setup}. 
 }

On the simulated data, we then implement various empirical Bayes strategies. We consider
the feasible procedures: \naive,
\indepgauss, \indepnpmle, \closegauss{} (parametric),
\closegauss, and
\closenpmle, as well as the infeasible \oracle.
Here,
\begin{itemize}
  \item \naive{} sets $\hat\theta_i = Y_i$.

\item \indepgauss{} weighs the estimation of the hyperparameters $(m_0, s_0)$ with
  $1/\sigma_i^2$, following \citet{bergman2019creating}.

  \item \closegauss{} (parametric) implements
\closegauss, where \cref{item:close1} models the conditional moments parametrically as
$m_0 (\sigma_i; a) = a_1 + a_2 \log \sigma_i$
and $s_0^2(\sigma_i; b) = \exp(b_1 + b_2 \log \sigma_i)$, and estimates $m_0, s_0$ via
least-squares.\footnote{That is, we fit $a_1, a_2$ via minimizing $\sum_i (Y_i - a_1 - a_2
\log \sigma_i)^2$. We then fit $b_1, b_2$ via minimizing $
\sum_{i} \br{(Y_i - \hatm(\sigma_i))^2 - \sigma_i^2 - \exp(b_1 + b_2 \log (\sigma_i))}^2
$.
We thank an anonymous referee for this suggestion.
}

\item The conditional moments $\eta_0 = (m_0(\cdot), s_0(\cdot))$ in \closegauss{} and 
\closenpmle{} are estimated via local linear
regression, where bandwidth is selected via plug-in IMSE-optimal bandwidth, as
implemented in 
\citet{calonico2019nprobust}.\footnote{\label{fn:implementation}Specifically, $\hatm = \hat \E
[Y_i \mid \log\sigma_i]$ and  $\hats^2(\sigma_i) = \max(
\hat \E[(Y_i - \hatm(\sigma_i))^2 \mid
\log
\sigma_i] - \sigma_i^2, \tilde{s}^2(\sigma_i)),$ where $\hat \E[\cdot \mid \log
\sigma_i]$ implements local linear regression and $\tilde{s}(\sigma_i)$ implements a
 data-driven truncation of $\hats^2$, detailed in \cref
 {sec:nuisance_estimation}. Replacing the truncation point $\tilde s(\sigma_i)$ with zero
 (that is, we exclude the observations with $\hats(\sigma_i) = 0$ from estimating $\hat
 G_n$, and treat these observations as having empirical Bayes posterior degenerate at
 $\hatm(\sigma_i)$) does not appear to qualitatively affect our results. }

\item Since we know
 the ground truth data-generating process, we can also compute the
\oracle{} procedure that uses posterior means under the true $P_0$.

  \item None of the feasible procedures have access to $\beta$, which they must estimate
   in the same way using weighted least squares with weight $1/\sigma_i^2$, following
  \citet{bergman2019creating}.

\end{itemize} 

\begin{figure}[htb]
  \centering
  \includegraphics[width=\textwidth]{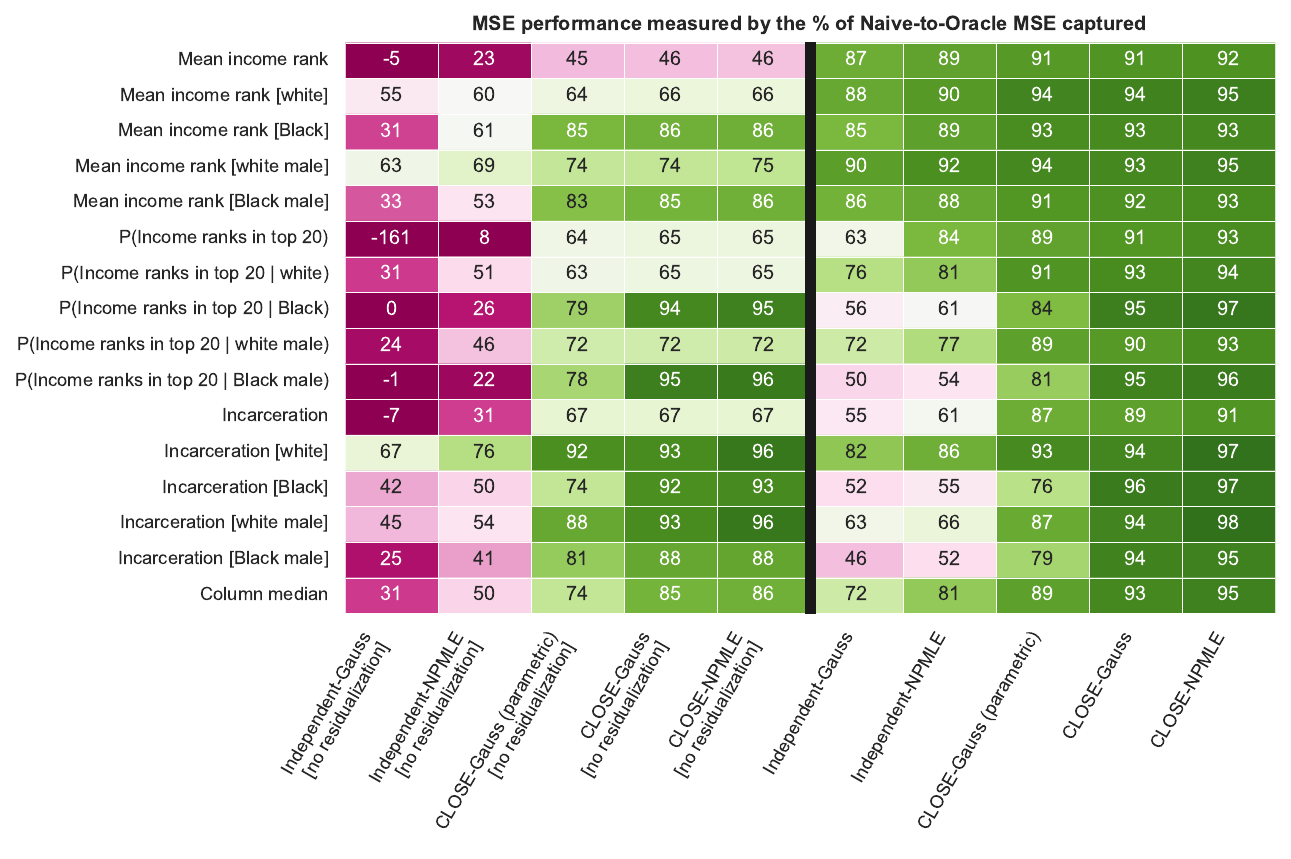}

  \begin{proof}[Notes]
  Each column is an empirical Bayes strategy that we consider, and each row is a different
  definition of $\vartheta_i$. The table shows relative performance, defined as the
  squared error improvement over \naive, normalized as a percentage of the improvement of
  \oracle{} over \naive. The last row shows the column median. Results are averaged over
  1,000 Monte Carlo draws.
  \end{proof}
  \caption{Relative squared error Bayes risk for various empirical Bayes posterior means}
  \label{fig:mse_table}
\end{figure}

\Cref{fig:mse_table} plots the results from this calibrated simulation, focusing on MSE
 performance. For each method and each target variable, we display a relative measure of
 MSE gain. For each method, we calculate its MSE gain over \naive{}, normalized by the
 MSE gain of \oracle{} over \naive. If we think of the
\oracle--\naive{} difference as the total size of the ``statistical pie,'' then
\cref{fig:mse_table} shows how much of this pie each method captures.  

The first five columns show the relative mean-squared error performance {without}
residualizing against covariates, applying empirical Bayes methods directly on $ (\tilde
Y_i, \sigma_i)$. We see that methods which assume precision independence perform worse
than
methods based on
\close.\footnote{\label
{foot:weighted_means}It may be surprising that \indepgauss{} can perform worse than
\naive{} even on MSE, since Gaussian empirical Bayes can be thought of as optimizing
among a class of linear shrinkage estimators that include
\naive. We note that, as in
\citet{bergman2019creating}, when we estimate the prior mean and prior variance, we
\emph{weight} the data with precision weights proportional to $1/\sigma_i^2$. When the
independence between $\theta$ and $\sigma$ holds, these precision weights typically
improve efficiency. However, the weighting does mean that the resulting posterior means
are no longer optimal, even asymptotically, among the class of linear shrinkage rules
under precision dependence. To take an extreme example, if a particular observation
has $\sigma_i \approx 0$, then that observation is highly influential for the prior mean
estimate. If $\E[\theta_i \mid \sigma_i]$ is very different for that observation than the
other observations, then the estimated prior mean is a bad target for shrinkage. } Across
the 15 variables, the median proportion of possible gains captured by
\indepgauss{} is only 31\%. This value is 50\% for \indepnpmle{}, and 86\% for
\closenpmle. Among the first five columns, \closenpmle{} uniformly dominates all three
 other methods. This indicates that the standard error $\sigma_i$ is highly predictive of
 $\theta_i$, and using that information can be very helpful in the absence of additional
 covariates.

The next five columns show performance when the methods do have access to covariate
information. For \meanrank, after covariate residualization, the dependence between
$\theta_i$ and $\sigma_i$ does not appear to substantially affect shrinkage decisions. 
\indepnpmle{} and
\close-methods perform similarly, capturing almost all of the available gains. For the other two outcome variables, \topprob {} and
 \incarceration, the dependence between $\theta_i $ and $
\sigma_i$ is stronger, and \close-based methods display substantial improvements over
 methods that assume precision independence. Among \close-methods, those that are more
 flexible appear to reap a small benefit, though simple parametric models for $(m_0, s_0,
 G_0)$ remain competitive and significantly improve upon methods that assume precision
 independence. The most flexible method,
\closenpmle, achieves near-oracle performance across the different definitions of
 $\theta_i$ and again uniformly dominates all other feasible methods.\footnote{\Cref
 {asub:weibull} contains an alternative data-generating process in which the $\theta_i
 \mid \sigma_i$ distribution is Weibull, which has thicker tails and higher skewness.
 Under such a scenario, \npmle-based methods more substantially outperform methods
 assuming Gaussian priors.}

\subsection{Validation exercise via coupled bootstrap}

\begin{figure}[!htb]

  \centering

  \includegraphics[width=\textwidth]{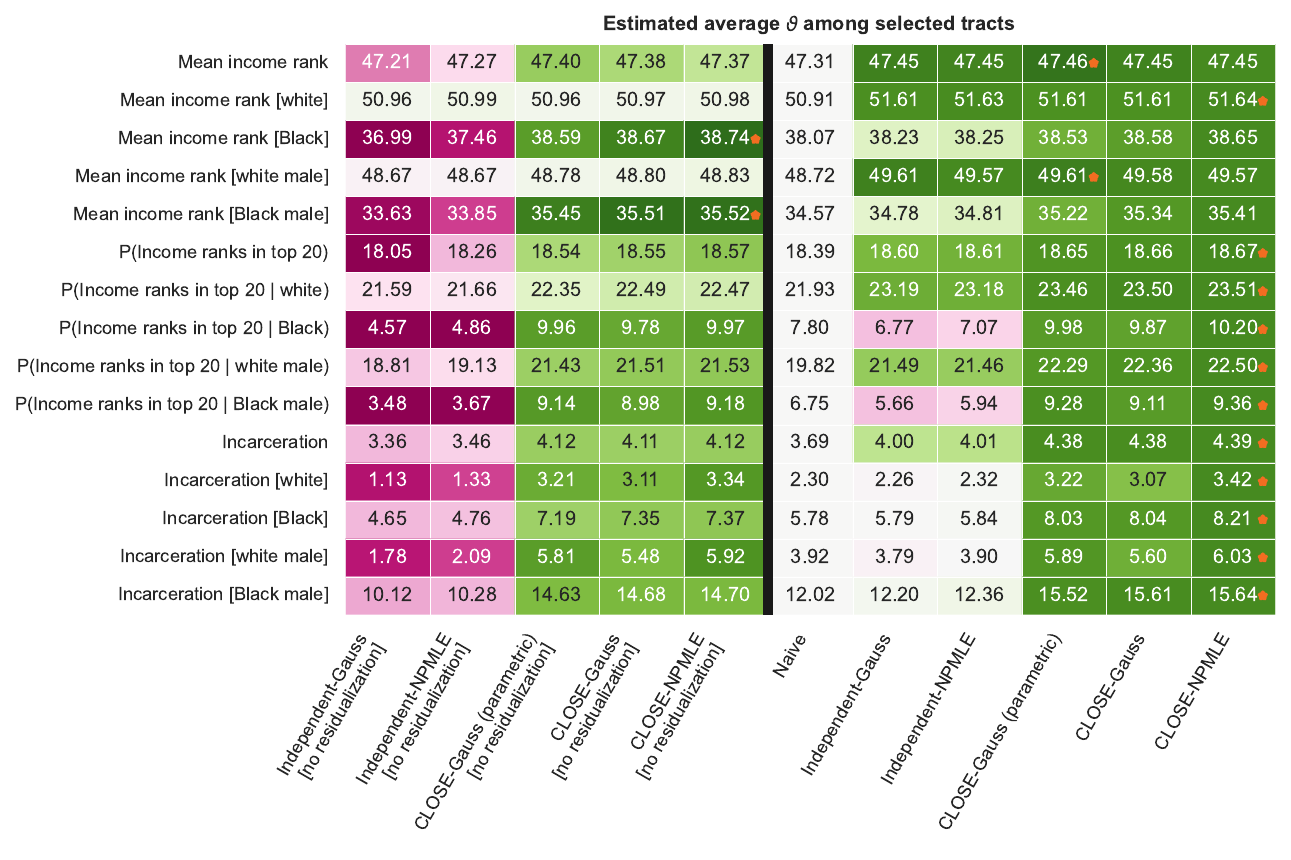}

  \begin{proof}[Notes] 
Each column is an empirical Bayes strategy that we consider, and each row is a different
  definition of $\vartheta_i$. The table shows coupled-bootstrap estimates of average
  $\vartheta_i$ among the Census tracts selected by each method---in terms of either
  percentage points or percentile ranks---over 1,000 draws of coupled bootstrap. All
  decision rules are estimated separately within CZs and select the top third of Census
  tracts within each CZ. The color scheme within each row treats the performance of
  \naive{} as zero (grey) and \closenpmle{} as one (dark green), and is hence not
  comparable across rows. On this scale, a method that overperforms \naive{} is colored
  green; otherwise it is colored magenta. The best performer for each row is additionally
  marked with an orange star.  %
  \end{proof}

  \caption{Performance of decision rules in top-$m$ selection exercise}
  \label{tab:selection_validation_within_cz}
\end{figure}

Our second empirical exercise uses the coupled bootstrap described in
\cref{sub:coupled_bootstrap} for the policy problem in \citet {bergman2019creating}.
Viewing the policy problem in \citet{bergman2019creating} as \topm, can \closenpmle{} make
better selections? 

Specifically, we imagine scaling up \citet{bergman2019creating}'s exercise and perform
empirical Bayes procedures for all Census tracts in the largest 20 Commuting Zones (CZs).
We
then select the top third of tracts \emph{within} each CZ, according to
empirical Bayesian posterior means for $\vartheta_i$. Additionally, to faithfully mimic
\citet{bergman2019creating}, here we perform all empirical Bayes procedures \emph
 {within CZ}.
Throughout, we choose $\omega$ to emulate a 90-10 train-test split on the micro-data. See
\cref{asub:sim_setup} for details on the policy exercise setup.

\Cref{tab:selection_validation_within_cz} shows the estimated performance of various
methods.
According to these estimates, \closenpmle {} generally improves over
\indepgauss.\footnote{\closenpmle{} is worse by an estimated
$0.006$ percentile ranks for \meanrank{} \pooled and worse by $0.04$
percentile ranks for
\meanrank {} for white men. In either case, the estimated disimprovement is small.}
Strikingly, \indepgauss{} with covariates underperforms \naive{} for four of the 15
variables, and
 \indepgauss{} without covariates underperforms for nearly all variables.

For the \meanrank{} variables, using \closenpmle{} generates substantial gains for
mobility measures for Black individuals (0.63 percentile ranks for Black men and 0.43
percentile ranks for Black individuals). To put these gains in dollar terms, at the income
level for experiment participants in \citet{bergman2019creating}, an incremental
percentile rank amounts to about \$1,000 per annum. Thus, the estimated gain in
terms of mean income rank is roughly \$400--600.
For the other two outcomes, \topprob{} and \incarceration,  the gains are even more
sizable. %
These gains are as high as 2--3 percentage points on average. Among \close-methods, we
again find that \closenpmle{} generally performs the best, though by small
margins.\footnote{Interestingly, the best performing method for \meanrank{} (\pooled) and 
\meanrank{} (white men) is \closegauss{} (parametric), and the best performing method
 for \meanrank{} (Black) and \meanrank{} (Black men) is \closenpmle, but without
 residualizing against covariates.} \Copy{empirical}{While \closenpmle{} is a simple
 default that works
  uniformly well, in this case, simple parametric models that allow for dependence also
  appear competitive.}

 We can think of the performance gap between \indepgauss{} and \naive {} as the
 \emph{value of basic empirical Bayes}. If practitioners find using the standard empirical
 Bayes method a worthwhile investment over screening on the raw estimates directly,
 perhaps they reveal that the value of basic empirical Bayes is economically significant.
 Across the 15 measures, the improvement of \closenpmle{} over
 \indepgauss{} is on median 260\% of the value of basic empirical Bayes, where the median
 is attained by
 \meanrank{} for Black individuals. Thus, the additional gain of \closenpmle
  {} over \indepgauss{} is substantial compared to the value of basic empirical Bayes. If
  the latter is economically significant, then it is similarly worthwhile to use
  \closenpmle{} instead.

\section{Conclusion}
\label{sec:conclusion}

This paper studies empirical Bayes methods in the heteroskedastic Gaussian location model.
We argue that precision independence---the assumption that the precision of estimates does
not
predict the true parameter---is often empirically rejected. Empirical Bayes 
methods that rely on precision independence can generate worse posterior mean estimates.
Screening decisions based on these estimates can suffer as a result. They may even be
worse than the selection decisions made with the unshrunk estimates directly.

Instead of treating $\theta_i$ as independent from $\sigma_i$, we model its conditional
distribution as a location-scale family in $\sigma$-dependent location and scale
parameters. This assumption leads naturally to a family of empirical Bayes strategies
that we call \close. The \close-framework naturally subsumes and generalizes several
existing proposals for accommodating precision dependence. We prove that
\closenpmle{} attains minimax-optimal rates in Bayes regret, extending previous
theoretical results. That is, it approximates infeasible oracle Bayes posterior means as
competently as statistically possible. Additionally, we show that an idealized version of
\closenpmle{} is robust, with finite worst-case Bayes risk. Finally, we further connect
our main theoretical
results   to ranking-type decision problems in
\citet{bergman2019creating}. 

Simulation and validation exercises demonstrate that \closenpmle{} generates sizable gains
relative to the standard parametric empirical Bayes shrinkage method. Across calibrated
simulations, \closenpmle{} attains close-to-oracle mean-squared error performance. In a
hypothetical, scaled-up version of \citet{bergman2019creating}, across a wide range of
economic mobility measures, \closenpmle{} consistently selects more mobile tracts than
does the standard empirical Bayes method. The gains in the average economic mobility among
selected tracts, relative to the standard empirical Bayes procedure, are often comparable
to---or even multiples of---the value of basic empirical Bayes.

\appendix

\numberwithin{lemma}{section}
\numberwithin{theorem}{section}
\numberwithin{cor}{section}
\numberwithin{prop}{section}
\numberwithin{as}{section}
\numberwithin{rmk}{section}
\numberwithin{figure}{section}

\begin{appendices}

\section{Proof outline for \cref{cor:maintext}}
\label{asec:proof_main}

The proof of \cref{cor:maintext} depends on numerous results deferred to the Online
Appendix. An outline is stated here. For constants $\Delta_n, M_n, C$ to be chosen, define
the following events: For $\Dinfty \equiv \max(
\norm{\hatm -
m_0}_\infty, \norm{\hats-s_0}_\infty)$, 
\begin{align}
A_n &\equiv \br{\Dinfty \le \Delta_n,\, \bar Z_n \equiv \max_{i \in [n]}
(|Z_i| \vee 1) \le
M_n}  \label{eq:barzn_def} \\
\A_n(C) &\equiv
\br{\norm{\hateta - \eta_0}_\infty  \le C
n^{-\frac{p}{2p+1}} (\log n)^{\beta_0}}.
 \label{eq:bold_A_def}
\end{align}
With the choice $\Delta_n = C_{1} n^{-\frac{p}{2p+1}} 
(\log n)^{\beta_0},$ we have that $
A_n = \A_n(C_1) \cap \br{\bar Z_n \le M_n}. 
$
The event $\A_n(C)$ indicates that the first-step estimates $\hateta$ are accurate. By
\cref{as:holder}(2), there is some sufficiently large constant $C$ such that this event
occurs with high probability: $\P [\A_n
(C)]
> 1-
\frac{1} {n^2}$. The event $A_n$ also occurs with high probability since the additional
 requirement $\bar Z_n \le M_n$ can be made probable, by choosing some $M_n$ logarithmic
 in $n$, thanks to \cref{as:moments}.

To prove \cref{cor:maintext}, we consider the events $A_n, A_n^\comp$ separately. On
$A_n^\comp$, we use the fact that the empirical Bayes posterior means $\hat\theta_i$ and
the oracle posterior means $\theta_i^*$ are no farther than the range of the data $\max
Y_i - \min Y_i$, which is logarithmic in $n$ under \cref{as:moments}
(\cref{lemma:reg_on_anc}). Since $A_n^\comp$ is assumed to be unlikely, regret on
$A_n^\comp$ is sufficiently small.

On the event $A_n$, the first-step estimates $\hateta$ are accurate, and the data $Z_i$
are not too large. The bulk of the argument thus controls regret on $A_n$, stated
separately in the following theorem, whose proof is deferred to \cref
{sec:appendix_setup}.

\begin{rmksq}[Notation]
 For $A_n, B_n \ge 0$, we use $A_n \lesssim B_n$ to mean that some universal $C$ exists
such that $A_n \le C B_n$ for all $n$, and we use $A_n \lesssim_x B_n$ to mean that some
universal $C_x$ exists such that $A_n \le C_x B_n$ for all $n$.\footnote{In logical
statements, appearances of $\lesssim$ implicitly prepend ``there exists a universal
constant'' to the statement. For instance, statements like ``under certain assumptions,
$\P(A_n \lesssim B_n) \ge c_0$'' should be read as ``under certain assumptions, there
exists a constant $C>0$ such that for all $n$, $\P (A_n \le CB_n) \ge c_0$.''} 
\end{rmksq}

\begin{restatable}{theorem}{thmregretonAn}
\label{thm:regret_on_An}
Suppose \cref{as:npmle,as:holder,as:moments,as:variance_bounds} hold. Fix some $\beta > 0,
C_{1} > 0$, there exists choices of a constant $C_{\H,2}$ such that, for $\Delta_n = C_
{1} n^{-p/ (2p+1)} (\log n)^\beta$, $M_n = C_{\H, 2} (\log n)^{1/\alpha}$, and
corresponding $A_n$, %
\[
\E\bk{
    \reg(\hat G_n, \hateta) \one(A_n)
} \leH  n^{-\frac{2p}{2p+1}} (\log n)^{\frac{2+\alpha}{\alpha} + 3 + 2\beta}.
\]
\end{restatable}

We now outline how to prove \cref{thm:regret_on_An} and provide a proof for 
\cref{cor:maintext} given \cref{thm:regret_on_An}.

\subsection{Step 1: convert regret on $\theta_i$ to regret on $\tau_i$} To prove 
\cref{thm:regret_on_An}, note that the empirical Bayes posterior means are of the
form\[
\hat \theta_{i, \hat G_n, \hateta} = \hatm(\sigma_i) + \hats(\sigma_i) \cdot \hat\tau_{i,
\hat
G_n,
\hateta},
\]
where $\hat\tau_{i, \hat
G_n,
\hateta}$ denotes the posterior mean of $\tau_i \mid \hat Z_i, \hat\nu_i$, where $\tau_i
\sim \hat G_n$ and $\hat Z_i \mid \tau_i, \hat\nu_i \sim \Norm(\tau_i, \hat\nu_i^2)$. On
the event $A_n$, $\hatm,\hats$ are close to $m_0, s_0$, and thus controlling $\reg$
amounts to controlling MSE on $\tau$'s: $
\E\bk{
  (\tau_i^* - \hat\tau_{i, \hat
G_n,
\hateta})^2
},
$
where $\tau_i^* = \hat\tau_{i, G_0, \eta_0}$ is the oracle posterior mean for $\tau_i$.

To do so, we adapt the argument in \citet{soloff2021multivariate} and \citet
{jiang2020general}. To introduce this argument, recall that $\psi_i$ denotes
the log-likelihood in \cref{as:npmle} and define \[
\sub_n(G) = \pr{\frac{1}{n} \sum_{i=1}^n \psi_i(Z_i, \eta_0, G) -  \frac{1}{n} \sum_
{i=1}^n \psi_i(Z_i, \eta_0, G_0)}_+
\numberthis \label{eq:likelihood_suboptimality}
\] as the log-likelihood suboptimality of $G$ against the true distribution $G_0$,
 evaluated on $Z_i, \nu_i$, which depend on the true conditional moments $\eta_0$. For
 generic $G$ and $\nu > 0$, define \[ f_{G, \nu}(z) = \int_{-\infty}^{\infty} \varphi\pr
 {
\frac{z-\tau}{\nu}} \frac{1}{\nu} \,G(d\tau).
\numberthis \label{eq:density_convolved}
\]
as the  density of some mixed Gaussian variable $Z \sim \Norm(0,\nu^2) \star G$. 
Let the average squared Hellinger distance be \[
\barh^2(\fs{G_1}, \fs{G_2}) = \frac{1}{n} \sum_{i=1}^n h^2\pr{
    f_{G_1, \nu_i}, f_{G_2, \nu_i}
} \quad h^2(f,g) \equiv \frac{1}{2}\int_{-\infty}^{\infty} (\sqrt{f(x)} - \sqrt{g(x)})^2
\,dx.\numberthis \label{eq:avg_hellinger}
\]

Loosely speaking, \citet{soloff2021multivariate}, following \citet{jiang2009general}, show
that
\begin{enumerate}[wide]
  \item With high probability, all approximate maximizers of the likelihood have low
  average Hellinger distance: \[
\P\bk{
  \text{There exists $G$ where } \sub_n(G) < C_1 \delta_n^2 \text{ but } \barh^2 (f_{G,
  \cdot}, f_ {G_0,
  \cdot}) > C_2 \delta_n^2
} < \frac{1}{n} \numberthis \label{eq:hellinger_result_loose}
  \] for some rate function $\delta_n^2 \lesssim \frac{1}{n} (\log n)^C$ (Theorem 7 in
  \citet{soloff2021multivariate}).
  \item For a given $G$, $
\E[(\tau_i^* - \hat \tau_{i, G, \eta_0})^2] \lesssim (\log n)^C \pr{\barh^2(f_{G,
  \cdot}, f_ {G_0,
  \cdot})}
  $  (Theorem 9 in \citet{soloff2021multivariate}).
\end{enumerate}
Therefore, an approximate maximizer $\hat G_n^*$ of the likelihood $\sub_n(G)$ should have
low average Hellinger distance to $G_0$ and thus should output similar posterior
means.

\subsection{Step 2: show $\hat G_n$ is an approximate maximizer of the true likelihood}

To use this argument for \cref{thm:regret_on_An}, a key challenge is that $\hat G_n$ only
maximizes the \emph{approximate} likelihood $\frac{1}{n}\sum_i \psi_i (Z_i, \hateta, G)$,
which only has $\hateta \approx \eta_0$ on $A_n$, but $\hateta \neq \eta_0$. A key result
is an oracle inequality for the likelihood (\cref{cor:suboptimality}), where, loosely
speaking, \[
\P\bk{A_n, \sub_n(\hat G_n) \gtrsim_{\H} \varepsilon_n} = O(1/n) 
\numberthis 
\label{eq:cor_suboptimality_loose}
\] for some $\varepsilon_n \lesssim (\log n)^C \pr{n^{-2p/(2p+1)} + n^{-p/(2p+1)} \barh(f_
 {\hat G_n,
\cdot}, f_{G_0,
\cdot})}$. 
This result states that the likelihood suboptimality of the feasible \npmle{}
$\hat G_n$ cannot be much higher than its average Hellinger distance to $G_0$. 

The bound \eqref{eq:cor_suboptimality_loose} is a refinement of a simple linearization
argument applied to $\eta \mapsto \frac{1}{n} \sum_{i=1}^n \psi_i(Z_i, \eta, \hat G_n)$.
Heuristically speaking, a first-order Taylor expansion yields \[
\frac{1}{n} \sum_{i=1}^n \psi_i
(Z_i, \hateta, \hat G_n) - \frac{1}{n} \sum_{i=1}^n \diff{\psi_i}{\eta}\evalbar_
{\eta=\eta_0} (\hat\eta_i -
\eta_{0i})  \approx \frac{1}{n} \sum_{i=1}^n \psi_i(Z_i, \eta_0, \hat G_n).
\]
Here, $\frac{1}{n} \sum_{i=1}^n \psi_i
(Z_i, \hateta, \hat G_n)$ is large by definition of $\hat G_n$. Thus, the right-hand side
would be large following a bound on the first-order term\[
\absauto{\frac{1}{n} \sum_{i=1}^n \diff{\psi_i}{\eta}\evalbar_
{\eta=\eta_0} (\hat\eta_i -
\eta_{0i})}.
\]
A naive bound on this term, using only the fact that $|\hat\eta_i -
\eta_{0i}| \le \norm{\hateta - \eta_0}_\infty$, would lead to a suboptimal regret rate of
 $O(n^{-p/(2p+1)} (\log n)^C)$. Our more refined analysis additionally leverages the fact
 that
\[
\E\bk{\diff{\psi_i(Z, \eta, G_0)}{\eta}\evalbar_
{\eta=\eta_0}} = -\E\bigg[
  \frac{Z-\tau}{\nu} \underbrace{\diff{\br{(Z(\eta)-\tau)/\nu(\eta)}}{\eta}}_{
  [-1/\sigma, -\tau/\sigma]'}
\bigg] = 0,
\]
and thus the derivative $\diff{\psi_i}{\eta}$ is sufficiently small if $\hat G_n \approx
G_0$ in Hellinger distance.

\subsection{Step 3: adapt Hellinger distance bound}

\cref{cor:suboptimality} makes sure that $\hat G_n$ probably achieves high likelihood, but
the bound depends on
$\barh^2$. Since \eqref{eq:hellinger_result_loose} uses a likelihood bound for $G$ to
control $\barh^2$, we need to additionally finesse \eqref{eq:hellinger_result_loose} to
accommodate the fact that the likelihood bound depends on $\barh^2$.

Second, we adapt \eqref{eq:hellinger_result_loose} to show that, loosely speaking, with high probability $\hat G_n$ has low average
Hellinger distance to $G_0$ 
(\cref{cor:hellinger_large_dev}): \[
\P\bk{
  A_n, \barh^2(f_{\hat G_n, \cdot}, f_{G_0, \cdot}) \gtrsim_\H n^{-p/(2p+1)} (\log n)^{C}
} = O\pr{\frac{(\log n)^C}{n} }.
\] 
Thus, this allows us to show that $\E[(\tau_i^* - \hat \tau_{i, G, \eta_0})^2 \one
(A_n)]$ is small, after additional empirical process arguments in
\cref{sec:appendix_setup}.

This section concludes with a proof for \cref{cor:maintext} given these results.

\ifrepeatthm \cormaintext*\fi

\begin{proof}[Proof of \cref{cor:maintext}]

Let $\Delta_n = C_
{1,\Hyperparams} n^{-\frac{p}{2p+1}} (\log n)^{\beta_0}$, where $C_{1,\H}$ is the
constant in \cref{as:holder}(2), and $M_n = C (\log n)^ {1/\alpha}$
for some $C$ chosen by our application of \cref{thm:regret_on_An}.
Decompose \begin{align*}
&\E[\reg(\hat G_n, \hateta)] \\
&=  \E[\reg(\hat G_n, \hateta) \one(A_n)] + \E[\reg(\hat
G_n, \hateta) \one(\A_n^\comp \cup \br{\bar Z_n > M_n})]  \\
&\le \E[\reg(\hat G_n, \hateta) \one(A_n)] + \E[\reg(\hat
G_n, \hateta) \one(\A_n^\comp)] \\
&\quad+ \E[\reg(\hat G_n, \hateta) \one(\bar Z_n > M_n)]
\\
&\leH n^{-\frac{2p}{2p+1}} (\log n)^{\frac{2+\alpha}{\alpha} + 3 + 2\beta_0} + \frac{2}{n} (\log n)^
{2/\alpha} \tag{\cref{thm:regret_on_An,lemma:reg_on_anc}} \\
&\leH n^{-\frac{2p}{2p+1}} (\log n)^{\frac{2+\alpha}{\alpha} + 3 + 2\beta_0}.
\end{align*}  
The application of \cref{lemma:reg_on_anc} uses the implication of \cref{as:holder}(2)
 that \[\P (\A_n(C_{1,\Hyperparams})^\comp) = \P (\Dinfty > \Delta_n) \le \frac{1}
 {n^2}.\qedhere\]
\end{proof}

\section{Proofs of other results stated in the main text}

\ifrepeatthm
  \thmminimaxlower*
\fi

\begin{proof}[Proof of \cref{thm:minimaxlower}]
We consider a specific choice of $G_0, \sigma_{1:n},$ and $s_0$. Namely, suppose $G_0
\sim \Norm(0,1)$, $\sigma_{1:n}$ are equally spaced in $[\sigl, \sigu]$, and $s_0(\sigma)
= (s_\ell + s_u) / 2 \equiv s_0$ is constant.  With these choices, the oracle posterior
means $\theta_i^*$ are equal to \[
\theta_i^* = \frac{s_0^2}{s_0^2 + \sigma_i^2} Y_i + \frac{\sigma_i^2}{s_0^2 + \sigma_i^2}
m_0(\sigma_i).
\]
For a given vector of estimates $\tilde \theta_{1:n}$, we can form $
\hatm(\sigma_i) = \frac{s_0^2 + \sigma_i^2}{\sigma_i^2} \pr{\tilde \theta_i -
\frac{s_0^2}{s_0^2 + \sigma_i^2} Y_i}.$ Note that, for this choice, \[
\E\bk{\frac{1}{n} \sum_{i=1}^n (\tilde \theta_i - \theta_i^*)^2} 
\gtrsim_{\sigl, s_u} \E
\bk{\frac{1}{n} \sum_{i=1}^n (\hat m(\sigma_i) - m_0(\sigma_i))^2}.
\]
Therefore, the minimax rate must be lower bounded by the minimax rate of estimating $m_0$
at $\sigma_{1:n}$, where the right-hand side takes the infimum over all estimators of
$m_0$ with data $(Y_i, \sigma_i)$: \[
\inf_{\hat\theta_{1:n}} \sup_{\sigma_{1:n}, P_0} \E\bk{
    \frac{1}{n} \sum_{i=1}^n (\hat\theta_i - \theta_i)^2 - (\theta_i^* - \theta_i)^2
} \gtrsim_{\sigl, s_u} \inf_{\hatm} \sup_{m_0}  \E\bk{\frac{1}{n} \sum_{i=1}^n (\hat
m(\sigma_i) - m_0 (\sigma_i))^2}.
\]
Using classical minimax results, \cref{lemma:minimax} shows that the right-hand side is
lower bounded by $n^ {-2p/ (2p+1)}$, which completes the proof.
\end{proof}

\ifrepeatthm\worstcaserisk*\fi

\begin{proof}[Proof of \cref{thm:worstcaserisk}]
Note that $\hat\theta_{i, G_0^*, \eta_0} =
s_0(\sigma_i)
\hat\tau_{i, G_0^*, \eta_0} + m_0(\sigma_i),$ where $\tau^*_{i,G,\eta}$ is the posterior
mean for $\tau_i$ under $(G,\eta)$, and $
\theta_i = s_0(\sigma_i) \tau_i + m_0(\sigma_i).
$ Thus, \[
\frac{1}{n}\sum_{i=1}^n (\hat\theta_i - \theta_i)^2 = \frac{1}{n}\sum_{i=1}^n s_0^2
(\sigma_i) (\hat\tau_{i, G_0^*,
\eta_0} - \tau_i)^2.
\]
\Cref{thm:max_risk} shows that $
\E_{G_i} \bk{
  (\hat\tau_{i, G_0^*, \eta_0} - \tau_i)^2
} \le C_{\lambda, \epsilon}$ for all $G_0^* \in \mathcal G_0$.  
Taking the expected value with respect to $P_0 \in \mathcal P(m_0, s_0)$ and apply the
bound $C_{\lambda, \epsilon}$, we have that \[
\E\bk{\frac{1}{n}\sum_{i=1}^n (\hat\theta_i - \theta_i)^2} \le 
C_{\lambda, \epsilon}
\frac{1}{n}\sum_{i=1}^n s_0^2
(\sigma_i).
\]

By \cref{lemma:minimax_close_gauss}, we have that  \[
\frac{1}{n} \sum_{i=1}^n \frac{\sigma_i^2}{\sigma_i^2 + s_0^2(\sigma_i)} s_0^2(\sigma_i)
=  \inf_{\hat\theta_{1:n}} \sup_{P_0 \in \mathcal P(m_0, s_0)} \E_{P_0} \bk{\frac{1}{n}
\sum_{i=1}^n (\hat\theta_i - \theta_i)^2}.
\]
Note that, for some $c_{\bar\rho} > 0$, \[
\frac{1}{n} \sum_{i=1}^n \frac{\sigma_i^2}{\sigma_i^2 + s_0^2(\sigma_i)} s_0^2(\sigma_i)
 = \frac{1}{n} \sum_{i=1}^n \frac{1}{1 + s_0^2(\sigma_i)/\sigma_i^2} s_0^2
 (\sigma_i) \ge c_{\bar\rho} \frac{1}{n} \sum_{i=1}^n s_0^2(\sigma_i).
\]
Hence, \begin{align*}
\E\bk{\frac{1}{n}\sum_{i=1}^n (\hat\theta_i - \theta_i)^2} &\le \frac{C_{\lambda, \epsilon}}{c_
 {\bar\rho}} \frac{1}{n} \sum_{i=1}^n \frac{\sigma_i^2}{\sigma_i^2 + s_0^2
 (\sigma_i)} s_0^2(\sigma_i)  \\ &= C_{\bar\rho, \lambda, \epsilon} \inf_{\hat\theta_{1:n}} \sup_
 {P_0 \in
\mathcal P (m_0, s_0)} \E_{P_0} \bk{\frac{1}{n}
\sum_{i=1}^n (\hat\theta_i - \theta_i)^2}.
\end{align*}
This concludes the proof.
\end{proof}

\ifrepeatthm
  \mserelevance*
\fi 

\begin{proof}[Proof of \cref{thm:mserelevance}]

\begin{enumerate}[wide]
  \item 
By the law of iterated expectations, since $\hat\theta_i, \theta_i^*$ are both measurable
with respect to the data,\footnote{For a randomized decision rule $\hat\theta_i$ that is
additionally measurable with respect to some $U$ independent of $ (\theta_i, Y_i,
\sigma_i)_{i=1}^n$, this step continues to hold since $\E[\theta_i \mid U, Y_i, \sigma_i]
= \theta_i^*$.}
  \begin{align*}
\E[\utilmaxreg] &= \E \bk{\frac{1}{n} \sum_{i=1}^n \br{\one(\theta_i^* \ge 0) - \one
(\hat\theta_i \ge 0)} \theta_i^*}
  \end{align*}
  Note that, for $\one(\theta_i^* \ge 0) - \one
(\hat\theta_i \ge
0)$ to be nonzero, $0$ is between $\hat\theta_i$ and $\theta_i^*$. Hence,
$|\theta_i^*| \le |\theta_i^* - \theta_i|$ and thus by Jensen's inequality \[
\E[\utilmaxreg] \le \E\bk{\frac{1}{n} \sum_{i=1}^n |\theta_i^* - \theta_i|  } \le \pr{\E
\bk{
\frac{1}{n} \sum_{i=1}^n (\theta_i^* - \theta_i)^2  }}^{1/2}.
\]

\item Let $\mathcal J^*$ collect the indices of the top-$m$ entries of $\theta_i^*$ and
let $\hat {\mathcal J}$ collect the indices of the top-$m$ entries of $\hat\theta_i$.
Then,
by law of iterated expectations, \[
\frac{m}{n} \E [\topmreg] = \frac{1}{n}\sum_{i=1}^n \E \bk{\br{\one(i \in \mathcal J^*) -
\one(i \in
\hat{\mathcal J})} \theta_i^*}.
\]
Observe that this can be controlled by applying \cref{prop:rearrangement}, where $w_i = 0$
for all $i \le n-m$ and $w_i = 1$ for all $i > n-m$. In this case, $\norm{w} = \sqrt{m}$.
Hence, \[
\frac{m}{n}\E[\topmreg] \le 2\sqrt{\frac{m}{n}} \E\bk{\pr{\frac{1}{n} \sum_
{i=1}^n (\hat\theta_i - \theta^*_i)^2}^{1/2}} \le 2\sqrt{\frac{m}{n}} \pr{\E\bk{\frac{1}{n} \sum_
{i=1}^n (\hat\theta_i - \theta^*_i)^2}}^{1/2}.
\]
Divide both sides by $m/n$ to obtain the result.\qedhere
\end{enumerate}
\end{proof}

\begin{prop}
\label{prop:rearrangement}
Suppose $\sigma(\cdot)$ is a permutation such that $\hat\theta_
{\sigma(1)} \le \cdots \le \hat\theta_{\sigma(n)}
$. %
Then, for any $w_1,\ldots, w_n \in \R$,
\[
\frac{1}{n} \sum_{i=1}^n w_i \theta^*_{(i)} - \frac{1}{n} \sum_{i=1}^n w_i \theta^*_
{\sigma(i)} \le \frac{2\norm{w}_2}{\sqrt{n}} \sqrt{\frac{1}{n} \sum_
{i=1}^n (\hat\theta_i - \theta^*_i)^2}.
\]
where $\theta_{(1)}^* \le \theta_{(2)}^* \le \ldots \le \theta_{(n)}^*$ are the order
statistics for $\br{\theta_1^*,\ldots, \theta_n^*}$ and $ \norm{w}_2^2
= w_1^2 + \cdots + w_n^2$.
\end{prop}

\begin{proof}
We compute
\begin{align*}
\frac{1}{n} \sum_{i=1}^n w_i \theta^*_{(i)} - \frac{1}{n} \sum_{i=1}^n w_i \theta^*_
{\sigma(i)} &\le \absauto{\frac{1}{n}\sum_{i=1}^n w_i \theta_{(i)}^* - \frac{1}{n} \sum_
{i=1}^n w_i
\hat\theta_{\sigma
(i)}} + \absauto{\frac{1}{n} \sum_{i=1}^n w_i(\hat \theta_{\sigma(i)} - \theta^*_{\sigma
(i)}) }\\
& \le \frac{\norm{w}_2}{\sqrt{n}} \sqrt{\frac{1}{n} \sum_{i=1}^n (\theta_{(i)}^* -
\hat\theta_{\sigma(i)})^2} + \frac{\norm{w}_2}{\sqrt{n}} \sqrt{\frac{1}{n} \sum_
{i=1}^n (\hat\theta_i - \theta^*_i)^2} \tag{Cauchy--Schwarz}\\
& \le 2\frac{\norm{w}_2}{\sqrt{n}} \sqrt{\frac{1}{n} \sum_
{i=1}^n (\hat\theta_i - \theta^*_i)^2}.
\end{align*}
The last step follows from the observation that the sorted difference is dominated by the
unsorted difference, $
\sum_{i=1}^n (\theta_{(i)}^* -
\hat\theta_{\sigma(i)})^2 \le \sum_
{i=1}^n (\hat\theta_i - \theta^*_i)^2
$,
which is true by the rearrangement inequality.\footnote{For all real numbers $x_1 \le
\cdots \le x_n, y_1\le
\cdots \le y_n$, we have $\sum_i x_i y_{\pi(i)} \le \sum_i x_i y_i$ for any
permutation $\pi$.}
\end{proof}

\begin{rmksq}[Mover interpretation of \cref{thm:mserelevance}]
\label{rmk:nonuniform}
Recall that we can think of \topm{} as the decision problem in
\citet{bergman2019creating} (\cref{rmk:mover}). The utility function represents the
 expected mobility of a mover, assuming that the mover moves randomly into one of the
 high mobility Census tracts. Our proof of \cref{thm:mserelevance} allows for a slightly
 more general decision problem. Suppose the decision now is to provide a full ranking of
 Census tracts for potential movers and maximize the expected mobility for a mover.
 Suppose that the probability that a mover moves to a tract depends decreasingly and
 solely on the tract's rank. To be more concrete, suppose the mover has probability
 $\pi_1$ of moving to the highest-ranked tract, $\pi_2 \le \pi_1$ to the second-highest,
 and so
 forth. Then, with the same argument, the corresponding regret is dominated by $2\sqrt
 {n \sum_{i=1}^n \pi_i^2} \cdot 
\pr{\E\bk{\frac{1}{n} \sum_{i=1}^n (\hat\theta_i - \theta_i^*)^2}}^{1/2}$, which
generalizes \eqref{eq:topm_regret_bound}.
\end{rmksq}

\ifrepeatthm\unbiased*\fi

\begin{proof}[Proof of \cref{prop:unbiased}] 
These are straightforward calculations of the
expectation. 
Since every expectation and variance is conditional on $\theta_{1:n},
Y^{(1)}_{1:n}, \sigma_{1:n, (1)}, \sigma_{1:n, (2)}$, we write $\E[\cdot \mid \mathcal F]$
and $\var(\cdot \mid \mathcal F)$ without ambiguity.

\begin{enumerate}[wide]
  \item (\cref{ex:mse}) The unbiased estimation follows directly from the calculation
  \begin{align*}
  \E\bk{
(Y_i^{(2)} - \delta_i(Y_{1:n}^{(1)}))^2
  \mid \mathcal F} = (\theta_i^{(2)} -  \delta_i(Y_{1:n}^{(1)}))^2 + \sigma_{i, (2)}^2
  \end{align*}
  The conditional variance statement holds by definition.
  \item (\cref{ex:utilmax}) The unbiased estimation follows directly from the calculation
  \[
 \E\bk{
 \delta_i(Y_{1:n}^{(1)})Y_i^{(2)} 
  \mid \mathcal F} = \delta_i(Y_{1:n}^{(1)}) \theta_i.
  \]
  The conditional variance statement follows from \[
\var\bk{
 \delta_i(Y_{1:n}^{(1)})
Y_i^{(2)}
  \mid \mathcal F} = \delta_i(Y_{1:n}^{(1)}) \sigma_{1:n, (2)}^2.
  \]

  \item (\cref{ex:topm}) The loss function for \cref{ex:topm} is the same as that for 
  \cref{ex:utilmax} up to a factor of $n/m$. Since we condition on $Y_{1:n}^{(1)}$, the
  argument is thus analogous.
  \qedhere
\end{enumerate}
\end{proof}
\end{appendices}

\printbibliography

\newpage 

\begin{refsection}

\begin{appendices}

\makeatletter
\def\@seccntformat#1{\@ifundefined{#1@cntformat}%
   {\csname the#1\endcsname\space}%
   {\csname #1@cntformat\endcsname}}%
\newcommand\section@cntformat{\thesection.\space} %
\makeatother
\renewcommand{\thesection}{OA\arabic{section}}
\counterwithin{equation}{section}
\counterwithin{figure}{section}
\counterwithin{table}{section}

\begin{center}
\textbf{\large Online Appendix to \\ ``\newtitle''}

Jiafeng Chen

\today
\end{center}

\DoToC

\newpage

\part{Proof of \cref{cor:maintext}}

\section{Review of notation and proofs of \cref{lemma:reg_on_anc,thm:regret_on_An}}
\label{sec:appendix_setup}

We recall some notation in the main text, and introduce additional notation. Recall that
we assume $n\ge 7$. We observe $(Y_i, \sigma_i)_{i=1}^n,$ where $(Y_i, \sigma_i) \in
\R\times \R_{> 0}$ such that
\[
Y_i \mid (\theta_{i}, \sigma_{i}) \sim \Norm(\theta_i, \sigma_i^2)
\]
and $(Y_i, \theta_i, \sigma_i)$ are mutually independent. Assume that the joint
distribution for $(\theta_i, \sigma_i)$ takes the
location-scale form \eqref{eq:location_scale} \[
\theta_i \mid (\sigma_1,\ldots, \sigma_n) \sim G_{0}\pr{\frac{\theta_i - m_0(\sigma_i)}
{s_0
(\sigma_i)}}.
\]

Define shorthands $m_{0i} = m_0(\sigma_i)$ and $s_{0i} = s_0(\sigma_i)$. Define the
transformed parameter $\tau_i = \frac{\theta_i - m_{0i}}{s_{0i}}$, the transformed data
$Z_i = \frac{Y_i - m_{0i}}{s_{0i}}$, and the transformed variance $\nu_i^2 =
\frac{\sigma_i^2}{s_{0i}^2}$. By assumption, \[ Z_i \mid (\tau_i, \nu_i) \sim
\Norm(\tau_i, \nu_i^2) \quad \tau_i \mid \nu_1,\ldots, \nu_n
\iid G_0.
\]
Let $\hateta = (\hatm, \hats)$ denote estimates of $m_0$ and $s_0$. Likewise, let
$\hateta_i = (\hatm_i, \hats_i) = (\hatm(\sigma_i), \hats(\sigma_i))$. For a given
$\hateta$, define \[
\hat Z_i = \hat Z_i(\hateta) = \hat Z_i(Z_i, \hateta) = \frac{Y_i -
\hatm_i}{\hats_i} = \frac{s_{0i} Z_i + m_{0i}
- \hatm_i}{\hats_i} \quad \hat\nu_i^2 = \hat\nu_i^2(\hateta) = \frac{\sigma_i^2}
{\hats_i^2}.
\]
We will condition on $\sigma_{1:n}$ throughout, and hence we treat them as fixed. Let $\nu_\ell, \nu_u$ be the corresponding bounds on $\nu_i = \frac{\sigma_i}{s_0(\sigma_i)}$, implied by \cref{as:variance_bounds}.

For generic values $\eta = (m,s)$ and
distribution $G$,  define the log-likelihood function \[
\psi_i(z,\eta, G) = \log \int_{-\infty}^{\infty} \varphi\pr{
\frac{\hat Z_i(\eta) -
\tau}{\hat\nu_i(\eta)}} \,G(d\tau) = \log \pr{\hat\nu_i(\eta) \cdot f_{G, \hat\nu_i
(\eta)}(\hat Z_i(\eta))},
\]
where we recall $f_{G,\nu}$ from \eqref{eq:density_convolved}. As
a shorthand, we write $f_{i, G} = f_{G, \nu_i}(Z_i)$ and $f'_{i,G} = f_{G, \nu_i}' (Z_i)$.

Fix some generic $G$ and $\eta = (m,s)$. The empirical Bayes posterior mean ignores
the fact that $G, \eta$ are potentially estimated. The posterior mean for $
\theta_i = s_i \tau + m_i
$
is \[
\hat \theta_{i, G, \eta} \equiv m_i + s_i \PE_{G, \hat\nu_i(\eta)}[\tau \mid
\hat Z_i(\eta)].
\]
Here, we define $\PE_{G, \nu}\bk{h(\tau, Z) \mid z}$ as the function of
$z$ that
equals the posterior mean for $h(\tau, Z)$ under the data-generating model $\tau \sim G$
and $Z \mid \tau \sim \Norm(\tau, \nu)$. Explicitly, \[
\PE_{G, \nu}\bk{h(\tau, Z) \mid z} = \frac{1}{f_{G, \nu} (z)} \int h(\tau, z)
\varphi\pr{
    \frac{z-\tau}{\nu}
}\frac{1}{\nu}\, G(d\tau).
\]
Explicitly, by Tweedie's formula, \[
\PE_{G, \hat \nu_i(\eta)}[\tau_i \mid \hat Z_i(\eta)] = \hat Z_i(\eta) +
\hat\nu_i^2(\eta) \frac{ f'_
{G, \hat\nu_i(\eta)}(\hat Z_i(\eta))}{f_
{G, \hat\nu_i(\eta)}(\hat Z_i(\eta))}.
\]
Hence, since $\hat Z_i(\eta) = \frac{Y_i - m_i}{s_i}$, \[
\hat \theta_{i, G, \eta} = Y_i + s_i \hat\nu_i^2(\eta) \frac{ f'_
{G, \hat\nu_i(\eta)}(\hat Z_i(\eta))}{f_
{G, \hat\nu_i(\eta)}(\hat Z_i(\eta))}.
\] 

Define $\theta_i^* = \hat \theta_{i, G_0, \eta_0}$ as the oracle Bayesian's posterior
 mean. Fix some positive number $\rho > 0$, define a regularized posterior mean as \[
\hat\theta_{i, G, \eta, \rho} = Y_i + s_i \hat\nu_i^2(\eta) \frac{ f'_
{G, \hat\nu_i(\eta)}(\hat Z_i(\eta))}{f_
{G, \hat\nu_i(\eta)}(\hat Z_i(\eta)) \vee \frac{\rho}{\hat\nu_i(\eta)}} \numberthis
\label{eq:regularized_posterior_means}
\]
and define $\theta^*_{i, \rho} = \hat \theta_{i, G_0, \eta_0, \rho}$ correspondingly.
Similarly, we define \[
\hat\tau_{i, G, \eta, \rho} = \hat Z_i(\eta) + \hat\nu_i^2(\eta) \frac{ f'_
{G, \hat\nu_i(\eta)}(\hat Z_i(\eta))}{f_
{G, \hat\nu_i(\eta)}(\hat Z_i(\eta)) \vee \frac{\rho}{\hat\nu_i(\eta)}} \quad \tau^*_
{i, \rho} = \hat\tau_{i, G_0, \eta_0, \rho} \numberthis \label{eq:regularized_tau}
\]
We also define \begin{align}
\invphi(\rho) = \sqrt{\log \frac{1}{2\pi \rho^2}}
\quad \rho \in (0, (2\pi)^{-1/2}) \label{eq:invphidef}
\end{align}
so that $\varphi(\invphi(\rho))=\rho$. Observe that $\invphi(\rho) \lesssim
\sqrt{\log (1/\rho)}$.

Recall the event $A_n$ and the quantity $\bar Z_n$ in
\eqref{eq:barzn_def}. Many of the following statements are true for $A_n$ defined with
generic $\Delta_n, M_n$. However, to obtain our rate expression in the end, recall that we
set $\Delta_n, M_n$ to be of the following form:
\[\Delta_n = C_
{\Hyperparams} n^
{-
\frac{p}{2p+1}} (\log
n)^{\beta} \text{ and } M_n =
(C_\Hyperparams + 1) (C_{2,\Hyperparams}^{-1}\log n)^{1/\alpha}.\numberthis \label{eq:DMrate}\]
Here, $C_\Hyperparams$ is to be chosen, and $C_{2, \Hyperparams}$ is some constant
determined by \cref{thm:suboptimality}. Correspondingly, we also have a choice 
\[
\rho_n =
\frac{1}{n^3} e^{-C_{\Hyperparams,\rho} M_n^2 \Delta_n}\minwith \frac{1}{e\sqrt{2\pi}}
\numberthis
\label{eq:rhodef}
,\]
where the constant $C_{\Hyperparams,\rho}$ is chosen to satisfy the following result, proved in
\cref{sec:likelihood}.

\begin{restatable}{lemma}{lemmalb}
\label{thm:lb}
Suppose $|\bar Z_n| = \max_{i\in[n]} |Z_i| \vee 1 \le M_n$, $\norm{\hats - s_0}_\infty \le
\Delta_n$, and $\norm{\hatm-m_0}_\infty \le \Delta_n$.
Let $\hat G_n$ satisfy \cref{as:npmle} and $\hateta$ satisfy \cref{as:holder}. Then, under
\cref{as:Delta_M_rate}, \footnote{This assumption is satisfied with our choices in
\eqref{eq:DMrate}.}
\begin{enumerate}
    \item $|\hat Z_i \vee 1| \leH M_n$
    \item There exists $C_\Hyperparams$ such that with $\rho_n = \frac{1}{n^3}\exp
    \pr{-C_
    {\Hyperparams} M_n^2 \Delta_n} \minwith \frac{1}{e\sqrt{2\pi}}$, we have that \[
f_{\hat G_n, \nu_i}(Z_i) \ge \frac{\rho_n}{\nu_i}.
    \]
    \item The choice of $\rho_n$ satisfies $\log (1/\rho_n) \rateeq_{\Hyperparams} \log
    n$,
    $\invphi(\rho_n) \rateeq_{\Hyperparams} \sqrt{\log n}$, and $\rho_n \leH n^{-3}$.
\end{enumerate}
\end{restatable}

We now state and prove \cref{lemma:reg_on_anc,thm:regret_on_An}, which
are crucial claims in the proof of \cref{cor:maintext}. The first claim, 
\cref{lemma:reg_on_anc}, controls regret on the event $A_n^\comp$.

\begin{lemma}
\label{lemma:reg_on_anc}
Under \cref{as:npmle,as:holder,as:moments,as:variance_bounds}, for $\beta \ge 0$, suppose
$\Delta_n, M_n$ are of the form \eqref{eq:DMrate} such that $\P(\bar Z_n > M_n) \le n^
{-2}$, we can decompose \begin{align*}
\E[\reg(\hat G_n, \hateta) \one(\Dinfty > \Delta_n)] &\leH \P(\Dinfty >
\Delta_n)^{1/2} (\log n)^{2/\alpha} \\
\E[\reg(\hat G_n, \hateta) \one(\bar Z_n > M_n)] &\leH \frac{1}{n} (\log n)^
{2/\alpha}.
\end{align*}
\end{lemma}

\begin{proof}
Observe that, for an event $A$ on the data $Z_{1:n}$, \begin{align*}
\E\bk{\reg(\hat G_n, \hateta)\one(A)} &= \E\bk{\frac{1}{n} \sum_{i=1}^n (\hat\theta_
{i,
\hat
G, \hateta} - \theta_i^*)^2 \one(A)} \\
& \le  \E\bk{\pr{\frac{1}{n} \sum_{i=1}^n (\hat\theta_
{i,
\hat
G, \hateta} - \theta_i^*)^2}^{2}}^{1/2} \P(A)^{1/2}
\end{align*}
by Cauchy--Schwarz. Since $\norm{\hateta -\eta_0}_\infty \leH 1$,  a crude bound
(\cref{lemma:distance_bound_hard}) shows that \[
\pr{\frac{1}{n} \sum_{i=1}^n (\hat\theta_
 {i,
 \hat
 G, \hateta} - \theta_i^*)^2}^{2} \leH \bar Z_n^{4}.
\]
Apply \cref{lemma:tail_bound_max} to find that $\E[\bar Z_n^4] \leH (\log n)^{4/\alpha}$.
This proves both claims.
\end{proof}

The main theorem of this part in the Online Appendix is stated and proved in
the following section. It characterizes regret behavior on the event $A_n$,
for $\Delta_n, M_n$ chosen as in \eqref{eq:DMrate}.

\subsection{Proof of \cref{thm:regret_on_An}}

We first state a result that is key to our remaining arguments, which we verify in the
Supplementary Material (\cref{sec:hellinger}).

\begin{restatable}{cor}{corhellinger}
\label{cor:hellinger_large_dev}

Assume \cref{as:npmle,as:holder,as:moments,as:variance_bounds} hold and suppose $\Delta_n,
M_n$ take the form \eqref{eq:DMrate}. Define the rate sequence \[\delta_n = n^{-p/(2p+1)}
(\log n)^{\frac{2 + \alpha}{2 \alpha}
+
\beta}.\numberthis\label{eq:delta_rate_def}\]
Then, there exists some constant $B_{\Hyperparams}$, depending solely on
$C^*_\Hyperparams$ in \cref{cor:suboptimality}, $\beta$, and $p, \nu_\ell,
\nu_u$ such that
\[
\P\bk{
    A_n, \barh(\fs{\hat G_n}, \fs{G_0}) > B_\Hyperparams \delta_n } \le \pr{\frac{\log
    \log n}{\log 2} + 10}\frac{1}{n}.
\]
\end{restatable}

\thmregretonAn*

\begin{proof}
We choose $M_n$ to be of the form \eqref{eq:DMrate}. Note that we can decompose
\begin{align*}
\reg(G,\eta)
&= \frac{1}{n}\sum_{i=1}^n (\hat\theta_{i, G, \eta} - \theta_i)^2 - \frac{1}
{n}\sum_{i=1}^n (\theta^*_i - \theta_i)^2
\\
&= \frac{1}{n} \sum_{i=1}^n (\hat\theta_{i, G, \eta} - \theta^*_i)^2 + \frac{2}{n} \sum_
{i=1}^n (\theta^*_i - \theta_i)(\hat\theta_{i, G, \eta} - \theta^*_i)
\numberthis\label{eq:regdecompose}
\end{align*} 
Note that the second term in the decomposition
\eqref{eq:regdecompose}, truncated to $A_n$, is mean zero: \[
\E\bk{\one(A_n) \frac{2}{n} \sum_
{i=1}^n (\theta^*_i - \theta_i)(\hat\theta_{i, \hat G_n, \hateta} - \theta^*_i)} = 0,
\]
 since $\E[(\theta^*_i - \theta_i) \mid Y_{1},\ldots, Y_n] = 0$. Thus, we can
focus on \[
\E[\reg(\hat G_n,\hateta)\one(A_n)] = \E\bk{\frac{\one(A_n)}{n} \sum_{i=1}^n (\hat\theta_
{i, \hat G_n, \hateta} - \theta^*_i)^2} \equiv
\frac{1}{n}\E[\one(A_n) \norm{\hat\theta_{\hat G_n, \hateta} - \theta^*}^2] \numberthis
\label{eq:regret_diff_term},
\]
where we let $\hat\theta_{\hat G_n, \hateta}$ denote the vector of estimated posterior
means and let $\theta^*$ denote the corresponding vector of oracle posterior means. Let
the subscript $\rho_n$ denote a vector of regularized posterior means as in
\eqref{eq:regularized_posterior_means}. Here, we set $\rho_n$ as in \eqref{eq:rhodef}.  Thus, we may
further decompose
by triangle
inequality: \begin{align*}
\norm{\hat\theta_{\hat G_n, \hateta} - \theta^*} &\le
\norm{\hat\theta_{\hat G_n, \hateta} - \hat\theta_{\hat G_n, \eta_0}}
+ \norm{\hat\theta_{\hat G_n, \eta_0} - \hat\theta_{\hat G_n, \eta_0, \rho_n}}
+ \norm{\hat\theta_{\hat G_n, \eta_0, \rho_n} - \theta^*_{\rho_n}}
+ \norm{\theta^*_{\rho_n} - \theta^*}.
\end{align*}

We denote each term in the decomposition of \eqref{eq:regret_diff_term} by $\xi_1,\ldots,
\xi_4$: \begin{align}
\xi_1 &= \frac{\one(A_n)}{n} \norm{\hat\theta_{\hat G_n, \hateta} - \hat\theta_{\hat G_n, \eta_0}}^2
\\
\xi_2 &= \frac{\one(A_n)}{n} \norm{\hat\theta_{\hat G_n, \eta_0} - \hat\theta_{\hat G_n, \eta_0, \rho_n}}^2
\\
\xi_3 &= \frac{\one(A_n)}{n}\norm{\hat\theta_{\hat G_n, \eta_0, \rho_n} - \theta^*_{\rho_n}}^2 \\
\xi_4 &= \frac{\one(A_n)}{n}\norm{\theta^*_{\rho_n} - \theta^*}^2.
\end{align}
We have that $
\eqref{eq:regret_diff_term} \le 4 (\E\xi_1 + \E\xi_2 + \E\xi_3 + \E\xi_4) = 4 (\E\xi_1 + \E\xi_3 + \E\xi_4).$

The individual $\xi_j$'s are bounded by the arguments in the remainder of this
section. The key term leading to the final rate is $\E[\xi_3]$:
\begin{itemize}[wide]
  \item We show in \cref{sub:xi1} that $\xi_1 \leH M_n^2 (\log n)^2 \Delta_n^2$, and thus $\E \xi_1 \leH M_n^2 (\log n)^2
\Delta_n^2$. 
\item \Cref{thm:lb}(2) implies that, given the choice $\rho_n$ in \eqref{eq:rhodef}, the
regularized posterior means and the unregularized posterior means are equal $\hat\theta_
{\hat G_n,
\eta_0, \rho_n} = \hat\theta_{\hat G_n,
\eta_0}$, since the truncation does not bind. Therefore, $\xi_2 = 0$.

\item We show in \cref{sub:xi3} that $\E\xi_3 \leH
(\log n)^3 \delta_n^2$. Here, $\delta_n$ is the rate in \eqref{eq:delta_rate_def}.

\item Finally, we show in \cref{sub:xi4} that $\E\xi_4 \leH 
\frac{1}{n}$. 
\end{itemize}

Lastly, we observe that by the definition of $\delta_n$ in \eqref{eq:delta_rate_def}, the
upper bound for $\E[\xi_3]$ is the dominating rate. Plugging the definition of
$\delta_n^2$ yields that \[
\eqref{eq:regret_diff_term} = \E[\reg(\hat G_n,\hateta)\one(A_n)] \leH n^{-\frac{2p}{2p+1}} (\log n)^{\frac{2 + \alpha}{\alpha} +
 3 + 2\beta}.
\qedhere
\]

\begin{rmksq}[Remainder of proof]
The proof for \cref{thm:regret_on_An} hinges on the key result in \cref{sub:xi3} for
bounding $\xi_3$. Effectively, the argument first relates $\xi_3$ to the corresponding
regret for the transformed parameters $\tau_i$ \eqref{eq:regularized_tau}: \[
\norm{\tau_{\hat G_n, \eta_0, \rho_n} - \tau_{\rho_n}^*}^2. 
\]
To prove a bound for this object, we truncate to the event where $\barh^2(\fs{\hat
G_n}, \fs{G_0})$ is small and use the fact that, loosely speaking, the $\norm{\tau_{\hat
G_n, \eta_0, \rho_n}
- \tau_{\rho_n}^*}^2$ can be bounded by $\barh^2(\fs{\hat G_n}, \fs
   {G_0})$. For this argument to work, the key is that the event where
   $\barh^2(\fs{\hat G_n}, \fs{G_0})$ is small has high probability, which is
   shown in 
\cref{cor:hellinger_large_dev}. Lastly, to prove \cref{cor:hellinger_large_dev}, we need
to first establish that $\hat G_n$---estimated off $(\hat Z_i, \hat\nu_i)$---does not have
high likelihood suboptimality $\sub_n(\hat G_n)$. This is the most laborious part of the
proof (\cref{cor:suboptimality}). 
\end{rmksq}

\end{proof}

\begin{lemma}
\label{sub:xi1}

Under the assumptions of \cref{thm:regret_on_An}, in the proof of \cref{thm:regret_on_An},
$\xi_1 \leH M_n^2 (\log n)^2 \Delta_n^2$.
\end{lemma}
\begin{proof}
Note that, by an application of Taylor's theorem, \begin{align*}
\absauto{
    \hat\theta_{i, \hat G_n, \hateta} - \hat\theta_{i, \hat G_n, \eta_0}
} &= \sigma_i^2  \absauto{
    \frac{f'_{\hat G_n, \hat\nu_i}(\hat Z_i)}{\hats_i f_{\hat G_n, \hat\nu_i}(\hat Z_i)} -
    \frac{f'_{\hat G_n, \nu_i}(Z_i)}{s_{0i} f_{\hat G_n, \nu_i}(Z_i)}
}  \\
&= \sigma_i^2 \absauto{\pr{\diff{\psi_i}{m_i} \evalbar_{\hat G_n, \hat\eta} -
\diff{\psi_i}{m_i} \evalbar_{\hat G_n, \eta_0}}}
\tag{\cref{eq:dpsidm}}
\\
& = \sigma_i^2 \absauto{
    \diff{^2\psi_i}{m_i \partial s_i}\evalbar_{\hat G_n, \tilde \eta_i} (\hats_i - s_{0i})
    + \diff{^2\psi_i}{m_i^2}\evalbar_{\hat G_n, \tilde \eta_i} (\hatm_i - m_{0i})
},
\end{align*}
where we use $\tilde\eta_i$ to denote some intermediate value lying on the line segment
between $\hateta_i$ and $\eta_{0i}$. By \cref{lemma:secondderivatives}, we can bound the
two derivative terms and obtain \[
\one(A_n) \absauto{
    \hat\theta_{i, \hat G_n, \hateta} - \hat\theta_{i, \hat G_n, \eta_0}
} \leH M_n (\log n) \Delta_n.
\]
Hence, squaring both sides, we obtain $
\xi_1 \leH M_n^2 (\log n)^2 \Delta_n^2.
$
\end{proof}

\begin{lemma}
\label{sub:xi4}

Under the assumptions of \cref{thm:regret_on_An}, in the proof of \cref{thm:regret_on_An},
$\E\xi_4 \leH \frac{1}{n}.$
\end{lemma}
\begin{proof}
Note that \begin{align*}
\E[(\theta^*_{i,\rho_n} - \theta^*_i)^2] &= s_{0i}^2 \int  \pr{ \nu_i^2 \frac{f'_{G_0,
\nu_i} (z)} {f_
{G_0,
\nu_i}
(z)}}^2 \pr{1-\frac{f_{G_0, \nu_i}}{f_{G_0, \nu_i} \vee \frac{\rho_n}{\nu_i}}}^2 f_{G_0,
\nu_i}
(z)\,dz \\
&\le s_{0i}^2 \E\bk{\pr{ \nu_i^2 \frac{f'_{G_0, \nu_i}(Z)}{f_
{G_0,
\nu_i}
(Z)}}^4}^{1/2} \cdot \P\bk{f_{G_0, \nu_i}(Z) < \rho_n /\nu_i}^{1/2} \tag{Cauchy--Schwarz}
\\
&\leH \E[(\tau-Z)^4]^{1/2} \cdot \rho_n^{1/3} \var(Z)^{1/6} \tag{Tweedie's formula,
Jensen's
inequality, and
\cref{lemma:chebyshev}}\\
&\leH \frac{1}{n}.
\end{align*}
In particular, the third line follows since by Tweedie's formula and Jensen's inequality 
\[
\E\bk{\pr{ \nu_i^2 \frac{f'_{G_0, \nu_i}(Z)}{f_
{G_0,
\nu_i}
(Z)}}^4} = \E\bk{\E_{G_0, \nu_i}[\tau - Z \mid Z]^4} \le \E[(\tau-Z)^4] \leH 1. 
\]
Therefore, $\E[\xi_4] \leH \frac{1}{n}.$
\end{proof}

\subsection{Controlling $\xi_3$}
\label{sub:xi3}

\begin{lemma}

Under the assumptions of \cref{thm:regret_on_An}, in the proof of \cref{thm:regret_on_An},
$\E\xi_3 \leH  (\log n)^3 \delta_n^2$, where $\delta_n$ is defined in \eqref{eq:delta_rate_def}.
\end{lemma}
\begin{proof}
Observe that $
\absauto{
    \hat\theta_{i, \hat G_n, \eta_0, \rho_n} - \theta^*_{i, \rho_n}
} = s_{0i} \absauto{
    \hat \tau_{i, \hat G_n, \eta_0, \rho_n} - \tau^*_{i, \rho_n}
}
$
where $\hat \tau_{i, \hat G_n, \eta_0, \rho_n}$ is the regularized posterior with prior
$\hat G_n$ at conditional moments $\eta_0$ and $\tau^*_{i, \rho_n} = \hat\tau_{i, G_0,
\eta_0, \rho_n}$, where we recall \eqref{eq:regularized_tau}.

Thus, we shall focus on controlling \[
\one(A_n) \norm{\hat \tau_{\hat G_n, \eta_0, \rho_n} - \tau^*_{\rho_n}}^2.
\]
Fix the rate function $\delta_n$ in \eqref{eq:delta_rate_def} and the constant
$B_\Hyperparams$ in \cref{cor:hellinger_large_dev} (which in turn depends on
$C^*_\Hyperparams$ in \cref{cor:suboptimality}). Let $B_{n} = \{\barh(\fs{\hat G_n},
\fs{G_0}) < B_\Hyperparams \delta_n\}$ be the event of a small average squared Hellinger
distance. Let $G_1,\ldots, G_N$ be a finite set of prior
distributions (chosen to be a net of $\br{G: \barh(\fs{G}, \fs{G_0}) \le \delta_n}$ in
some distance), and let $\tau^ {
(j)}_{\rho_n}$ be the posterior mean vector corresponding
to prior $G_j$ with conditional moments $\eta_0$ and regularization $\rho_n$.

Now, note that, for any $j$, \begin{align*}
&\one(A_n) \norm{\hat\tau_{\hat G_n, \eta_0, \rho_n} - \tau_{\rho_n}^*} \\
&\le 
\norm{\hat\tau_
{\hat G_n,
\eta_0, \rho_n} - \tau_{\rho_n}^*}\one(A_n \cap B_n^\comp) + \one(A_n \cap B_n) \pr{
\norm{\hat\tau_
{\hat G_n,
\eta_0, \rho_n} - \tau_{\rho_n}^*} - \norm{\tau_{\rho_n}^{(j)} - \tau_{\rho_n}^*}}_+ \\ 
&\quad+  \pr{ \norm{\tau_{\rho_n}^{(j)} - \tau_{\rho_n}^*} - \E[ \norm{\tau_{\rho_n}^{
(j)} - \tau_ {\rho_n}^*} ]}_+ +\E[ \norm{\tau_{\rho_n}^{(j)} - \tau_{\rho_n}^*} ].
\end{align*}
Then \[
\frac{\one(A_n)}{n} \norm{\hat \tau_{\hat G_n, \eta_0, \rho_n} - \tau^*_{\rho_n}}^2 \le
\frac{4}{n}\pr{\zeta_1^2 + \zeta_2^2 + \zeta_3^2 + \zeta_4^2}
\]
where \begin{align}
\zeta_1^2 &=  \norm{\hat \tau_{\hat G_n, \eta_0, \rho_n} - \tau^*_
{\rho_n}}^2 \one\pr{A_n \cap B_n^\comp}  \label{eq:zeta1}\\
\zeta_2^2 &= \pr{
    \norm{\hat\tau_{\hat G_n, \eta_0, \rho_n} - \tau^*_{\rho_n}} - \max_{j \in [N]}
    \norm{\tau_{\rho_n}^{(j)} - \tau_{\rho_n}^*}
}_+^2 \one(A_n \cap B_n)  \label{eq:zeta2}\\
\zeta_3^2 &= \max_{j\in [N]} \pr{
    \norm{\tau_{\rho_n}^{(j)} - \tau_{\rho_n}^*} - \E\bk{\norm{\tau_{\rho_n}^{(j)} - \tau_
    {\rho_n}^*}}
}^2_+  \label{eq:zeta3}\\
\zeta_4^2 &=  \max_{j\in [N]} \pr{\E\bk{\norm{\tau_{\rho_n}^{(j)} - \tau_
    {\rho_n}^*}}}^2. \label{eq:zeta4}
\end{align}
The decomposition $\zeta_1$ through $\zeta_4$ is exactly analogous to Section D.3 in the
supplementary materials to 
\citet{soloff2021multivariate} and to the proof of Theorem 1 in \citet{jiang2020general}. In
particular, $\zeta_1$ is the gap on the ``bad event'' where the average squared Hellinger
distance is large, which is manageable since $\one(A_n \cap B_n^\comp)$ has small
probability by \cref{cor:hellinger_large_dev}. $\zeta_2$ is the distance from the
posterior means at $\hat G_n$ to the closest posterior mean generated from the net
$G_1,\ldots, G_N$; $\zeta_2$ is small if we make the net $\br{G_1,\ldots, G_N}$ very fine.
$\zeta_3$ measures the
distance between $\norm{\tau_{\rho_n}^{(j)} - \tau_{\rho_n}^*}$ and its expectation;
$\zeta_3$ can be controlled by (i) a large-deviation inequality and (ii) controlling the
metric entropy of the net (\cref{prop:covering_ms}). Lastly, $\zeta_4$ measures the
expected distance between $\tau_{\rho_n}^{(j)} $ and $\tau_{\rho_n}^*$; it is small since
$G_j$ are fixed priors with small average squared Hellinger distance.

However, our argument for $\zeta_3$ is slightly different and avoids an argument in
\citet{jiang2009general} which appears to not apply in the heteroskedastic setting.
See \cref{rmk:gap}.

The subsequent subsections control $\zeta_1$ through $\zeta_4$, and find that $\zeta_4
\leH (\log n)^3 \delta_n^2$ is the dominating term.
\end{proof}

\subsubsection{Controlling $\zeta_1$}

First, we note that, \[
\pr{\hat\tau_{i, \hat G_n, \eta_0, \rho_n} - \tau^*_{i,\rho_n}}^2 \one(A_n \cap B_n^\comp) \leH
\log (1/\rho_n)  \one(A_n \cap B_n^\comp) \leH \log n \one(A_n \cap B_n^\comp).
\tag{\cref{thm:lb,thm:jianglemma2}}
\]
By \cref{cor:hellinger_large_dev}, $\P(A_n \cap B_n^\comp) \le \pr{\frac{\log\log n}{\log 2}
+ 9} \frac{1}{n}$, and hence $
\frac{1}{n}\E\zeta_1^2 \leH \frac{\log n  \log\log n}{n}.
$

\subsubsection{Controlling $\zeta_2$}

Choose $G_1, \ldots, G_N$ to be a minimal $\omega$-covering of $\br{G : \barh(\fs{G},
\fs{G_0}) \le \delta_n}$ under the pseudometric \[
d_{M_n, \rho_n}(H_1, H_2) = \max_{i\in [n]} \sup_{z : |z| \le M_n} \abs[\Bigg]
{
\frac{\nu_i^2 f'_{H_1, \nu_i}(z)}{f_{H_1, \nu_i}(z) \vee (\rho_n/\nu_i)}
    - \frac{\nu_i^2 f'_{H_2, \nu_i}(z)}{f_{H_2, \nu_i}(z) \vee (\rho_n/\nu_i)}
} \numberthis \label{eq:distance_bayesrule}
\]
where $N \le N\pr{\omega/2, \mathcal P(\R), d_{M_n, \rho_n}}$.\footnote{This is by the
monotonicity relation of covering numbers. See Exercise 4.2.10 in 
\citet{vershynin2018high}.} We note that
\eqref{eq:distance_bayesrule} and $d_{m,\infty, M}$ \eqref{eq:def_dk} are different only
by constant factors, in the sense that $d_{M_n, \rho_n}(H_1, H_2) \rateeq_{\Hyperparams}
 d_{m,\infty, M}(H_1, H_2)$ for all $H_1, H_2$. 
Therefore, \cref{prop:covering_ms} implies that \[
\log N\pr{\frac{\delta \log(1/\delta)}{\rho_n} \sqrt{\log(1/\rho_n)}, \mathcal P(\R), d_
{M_n,
\rho_n}} \leH \log(1/\delta)^2 \max\pr{1, \frac{M_n}{\sqrt{\log (1/\delta)}}} \numberthis
\label{eq:covering_number_bayesrule}.
\]
Then, \begin{align*}
\frac{1}{n} \zeta_2^2 &\le \frac{\one(A_n \cap B_n)}{n} \min_{j\in[N]} \pr{ 
\norm{\hat\tau_ {\hat G_n, \eta_0, \rho_n} - \tau^*_{\rho_n}} - 
    \norm{\tau_{\rho_n}^{(j)} - \tau_{\rho_n}^*} }^2 \tag{$(a-\max_j b_j)_+ \le
    \min_j |a-b_j|$} \\
&\le \one(A_n \cap B_n) \frac{1}{n} \min_{j \in [N]} 
\norm{\hat\tau_
{\hat G_n,
\eta_0, \rho_n} - \tau_{\rho_n}^{(j)}}^2\tag{Triangle inequality :
 $\abs{\norm{a-b} - \norm{b-c}} \le \norm{a-c}$}
 \\
 &= \one(A_n \cap B_n) \min_{j\in [N]} \frac{1}{n} \sum_{i=1}^n \one\pr{|Z_i| \le M_n}\pr
{
\frac{\nu_i^2 f'_{\hat G_n, \nu_i}(Z_i)}{f_{\hat G_n, \nu_i}(Z_i) \vee (\rho_n/\nu_i)}
    - \frac{\nu_i^2 f'_{G_j, \nu_i}(Z_i)}{f_{G_j, \nu_i}(Z_i) \vee (\rho_n/\nu_i)
    }
}^2 \\
&\le \omega^2 \leH \frac{\delta^2 \log(1/\delta)^2}{\rho_n^2} \log(1/\rho_n) \tag{Reparametrize
$\omega = 2\delta \log(1/\delta) \rho_n^{-1} \sqrt{\log(1/\rho_n)}$}.
\end{align*}

\subsubsection{Controlling $\zeta_3$}
\label{sub:zeta_3}

We first observe that $V_{ij} \equiv |\tau_{i,\rho_n}^{(j)} - \tau_{i, \rho_n}^*| \leH
\sqrt{\log n}$, by
\cref{thm:jianglemma2}. Let $V_j = (V_{1j},\ldots, V_{nj})'$, we have that $
\zeta_3 = \max_j (\norm{V_{j}} - \E\norm{V_j})_+. $ Let $K_n = C_\Hyperparams \sqrt
 {\log n} \ge \max_{ij} |V_{ij}|$. Since $G_j, G_0$ are both fixed, $V_{1j},\ldots, V_
 {nj}$ are mutually independent, and $V_{ij}/K_n \in [0,1]$.  Then, observe that by \cref
 {lemma:norm_concentration}, \begin{align*}
\P\pr{\norm{V_j} > \E[\norm{V_j}] + u} = \P\pr{
    \normauto{\frac{V_j}{K_n}} \ge \E\normauto{\frac{V_j}{K_n}} + \frac{u}{K_n}
}\le \exp\pr{-\frac{u^2}{2K_n^2}}.
\end{align*}
By union bound, \[
\P\pr{\zeta_3^2 > x} \le N \exp\pr{-\frac{x}{2K_n^2}}.
\]
Therefore, \begin{align*}
\E[\zeta_3^2] &= \int_0^\infty \P(\zeta_3^2>x) \,dx \le  \int_0^\infty \min\pr{1, N\exp
\pr{-
\frac{x}
{2K_n^2}}} \,dx \\ &= 2K_n^2 \log N + \int_
{2K_n^2 \log N}^{\infty} N\exp\pr{-
\frac{x}
{2K_n^2}}\,dx \\
& \leH \log n \log N.
\end{align*}
Now, if we take $\delta = \rho_n/n$, then $
\frac{1}{n} \E[\zeta_2^2 + \zeta_3^2] \leH \frac{(\log n)^{5/2} M_n}{n}
$.

\begin{rmksq}
\label{rmk:gap}

For the analogous term under homoskedasticity, \citet{jiang2009general}  observe that $
\norm{\tau_ {\rho_n}^ {(j)} -
\tau_{\rho_n}^*}$ is a Lipschitz function of the noise component $Z_i - \tau_i$. As a result, a Gaussian
isoperimetric inequality
(Theorem 5.6 in \citet{boucheron2013concentration}) bounds \[
\P\pr{
  \norm{\tau_{\rho_n}^{(j)} - \tau_{\rho_n}^*} \ge \E\bk{ \norm{\tau_{\rho_n}^{(j)} -
  \tau_
  {\rho_n}^*} \mid \tau_1,\ldots, \tau_n } + x
},
\] independently of $n$---a fact used in Proposition 4 of \citet
 {jiang2009general}. Note that the concentration of $\norm{\tau_{\rho_n}^{(j)}
- \tau_{\rho_n}^*}$ is towards its conditional mean \[ \E\bk{ \norm{\tau_{\rho_n}^{(j)} -
  \tau_
  {\rho_n}^*} \mid \tau_1,\ldots, \tau_n }.\] In the homoskedastic setting where $\nu_i =
  \nu$, \[
 \E\bk{ \norm{\tau_{\rho_n}^{(j)} -
  \tau_
  {\rho_n}^*} \mid \tau_1,\ldots, \tau_n } = \E_{G_{0,n}} \bk{ \norm{\tau_{\rho_n}^{(j)} -
  \tau_
  {\rho_n}^*} } \numberthis \label{eq:homosk_identity}
  \]
  where $G_{0,n} = \frac{1}{n} \sum_i \delta_{\tau_i}$ is the empirical distribution of
  the $\tau$'s. However, \eqref{eq:homosk_identity} no longer holds in the heteroskedastic
  setting, and to adapt this argument, we need to additionally control the difference
  between $\E\bk{\norm{\tau_{\rho_n}^{(j)} -
  \tau_
  {\rho_n}^*} \mid \tau_1,\ldots, \tau_n }$ and $\E\bk{ \norm{\tau_{\rho_n}^{(j)} -
  \tau_
  {\rho_n}^*} }$. The argument in \citet{jiang2020general} (p.2289) appears to use the
  Gaussian concentration of Lipschitz functions argument without this additional step.
  Instead, we establish control of $\zeta_3$ by observing that entries of $\tau_{\rho_n}^{
  (j)} - \tau_{\rho_n}^*$ are bounded and applying the convex Lipschitz concentration
  inequality. 
\end{rmksq}

\subsubsection{Controlling $\zeta_4$}
Consider a change of variables where we let $w_i = z/\nu_i$ and $\lambda_i = \tau /
\nu_i$. Let $G_{(i)}$ be the distribution of $\lambda_i$ under $G$, where 
$
G_{(i)}(d\lambda) = G(d\tau).
$ Then \[
f_{G,\nu_i}(z) = \int \frac{1}{\nu_i} \varphi\pr{w_i - \lambda_i} G(d\tau) = \frac{1}
{\nu_i} \int \varphi\pr{w_i - \lambda_i} G_{(i)}
(d \lambda_i) = \frac{1}
{\nu_i} f_{G_{(i)}, 1}(w_i)
\]
and $
f_{G, \nu_i}'(z) = \frac{1}{\nu_i^2} f'_{G_{(i)}, 1} (w_i)
$. Hence, \begin{align*}
\E\bk{(\tau_{i, \rho_n}^{(j)} - \tau_{i, \rho_n}^*)^2} &= \nu_i^2 \E\bk{\pr{
    \frac{f'_{G_{j,(i)}, 1}(w_i)}{f_{G_{j,(i)}, 1}(w_i) \vee \rho_n} - \frac{f'_{G_{0,(i)},
    1}
    (w_i)}{f_{G_{0,(i)}, 1}(w_i) \vee \rho_n}
}^2} \\
&\leH \max\pr{(\log 1/\rho_n)^3, \abs{\log h(f_{G_{j,(i)}, 1}, f_{G_{0,(i)}, 1})}}h^2(f_{G_
{j,(i)},
1},
f_{G_{0,(i)}, 1}) \tag{\cref{lemma:misspec_prior}} \\
&= \max\pr{(\log 1/\rho_n)^3, \abs{\log h(f_
{G_j, \nu_i}, f_{G_0, \nu_i})}}h^2(f_
{G_j, \nu_i}, f_{G_0, \nu_i}) \tag{Hellinger distance is invariant to
change of variables} \end{align*}
Let $h_i = h(f_
{G_j, \nu_i}, f_{G_0, \nu_i})$. Hence, \begin{align*}
\frac{1}{n} \E[\zeta_4^2] &\leH \frac{(\log n)^3}{n} \sum_{i : \abs{\log h_i} \le 
(\log 1/\rho_n)^3} h_i^2 +
\frac{1}{n} \sum_{i : \abs{\log h_i} > (\log 1/\rho_n)^3} \abs{\log h_i} h_i^2 \\
&\le (\log n)^3 \barh^2(\fs{G_j},
\fs{G_0}) + \frac{1}{n} \sum_{i : \abs{\log h_i} >(\log 1/\rho_n)^3} \frac{1}{e} h_i 
\tag{$x\abs{\log x}
\le e^{-1}$ for $x\in[0,1]$}
\end{align*}
Note that \[ \abs{\log h_i} > (\log 1/\rho_n)^3 \implies h_i < \exp\pr{-\log(1/\rho_n)^3}
<
\rho_n^{(\log 1/\rho_n)^2} \leH \rho_n^3 \leH n^{-1} \tag{\cref{as:Delta_M_rate}}. \]
Therefore the first term dominates, and thus $
\frac{1}{n} \E[\zeta_4^2] \leH (\log n)^3 \delta_n^2.
$

\subsection{Auxiliary lemmas}

\begin{lemma}
\label{lemma:distance_bound_hard}

Let $\hat \theta_{i, \hat G, \hateta}$ be the posterior mean at prior $\hat G$ and
conditional moments estimate at $\hateta$. Let $\theta_i^* = \hat\theta_{i, G_0, \eta_0}$
be the oracle posterior mean. Assume that $\hat G$ is supported within $[-\bar M_n, \bar
M_n]$ where $\bar M_n = \max_i |\hat Z_i(\hateta) \vee 1|$. Recall that $\Dinfty =
\max(\norm{\hatm-m_0}_\infty,
\norm{\hats - s_0}_\infty)$. Suppose
\begin{enumerate}
  \item $\Dinfty \leH 1$;
  \item \Cref{as:moments,as:variance_bounds} holds;
  \item $\hats \gtrsim_{\Hyperparams} s_{\ell n}$ for some fixed sequence $s_{\ell n} > 0$.
\end{enumerate}
Then, letting $\bar Z_n = \max_i |Z_i| \vee 1$, \[
\absauto{\hat \theta_{i, \hat G, \hateta} - \theta_i^*} \leH \bar s_{\ell n}^{-1} \bar Z_n
.\]
Moreover, the assumptions are satisfied by
\cref{as:npmle,as:holder,as:moments,as:variance_bounds}
with $s_{\ell n} = s_{0\ell} \rateeq 1$.
\end{lemma}
\begin{proof}
Observe that \begin{align*}
\absauto{
    \hat\theta_{i, \hat G_n, \hateta} - \hat\theta_{i, G_0, \eta_0}
} &=  \absauto{
    \hat s_i \frac{\hat\nu_i^2 f'_{\hat G_n, \hat\nu_i}(\hat Z_i)}{ f_{\hat G_n,
    \hat\nu_i}(\hat Z_i)} - s_{0i}
    \frac{\nu_i^2 f'_{G_0 \nu_i}(Z_i)}{f_{G_0, \nu_i}(Z_i)}
} \\
&\leH  \bar M_n + \bar Z_n \leH \max_i \max(|\hat Z_i|, |Z_i|, 1)
\end{align*}
by the boundedness of $\hat G_n$ and 
\cref{lemma:posterior_moments}.
Note that $|\hat Z_i(\hateta)| = \absauto{\frac{s_{0i}}{\hats_i} Z_i + \frac{m_{0i} -
\hatm_i}
{\hats_i}} \leH s_{\ell n}^{-1} \max(|Z_i|, 1) = s_{\ell n}^{-1} \bar Z_n.$  Therefore, $
\abs{
    \hat\theta_{i, \hat G_n, \hateta} - \hat\theta_{i, G_0, \eta_0}
} \leH s_{\ell n}^{-1} \bar Z_n
$.
\end{proof}

\begin{lemma}
\label{lemma:tail_bound_max}
Let $\bar Z_n = \max_i |Z_i| \vee 1$. Under \cref{as:moments}, for $t > 1$ \[
\P(\bar Z_n > t) \le n\exp\pr{-C_{A_0, \alpha, \nu_u} t^\alpha}
\quad \text{ and } \quad
\E[\bar Z_n^p] \lesssim_{p, \Hyperparams} (\log n)^{p/\alpha}.
\]
Moreover, if $M_n = (C_\Hyperparams + 1) (C_{2,\Hyperparams}^{-1} \log n)^{1/\alpha}$ as
in
\eqref{eq:DMrate}, then for all sufficiently large choices of $C_\Hyperparams$, $\P(\bar
Z_n > M_n) \le n^{-2}.$
\end{lemma}
\begin{proof}
The first claim is immediate under \cref{lemma:tail_bound} and a union bound. The second
claim follows from the observation that \[
\E[\max_i (|Z_i| \vee 1)^p] \le \pr{\sum_{i=1}^n \E[(|Z_i| \vee 1)^{pc}]}^{1/c} \le n^
{1/c}
C_\Hyperparams^{p} (pc)^{p/\alpha}.
\]
where the last inequality follows from simultaneous moment control. Choose $c = \log n$
with $n^{1/\log n} = e$ to finish the proof. For the ``moreover'' part, we have that \[
\P(Z_n > M_n) \le \exp\pr{
    \log n - C_{A_0, \alpha, \nu_u} (C_{\Hyperparams} + 1)^{\alpha} C_{2,\Hyperparams}^
    {-1} \log n
}
\]
and it suffices to choose $C_{\Hyperparams}$ such that $(C_{\Hyperparams} + 1)^{\alpha} >
\frac{3 C_{2,\Hyperparams} }{C_{A_0, \alpha, \nu_u}}$ so that $\P(Z_n > M_n) \le e^
{-2\log n} = n^{-2}$.
\end{proof}

\begin{lemma}
\label{lemma:norm_concentration}
Let $W = (W_1,\ldots, W_n)$ be a vector containing independent entries, where $W_i \in
[0,1]$. Let $\norm{\cdot}$ be the Euclidean norm. Then, for all $t > 0$ \[
\P\bk{\norm{W} > \E\norm{W} + t} \le e^{-t^2/2}.
\]
\end{lemma}
\begin{proof}
We wish to use Theorem 6.10 of \citet{boucheron2013concentration}, which is a
dimension-free concentration inequality for convex Lipschitz functions of bounded random
variables. To do so, we observe that $w \mapsto \norm{w}$ is Lipschitz with respect to
$\norm{\cdot}$, since \[
\norm{w + a} \le \norm{w} + \norm{a} \quad \norm{w} = \norm{w+a-a} \le \norm{w+a} +
\norm{a} \implies |\norm{w+a} - \norm{w}| \le \norm{a}.
\]
Moreover, trivially $\norm{\lambda w + (1-\lambda) v} \le \lambda\norm{w} + (1-\lambda)
\norm{v}$ for $\lambda \in [0,1]$, and hence $w\mapsto \norm{w}$ is convex. Convexity
implies the convexity required in Theorem 6.10 of \citet{boucheron2013concentration}. This
checks all conditions and the claim follows by applying Theorem 6.10 of
\citet{boucheron2013concentration}.
\end{proof}

\begin{lemma}
\label{lemma:misspec_prior}
Let $f_{H} = f_{H,1}$. Then, for $0 < \rho_n \le \frac{1}{\sqrt{2\pi e^2}}$, \begin{align*}
\int \bk{ \frac{f'_{H_1}(x)}{f_{H_1}(x) \vee \rho_n} - \frac{f'_{H_0}
    (x)}{f_{H_0}(x) \vee \rho_n} }^2 f_{H_0}(x)\,dx 
    \lesssim \pr{(\log 1/\rho_n)^{3} \vee
    \abs{\log h \pr{f_{H_1}, f_{H_0}}}} h^2 \pr{f_{H_1}, f_{H_0}}
\end{align*}
where we define the right-hand side to be zero if $H_1 = H_0$.
\end{lemma}

\begin{proof}
This claim is an intermediate step of Theorem 3 of \citet{jiang2009general}. In (3.10) in
\citet{jiang2009general}, the left-hand side of this claim is defined as $r^2(f_{H_1},
\rho_n)$. Their subsequent calculation, which involves Lemma 1 of \citet{jiang2009general},
proceeds to bound \[ r(f_{H_1}, \rho_n) \le 2\sqrt{2} e h(f_{H_1}, f_{H_0}) \max
\pr{\invphi^3(\rho_n), \sqrt{2}a} + 2\invphi(\rho_n) \sqrt{2} h(f_{H_1}, f_{H_0}),
\]
for $a^2 = \max\pr{
    \invphi^2(\rho_n) + 1, \abs{\log h^2\pr{f_{H_1}, f_{H_0}}} }$.
    Collecting
     the powers on $h, \log h$, squaring, and using $\invphi(\rho_n) \lesssim
\sqrt{\log (1/\rho_n)}$ proves the claim.
\end{proof}

\newpage

\part{Additional discussions and empirical results}

\section{Additional discussions}

\subsection{Correlated $Y_i$ and correlated $\theta_i$}
\label{rmk:independence}

The assumptions \eqref{eq:gaussian_heteroskedastic_location} and \eqref{eq:eb_sampling}
imply that $(Y_i, \theta_i, \sigma_i)$ are i.i.d. across $i$. In general, we may consider
a joint distribution on $\theta_{1:n} \mid \Sigma$, conditional on which the estimates may
also be correlated $Y_ {1:n} \mid \theta_ {1:n} \sim \Norm(\theta_{1:n}, \Sigma)$ \citep
[see, among others,] [for settings in which the non-independence appears natural] 
{bonhommedenis,muller2022spatial}.
In
principle, given a flexible model for the distribution $\theta_{1:n} \mid \Sigma$, we
would estimate this model from the data $(Y_{1:n}, \Sigma)$, and likewise compute
an estimated posterior. In particular, if $\theta_{1:n} \mid \Sigma$ is described by
\eqref{eq:location_scale} but $Y_ {1:n}$ may be correlated conditional on $\theta_{1:n},
\Sigma$, then \cref{item:close1,item:close2} continue to estimate $\theta_{1:n} \mid
\Sigma$, and one can adapt \cref{item:close3} accordingly to exploit the correlation
 between $Y_i$'s.

Modeling the joint distribution of $(Y_{1:n}, \theta_{1:n}, \Sigma)$ may be difficult.
Interestingly, modeling only the marginal distributions of each coordinate $ (Y_i,
\theta_i, \sigma_i)$---as we have done---turns out to be sufficient for optimal
decision-making, if we restrict the class of decision rules. 
For a compound decision
problem, where $L (\bdelta, \theta_{1:n}) = \frac {1}{n} \sum_{i=1}^n \ell(\delta_i,
\theta_i)$, consider restricting to \emph{separable} decision rules $\delta_i =
\delta_i(Y_i, \sigma_i)$, where we make decisions about $\theta_i$ using only $(Y_i,
\sigma_i)$. In that case, the best decision $\delta_i(\cdot, \cdot)$ minimizes the
posterior risk $\E[\ell(\delta(Y_i, \sigma_i),
\theta_i) \mid Y_i, \sigma_i]$. This result provides some reassurance that ignoring the
 dependence across coordinates is inefficient but may not be harmful, since naive
 decisions are often separable, and the optimal separable rule would dominate it.
 Restricting to separable decision rules can also be motivated by practical or fairness
 considerations: For instance, it may be unfair to base human resource decisions about a
 given teacher on the value-added estimate of another.

We formalize the above paragraph in \cref{lemma:correlation}. Suppose $(Y_ {1:n},
\theta_{1:n},
\Sigma)$ follow some joint distribution $Q_0$ under which
$Y_{1:n} \mid \theta_{1:n}, \Sigma \sim \Norm(\theta_{1:n}, \Sigma)$. We observe $(Y_
{1:n}, \Sigma)$. Consider a compound decision problem with a separable decision rule,
such that \[
L(\bdelta, \theta_{1:n}) = \frac{1}{n} \sum_{i=1}^n \ell(\delta_i(Y_i, \sigma_i),
\theta_i).
\]
Let $\sigma_i^2$ denote the $i$\th{} diagonal element of $\Sigma$. 

\begin{lemma}
\label{lemma:correlation}
Let $(Y_i, \theta_i, \sigma_i^2) \sim H_{0i}$ under $Q_0$. Assume $Q_0$ is such that the
posterior mean
$\E_{H_{0i}} \bk{
  \ell(a, \theta_i) \mid Y_i, \sigma_i 
}$ exists almost surely. 
Consider the decision rule $\bdelta$ that minimizes posterior risk under $\theta_i
\mid Y_i, \sigma_i$ over the action space\[
\delta_i(y, \sigma) \in \argmin_a \E_{H_{0i}} \bk{
  \ell(a, \theta_i) \mid Y_i = y, \sigma_i = \sigma
}.
\]
Then such a decision is optimal for Bayes risk under $Q_0$:\[
\E_{Q_0}[L(\bdelta, \theta_{1:n})] = \min_{\tilde {\bdelta} \text{ separable}} \E_{Q_0}[L
(\tilde {\bdelta}, \theta_{1:n})]. 
\]
\end{lemma}

\begin{proof}
By definition, for any separable $\tilde{\bdelta}$ \begin{align*}
\E_{Q_0}\bk{
  \ell(\delta_i(Y_i, \sigma_i), \theta_i)
} &= \E_{H_{0i}}\bk{
  \ell(\delta_i(Y_i, \sigma_i), \theta_i)
}  \\
&= \E\bk{
  \E\bk{\ell(\delta_i(Y_i, \sigma_i), \theta_i) \mid Y_i, \sigma_i}
}  \tag{Law of iterated expectations}\\
&\le \E\bk{
  \E\bk{\ell(\tilde\delta_i(Y_i, \sigma_i), \theta_i) \mid Y_i, \sigma_i}
} \tag{$\delta_i$ is Bayes rule} \\
&= \E_{Q_0}\bk{
  \ell(\tilde\delta_i(Y_i, \sigma_i), \theta_i)
}.
\end{align*}
Thus by linearity of expectation, \[
\E_{Q_0}[L(\bdelta, \theta_{1:n})] \le \E_{Q_0}[L(\tilde{\bdelta}, \theta_{1:n})].\qedhere
\]
\end{proof}

\subsection{Alternatives to \close}
\label{sub:alt_methods}

Let us turn to a few specific alternative methods that consider failure of precision
independence. We argue that they do not provide a free-lunch improvement over our
assumptions.

\begin{altsq}[Working with $t$-ratios]
We may consider normalizing $\sigma_i$ away by working with
$t$-ratios $
T_i \equiv \frac{Y_i}{\sigma_i} \mid (\sigma_i, \theta_i) \sim \Norm
\pr{\theta_i/\sigma_i, 1}.$ The resulting problem is homoskedastic by construction. It is
natural to consider performing empirical Bayes shrinkage assuming that
$\frac{\theta_i}{\sigma_i} \iid H_0$, and use, say, $\sigma_i \PE_{\hat H_n}\bk{\frac
{\theta_i}{\sigma_i} \mid T_i}$ as an estimator for the posterior mean of $\theta_i$
\citep{jiang2010empirical}. However, without imposing $\theta_i/\sigma_i \indep
\sigma_i$ (which we discuss in \cref{rmk:alts_to_close}), such an approach approximates
the optimal decision rule within a restricted class on a different objective.

Let us restrict decision rules to those of the form $\delta_{i, \text{t-stat}}(Y_i,
\sigma_i) = \sigma_i h (Y_i/\sigma_i)$. The oracle Bayes choice of $h$ is $h^\star(T_i)
=\frac{\E\bk{\sigma_i \theta_i \mid T_i}}{
\E[\sigma_i^2 \mid T_i]}$. However, $h^\star$ is {not} the posterior mean of
$\theta_i/\sigma_i$ given the $t$-ratio $T_i$, unless $\sigma_i^2 \indep
\theta_i/\sigma_i$.
On the other hand, the loss function that does rationalize the posterior mean $h(T_i) =
\E[\theta_i/\sigma_i \mid T_i]$ is the precision-weighted compound loss $L(\bdelta,
\theta_{1:n}) = \frac{1}{n} \sum_{i=1}^n
\sigma_i^{-2} (\delta_i - \theta_i)^2$. Thus, rescaling posterior means on $t$-ratios
achieves optimality for a weighted objective among a restricted class of decision rules
$\delta_{i, \text{t-stat}}$.
\end{altsq}

\begin{altsq}[Variance-stabilizing transforms]

Second, we may consider a variance-stabilizing transform when the underlying micro-data are
Bernoulli and $\theta_i$ is a Bernoulli mean
\citep{efron1975data,10.1214/07-AOAS138}. Specifically, we rely on the asymptotic
 approximation \[
\sqrt{n_i} (Y_i - \theta_i)\underset{n_i\to\infty}{\dto} \Norm(0, \theta_i (1-\theta_i)).
\]
A variance-stabilizing transform can disentangle the dependence: Let
$W_i = 2\arcsin(\sqrt{Y_i})$ and $\omega_i = 2\arcsin(\sqrt{\theta_i})$, and, by the delta
method, \[
\sqrt{n_i} \pr{W_i - \omega_i} \underset{n_i\to\infty}{\dto} \Norm(0, 1). \quad \text{
Thus, approximately, } W_i \mid \omega_i, n_i \sim \Norm\pr{\omega_i, \frac{1}{n_i}}.
\] One might consider an empirical Bayes approach on the resulting $W_i$. Note that $W_i$
 may still violate precision independence, since $\omega_i$ may not be independent of
 $n_i$.
 Moreover, squared error loss on estimating $\omega_i = 2\arcsin(\sqrt{\theta_i})$ is
 different from squared error loss on estimating $\theta_i$. We do not know of any
 guarantees for the loss function on $\theta_i$, $\frac{1}{n} \sum_ {i=1}^n
 (\delta_i - \sin^2 (\omega_i / 2))^2$, when we perform empirical Bayes analysis on
 $\omega_i$.
\end{altsq}

\begin{altsq}[Treating the standard error as estimated]
\label{alt:gukoenker}
Lastly, if the researcher has access to micro-data,
\citet{gu2017empirical} and \citet{fu2020nonparametric} propose empirical Bayes strategies
that treat $\sigma_i$ as noisy as well, in which we know the likelihood of $(Y_i,
\sigma_i)$. This approach allows for dependence between $\theta_i$ and $\sigma_i$ but
assumes independence between $(\theta_i, \sigma_i)$ and some other known 
 parameter. To describe their model, we introduce more notation. Let $Y_
 {ij}, j=1,\ldots,n_i,$ denote the micro-data for population $i$, where, for each $i$, we
 are interested in the mean of $Y_{ij}$. Let $Y_i$ denote their sample mean and $S_i^2$
 denote their sample variance, where $\sigma_i^2 = S_i^2/n_i$. Let $\sigma_{i0}^2$ denote
 the true variance of observations from population $i$.

Both papers work under Gaussian assumptions on the micro-data. This parametric
assumption\footnote{The parametric restriction on the micro-data $Y_{ij}$ can be relaxed
by appealing to the asymptotic distribution of $(Y_i, S_i^2)$---resulting in the Gaussian
likelihood $(Y_i, S_i^2) \mid \btheta_i, \Sigma_i \sim \Norm(\btheta_i, \Sigma_i)$. In
general, however, $\Sigma_i$ also depends on $n_i$ and higher moments of $Y_{ij}$, which
again may not be independent of $\btheta_i$.} on the micro-data---which is stronger than
we require---implies that $Y_i \indep S_i^2 \mid (\sigma_{i0}, \theta_i, n_i)$ with
marginal distributions:
\begin{align*}
Y_i \mid \sigma_{i0}, \theta_i, n_i  \sim \Norm
\pr{\theta_i, \frac{\sigma_{i0}^2}{n_i}
}
\qquad
S_{i}^2 \mid \sigma_{i0}, \theta_i, n_i \sim
\Gammadist\pr{
  \frac{n_i - 1}{2}, \frac{1}{2\sigma_{i0}^2}
}.
  \end{align*}
They then propose empirical Bayes methods treating $\bY_i \equiv(Y_i, S_i^2)$
 as noisy estimates for parameters $\btheta_i \equiv (\theta_i, \sigma_{i0}^2)$. This
 formulation allows $\btheta_i$ to have a flexible distribution, and thus allows for
 dependence between $\theta_i$ and $\sigma_{i0}^2$. However, since the known sample size
 $n_i$ enters the likelihood of $\bY_i$, this approach still assumes that
 $n_i \indep \btheta_i$.
\end{altsq}

\begin{altsq}[$f$-modeling]
\label{alt:f-modeling}
A final alternative is to exploit Tweedie's formula \citep{efron2022exponential}, which
  implies that an estimate of the conditional distribution $Y_i \mid
 \sigma_i$ is all one needs for computing the posterior means
 \citep{brown2009nonparametric,liu2020forecasting,luo2023empirical}. However, conditional
 density estimation is a challenging problem, and most available methods do not exploit
 the restriction that $Y_i
\mid
\sigma_i$ is a Gaussian convolution. 
\end{altsq}

This discussion is not to say that \close{} is necessarily preferable to these alternatives.
It highlights that the possible dependence between $\theta_i$ and $\sigma_i$ cannot be
easily resolved. Existing alternatives compromise
on optimality, use a different loss function, or implicitly assume $\theta_i$ is
independent from components of $\sigma_i^2$ (e.g., $n_i$). Of course, depending on the
empirical context, these may well be reasonable features.

\subsection{Transformation-based rationalization of the location-scale assumption
\eqref{eq:location_scale}}
\label{sub:model_free}

The following lemma shows that, essentially, only affine transforms preserve exponential
family structure on $Y_i$. Exponential family structure is important since
generalizations of Tweedie's formula holds for such distributions \citep
{efron2011tweedie}, and thus they connect posterior means to the marginal distribution of
the data. If some affine transform yields precision independence---so as to allow for
methods
that assume precision independence---then the location-scale assumption \eqref
{eq:location_scale} must hold. 

\begin{lemma}
\label{lemma:transform}
Let $Y_i \mid \theta_i, \sigma_i \sim \Norm(\theta_i, \sigma_i^2)$ and $(\theta_i,
\sigma_i)$ drawn i.i.d. Consider $h (Y_i,
\sigma_i)$ such that $h(\cdot, \sigma_i)$ is differentiable and strictly increasing. Let
$Z_i = h(Y_i, \sigma_i)$. Then the distribution of $Z_i$, parametrized by $\theta_i$, 
is an exponential family \[
p(z \mid \theta_i, \sigma_i) = \exp\pr{
  \eta(\theta_i, \sigma_i) z + A(\eta(\theta_i, \sigma_i), \sigma_i) 
} f_0(z, \sigma_i) \numberthis \label{eq:expofam_full}
\]
only if $h(Y_i, \sigma_i) = a(\sigma_i) + b(\sigma_i) Y_i$. 

Moreover, suppose $\theta_i \mid \sigma_i$ has finite first and second moments. When 
\eqref{eq:expofam_full} holds for some $h(Y_i, \sigma_i) = a(\sigma_i)+b(\sigma_i) Y_i$,
the distribution of $Z_i = h(Y_i, \sigma_i)$ is of the form \[
Z_i \mid \theta_i, \sigma_i \sim \Norm(\tau(\theta_i, \sigma_i), \nu^2(\sigma_i)). 
\]
The distribution of $\tau(\theta_i, \sigma_i) \mid \sigma_i$ does not depend on $\sigma_i$
only if \eqref{eq:location_scale} holds.
\end{lemma}

\begin{proof}
The first part of the statement follows immediately from \cref{lemma:transform2}. For the
second part, it is easy to see that $
\tau_i \equiv \tau(\theta_i, \sigma_i) = \theta_i b(\sigma_i) + a(\sigma_i)
$.
If $\tau(\theta_i, \sigma_i) \mid \sigma_i$ does not depend on $\sigma_i$, then $
\E[\tau \mid \sigma_i], \var(\tau\mid\sigma_i)
$
are constant in $\sigma_i$. This means that $b(\sigma_i) = 1/\sqrt{\var(\theta_i \mid
\sigma_i)}$ and $a(\sigma_i) = -\E[\theta_i \mid \sigma_i] b(\sigma_i)$. Rearrange, we
have \[
\theta_i = s_0(\sigma_i) \tau_i + m_0(\sigma_i) \quad \tau_i \mid \sigma_i \iid G_0,
\]
which is the condition \eqref{eq:location_scale}.
\end{proof}

\begin{lemma} 
\label{lemma:transform2}
Consider a family of densities $\br{\Norm
 (\theta, \sigma^2): \theta \in \R}$. Let $h(\cdot): \R \to \R$ be a strictly increasing
 and
 differentiable function. Let $Z = h(Y)$ where $Y
\sim \Norm
(\theta, \sigma^2)$. The corresponding family of densities for $Z$ is an exponential
family: \[
p_Z(z) = f_0(z; \sigma) \exp\pr{
  z\eta(\theta, \sigma) + A(\eta; \sigma)
} \numberthis \label{eq:expofam}
\]
for some canonical parameter $\eta(\theta; \sigma)$
if and only if $h(\cdot)$ is affine. 
\end{lemma}

\begin{proof}
The ``if'' part is immediate. We will focus on the ``only if'' part. By the change of
variables formula, the density of $Z$ is \[
p_Z(z) = \frac{1}{\sqrt{2\pi} \sigma}\exp\pr{
  -\frac{1}{2} \frac{h^{-1}(z)^2}{\sigma^2} + h^{-1}(z) \frac{\theta}{\sigma^2} - 
  \frac{\theta^2}
  {2 \sigma^2}
} \frac{dh^{-1}(z)}{dz}.
\]
The log-likelihood ratio of this family is \[
\log \frac{p_Z(z \mid \theta_1)}{p_Z(z \mid \theta_2)} = {
  h^{-1}(z) \frac{\theta_1 - \theta_2}{\sigma^2} - \frac{1}{2\sigma^2} (\theta_1^2 -
  \theta_2^2)
}.
\]
For an exponential family \eqref{eq:expofam}, the log-likelihood ratio is $
z \pr{\eta_1 - \eta_2} + A(\eta_1; \sigma) - A
(\eta_2;\sigma),
$
where $\eta_j = \eta(\theta_j; \sigma)$. Equating the two and differentiating in $z$, we
have that \[
\frac{dh^{-1}(z)}{dz} \frac{\theta_1 - \theta_2}{\sigma^2} = \eta_1 - \eta_2
\]
for all $\theta_1, \theta_2 \in \R$. Since the right-hand side is free of $z$, we conclude
that $\frac{dh^{-1}(z)}{dz}$ must be a constant. Thus $h^{-1}$---and hence $h$---is
affine.
\end{proof}

\section{Additional empirical exercises}

\subsection{Positivity of $s_0(\cdot)$ in the Opportunity Atlas data}
\label{asub:variance_right_tail}

\begin{figure}[htb]
  \centering
  \includegraphics[width=0.8\textwidth]{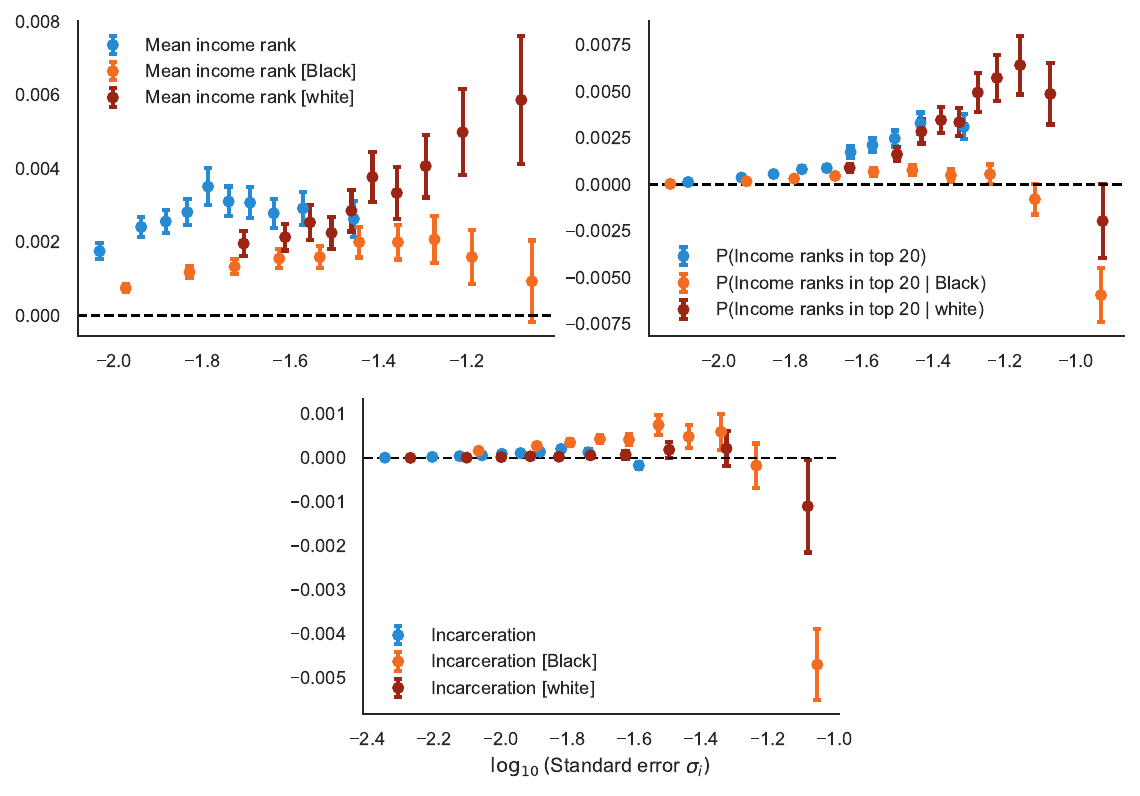}
  \caption{Estimated conditional variance $s_0^2(\sigma)$, binned into deciles, with 95\% uniform confidence intervals shown.}
  \label{fig:variance_right_tail}
\end{figure}

In the Opportunity Atlas data, we often observe that the estimated conditional variance is
negative: $\hat s_0^2 < 0$. To test if this is due to sampling variation or underdispersion
of the Opportunity Atlas estimates relative to the estimated standard error, we consider
the following upward-biased estimator of $s_0^2(\sigma_i)$. Without loss, let us sort the
$Y_i, \sigma_i$ by $\sigma_i$, where $\sigma_1 \le \cdots \le \sigma_n$. Let $S_i = \frac
{1}{2} \bk{(Y_{i+1} - Y_i)^2 -(\sigma_i^2 + \sigma_{i+1}^2)}$. Note that \begin{align*}
\E[S_i \mid \sigma_{1:n}] &= \frac{1}{2}\E[(\theta_{i+1} - \theta_{i})^2 \mid \sigma_
 {1:n}] \\
 &= \frac{s_0^2(\sigma_{i+1}) + s_0^2(\sigma_{i})}{2} + \frac{1}{2} (m_0(\sigma_
 {i+1}) - m_0(\sigma_i))^2 \ge \frac{s_0^2(\sigma_{i+1}) + s_0^2(\sigma_{i})}{2}.
\end{align*} Hence $S_i$ is an overestimate of the successive averages of $s_0
 (\sigma)$. \Cref{fig:variance_right_tail} plot the estimated conditional expectation of
 $S_i$ given $\sigma_i$, using a sample of $(S_1, S_3, S_5, \ldots)$ so that the $S_i$'s
 used are mutually independent. We see in \cref{fig:variance_right_tail} that for many
 measures of economic mobility, we can reject $\E[S_i \mid \sigma_i] \ge 0$, indicating
 some underdispersion in the data.

\subsection{Simulation exercises setup}

\label{asub:sim_setup}

This section describes the details of the simulation exercises in \cref{sec:empirical}. 

We first consider details on data preprocessing:
\begin{enumerate}[wide]
  \item The data used is the publicly available tract-level data from 
  \citet{chetty2018opportunity}.
  \item Limit to Census tracts in the 20 largest Commuting Zones, ranked by the number of tracts in the dataset. They are: Phoenix, San Francisco, Los
Angeles, Bridgeport, Washington DC, Miami, Tampa, Atlanta, Chicago, Boston, Detroit,
Minneapolis, Philadelphia, Newark, New York, Cleveland, Pittsburgh, Houston, Dallas, and 
Seattle.
  \item For a given outcome variable, we truncate to tracts with $\sigma_i^2$ in the
  bottom 99.5\% with all available covariates in \cref{fn:covariates}. \Cref
  {tab:sample_sizes} displays the sample sizes used.
\end{enumerate}

\begin{table}[tb]
  \caption{Number of tracts included for each outcome variable}
  \label{tab:sample_sizes}
  \centering

\footnotesize
  \begin{tabular}{lr}
\toprule
{} &  Sample size \\
\midrule
Mean income rank                       &        10056 \\
Mean income rank [white male]          &         7521 \\
Mean income rank [Black male]          &         7547 \\
Mean income rank [Black]               &        10056 \\
Mean income rank [white]               &         8138 \\
Incarceration [Black male]             &         6634 \\
Incarceration [white male]             &         7308 \\
Incarceration [Black]                  &         9205 \\
Incarceration [white]                  &         7968 \\
Incarceration                          &        10056 \\
P(Income ranks in top 20 $|$ Black male) &         7547 \\
P(Income ranks in top 20 $|$ white male) &         7521 \\
P(Income ranks in top 20 $|$ Black)      &        10056 \\
P(Income ranks in top 20 $|$ white)      &         8138 \\
P(Income ranks in top 20)              &        10056 \\
\bottomrule
\end{tabular}

\end{table}

The covariates used are listed in \cref{fn:covariates}. The ``number of children''
variables are included in both levels and logs. This set of covariates is not precisely the
same as what is used in \citet{bergman2019creating}. \citet{bergman2019creating}
additionally use economic mobility estimates for a later birth cohort, which are not
included in the publicly released version of the Opportunity Atlas. The ``number of
children'' variables are used by \citep{chetty2018opportunity} as a population weighting
variable; they contain some information on the implicit micro-data sample sizes $n_i$.

\begin{table}[tb]
  \caption{Covariates and corresponding variable labels from \citet{chetty2018opportunity}}
  \label{fn:covariates}
  \centering
\scriptsize
\begin{tabularx}{\textwidth}{YY}
\toprule
\textbf{Covariate description} & \textbf{Variable label} \\ \midrule
Poverty rate in 2010 & \texttt{poor\_share2010} \\ \hline
Share of Black individuals in 2010 & \texttt{share\_black2010} \\ \hline
Mean household income in 2000 & \texttt{hhinc\_mean2000} \\ \hline
Log wage growth for high school graduates & \texttt{ln\_wage\_growth\_hs\_grad} \\ \hline
Fraction with college or post-graduate degrees in 2010 & \texttt{frac\_coll\_plus2010} \\ \hline
Mean parent family income rank & \texttt{par\_rank\_pooled\_pooled\_mean} \\ \hline
Mean parent family income rank for Black individuals & \texttt{par\_rank\_black\_pooled\_mean} \\ \hline
Number of all children under 18 with parents whose household income is below median in
2000 & \texttt{kid\_pooled\_pooled\_blw\_p50\_n} \\ \hline
Number of Black children under 18 with parents whose household income is below median in
2000 & \texttt{kid\_black\_pooled\_blw\_p50\_n} \\ \bottomrule
\end{tabularx}

\begin{proof}[Notes]
This table links the covariates to their codebook labels in \citet{chetty2018opportunity}.
See their
\href{https://opportunityinsights.org/wp-content/uploads/2019/07/Codebook-for-Table-9.pdf}
 {Codebook for Table 9} and 
\href{https://opportunityinsights.org/wp-content/uploads/2019/07/Codebook-for-Table-4.pdf}
 {Codebook for Table 4} for the corresponding precise definitions of each covariate
 \citep{opportunity_insights_data}.
\end{proof}

\end{table}

We now describe the data-generating process for the calibrated simulation exercise. For a
given outcome variable,

\begin{enumerate}[label=(CS-\arabic*),wide]
  \item Let $\tilde Y_i$ denote the raw estimates. 
Residualize $\tilde Y_i$ against some covariates $X_i$ to obtain $\beta$ and residuals $Y_i$. 
\item Estimate the conditional moments $m_0, s_0$ on $(Y_i,\sigma_i)$ via local linear
  regression, described in \cref{sec:nuisance_estimation}.

  \item \label{item:npmle-within-bins} Partition $\sigma$ into vingtiles. Within each
   vingtile $j$, estimate an \npmle{} $G_j$ over the data in that vingtile \[\pr{\frac{Y_i
   - m_0
   (\sigma_i)}{s_0(\sigma_i)}, \frac{\sigma_i}{s_0(\sigma_i)}}.\]
   \item  Normalize $G_j$ to
   have zero mean and unit variance by shifting its support and dividing the support by
   its standard deviation. 
   
   \item \label{item:npmle-within-bins-final} Sample $\tau_i^* \mid \sigma_i \sim G_ {j}$
   if observation $i$
   falls within vingtile $j$.

  \item Let $\vartheta_i^* = s_0(\sigma_i) \tau_i^* + m_0(\sigma_i) + \beta ' X_i$ and
  let $\tilde Y_i^* \mid \vartheta_i^*, \sigma_i \sim \Norm(\vartheta_i^*, \sigma_i^2)$.
\end{enumerate}

  The estimated
  $\beta, m_0, s_0$ will serve as
  the basis for the true data-generating process in the simulation, and as a result we do
  not denote it with hats. 
\Cref {fig:example_sim} shows an overlay
of real and simulated data for one of the variables we consider. Visually, at least,  the
simulated data resemble the real estimates.

\begin{figure}[htb]
  \centering
  \includegraphics[width=0.8\textwidth]{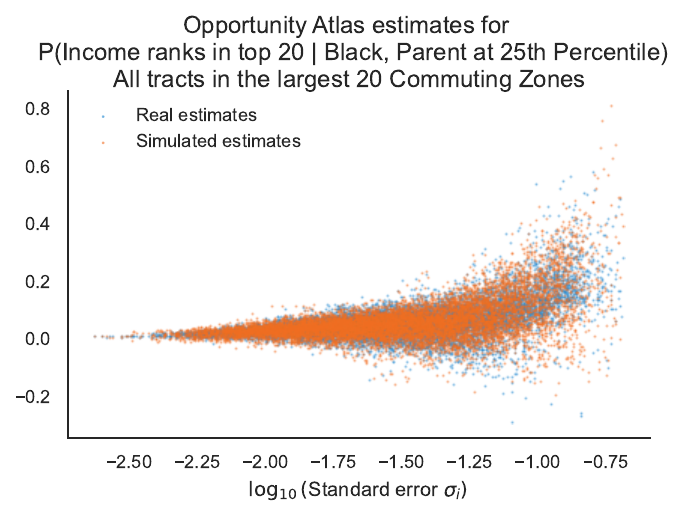}
  \caption{A draw of real vs. simulated data for estimates of \topprob{} for Black
  individuals}
  \label{fig:example_sim}
\end{figure}

Finally, we describe details for the policy exercise. Fix a given outcome variable:
\begin{enumerate}[wide,label=(CB-\arabic*)]
  \item Let $\tilde Y_i, \sigma_i$ denote the observed data. 
  \item Define $\tilde Y_{i,1} = Y_i + \frac{1}{3} \sigma_i W_i$ and $\tilde Y_{i,2} =
  Y_i - 3 \sigma_i W_i$ as the coupled bootstrap draws, for $\omega= 1/9$.
  Correspondingly, let $\sigma_{i,1} = \sqrt{10/9} \sigma_i$ and $\sigma_{i,2} = 
  \sqrt{10} \sigma_i$.
\end{enumerate}
For the policy exercise, we separate the units into Commuting Zones (CZs). For each CZ
separately, we treat the data as $(\tilde Y_{i,1}, \sigma_{i,1}, X_i)$ within that CZ. We
compute decision rules by only using data within the CZ (including the residualization by
covariates). For the selection exercise, we select by top third within each CZ.

\subsection{Different simulation setup}
\label{asub:weibull}

\begin{figure}[htb]
  \centering
  \includegraphics[width=\textwidth]{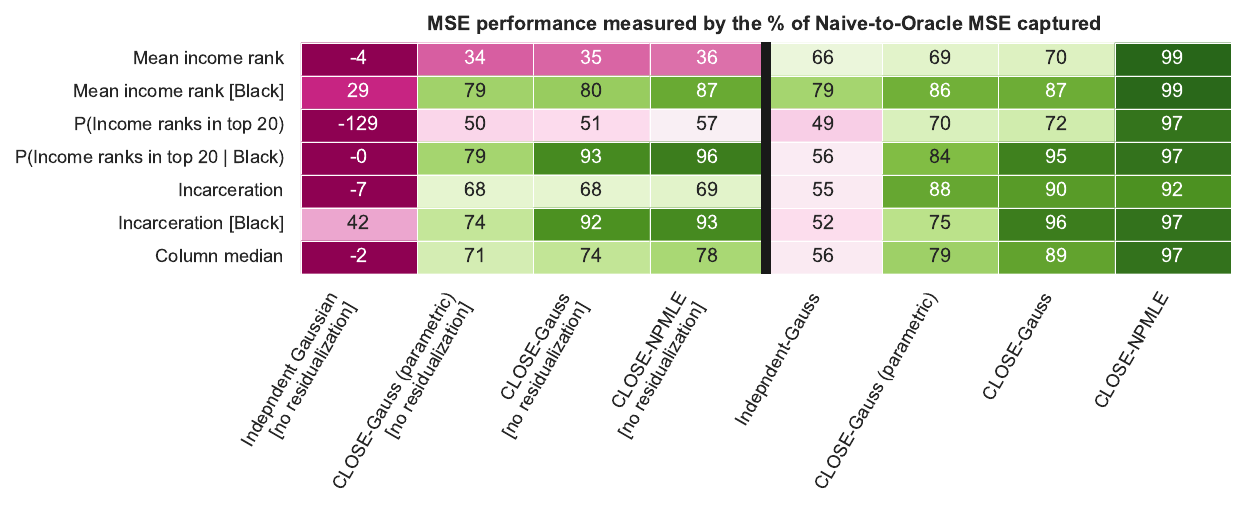}
  \caption{Analogue of \cref{fig:mse_table} for the data-generating process in \cref{asub:weibull}. Here the results average over 100 replications.}
  \label{fig:mse_calibrated_weibull}
\end{figure}

We have also conducted a Monte Carlo exercise where we replace
\cref{item:npmle-within-bins}--\cref{item:npmle-within-bins-final}
with the following step:
\begin{itemize}[wide]
  \item For each $\sigma_i$, let \[
\alpha_i = \frac{1}{2} + \frac{1}{2} \frac{m_0(\sigma_i) - \min_i m_0(\sigma_i)}{\max_i
m_0
(\sigma_i) - \min_i
(\sigma_i)}
\in [1/2,1]
  \]
  We sample $\tau_i^* \mid \sigma_i$ as a scaled and translated $\mathrm{Weibull}$
  distribution with shape $\alpha_i$. The scaling and translation ensures that $\tau_i
  \mid \sigma_i$ has mean zero and variance one. Because we choose the Weibull
  distribution, the shape parameter $\alpha_i$ corresponds exactly to $\alpha$ in
  \cref{as:moments}. Our choices of $\alpha_i$ implies that $\tau_i \mid
  \sigma_i$ has thicker tails than exponential and does not have a moment-generating
  function.
\end{itemize} The Weibull distribution has thicker tails and is skewed, and as
 a result, \npmle-based methods tend to greatly outperform methods based on
 assuming Gaussian priors. \Cref{fig:mse_calibrated_weibull} shows the
 analogue of \cref{fig:mse_table} for this data-generating process. Indeed,
 we see that \indepnpmle{} improves over \indepgauss{} considerably, and
 similarly for \closenpmle{} and \closegauss.

\subsection{Treating $\sigma_i$ symmetrically with covariates}
\label{asub:covariate_additive}

Here, we consider \close{} with covariates \eqref{eq:location-scale-covariates}. In
principle, we could model $m_0(\sigma_i, X_i), s_0(\sigma_i, X_i)$ fully
nonparametrically. However, such a model may be difficult to estimate given there are 9
additional covariates. Alternatively, we consider an additive model for the
covariates:\footnote{We thank the editor, Michal Koles\'ar, for this suggestion.} \[
m_0(\sigma_i, X_i) = g_0(\sigma_i) + \sum_k g_k (X_{ik}) \quad s_0^2(\sigma_i, X_i) = h_0
(\sigma_i) + \sum_k h_k (X_{ik}). \numberthis \label{eq:additive_model}
\]
For each covariate and for $\sigma$, the functions $g_k(\cdot), h_k(\cdot)$ are
approximated with cubic splines with knots at the 25\th{}, 50\th{}, and 75\th{}
percentiles of the covariate.

We estimate $m_0(\cdot)$ by least squares projection of $Y_i$ onto the basis functions in
$(\sigma_i, X_i)$. We estimate $s_0^2(\cdot)$ by least squares projection of $(Y_i -
\hat m(\sigma_i, X_i))^2 - \sigma_i^2$ onto the basis functions in $(\sigma_i, X_i)$. We
truncate fitted values for $s_0^2$ at zero. In practice, a substantial portion of them are
negative (cf. \cref{asub:variance_right_tail}). Additionally, we consider a similar
procedure, but one in which $g_0(\cdot), h_0(\cdot)$ are constant. We think of this
procedure as \indepgauss{} with flexible controls for the covariates.

\begin{figure}[tb]
  \centering
  \includegraphics[width=\textwidth]{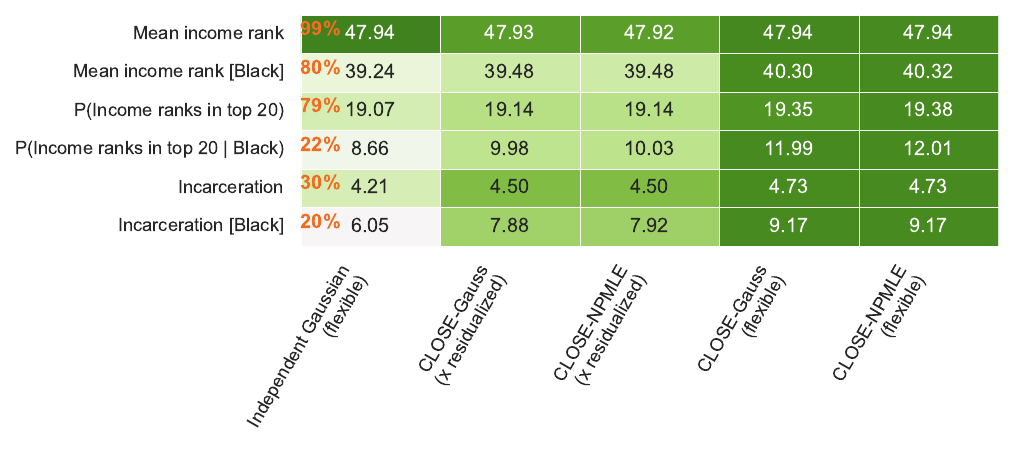}

  \begin{proof}[Notes] 

 ``Flexible'' denotes procedures that use the additive model \eqref{eq:additive_model} for
   $X_i$. ``$x$ residualized'' refers to procedures that, like those in the main text,
   residualize against the covariates in a pre-processing step.    The percentage in
   orange shows the proportion of units with strictly positive estimated conditional
   variance $\hat s^2(\sigma_i, X_i)$ for the flexible procedures.
  \end{proof}
  \caption{Analogue of \cref{tab:selection_validation_within_cz} for the setup in 
  \cref{asub:covariate_additive}.}
  \label{fig:rank_table_covariate_additive}
\end{figure}

We repeat the coupled bootstrap-based ranking exercise in
\cref{tab:selection_validation_within_cz}. However, here, due to considerably more
flexible procedures, we no longer consider a within-Commuting Zone version of the
exercise. Rather, we estimate everything on all the tracts, and select the top third over
all tracts. 

\Cref{fig:rank_table_covariate_additive} plots the results of this exercise, analogous to
\cref{tab:selection_validation_within_cz}. We find that, compared to the
residualize-then-shrink approach in the main text, modeling $X_i$'s additionally well can
have large benefits. We caution that, since much of the conditional signal variance
$s_0^2$ appears to be quite small once we use \eqref{eq:additive_model} for the
covariates, empirical Bayes posterior means are not very different from using $\hat
m_0(\sigma_i, X_i)$ directly. Importantly, for this application, including $\sigma_i$ in
modeling $m_0, s_0$ remains important. The approach that simply models $X_i$'s flexibly
for $m_0, s_0$ but omits $\sigma_i$ does not outperform the \close{} approaches in the
main text.

\printbibliography
\stopcontents
\end{appendices}
\end{refsection}

\FloatBarrier
\newpage

\small
\newgeometry{margin=1in}
\onehalfspacing
\begin{refsection}

\begin{appendices}
\makeatletter

\def\@seccntformat#1{\@ifundefined{#1@cntformat}%
   {\csname the#1\endcsname\space}%
   {\csname #1@cntformat\endcsname}}%
\newcommand\section@cntformat{\thesection.\space} %
\makeatother
\renewcommand{\thesection}{SM\arabic{section}}
\counterwithin{equation}{section}
\counterwithin{figure}{section}
\counterwithin{table}{section}

\begin{center}
\textbf{\large Supplementary Material to \\ ``\newtitle''}

Jiafeng Chen

\today
\end{center}

\DoToC

\newpage

\part{Important preliminary results for \cref{cor:maintext}}
\section{An oracle inequality for the likelihood}
\label{sec:likelihood}

Recall that for fixed sequences $\Delta_n, M_n$,  we define $A_n = \br
{\Dinfty \le \Delta_n, \bar Z_n \le M_n}$ in \eqref{eq:barzn_def}. This section bounds 
\[
\P\bk{A_n, \sub_n(\hat G_n) \gtrsim_{\Hyperparams} \epsilon_n},
\]
where we recall $\sub_n$ from \eqref{eq:likelihood_suboptimality}, for some rate function
$\epsilon_n$.  It is convenient to state a set of high-level assumptions on the rates
$\Delta_n, M_n$. These are satisfied for the choice \eqref{eq:DMrate}.

\begin{as}
\label{as:Delta_M_rate}
Assume that
\begin{enumerate*}
    \item $\frac{1}{\sqrt n} \leH \Delta_n \leH
    \frac{1}{M_n^3} \leH 1$, and
    \item $\sqrt{\log n} \leH M_n$.
\end{enumerate*}
\end{as}
Our main result in this section is the following oracle inequality.
\begin{theorem}
\label{thm:suboptimality}
Let $\Dinfty = \max(\norm{\hatm - m_0}_\infty, \norm{\hats - s_0}_{\infty})$ and
$\bar Z_n = \max_{i\in [n]} |Z_i| \vee 1$. Suppose $\hat G_n$ satisfies \cref{as:npmle}.
Under \cref{as:moments,as:holder,as:variance_bounds,as:Delta_M_rate}, there exists
constants $C_{1,\Hyperparams}, C_{2, \Hyperparams} > 0$ such that the following tail bound
holds: Let \begin{align*}
\epsilon_n &= M_n \sqrt{\log n} \Delta_n \frac{1}{n}\sum_{i=1}^n h
\pr{f_
{\hat G_n, \nu_i}, f_{G_0, \nu_i}}
 + \Delta_n M_n \sqrt{\log n} e^{-C_{2, \Hyperparams} M_n^\alpha} +
 \Delta_n^2 M_n^2 \log n + M_n^2
 \frac{\Delta_n^{1-\frac{1}{2p}}}{\sqrt{n}}. \numberthis
 \label{eq:general_rate_sub}
\end{align*}
Then, \[
\P\bk{
    \bar Z_n \le M_n, \Dinfty \le \Delta_n, \sub_n(\hat G_n) > C_{1, \Hyperparams}
    \epsilon_n
} \le \frac{9}{n}.
\]
\end{theorem}

The following corollary plugs in the rates \eqref{eq:DMrate} for $\Delta_n, M_n$ and
verifies that they satisfy \cref{as:Delta_M_rate}.
\begin{cor}
\label{cor:suboptimality}
For $\beta \ge 0$, suppose $\Delta_n, M_n$ are of the form \eqref{eq:DMrate}. Then there
exists a $C_\Hyperparams^*$ such that the following tail bound holds. Recall the average
Hellinger distance $\barh$ from
  \eqref{eq:avg_hellinger}. Suppose $\hat G_n$ satisfies \cref{as:npmle}. Under 
  \cref{as:holder,as:moments,as:variance_bounds}, define
  $\varepsilon_n$
  as:
\[
\varepsilon_n = n^{-\frac{p}{2p+1}} (\log n)^{\frac{2 + \alpha}{2 \alpha} +
\beta} \barh
\pr{\fs{\hat G_n}, \fs{G_0}} + n^{-\frac{2p}{2p+1}} (\log n)^{\frac{2+\alpha}{\alpha} + 2 \beta
},
\numberthis
 \label{eq:specific_rate_sub}
  \]
  we have that, $
\P\bk{
    A_n, \sub_n(\hat G_n) > C_{\Hyperparams}^*
    \varepsilon_n
} \le \frac{9}{n}.
$
The constants  in $\Delta_n, M_n$ affects the conclusion of the statement only through
affecting the constant $C_{\Hyperparams}^*$.
\end{cor}

\begin{proof}
We first show that the specification of $\Delta_n$ and $M_n$ means that the requirements
of \cref{as:Delta_M_rate} are satisfied. Among the requirements of \cref{as:Delta_M_rate}:
\begin{enumerate}[wide]
    \item is satisfied since the polynomial part of $\Delta_n$ converges to zero slower
    than $n^{-1/2}$, but converges to zero faster than any logarithmic rate. $M_n$ is a
    logarithmic rate.
    \item is satisfied since $\alpha \le 2$.
\end{enumerate}

We also observe that by Jensen's inequality, \[
\frac{1}{n} \sum_{i=1}^n h(f_{\hat G_n, \nu_i}, f_{G_0, \nu_i}) \le \barh(\fs{\hat G_n},
\fs{G_0}),
\]
and so we can replace the corresponding factor in $\epsilon_n$ by $\barh$.
Now, we plug the rates $\Delta_n, M_n$ into $\epsilon_n$. We find that the second term in
$\epsilon_n$
is dominated: \[
\Delta_n M_n^2 e^{-C_{2,\Hyperparams} M_n^\alpha} = \Delta_n M_n^2 e^{
    -(C_\Hyperparams + 1)^\alpha (\log n)
} \le \Delta_n M_n^2 n^{-1} \leH
\Delta_n^2 M_n^2 \log n
\]
since $\log n > 1$ as $n > \sqrt{2\pi} e$ by \cref{as:npmle}. The fourth term is also
dominated by the third term. Thus, plugging in the rates for
the other terms, we find that $
\epsilon_n \leH \varepsilon_n. $ Therefore, \cref{cor:suboptimality} follows from
\cref{thm:suboptimality}.
\end{proof}

\subsection{Derivative computations}
\label{sub:derivatives} 
It is sometimes useful to relate the derivatives of $\psi_i$ to
 $\PE_{G, \eta}$. We compute the following derivatives. Since they are all evaluated at
 $G, \eta$, we let $\hat
\nu = \hat\nu_i(\eta)$ and $\hat z = \hat Z_i(\eta)$ as a shorthand.
\begin{align}
\diff{\psi_i}{m_i}\evalbar_{\eta, G} &= -\frac{1}{s_i} \frac{f'_{G, \hat\nu}(\hat z)}{f_
{G,
\hat\nu} (\hat z)} \\
&= \frac{s_i}{\sigma_i^2} \PE_{G, \hat\nu}[Z-\tau
\mid \hat z] \label{eq:dpsidm}
\\
\diff{\psi_i
}{s_i} \evalbar_{{\eta, G}} &= \frac{1}{\sigma_i \hat\nu_i(\eta) f_{G,
\hat\nu(\eta)}(\hat Z_i(\eta))} \underbrace{\int (\hat Z_i(\eta) - \tau) \tau \varphi\pr{
\frac{\hat
Z_i(\eta) - \tau}{\hat \nu_i(\eta)}} \frac{1}{\hat \nu_i(\eta)} \,G(d\tau)}_{Q_i(Z_i,
\eta, G)}\label{eq:dpsids} \\
&= \frac{1}{\sigma_i \hat \nu} \PE_{G, \hat\nu}[(Z-\tau)\tau \mid \hat z]
\\
\diff{^2\psi_i}{m_i^2} \evalbar_{\eta, G}
&= \frac{1}{s^2_i} \bk{\frac{f''_{G,\hat\nu}(\hat z)}{f_{G,\hat\nu}(\hat z)} -
\pr{\frac{f'_{G,\hat\nu}(\hat z)}{f_{G,\hat\nu}(\hat z)}}^2}
\numberthis \label{eq:d2psidm2} \\
&= \frac{1}{s_i^2} \bk{
    \frac{1}{\hat\nu^4} \PE_{G,\hat\nu}[(\tau-Z)^2 \mid \hat z] - \frac{1}{\hat\nu^2} -
    \frac{1}{\hat\nu^4} \pr{\PE_{G, \hat\nu} [(\tau-Z) \mid \hat z]}^2
}
\\
\diff{^2 \psi_i}{m_i \partial s_i} \evalbar_{\eta, G}
&= 
\frac{1}{\sigma^2  f_{G, \hat \nu(\eta)} } \int \br{\pr{\frac{\hat Z(\eta) - \tau}
{\hat\nu(\eta)}}^2 - 1 } \tau \varphi\pr{\frac{\hat Z(\eta) - \tau}
{\hat\nu(\eta)}} \frac{1}{\hat\nu(\eta)} G(d\tau) + \frac{1}{\sigma^2} \frac{Q_i(Z_i,
\eta, G)} {f_{G, \hat \nu(\eta)} } \cdot \frac{f'_{G, \hat \nu(\eta)}}{f_{G, \hat \nu(\eta)}}
\numberthis \label{eq:cross_psi_deriv}
\\
&=\frac{1}{\sigma^2} \PE_{G, \hat\nu} \bk{\br{\pr{\frac{Z-\tau}{\hat\nu}}^2 - 1} \tau \mid
 \hat z} +
\frac{1}{\sigma^2\hat \nu(\eta)^2} \PE_{G, \hat\nu}[(Z-\tau)\tau \mid \hat z] \PE_{G, \hat\nu}[\tau-Z \mid \hat
z]
\\
\diff{^2 \psi_i}{s_i^2} \evalbar_{\eta, G}
&= 
\frac{1}{\sigma^2 f_{G, \hat \nu(\eta)}} \int \br{\pr{\frac{\hat Z(\eta) - \tau}
{\hat\nu(\eta)}}^2 - 1 }
\tau^2 \varphi\pr{\frac{\hat Z(\eta) - \tau}
{\hat\nu(\eta)}} \frac{1}{\hat\nu(\eta)} G(d\tau) + \frac{s_i^2} {\sigma^4} \pr{\frac{Q_i
(Z_i, \eta, G)}{f_{G, \hat \nu(\eta)}}}^2
 \numberthis 
 \label{eq:d2psids2} \\
 &=\frac{1}{\sigma^2} \PE_{G, \hat\nu} \bk{
 \tau^2 \br{\pr{\frac{Z-\tau}{\hat\nu}}^2 - 1}
 } - \frac{s_i^2}{\sigma^4} \PE_{G, \hat\nu} [(Z -\tau)\tau \mid \hat z]^2.
\end{align}

It is also useful to note that
\begin{align}
\frac{f'_{G, \nu}(z)}{f_{G,\nu}(z)} &= \frac{1}{\nu^2} \PE_{G,\nu}[(\tau - Z) \mid z]
\label{eq:firstderivative}\\
\frac{f''_{G, \nu}(z)}{ f_{G, \nu}(z)} &= \frac{1}{\nu^4} \PE_{G,\nu} [(\tau - Z)^2 \mid z]
- \frac{1}{\nu^2}.
\end{align}

\subsection{Proof of \cref{thm:suboptimality}}

\subsubsection{Decomposition of $\sub_n(\hat G_n)$}
Observe that, by \eqref{eq:approx_mle} in \cref{as:npmle}, \[
\frac{1}{n} \sum_{i=1}^n \psi_i(Z_i, \hateta, \hat G_n) - \frac1n \sum_{i=1}^n \psi_i(Z_i,
\hateta, G_0) \ge \kappa_n
\]
For random variables $a_n, b_n$ such that when $\bar Z_n \le M_n, \Dinfty \le \Delta_n$, 
\begin{align*}
\abs[\bigg]{\frac{1}{n} \sum_{i=1}^n \psi_i(Z_i, \hateta, \hat G_n) - \psi_i(Z_i,
\eta_0, \hat G_n)} \le a_n
\quad 
\abs[\bigg]{\frac{1}{n} \sum_{i=1}^n \psi_i(Z_i, \hateta, G_0) - \psi_i(Z_i,
\eta_0, G_0)} \le b_n,
\end{align*}
on the event $\bar Z_n \le M_n, \Dinfty \le \Delta_n$, \[
\frac{1}{n} \sum_{i=1}^n \psi_i(Z_i, \eta_0, \hat G_n) - \frac1n \sum_{i=1}^n \psi_i(Z_i,
\eta_0, G_0) \ge -a_n - b_n - \kappa_n
\]
and therefore $
\sub_n(\hat G_n) \le a_n + b_n + \kappa_n
$. 
Therefore, it suffices to show large deviation results for $a_n$ and $b_n$, where $a_n$
is chosen to be \eqref{eq:an_defn} and $b_n$ is chosen to be \eqref{eq:bn_defn}: \begin{align*}
&\P\bk{
  \bar Z_n \le M_n, \Dinfty \le \Delta_n, \sub_n(\hat G_n) \gtrsim_\H \epsilon_n
} \\&\le \P\bk{
  \bar Z_n \le M_n, \Dinfty \le \Delta_n, a_n+b_n+\kappa_n \gtrsim_\H \epsilon_n
} \\&\le \P\bk{a_n+b_n+\kappa_n \gtrsim_\H \epsilon_n}. 
\end{align*}

\subsubsection{Taylor expansion of $\psi_i(Z_i, \hat\eta, \hat G_n) - \psi_i(Z_i, \eta_0,
\hat G_n)$}
Define
$\Delta_{mi} = \hatm_i - m_{0i}$, $\Delta_{si} = \hats_i - s_{0i}$, and $\Delta_i =
[\Delta_{mi}, \Delta_{si}]'$. Recall $\Dinfty = \max(
\norm{s-s_0}_\infty, \norm{m-m_0}_\infty)$ as in \eqref{eq:barzn_def}.
Since $\psi_i(Z_i, \eta, G)$ is smooth in $(m_i, s_i) \in \R \times \R_{> 0}$, we can take
a second-order Taylor expansion:
\begin{align}
\psi_i\pr{Z_i, \hateta, \hat G_n} - \psi_i\pr{Z_i, \eta_0, \hat G_n}  = \diff{\psi_i}
{m_i}\evalbar_{\eta_0, \hat G_n} \Delta_{mi} + \diff{\psi_i}{s_i}\evalbar_
{\eta_0, \hat G_n} \Delta_{si} + \underbrace{\frac{1}{2} \Delta_i' H_i(\tilde\eta_i,
\hat G_n)
\Delta_i}_{R_{1i}} \label{eq:decomposition_hatG}
\end{align}
where $H_i(\tilde \eta_i ,\hat G_n)$ is the Hessian matrix $\diff{^2\psi_i}
{\eta_i \partial
\eta_i'}$ evaluated at some intermediate value $\tilde \eta_i$ lying on the line segment
between $\hateta_i$ and $\eta_{0i}$.

We further decompose the first-order terms into an empirical process term and a
mean-component term. By \cref{thm:lb}, \eqref{eq:dpsidm}, and \eqref{eq:dpsids}, for the
choice $\rho_n$ in \eqref{eq:rhodef} we have that the denominators to the first
derivatives can be truncated at $\rho_n$, as $f_{i, \hat G_n} \ge \rho_n/\nu_i$ so that
the truncation does not bind:
\begin{align}
\diff{\psi_i}{m_i}\evalbar_{\eta_0, \hat G_n}
&= -\frac{1}{s_i} \frac{
f'_{i, \hat G_n}
}{f_{i, \hat G_n} \vee \frac{\rho_n}{\nu_i}} \equiv D_{m,i}(Z_i,
\hat G_n, \eta_0, \rho_n)
\label{eq:dm_defn}
\\
\diff{\psi_i}{s_i}\evalbar_{\eta_0, \hat G_n}
&= \frac{s_i}{\sigma_i^2} \frac{
Q_i(Z_i, \eta_0, \hat G_n)
}{ f_{i, \hat G_n} \vee \frac{\rho_n}{\nu_i}} \equiv D_{s,i}(Z_i,
\hat G_n, \eta_0, \rho_n).
\label{eq:ds_defn}
\end{align}
where we recall $Q_i$ from \eqref{eq:dpsids}.

Let \[
\bar D_{k,i} (\hat G_n, \eta_0, \rho_n) = \int_{-\infty}^{\infty} D_{k,i}(z,
\hat G_n, \eta_0, \rho_n) \, f_{G_0, \nu_i}(z) dz \quad \text{ for $k \in
\br{m,s}$} \numberthis \label{eq:D_bar_def}
\]
be the population mean of $D_{k,i}$.
Then, for $k \in
\br{m,s}$, we can decompose \[
\diff{\psi_i}{k_i}\evalbar_{\eta_0, \hat G_n} \Delta_{ki} =
\bk{
    D_{k,i} (Z_i, \hat G_n, \eta_0, \rho_n) - \bar D_{k,i} (\hat G_n,
    \eta_0, \rho_n)
} \Delta_{ki} + \bar D_{k,i} (\hat G_n,
    \eta_0, \rho_n) \Delta_{ki}
\]
Hence, we can decompose the first-order terms in \eqref{eq:decomposition_hatG}:
\begin{align*}
\frac{1}{n} \sum_{i=1}^n \diff{\psi_i}{k_i}\evalbar_{\eta_0, \hat G_n}
\Delta_{ki} & =  \frac1n \sum_{i=1}^n \bar D_{k,i} (\hat G_n,
    \eta_0, \rho_n) \Delta_{ki} \\
    &\quad\quad+ \frac1n \sum_{i=1}^n \bk{
    D_{k,i} (Z_i, \hat G_n, \eta_0, \rho_n) - \bar D_{k,i} (\hat G_n,
    \eta_0, \rho_n)
} \Delta_{ki} \\ &\equiv U_{1k} + U_{2k} .
\end{align*}
Let the second order term in \eqref{eq:decomposition_hatG} be denoted as $R_1 = \frac1n
\sum_i R_ {1i}$. We let \[a_n = |R_1| + \sum_{k
\in \br{m,s}} |U_{1k}| + |U_{2k}| .
\numberthis \label{eq:an_defn}
\]

\subsubsection{Taylor expansion of $\psi_i(Z_i, \hat\eta, G_0) - \psi_i(Z_i, \eta_0,
G_0)$}
Like \eqref{eq:decomposition_hatG}, we similarly decompose \begin{align}
\psi_i(Z_i, \hat \eta, G_0) - \psi_i(Z_i, \eta_0, G_0) &= \diff{\psi_i}{m_i}\evalbar_
{\eta_0, G_0} \Delta_{mi} + \diff{\psi_i}{s_i}\evalbar_
{\eta_0, G_0} \Delta_{si} + \underbrace{\frac{1}{2} \Delta_i' H_i(\tilde\eta_i,
G_0)
\Delta_i}_{R_{2i}} \\
&= \sum_{k\in \br{m,s}} D_{k,i}(Z_i, G_0, \eta_0, 0) \Delta_{ki} + R_{2i}  \equiv U_{3mi} + U_{3si} + R_{2i}.
\end{align}
Let $U_{3k} = \frac{1}{n}\sum_i U_{3ki}$ for $k \in \br{m,s}$ and let $R_2 = \frac{1}{n}
\sum_i R_{2i}$. We let \[b_n = |R_2| + \sum_{k \in \br{m,s}} |U_{3k}| + |U_{3k}|
\numberthis \label{eq:bn_defn}.\]

\subsubsection{Bounding each term individually}
\label{ss:proof_intuition_oracle_ineq}
By our decomposition, we can write \[
a_n + b_n + \kappa_n \le \kappa_n + |R_{1}| + |R_{2}| + \sum_{k \in \br{m,s}} |U_{1k}| +
|U_{2k}| +
|U_{3k}|.
\]
To summarize, we have that, for $k = m,s,$ \begin{align}
U_{1k} &= \frac{1}{n}\sum_{i=1}^n \bar D_{k,i}(\hat G_n, \eta_0, \rho_n) \Delta_{ki} \\ 
U_{2k} &= \frac{1}{n}\sum_{i=1}^n \bk{D_{k,i}(Z_i, \hat G_n, \eta_0, \rho_n) - \bar D_
{k,i}
(\hat
G_n, \eta_0, \rho_n)} \Delta_ {ki} \\
U_{3k} &= \frac{1}{n}\sum_{i=1}^n D_{k,i}(Z_i, G_0, \eta_0, 0) \Delta_{ki} \\
R_1 &= \frac{1}{2n}\sum_{i=1}^n \Delta_i'H_i(\tilde\eta_i, \hat G_n) \Delta_i \\ 
R_2 &= \frac{1}{2n}\sum_{i=1}^n \Delta_i'H_i(\tilde\eta_i, G_0) \Delta_i
\end{align}

The ensuing subsections bound each term individually. Here we give an overview of the
main ideas:
\begin{enumerate}[wide]
    \item We bound $\one(A_n) |U_{1m}|$ in \cref{lemma:Dmi_final} by observing that $|\bar
    D_{mi}(\hat G_n, \eta_0, \rho_n)|$ is small when $\hat G_n$ is close to $G_0$, since
    $\bar D_ {mi} (G_0, \eta_0,
    0) = 0$. To do so, we need to control the differences \begin{align*}
    \bar D_{mi}(\hat G_n, \eta_0, \rho_n) - \bar D_{mi}(G_0, \eta_0, \rho_n) \text{ and }
\bar D_{mi}(G_0, \eta_0, \rho_n) - \underbrace{ \bar D_{mi}(G_0, \eta_0, 0)}_{= 0} = \bar D_{mi}(G_0, \eta_0,
\rho_n).
    \end{align*}
    Controlling the first difference features the Hellinger distance. Controlling
    the second relies on the fact that $\P_{X \sim f(X)}(f(X) \le \rho)$ cannot be too
    large, by an argument in \cref{lemma:chebyshev}. Similarly, we bound $\one(A_n) |U_
    {1s}|$ in \cref{lemma:Dsi_final}.

    \item The empirical process terms $U_{2m}, U_{2s}$ are bounded with statements of the form\[
    \P(A_n, |U_{2k}| > r_n) \le 2/n.
    \] To do so, we upper bound $\one(A_n) U_{2k} \le \bar U_{2k}$. The upper bound is
     obtained by projecting $\hat G_n$ onto a $\omega$-net of $\mathcal P (\R)$ in terms
     of some pseudo-metric $d_{k,\infty, M_n}$, such that two distributions $G_1, G_2$
     has a small distance $d_{k,\infty, M_n}$ between them if they give similar $\bar D_ {k,i}$. The upper bound
     $\bar U_{2k}$ then takes the form (up to some other terms), for $\eta \in S$ over a
     H\"older space and $G_1,\ldots, G_N$ a $\omega$-net over $\mathcal P(\R)$ in $d_
     {k,\infty, M_n}$, \[
\omega \Delta_n + \max_{j \in [N]} \sup_{\eta \in S} \abs[\bigg]{\frac{1}{n} \sum_i (D_
{ki}(G_j, \eta, \rho_n) -
\bar D_{ki}(G_j, \eta, \rho_n)) (\eta_i -
\eta_{0i})} \quad N \le N(\omega, \mathcal P(\R), d_{k,\infty, M_n}).
    \]
    Large deviation of $\bar U_{2k}$ is further controlled by applying Dudley's tail bound
    \citep{vershynin2018high}, since the entropy integral over $S$ is well-behaved. The
    covering number $N$ is controlled via \cref{prop:covering_for_moments} and
    \cref{prop:covering_ms}, which are minor extensions to Lemma 4 and Theorem 7 in
    \citet{jiang2020general}. The covering number is of a manageable size since the
    induced distributions $f_{G, \nu_i}$ are very smooth.
    \item Since $\bar D_{k,i}(G_0, \eta_0,0) = 0$. $U_{3m}, U_{3s}$ are effectively also
empirical process terms, without the additional randomness in $\hat G_n$. Thus the
projection-to-$\omega$-net argument above is unnecessary for $U_{3m}, U_{3s}$, whereas the
bounding follows from the same Dudley's chaining argument. \Cref{lemma:u3} bounds
$U_{3k}$.
\item For the second derivative terms $R_{1}, R_2$, we observe that the second derivatives
 take the form of functions of posterior moments under either $\hat G_n$ or $G_0$. The
 posterior moments under prior $\hat G_n$ is bounded within constant factors of $M_n^q$
 since the support of $G_n$ is restricted. The posterior moments under prior $G_0$ is
 bounded by $|Z_i|^q$, for $\E|Z_i|^q \leH M_n^q$ as we show in 
 \cref{lemma:posterior_moments}, thanks to
 the simultaneous moment control for $G_0$. These second derivatives are bounded in \cref
 {lemma:r1,lemma:r2}.
\end{enumerate}

(1) and (4) above bounds $U_{1k}, R_1, R_2$ under $A_n$. (2) and (3) bounds $U_{2k},
U_{3k}$ probabilistically by bounding $\P[A_n, U_{jk} > t]$. By a union bound in
\cref{lemma:union_bound}, we can simply add the rates.

Doing so, we find that the first term in $\epsilon_n$
\eqref{eq:general_rate_sub} comes from $U_{1s}$, which dominates $U_{1m}$. The second term
comes from $U_{2k}$. The third term comes from $R_1$, which
dominates $R_2$; this term also dominates a term in the bound for $U_{2k}$. The fourth
term comes from $U_ {3s}$. The leading terms in $\epsilon_n$
dominate $\kappa_n$, recalling \cref{as:npmle}. This completes the proof.

\subsection{Bounding $U_{1m}$}

\begin{lemma}
\label{lemma:Dmi_final}
    Under \cref{as:npmle,as:moments,as:holder,as:variance_bounds}, assume additionally that
    $\Dinfty \le
    \Delta_n, \bar Z_n \le M_n$. Assume that the rates $\Delta_n, M_n$ satisfy 
    \cref{as:Delta_M_rate}.
    Then \[
    |U_{1m}| \equiv \absauto{\frac{1}{n} \sum_{i=1}^n \bar D_{mi} (\hat G_n, \eta_0, \rho_n) \Delta_{mi}} \leH \Delta_n \bk{
    \frac{\log n}{n} \sum_ {i=1}^n h(f_{G_0, \nu_i}, f_ {\hat G_n,\nu_I}) + \frac{M_n^
    {1/3}}
    {n}}
    \numberthis \label{eq:Dmi_final}.\]
\end{lemma}
\begin{proof}
Note that \begin{align*}
|\bar D_{m,i}(\hat G_n,\eta_0,\rho_n)| = |\eqref{eq:D_bar_def}|
&\lesssim_{s_{0\ell}} \abs[\Bigg]{\int \frac{f'_{\hat G_n, \nu_i}
(z)}{
f_{\hat G_n, \nu_i}(z)
 \vee \frac{\rho_n}{\nu_i}}
\, f_{G_0, \nu_i}(z) dz}
\\
&= \abs[\Bigg]{\int \frac{f'_{\hat G_n, \nu_i}
(z)}{
f_{\hat G_n, \nu_i}(z)
 \vee \frac{\rho_n}{\nu_i}}
\, [f_{G_0, \nu_i}(z) - f_{\hat G_n, \nu_i}(z) + f_{\hat G_n, \nu_i}(z)] dz} \\
&\le \abs[\Bigg]{\int \frac{f'_{\hat G_n, \nu_i}
(z)}{
f_{\hat G_n, \nu_i}(z)
 \vee \frac{\rho_n}{\nu_i}}
\, [f_{G_0, \nu_i}(z) - f_{\hat G_n, \nu_i}(z)] \,dz}  \numberthis \label{eq:dmi_first} \\
&\quad\quad+ \abs[\Bigg]{ \int \frac{f'_{\hat G_n, \nu_i}
(z)}{
f_{\hat G_n, \nu_i}(z)
 \vee \frac{\rho_n}{\nu_i}} f_{\hat G_n, \nu_i}(z) dz }. \numberthis
 \label{eq:dmi_second}
\end{align*}
By the bounds for \eqref{eq:dmi_first} and \eqref{eq:dmi_second} below, we have that by \cref{as:Delta_M_rate} \[
|U_{1m}| \leH \Delta_n \bk{\frac{\sqrt{\log n}}{n} \sum_{i=1}^n h(f_{G_0, \nu_i}, f_{\hat
G_n,
\hat\nu_i}) +
\frac{M_n^{1/3}}{n}}.
\]
\end{proof}

\subsubsection{Bounding \eqref{eq:dmi_first}}

Consider the first term \eqref{eq:dmi_first}:
\begin{align*}
[\eqref{eq:dmi_first}]^2
 &= \bk{\int
\frac{f'_{\hat G_n, \nu_i}
(z)}{
f_{\hat G_n, \nu_i}(z)
 \vee \frac{\rho_n}{\nu_i}}
  \pr{\sqrt{f_{G_0,\nu_i}(z)}
-    \sqrt{f_{\hat G_n, \nu_i} ( z)}}
\pr{\sqrt{f_{G_0,\nu_i}(z)}
+   \sqrt{f_{\hat G_n, \nu_i} ( z)}}
 \,dz}^2
 \\
 &\le
 \underbrace{\int \pr{\sqrt{f_{G_0,\nu_i}(z)}
 -    \sqrt{f_{\hat G_n,  \nu_i} ( z)}}^2 \,dz }_{2 h^2 (f_{G_0,\nu_i},  f_{\hat G_n, 
 \nu_i} ) }
 \\&\quad\quad\times \int \pr{\frac{f'_{\hat G_n, \nu_i}
(z)}{
f_{\hat G_n, \nu_i}(z)
 \vee \frac{\rho_n}{\nu_i}}}^2 \pr{\sqrt{f_{G_0,\nu_i}(z)}
+   \sqrt{f_{\hat G_n,  \nu_i} ( z)}}^2 \,dz
 \tag{Cauchy--Schwarz}
 \\
 &\lesssim h^2(f_{G_0, \nu_i}, f_{\hat G_n, \nu_i}) 
 \int \pr{\frac{f'_{\hat G_n, \nu_i}
(z)}{
f_{\hat G_n, \nu_i}(z)
 \vee \frac{\rho_n}{\nu_i}}}^2 (f_{G_0,\nu_i}(z)
+   f_{\hat G_n,  \nu_i} ( z))\,dz
 \numberthis \label{eq:u1m_computation}
\end{align*}
where the last step uses $(a+b)^2 \le 2a^2 + 2b^2$. By \cref{thm:jianglemma2,thm:lb}, \[
\pr{\frac{f'_{\hat G_n, \nu_i}
(z)}{
f_{\hat G_n, \nu_i}(z)
 \vee \frac{\rho_n}{\nu_i}}}^2 \lesssim \frac{1}{
    \nu_i
} \log(1/\rho_n) \leH \log n.
\]
Hence, $
\eqref{eq:dmi_first} \leH  h(f_{G_0, \nu_i}, f_{\hat
G_n,
\nu_i}) \sqrt{\log n}.
$

\subsubsection{Bounding \eqref{eq:dmi_second}}

The second term   \eqref{eq:dmi_second} is
\begin{align*}
\eqref{eq:dmi_second} &= \abs[\Bigg]{
    \int \frac{f'_{\hat G_n, \nu_i}
(z)}{
f_{\hat G_n, \nu_i}(z)} \pr{\frac{f_{\hat G_n, \nu_i}(z)}{f_{\hat G_n, \nu_i}(z) \vee
\frac{\rho_n}{\nu_i}} - 1} f_{\hat G_n, \nu_i}(z)\,dz
} \\
&\le \int \abs[\Bigg]{
    \frac{f'_{\hat G_n, \nu_i}
(z)}{
f_{\hat G_n, \nu_i}(z)}} \one\pr{f_{\hat G_n, \nu_i}(z) \le \rho_n/\nu_i} f_{\hat G_n,
\nu_i}(z)\,dz \\
&\le \underbrace{\pr{\E_{Z\sim f_{\hat G_n, \nu_i}} \bk{
    \pr{\PE_{\hat G_n, \nu_i}\bk{ \frac{(\tau - Z)}{\nu_i^2} \mid Z}}^2
}}^{1/2}}_{ \le \E_{\tau \sim \hat G_n, Z \sim \Norm(\tau, \nu_i)}[(\tau - Z)^2 / \nu_i^4]^{1/2} = \nu_i^{-1}}
\cdot \sqrt{\P_{f_{\hat G_n, \nu_i}} [f_{\hat G_n, \nu_i}(Z) \le \rho_n/\nu_i]}
\tag{Cauchy--Schwarz and \eqref{eq:firstderivative}}.
\end{align*}
By Jensen's inequality and law of iterated expectations, the first term is bounded by
$\frac{1}{\nu_i}$. By \cref{lemma:chebyshev}, the second term is bounded by $\rho_n^
{1/3}\var_{Z\sim f_{\hat G_n, \nu_i}}(Z)^{1/6}$. Now, $
\var_{Z\sim f_{\hat G_n, \nu_i}}(Z) \le \nu_i^2 + \mu_2^2(\hat G_n) \leH M_n^2
$.
Hence, by \cref{thm:lb}, \[
\eqref{eq:dmi_second}
\leH M_n^{1/3} \rho_n^{1/3} \leH M_n^{1/3}n^{-1}. 
\]

\subsection{Bounding $U_{1s}$}

\begin{lemma}
\label{lemma:Dsi_final}
Under \cref{as:npmle,as:moments,as:holder,as:variance_bounds,as:Delta_M_rate}, if $\Dinfty
\le \Delta_n$,
$\bar Z_n \le M_n$, then
\[
|U_{1s}| \leH \Delta_n \bk{\frac{M_n \sqrt{\log n} }{n} \sum_{i=1}^n h(f_{\hat G_n, \nu_i}, f_{G_0,
\nu_i}) +  \frac{M_n^{4/3}}{n}}.
\numberthis \label{eq:dsi_final}
\]
\end{lemma}

\begin{proof}
Similar to our computation with $\bar D_{m,i}$, we
decompose \begin{align*}
|\bar D_{s,i}(\hat G_n, \eta_0, \rho_n)| &\leH \abs[\bigg]{\int \frac{Q_i(z, \eta_0,
\hat G_n)}{f_{\hat G_n,
\nu_i}(z) \vee (\rho_n / \nu_i)}
(f_{G_0, \nu_i}(z) - f_{\hat G_n, \nu_i}(z) )\,dz} \numberthis
\label{eq:dsi_first}
\\ &\quad+ \abs[\bigg]{\int \frac{Q_i(z, \eta_0, \hat G_n)}{f_{\hat G_n,
\nu_i}( z) \vee (\rho_n /
\nu_i)} f_{\hat
G_n, \nu_i}(z)  \,dz},  \numberthis
\label{eq:dsi_third}
\end{align*}
where we recall $Q_i$ from 
\eqref{eq:dpsids}.
We conclude the proof by plugging in  our subsequent calculations.
\end{proof}

\subsubsection{Bounding \eqref{eq:dsi_first}}
The first term \eqref{eq:dsi_first} is bounded by \begin{align*}
[\eqref{eq:dsi_first}]^2 &\lesssim h^2(f_{G_0,\nu_i} ,f_{\hat G_n, \nu_i}) \int \pr{
\frac{Q_i (z, \eta_0,
\hat G_n)}{f_{\hat G_n,
\nu_i}(z) \vee (\rho_n / \nu_i)} }^2   \bk{f_{G_0, \nu_i}(z) + f_{\hat G_n, \nu_i}(z)} \,dz,
\end{align*}
similar to the computation in \eqref{eq:u1m_computation}.

By \cref{lemma:qbound,thm:lb}, \[
\pr{\frac{Q_i(z, \eta_0,
\hat G_n)}{f_{\hat G_n,
\nu_i}(z) \vee (\rho_n / \nu_i)} }^2 \leH M_n^2 \log
n \implies \int \pr{\frac{Q(z, \nu_i)}{f_{\hat G_n, \hat
\nu_i}(z) \vee (\rho_n /
\nu_i) } }^2   \bk{f_{G_0, \nu_i}(z) + f_{\hat G_n, \nu_i}(z)} \,dz  \leH M_n^2 \log
n.
\]
Hence, \[
\eqref{eq:dsi_first} \leH M_n h(f_{G_0,\nu_i}, f_{\hat G_n, \nu_i}) \sqrt{\log n}.
\numberthis
\label{eq:dsi_first_final}
\]

\subsubsection{Bounding \eqref{eq:dsi_third}}
Observe that \[
\eqref{eq:dsi_third} = \abs[\Bigg]{
    \int \underbrace{\frac{Q_i(z, \eta_0,
\hat G_n)}{f_{\hat G_n,
\nu_i}(z)}}_{\PE_{\hat G_n, \nu_i}[(\hat Z-\tau) \tau \mid \hat Z=z]}  \underbrace{\pr{ 
\frac{f_ {\hat G_n,
\nu_i}(z) }{f_{\hat G_n,
\nu_i}(z) \vee (\rho_n / \nu_i)} - 1}}_{|\cdot| \le \one(\nu_i f_{\hat G_n, \nu_i} \le
\rho_n)} f_ {\hat G_n,
\nu_i}(z) \,dz
}. 
\]
Similar to our argument for \eqref{eq:dmi_second}, by Cauchy--Schwarz, \begin{align*}
\eqref{eq:dsi_third} &\le \pr{
    \E_{f_{\hat G_n,
\nu_i}(z)} \bk{
    (\PE_{\hat G_n, \nu_i}\bk{
            (Z-\tau)\tau \mid Z
        })^2
}
}^{1/2} \sqrt{\P_{f_{\hat G_n,
\nu_i}(z)}(f_{\hat G_n,
\nu_i}(z) \le \rho_n/\nu_i)} \\&\leH M_n \cdot \rho_n^{1/3} M_n^{1/3} \leH \frac{M_n^{4/3}}{n}.
\end{align*}

\subsection{Bounding $U_{2m}, U_{2s}$}

\begin{lemma}
\label{lemma:u2k}
Under \cref{as:npmle,as:moments,as:holder,as:variance_bounds,as:Delta_M_rate}, 
for $k \in \br{m,s}$, 
\[
\P\bk{\Dinfty \le \Delta_n, \bar Z_n \le M_n,
|U_{2k}| \gtrsim_{\Hyperparams} r_n
  } \le \frac{2}{n}
\]
for $r_n = \Delta_n e^{-C_\H M_n^\alpha}\log n  
  + \frac{M_n^{3/2} (\log n)^{5/4}}{\sqrt{n}} \Delta_n
  + \frac{M_n \sqrt{\log n}}{\sqrt{n}} \Delta_n^{1-\frac{1}{2p}}.$
\end{lemma}

\begin{proof}
Let $k \in \br{m,s}$. 
We first show that if $\Dinfty \le \Delta_n$ and $\bar Z_n \le
M_n$, then for some $\bar U_{2k}$ to be chosen, $
|U_{2k}|  \le \bar U_{2k}.
$
We choose $\bar U_{2k}$ in \eqref{eq:baru2k} such that $\P[\bar U_{2k} > t]$ is small. Then \[
\P\bk{\Dinfty \le \Delta_n, \bar Z_n \le M_n,
|U_{2k}| > t}  \le \P[\bar U_{2k} > t]. 
\]
Thus the bound for $\bar U_{2k}$ would suffice. 

Let \begin{align*}
D_{k,i,M_n}(Z_i, \hat G_n, \eta_0, \rho_n) &= D_{k,i} ( Z_i, \hat G_n, \eta_0, \rho_n )
\one(|Z_i| \le M_n) \\
\bar D_{k,i, M_n}(\hat G_n, \eta_0, \rho_n) &=
\int D_{k,i} ( z, \hat G_n, \eta_0, \rho_n ) \one(|z| \le M_n) f_{G_0, \nu_i}(z)\,dz.
\end{align*}
They are truncated versions of $D_{k,i}(z, \hat G_n, \eta_0,
\rho_n)$. 

Observe that on $\bar Z_n \le M_n$, $D_{k,i,M_n}(Z_i, \hat G_n, \eta_0, \rho_n) = D_{k,i}
(Z_i, \hat G_n, \eta_0, \rho_n)$ for all $i$. Thus, on $\bar Z_n \le M_n$, we may
decompose \begin{align}
|U_{2k}| &\le 
\absauto{
  \frac1n \sum_{i=1}^n \br{D_{k,i,M_n}(Z_i, \hat G_n, \eta_0, \rho_n) - \bar D_{k,i,M_n}
  (\hat G_n, \eta_0, \rho_n)} \Delta_{ki}
}
\label{eq:pre_project_emp_proc}
 \\&\quad + \absauto{
  \frac1n \sum_{i=1}^n \br{\bar D_{k,i}(Z_i, \hat G_n, \eta_0, \rho_n) - \bar D_
  {k,i,M_n}
  (\hat G_n, \eta_0, \rho_n)} \Delta_{ki}
} \label{eq:Dk_tail}.
\end{align}

By \cref{lemma:qbound,thm:jianglemma2}, uniformly over all $G$, \[ |D_{k,i}
(z, G, \eta_0, \rho_n)| \leH |z| \sqrt{\log n} + \log n. \numberthis \label{eq:Dki_bound_}
\]
Thus, \[
\eqref{eq:Dk_tail} \leH \Delta_n \bigg({\sqrt{\log n} \max_{i \in [n]} \underbrace{\int_
{|z|
> M_n} |z|
f_
{G_0,\nu_i} (z)\,dz}_{\le \sqrt{\E[Z_i^2] \P\pr{|Z_i| > M_n}}\leH \P\pr{|Z_i| > M_n}^
{1/2}} + \log n \max_{i\in [n]}\P_ {G_0, \nu_i} (|Z_i| >
M_n)}\bigg)
\]
By \cref{lemma:tail_bound}, $\P_{G_0, \nu_i}(|Z_i| > M_n) \le
\exp\pr{
    -C_{\alpha, A_0, \nu_u} M_n^{\alpha}
}.$ Hence $
\eqref{eq:Dk_tail} \leH \Delta_n e^{-C_\H M_n^\alpha} \log n .
$

Returning to \eqref{eq:pre_project_emp_proc}, let $G_1,\ldots, G_N \in \mathcal P(\R)$ be
 a $\omega$-net of $\mathcal P(\R)$ under the pseudometric \[
d_{k, \infty, M_n}(G_1, G_2) = \max_{i\in [n]} \sup_{|z| \le M_n} \abs{D_{k,i}(z, G_1,
\eta_0, \rho_n) - D_{m,i}(z, G_2, \eta_0, \rho_n)}.
\numberthis \label{eq:def_dk}
 \]
Thus, we can take $N = N(\omega, \mathcal P(\R), d_{k, \infty, M_n})$. By construction,
there exists a $G_{j^*}$ for $j^* \in [N]$ such that on $\bar Z_n \le M_n$,\begin{align*}
\abs{D_{k,i, M_n}(Z_i, \hat G_n, \eta_0, \rho_n) - D_{k,i, M_n}(Z_i, G_{j^*}, \eta_0, \rho_n)} &\le
\omega \\ 
\implies \abs{\bar D_{k,i, M_n}(\hat G_n, \eta_0, \rho_n) - \bar D_{k,i, M_n}(G_{j^*},
\eta_0,
\rho_n)} &\le \omega.
\end{align*}
The second line in the above display holds since the integrand is bounded by $\omega$.
Hence, projecting $\hat
G_n$ to the $\omega$-net, we have that \[
\eqref{eq:pre_project_emp_proc} \le 2\omega \Delta_n + \max_{j \in [N]}
\abs[\bigg]{\frac{1}{n} \sum_{i=1}^n \br{D_
{k,i,
M_n}(Z_i, G_j, \eta_0, \rho_n) - \bar D_{k,i, M_n }(G_j, \eta_0, \rho_n)}
\Delta_{ki}}
\]
Define \begin{align*}
v_{i,j}(\eta) &\equiv \br{D_
{k,i,
M_n}(Z_i, G_j, \eta_0, \rho_n) - \bar D_{k,i, M_n }(G_j, \eta_0, \rho_n)}
\Delta_{ki}(\eta) \quad \Delta_{ki}(\eta) \equiv k_i(\sigma_i) - k_0(\sigma_i) \quad k\in
\br{m,s} \\
V_{n,j}(\eta) &\equiv \frac{1}{n} \sum_{i=1}^n v_{i,j}(\eta).
\end{align*}
We have that \[
\eqref{eq:pre_project_emp_proc} \lesssim \omega \Delta_n + \max_{j\in[N]} \sup_{\eta\in S}
|V_{n,j}(\eta)|
\]
where  \[S = \br{(m, s) : \norm{m-m_0}_\infty \le \Delta_n,
\norm{s-s_0}_\infty
\le \Delta_n, (m, s) \in \mathcal V}\numberthis \label{eq:S_def}\] for $\mathcal V$ in
\cref{as:holder}.
As a result, for some $\omega$ to be chosen, let us take \[
\bar U_{2k} = C_\Hyperparams \br{
  \Delta_n (\log n) e^{-C_\H M_n^\alpha} + \omega \Delta_n + \max_{j\in[N]} \sup_{\eta\in S}
|V_{n,j}(\eta)|
} \numberthis \label{eq:baru2k}
\]
and analyze its tail behavior.

First, let us bound the empirical process $\max_{j\in[N]} \sup_{\eta\in S} |V_{n,j}
(\eta)|$. Note that for a fixed $\eta, \eta_1, \eta_2 \in S$, we have that, since $
\eqref{eq:Dki_bound_} \leH M_n \sqrt{\log n}$ on $|z| \le M_n$, \begin{align*}
\absauto{V_{n,j}(\eta)} &\leH \frac{ M_n \sqrt{\log n}}{\sqrt{n}} \Delta_n \\
\norm{V_{n,j}(\eta_1) - V_{n,j}(\eta_2)}_{\psi_2} &\leH \frac{ M_n \sqrt{\log n}}{\sqrt{n}}
\norm{\eta_1 - \eta_2}_\infty \tag{$v_{i,j}$ are independent across $i$ and bounded}.
\end{align*}
Since $\eta \mapsto V_{n, j}
(\eta)$ is a process with subgaussian increments under $\norm{\eta}_\infty$ (see (8.1)
in \citet{vershynin2018high}), by Theorem 8.1.6 in
\citet{vershynin2018high}, for all $u > 0$, \[
\sup_{\eta \in S} |V_{n,j}(\eta)| \leH \frac{ M_n \sqrt{\log n}}{
\sqrt{n}} \bk{(1+u)\Delta_n + 
\int_0^{\infty} \sqrt{\log N(\epsilon, S, \norm{\cdot}_\infty) } \,d \epsilon
},
\]
holds with probability at least $1-2e^{-u^2}$.

Since $
\sqrt{\log N(\epsilon, S, \norm{\cdot}_\infty) } \lesssim \sqrt{\log N(\epsilon/2,
\mathcal V,
\norm{\cdot}_\infty)} \leH \sqrt{\log C_\H + (1/\epsilon)^{-1/p}}$ by 
\cref{as:holder} and Exercise 4.2.10 in \citet{vershynin2018high}, \[
\int_0^{\infty} \sqrt{\log N(\epsilon, S, \norm{\cdot}_\infty) } \,d \epsilon = \int_0^{2
\Delta_n} \sqrt{\log N
(\epsilon, S, \norm{\cdot}_\infty) } \,d \epsilon \leH
\Delta_n^{1-\frac{1}{2p}}.
\]
Note that by a union bound, \begin{align*}
\P\br{
  \max_{j \in [N]} \sup_{\eta \in S} 
 |V_{n,j}(\eta)| \gtrsim_\H \frac{ M_n \sqrt{\log n}}{
\sqrt{n}} \bk{(1+u)\Delta_n + \Delta_n^{1-\frac{1}{2p}}
}
} \le 2N e^{-u^2}
\end{align*}
Choose $u = \sqrt{\log N} + \sqrt{\log n} \ge \sqrt{\log N + \log n}$ such that the right
hand side is bounded by $2/n$.

Next, choose $\omega = M_n \frac{\sqrt{\log(1/\rho_n)}}{\rho_n} \frac{\rho_n}{\sqrt{n}}
\log \pr{\frac{\sqrt{n}
}{\rho_n}} \ge \frac{\sqrt{\log(1/\rho_n)}}{\rho_n} \frac{\rho_n}{\sqrt{n}}
\log \pr{\frac{\sqrt{n}
}{\rho_n}}$. By \cref{prop:covering_ms}, \begin{align*}
\log N(\omega, \mathcal P(\R), d_{m, \infty, M_n}) 
&\le \log N\pr{
  \frac{\sqrt{\log(1/\rho_n)}}{\rho_n} \frac{\rho_n}{\sqrt{n}}
\log \pr{\frac{\sqrt{n}
}{\rho_n}} , \mathcal P(\R), d_{m,\infty, M_n}
} \\
&\leH (\log n)^2 \max\pr{1, \frac{M_n}{\sqrt{\log(n)}}} \\
&\leH (\log n)^{3/2} M_n
\\ 
\log N(\omega, \mathcal P(\R), d_{s, \infty, M_n}) &\leH (\log n)^2 \max\pr{1, \frac{M_n}{
\sqrt{\log(n)}}} \\
&\leH (\log n)^{3/2} M_n  \tag{$\delta = \rho_n/\sqrt{n}$}
\end{align*}
Note that this choice is such that  $\omega \leH \frac{1}{\sqrt{n}}(\log n)^{3/2} M_n$ and
$
(1+u) \leH (\log n)^{3/4} M_n^{1/2} + \sqrt{\log n} \leH (\log n)^{3/4} M_n^{1/2}. 
$

Returning to \eqref{eq:baru2k}, since $V_{n,j}(\eta)$ is the only random expression in
\eqref{eq:baru2k}, this shows that \[
\P\br{
  \bar U_{2k} \gtrsim_\H
  \Delta_n e^{-C_\H M_n^\alpha}\log n 
  + \frac{M_n^{3/2} (\log n)^{5/4}}{\sqrt{n}} \Delta_n
  + \frac{M_n \sqrt{\log n}}{\sqrt{n}} \Delta_n^{1-\frac{1}{2p}}
} \le \frac{2}{n}.
\]
Here, note that the term $\omega \Delta_n$ is dominated by $\frac{M_n^{3/2} (\log n)^
{5/4}} {\sqrt{n}} \Delta_n$ since $M_n \gtrsim_\H \sqrt{\log n}$. This concludes the
proof. \qedhere

\end{proof}

\subsection{Bounding $U_{3m}, U_{3s}$}

\begin{lemma}
\label{lemma:u3}
Under \cref{as:moments,as:holder,as:variance_bounds,as:Delta_M_rate}, 
for $k \in \br{m,s}$,
\begin{align*}
\P\bk{
    \Dinfty \le \Delta_n, \bar Z_n \le M_n, |U_{3k}| \gtrsim_{\Hyperparams} \Delta_n
    \br{ e^{-C_{\Hyperparams} M_n^\alpha} + \frac{M_n^2}{\sqrt{n}}
\pr{
\Delta_n^{-1/(2p)}
+ \sqrt{\log n}}
}
} &\le \frac{2}{n}.
\end{align*}
\end{lemma}

\begin{proof}
The proof structure follows that of \cref{lemma:u2k}. Recall that 
\begin{align*}
U_{3k} &= \frac{1}{n} \sum_{i=1}^n D_{k,i}(Z_i, G_0, \eta_0, 0) \Delta_{ki}.
\\&=\underbrace{\frac{1}{n} \sum_{i=1}^n \br{D_{k,i, M_n}(Z_i, G_0, \eta_0, 0) - \bar D_
{k,i,M_n} (Z_i,
G_0, \eta_0, 0)} \Delta_{ki}}_{V_n(\eta)} + \bar D_
{k,i,M_n}(G_0, \eta_0, 0)\Delta_{ki} \tag{on the event $\bar Z_n \le M_n$}.
\end{align*}

Now, observe that, by Cauchy--Schwarz, \begin{align*}
\absauto{\bar D_{k,i,M_n} (G_0, \eta_0, 0)} \leH \int_{|z| \le M_n} T_k(z, \eta_0, G_0)
f_{G_0, \nu_i}(z)\,dz \le \P(Z_i > M_n)^{1/2} \pr{\E \bk{T_k^2(Z_i, \eta_0, G_0)}}^{1/2}
\end{align*}
where $T_m = \frac{|f'_{G_0, \nu_i}(z)|}{f_{G_0,\nu_i}(z_i)}$ and $T_s = \frac{|Q_i(z),
\eta_0, G_0)|}{f_{G_0,\nu_i}(z_i)}$. Note that since both $T_k$ take the form of $|\E_
{G_0} [f(Z,\tau) \mid Z]|$, we can use Jensen's inequality to bound $
\E[T_k^2] \le \E[f^2(Z_i,\tau)] \leH 1
$.
Hence, \[
\absauto{\bar D_{k,i,M_n} (G_0, \eta_0, 0)} \leH  e^{-C_\H M_n^\alpha}.
\tag{\cref{lemma:tail_bound}}
\]

We likewise analyze $V_n(\eta)$. Note that \[
\abs{V_n(\eta_1) - V_n(\eta_2)} \le \norm{\eta_1 - \eta_2}_\infty \frac{1}{n} \sum_{i=1}^n
\absauto{
  D_{k,i, M_n}(Z_i, G_0, \eta_0, 0) - \bar D_
{k,i,M_n} (Z_i,
G_0, \eta_0, 0)
}.
\]
By \cref{lemma:posterior_moments}, since $D_{k,i,M_n}$ is a posterior moment under $G_0$
truncated to
$|z| \le M_n$, \[
\absauto{
  D_{k,i, M_n}(Z_i, G_0, \eta_0, 0) - \bar D_
{k,i,M_n} (Z_i,
G_0, \eta_0, 0)
} \leH M_n^2.
\]
As a result, \[
\norm[\bigg]{
  \frac{1}{n} \sum_{i=1}^n
\absauto{
  D_{k,i, M_n}(Z_i, G_0, \eta_0, 0) - \bar D_
{k,i,M_n} (Z_i,
G_0, \eta_0, 0)
}
}_{\psi_2} \leH \frac{M_n^2}{\sqrt{n}}
\]
and thus $\norm{V_n(\eta_1) - V_n(\eta_2)}_{\psi_2} \leH \frac{M_n^2}{\sqrt{n}} \norm{
  \eta_1 - \eta_2
}_\infty$ has subgaussian increments. For a fixed $\eta$, $|V_n(\eta)| \leH \Delta_n
M_n^2/\sqrt{n}$.

By the same chaining argument as in the proof of \cref{lemma:u2k}, recalling $S$ in \eqref{eq:S_def},
\[
\sup_{\eta \in S} |V_n(\eta)| \leH \frac{M_n^2}{\sqrt{n}} \bk{
  \sqrt{\log n} \Delta_n + \Delta_n^{1-\frac{1}{2p}}
}
\]
with probability at least $1-2/n$. Here we choose $u = \sqrt{\log n}$ since we do not have
to project $G_0$ to some covering. Thus, we can let \[
\bar U_{3k} = C_\H \pr{\sup_{\eta \in S} |V_n(\eta)| + \Delta_n e^{-C_\H M_n^\alpha} }. 
\]
Bounding $\bar U_{3k}$ using the bound for $\sup_{\eta \in S} |V_n(\eta)| $ concludes the
proof. \qedhere

\end{proof}

\subsection{Bounding $R_1, R_2$}

\begin{lemma}
\label{lemma:r1}
Recall $R_{1i}$ from \eqref{eq:decomposition_hatG}. Then, under
\cref{as:npmle,as:moments,as:holder,as:variance_bounds,as:Delta_M_rate},
if $\Dinfty \le \Delta_n$ and $\bar Z_n \le M_n$, then
$
R_{1i} \leH \Delta_n^2 M_n^2 \log n.
$
\end{lemma}

\begin{proof}
Observe that $R_{1i} \leh \max\pr{\Delta_{mi}^2, \Delta_{si}^2} \cdot \norm{H_i(\tilde
\eta_i, \hat G_n)}_\infty$,
where $\norm{\cdot}_\infty$ takes the largest entry from a matrix by magnitude. By
assumption, the first term is bounded by $\Delta_n^2$. By \cref{lemma:secondderivatives},
the second derivatives are bounded by $M_n^2 \log n$. Hence $\norm{H_i(\tilde \eta_i, \hat
G_n)}_\infty \leH M_n^2 \log n$. This concludes the proof.
\end{proof}

\begin{lemma}
\label{lemma:r2}
Under \cref{as:moments,as:holder,as:variance_bounds,as:Delta_M_rate},  then \[
\P\pr{\Dinfty \le \Delta_n, \bar Z_n \le M_n, |R_2| \gtrsim_\Hyperparams
\Delta_n^2}
\le \frac{1}{n}.
\]
\end{lemma}
\begin{proof}
Recall that $\one(A_n) = \one(\Dinfty \le \Delta_n, \bar Z_n \le M_n)$.
Note that \[
\one(A_n)|R_2| \leH \Delta_n^2 \frac{1}{n}\sum_{i=1}^n \one(A_n)\norm{H_i}_\infty.
\] by $(1,\infty)$-H\"older inequality. Moreover, note that the second derivatives(\eqref
 {eq:d2psids2},  \eqref{eq:cross_psi_deriv} , \eqref{eq:d2psids2}) that occur in entries
 of $H_i$ are functions of posterior moments under $G_0$, evaluated at $Z_i =
\hat Z_i(\hat\eta_i)$. By
\cref{lemma:posterior_moments},
under $G_0,$ these posterior moments are bounded above by corresponding moments of
$\hat Z_i (\tilde\eta_i).$  Hence,  \[
\one(A_n) \norm{H_i}_\infty \leH \one(A_n) (\hat Z_i
(\tilde\eta_i) \vee 1)^{4}  \leH (Z_i \vee 1)^4.
\numberthis \label{eq:Hbound}
\]
Hence, $
\one(A_n)|R_2| \leH \Delta_n^2 \frac{1}{n} \sum_{i=1}^n (Z_i \vee 1)^4.$
Chebyshev's inequality implies that there exists some choice $C_\H$ such that \[
\P\bk{
\frac{1}{n} \sum_{i=1}^n (Z_i \vee 1)^4 \ge C_\H
} \le \frac{1}{n},
\]
since $\var(\frac{1}{n} \sum_{i=1}^n (Z_i \vee 1)^4) \leH \frac{1}{n}$. 
Hence, $
\P\pr{\Dinfty \le \Delta_n, \bar Z_n \le M_n, |R_2| \gtrsim_\Hyperparams
\Delta_n^2}
\le \frac{1}{n}.$ \qedhere
\end{proof}

\subsection{Complexity of $\mathcal P(\R)$ under moment-based distance}
The following is a minor generalization of Lemma 4 and Theorem 7 in
\citet{jiang2020general}. In particular, \citet{jiang2020general}'s Lemma 4
reduces to the case $q=0$ below, and \citet{jiang2020general}'s Theorem 7
relies on the results below for $q=0,1$. The proof largely
follows the proofs of these two results of \citet{jiang2020general}.

We first state the following fact readily verified by differentiation.
\begin{lemma}
\label{lemma:lipschitz_smooth_varphi}
For all integer $m \ge 0$:
    \[
    \sup_{t\in \R} |t^m \varphi(t)| = m^{m/2} \varphi(\sqrt{m}).
    \]
As a corollary, there exists absolute $C_m > 0$ such that $t \mapsto t^m
\varphi(t)$ is $C_m$-Lipschitz.
\end{lemma}

\begin{prop}
\label{prop:covering_for_moments}
Fix some $q \in \N \cup \br{0}$ and $M > 1$.
Consider the pseudometric \[d_{\infty, M}^{(q)}(G_1, G_2) = \max_{i\in[n]}
\underbrace{\max_{0 \le v
\le q}\sup_{|x|
\le M} \abs[\bigg]
{\int
 \frac{(u-x)^v}{\nu_i^v} \varphi\pr{\frac{x-u}{\nu_i}}  (G_1 - G_2)(du)
}
}_{d_{q,i,m}(G_1, G_2)}.\]
Let $\nu_\ell, \nu_u$ be the lower and upper bounds of $\nu_i$.
Then, for all $0 < \delta < \exp(-q/2) \minwith e^{-1}$,
\[
\log N(\delta \log^{q/2}(1/\delta), \mathcal P(\R), d_{\infty, M}^{(q)}) \lesssim_{q, \nu_u,
\nu_\ell}\log^2(1/\delta) \max\pr{\frac{M}{\sqrt{\log
(1/\delta)}}, 1}.
\]
\end{prop}

\begin{proof}
The proof strategy is as follows. 
First, we discretize $[-M, M]$ into a
union of small intervals $I_j$. Fix $G$. There exists a finitely supported distribution
$G_m$ that matches moments of $G$ on every $I_j$. It turns out that such a $G_m$ is close
to $G$ in terms of $d^{(q)}_{\infty, M}$. Next, we discipline $G_m$ by
approximating $G_m$ with $G_{m,\omega}$, a finitely supported distribution supported on the
fixed grid $\br{k\omega : k \in \Z} \cap [-M,M]$. 
This shows that there exists a $G_{m,\omega}$ with finite support on a grid that
approximates $G$ in $d_{\infty, M}^{(q)}$.
Finally, the set of all $G_{m,\omega}$'s may
be approximated by a finite set of distributions, and we count the size of this finite
set.

\subsubsection{Approximating $G$ with $G_m$}

First, let us fix some $\omega < \varphi(\sqrt{q}) \minwith \varphi(1)$ to be chosen.
Let $a = \frac{\nu_u}{\nu_\ell} \invphi(\omega) \ge 1$. Let $I_j = [-M +
(j-2)a\nu_\ell, -M + (j-1) a
\nu_\ell]$ be such that \[
I \equiv [-M - a\nu_\ell, +M + a\nu_\ell] \subset \bigcup_{j=1}^{j^*} I_j
\] where $I_j$ is a width $a\nu_\ell$ interval. Let $j^* = \lceil \frac{2M}
 {a\nu_\ell} + 2\rceil$ be the number of such intervals. 

For some $k^*$ to be chosen, there exists by Carath\'eodory's theorem a distribution $G_m$
with support on $I$ and no
more than \[m = (2k^* + q + 1) j^* + 1\] support points such that the moments up to $2k^*
+ q$ match \[
\int_{I_j} u^k dG(u) = \int_{I_j} u^k dG_m(u) \text{ for all $k = 0,\ldots, 2k^* + q$ and
$j = 1,\ldots, j^*$}.
\]

Then, by analyzing $x \in I_j \cap [-M ,M]$, we have that \begin{align}
d_{q,i,M}(G,G_m) &\le \max_{0 \le v \le q} \max_{j=1,\ldots,j^*}  \sup_{x \in I_j \cap 
[-M,M]} \Bigg [\abs [\bigg]{\int_
{(I_
{j-1}\cup I_j \cup I_
{j+1})^\comp} \pr{\frac{u-x}{\nu_i}}^v \varphi\pr{\frac{x-u}
{\nu_i}} (G(du) - G_m(du))}
\label{eq:outsideintervaldiff}
 \\&+ \abs[\bigg]{ \int_{I_
{j-1}\cup I_j \cup I_
{j+1}} \pr{\frac{u-x}{\nu_i}}^v \varphi\pr{\frac{x-u}
{\nu_i}}  (G(du) - G_m(du)) } \Bigg] \label{eq:intervaldiff}
\end{align}

Note that $
t \mapsto |t|^v \varphi(t)
$
is a decreasing function in $|t|$, as long as $|t| > \sqrt{v}$. Note that over $ u \not
\in I_{j-1} \cup I_j \cup I_{j+1}$ and $x \in I_j$, $
\frac{|u-x|}{\nu_i} \ge a\nu_\ell / \nu_u = \varphi_+(\omega) \ge \sqrt{q} \ge \sqrt{v}
$.
Hence, \[
\absauto{\pr{\frac{u-x}{\nu_i}}^v \varphi\pr{\frac{x-u}
{\nu_i}}} \le \varphi_+(\omega)^{q} \omega \implies  \eqref{eq:outsideintervaldiff} \le
2\varphi_+(\omega)^q \omega. 
\]

For \eqref{eq:intervaldiff}, note that by the Taylor series for $e^x$, \[
\varphi(t) = \sum_{k=0}^\infty \frac{(-t^2/2)^{k}}{\sqrt{2\pi} k!} = \sum_{k=0}^{k^*}
\frac{(-t^2/2)^{k}}{\sqrt{2\pi} k!} + R(t).
\]
Thus the second term \eqref{eq:intervaldiff} can bounded by the maximum-over-$v$ of the
absolute value of \[
\sum_{k=0}^{k^*}\int \frac{\pr{\frac{x-u}{\nu_i}}^{v+2k} (-1/2)^{k}}{
\sqrt{2\pi} k !}
[G(du) - G_m(du)] + \int R\pr{\frac{x-u}
{\nu_i}} \pr{
\frac{x-u}
{\nu_i}}^v [G(du) - G_m(du)]
\]
The first term in the line above is zero since the moments match up to $2k^*+q$.
Therefore,\[
\eqref{eq:intervaldiff} \le \max_{0\le v\le q} \abs[\bigg]{ \int_{I_{j-1}\cup I_j
\cup I_{j+1}} \pr{\frac{u-x}{\nu_i}}^v R\pr{\frac{x-u}{\nu_i}}  (G(du) - G_m(du))
} \numberthis \label{eq:integral_to_bound}.
\]

For $x \in I_j$, in \eqref{eq:intervaldiff}, $|(x-u)/\nu_i| \le 2a\nu_\ell / \nu_i \le
2a$. On $[-2a, 2a]$, if we choose $k^* > (2a)^2/2$, then $R(t)$ is an infinite series with
alternating signs and decreasing entries. Thus, $R(t)$
is bounded by the first term of truncation\[
|R(t)| \le \frac{(t^2/2)^{k^* + 1}}{\sqrt{2\pi} (k^* + 1)!} \quad |t| \le 2a. 
\]
Hence the integral
\eqref{eq:integral_to_bound} is upper bounded by \begin{align*}
\eqref{eq:intervaldiff} &
\le 2 \cdot (2a)^q \cdot \frac{ \pr{(2a)^2/2}^{k^*+1}}{
\sqrt{2\pi} (k^*+1)!} \tag{$(2a)^v \le (2a)^q$} \\
&\le
\frac{2(2a)^q}{(2\pi) \sqrt{k^*+1}}\pr{\frac{2a^2}{k^*+1} e}^{k^*+1}  \tag{Stirling's formula $
(k^*+1)! \ge \sqrt{2\pi( k^*+1)} \pr{\frac{k^*+1}{e}}^{k^*+1}
$
} \\&
\le \frac{(2a)^q}{
\pi \sqrt{k^*+1}}\pr{\frac{e}{3}}^{k^*+1} \tag{Choosing $k^*+1 \ge 6a^2 \ge 6$}
\\
& \le \frac{(2a)^q}{
\pi \sqrt{k^*+1}} \exp\pr{-\frac{1}{2} \frac{k^*+1}{6}} \tag{$(e/3)^6 \le
e^{-1/2}$} \\
& \le \frac{(2a)^q}{\sqrt{k^*+1} \sqrt{\pi/2}} \underbrace{\varphi(a\nu_\ell/\nu_u) }_
{\varphi(\varphi_+(\omega))}
\tag{$k^*+1 \ge
6a^2 \ge 6(a\nu_\ell/\nu_u)^2$} \le \frac{(2a)^q}{\sqrt{k^* + 1} \sqrt{\pi/2}} \omega \\
&\le \frac{2^q}{\sqrt{3\pi}}\pr{\frac{\nu_u}{\nu_\ell}}^{q-1} \invphi^{q-1}(\omega) \omega
\tag{$k^* + 1 \ge 6a^2$}
\end{align*}
These bound $\eqref{eq:outsideintervaldiff} + \eqref{eq:intervaldiff}$ since the bounds
do not depend on $j$ or $x$. Therefore, \[ d_{q,i,M}(G,G_m) \le \pr {2 + \frac{2^q}{
\sqrt{3\pi}} (\nu_u/\nu_\ell)^{q-1}} \cdot \invphi^q
 (\omega) \omega \lesssim_{q, \nu_u, \nu_\ell} \log^{q/2}(1/\omega)
\omega.
\]

\subsubsection{Disciplining $G_m$ onto a fixed grid}
Now, consider a gridding of $G_m$ via $G_{m,\omega}$. We construct $G_{m, \omega}$ to be the
following distribution. For a draw $\xi \sim G_m$, let $\tilde \xi = \omega \sgn
(\xi)\lfloor |\xi|/\omega\rfloor$. We let $G_{m,\omega}$ be the distribution of $\tilde
\xi$.  $G_{m,\omega}$ has at most $m = (2k^* + q + 1) j^* + 1$ support points since $G_m$
has at most that many, and all its support points are multiples of $\omega$.

Since \[
\int g(x, u) G_{m,\omega}(du) = \int g(x, \omega \sgn
(u)\lfloor |u|/\omega\rfloor) \, G_m(du)
\]
we have that \[
\abs[\bigg]{\int g(x, u) G_{m,\omega}(du) - \int g(x,u) G_m(du)} \le \int |g
(x, \omega \sgn
(u)\lfloor |u|/\omega\rfloor) - g(x, u)| \, G_m(du)
\]
In the case of $g(x, u) = \pr{(x-u)/\nu_i}^v \varphi((x-u)/\nu_i)$, this function is
Lipschitz by \cref{lemma:lipschitz_smooth_varphi}, we thus
have that, \[
d_{q,i,M}(G_m, G_{m,\omega}) \le \int C_q \frac{\omega}{\nu_i} G_m
(du)\lesssim_{\nu_\ell, q} \omega.
\]

So far, we have shown that there exists a distribution with at most $m$
support points, supported on the lattice points $\br{j\omega : j\in \Z,
|j\omega| \in I}$, that approximates $G$ up to \[
\omega^*\equiv C_{q, \nu_u, \nu_\ell}  \omega \log^{q/2} (1/\omega)
\]
in $d_{\infty, M}^{(q)}(\cdot, \cdot)$.

\subsubsection{Covering the set of $G_{m,\omega}$}
Let $\Delta^{m-1}$ be the $(m-1)$-simplex of probability vectors in $m$ dimensions.
Consider discrete distributions supported on the support points of $G_{m,\omega}$, which
can be identified with a subset of $\Delta^{m-1}$. Thus, there are at most $N(\omega,
\Delta^{m-1}, \norm{\cdot}_1)$ such distributions that form an $\omega$-net in
$\norm{\cdot}_1$. Now, consider a distribution $G_{m,\omega}'$ where \[
\norm{G_{m,\omega}' - G_{m,\omega}}_1 \le \omega.
\]
Since $t^q \varphi(t)$ is bounded, we have that \[
d_{q,i,M}(G_{m,\omega}' , G_{m,\omega}) \le \omega \max_{0 \le v\le q} v^
{v/2}\varphi (
\sqrt{v})
\lesssim_q \omega
\]
by \cref{lemma:lipschitz_smooth_varphi}.

There are at most \[
\binom{1 + 2 \lfloor (M + a \nu_\ell) / \omega \rfloor}{m}
\]
configurations of $m$ support points. Hence there are a collection of at
most  \[
\binom{1 + 2 \lfloor (M + a \nu_\ell) / \omega \rfloor}{m} N(\omega, \Delta^{m-1},
\norm{\cdot}_1)
\] distributions $\mathcal G$ where for all $G \in \mathcal P(\R)$, \[
\min_{H\in \mathcal G} d_{\infty, M}^{(q)}(G, H) \le \omega^*.
\]

\subsubsection{Putting together} 
We have shown that 
\begin{align*}
N(\omega^*, \mathcal P(\R), d^{(q)}_{\infty, M}) &\le \binom{1 + 2 \lfloor (M + a
\nu_\ell) / \omega \rfloor}{m} N(\omega, \Delta^{m-1},
\norm{\cdot}_1)
\\
 &\le
\pr{\frac{(\omega + 2) (\omega + 2(M + a \nu_\ell)) e}{m}}^m \omega^{-2m} (2\pi m)^
{-1/2} \tag{(6.24) in \citet{jiang2020general}}.
\end{align*}

Since $\omega < 1$ and $m \ge 2\frac{12a^2 + 3 + q}{a \nu_\ell} (M + a \nu_\ell)$ given
the choice $k^* +1 > 6a^2$, the first
term is bounded by a constant raised to $m$\th{} power: \[
\frac{(\omega + 2) (\omega + 2(M + a \nu_\ell)) e}{m} \le \frac{3e}{m} (1 + 2(M + a
\nu_\ell)) \lesssim \frac{a \nu_\ell}{12 a^2 + 3 + q} \lesssim \nu_\ell.\]
Therefore, \[\log N(\omega^*, \mathcal P(\R), d^{(q)}_{\infty, M}) \lesssim_
{\nu_\ell, \nu_u, q} m \cdot
|\log
(1/\omega)| + m \lesssim_{\nu_\ell, \nu_u, q} m \log (1/\omega).\]

Finally, since $m = (2k^* + q + 1) j^* + 1$, recall that we have required $k^* + 1\ge
6a^2$, and it suffices to pick $k^* = \lceil 6a^2 \rceil$. Then
 \[m \lesssim_{q, \nu_u, \nu_\ell}
\log(1/\omega) \max\pr{\frac{M}{\sqrt{\log
(1/\omega)}}, 1}.
\]
 Hence, \[
\log  N(\omega^*, \mathcal P(\R), d^{(q)}_{\infty, M}) \lesssim_{q, \nu_u, \nu_\ell} \log(1/\omega)^2 \max
\pr{\frac{M}{\sqrt{\log
(1/\omega)}}, 1}.
\]

Lastly, let $K$ equal the constant in $\omega^* = K \log (1/\omega)^{q/2} \omega$. Note that
we can take $K \ge 1$. For some $c > 1$ such that $\log (c K)^{q/2} < c $, we plug in
$
\omega = \frac{\delta}{cK}
$
such that whenever $\delta < cK(\varphi(1) \minwith \varphi(\sqrt{q})) \minwith e^{-q/2}$,
the
covering number bound holds for \[
\omega^* = \frac{\delta}{c} \log(cK/\delta)^{q/2} \le \delta \log(1/\delta)^{q/2}.
\]
In this case, \begin{align*}
 \log N\pr{
  \delta \log(1/\delta)^{q/2}, \mathcal P(\R), d^{(q)}_{\infty, M}
} &\le \log N\pr{
  \omega^* , \mathcal P(\R), d^{(q)}_{\infty, M}
} \\
&\lesssim_{q, \nu_u, \nu_\ell} \log(1/\omega)^2 \max
\pr{\frac{M}{\sqrt{\log
(1/\omega)}}, 1} \\
&\lesssim_{q, \nu_u, \nu_\ell} \log(1/\delta)^2 \max
\pr{\frac{M}{\sqrt{\log
(1/\delta)}}, 1}
\end{align*}
This bound holds for all sufficiently small $\delta$. Since $\delta \log(1/\delta)^{q/2}$
is increasing over $(0, e^{-q/2} \minwith e^{-1})$ and the right-hand side does not vanish
over the interval, we can absorb larger $\delta$'s into the constant.
\end{proof}

As a consequence, we can control the covering number in terms of $d_{k, \infty, M}$ for
$k \in \br{m,s}$.

\begin{prop}
\label{prop:covering_ms}
Consider $d_{\infty, M}^{(q)}$ in \cref{prop:covering_for_moments} and  $d_
{s,\infty, M}$ and $d_
{m,\infty, M}$ in \eqref{eq:def_dk} for some $M > 1$.
Suppose $d_{\infty, M}^{(2)} (H_1, H_2)\le \delta$. Then we have \begin{align*}
 d_
{m,\infty, M}(H_1, H_2) \leH \frac{\sqrt{\log
(1/\rho_n)}}{\rho_n} \delta.\\
 d_
{s,\infty, M}(H_1, H_2) \leH \frac{M\sqrt{\log
(1/\rho_n)} + \log(1/\rho_n)}{\rho_n} \delta.
\end{align*}

As a corollary, for all $\delta \in (0, 1/e)$,
\begin{align*}
\log N\pr{\frac{\delta \log(1/\delta)}{\rho_n} \sqrt{\log (1/\rho_n)}, \mathcal P(\R), d_
{m,
\infty, M}} &\leH \log(1/\delta)^2 \max\pr{1, \frac{M}{\sqrt{\log(1/\delta)}}} \\
\log N\pr{\frac{\delta \log(1/\delta)}{\rho_n} \pr{M\sqrt{\log (1/\rho_n)} + \log
(1/\rho_n)},
\mathcal P(\R),
d_
{s,
\infty, M}} &\leH \log(1/\delta)^2 \max\pr{1, \frac{M}{\sqrt{\log(1/\delta)}}}.
\end{align*}
\end{prop}

\begin{proof}
Fix some $|z| \le M$. Let $k \in \br{m,s}$. Let $T_{mi} = f'_{G,\nu_i}(z)$ and $T_
{si} = Q_i(z, \eta_0, G)$. Observe that \begin{align*}
&|D_{k,i}(z, G_1, \eta_0, \rho_n) - D_{k,i}(z,G_2,\eta_0,\rho_n)| \\
&\leH \absauto{
  \frac{T_{ki}(z, \eta_0, G_1)}{f_{G_1, \nu_i}(z) \vee (\rho_n/\nu_i)}  - \frac{T_{ki}(z,
  \eta_0, G_2)}{f_{G_2, \nu_i}(z) \vee (\rho_n/\nu_i)} 
}\\
&\leH \absauto{
  \frac{T_{ki}(z, \eta_0, G_1)}{f_{G_1, \nu_i}(z) \vee (\rho_n/\nu_i)} - \frac{T_{ki}(z,
  \eta_0, G_2)}{f_{G_1, \nu_i}(z) \vee (\rho_n/\nu_i)} + \frac{T_{ki}(z, \eta_0, G_2)}{f_
  {G_1, \nu_i}(z) \vee (\rho_n/\nu_i)} - 
\frac{T_{ki}(z, \eta_0, G_2)}{f_{G_2, \nu_i}(z) \vee (\rho_n/\nu_i)}
}
\\
&\leH \frac{1}{\rho_n} \abs{
  T_{ki}(z, \eta_0, G_1) - T_{ki}(z, \eta_0, G_2)
} + \frac{|T_{ki}(z, \eta_0, G_2)| }{\rho_n (f_{G_2, \nu_i}(z) \vee (\rho_n/\nu_i))} |f_
{G_1,\nu_i} (z) - f_
{G_2,\nu_i} (z)|. 
\tag{$|f_1\vee \rho - f_2 \vee \rho| \le |f_1 - f_2|$ }
\end{align*}
Now, if $d^{(2)}_{\infty, M}(G_1, G_2) \le \delta$, then \begin{align*}
 |f_{G_1,\nu_i}(z) - f_{G_2,\nu_i}(z)| &\le \delta \\
\frac{\abs{T_{mi}(z, \eta_0, G_2)}}{f_{G_2,\nu_i}(z) \vee (\rho_n/\nu_i)} &\leH \sqrt{\log (1/\rho_n)} 
\tag{\cref{thm:jianglemma2}} \\
\frac{\abs{T_{si}(z, \eta_0, G_2)}}{f_{G_2,\nu_i}(z) \vee (\rho_n/\nu_i)} &\leH M\sqrt{\log (1/\rho_n)} +
\log(1/\rho_n) \tag{\cref{lemma:qbound}} \\ 
\abs{T_{mi}(z, \eta_0, G_1) - T_{mi}(z, \eta_0, G_2)} &= \absauto{
  \int \frac{u-z}{\nu_i} \varphi\pr{\frac{z-u}{\nu_i}} (G_1 - G_2)(du)
}  \le \delta \\
\abs{T_{si}(z, \eta_0, G_1) - T_{si}(z, \eta_0, G_2)} &\leH \abs[\bigg]{ \int 
  \frac{\overbrace{(z-\tau)\tau}^{-(z-\tau)^2 + z(z-\tau)}}{\nu_i^2} \varphi\pr{
  \frac{z-\tau} {\nu_i}} (G_1-G_2) (du)
}  \leH M \delta.
\end{align*}
As a result, \begin{align*}
|D_{m,i}(z, G_1, \eta_0, \rho_n) - D_{m,i}(z,G_2,\eta_0,\rho_n)| &\leH \frac{1}{\rho_n} 
\bk{
  \delta \sqrt{\log(1/\rho_n)}
} \\ 
|D_{s,i}(z, G_1, \eta_0, \rho_n) - D_{s,i}(z,G_2,\eta_0,\rho_n)| &\leH \frac{\delta}
{\rho_n} (M + M\sqrt{\log (1/\rho_n)} + \log (1/\rho_n)) \\&\leH \frac{\delta}{\rho_n} 
(M\sqrt{\log(1/\rho_n)} + \log(1/\rho_n)).
\end{align*}
This proves the first claim. 

For each $k \in \br{m,s}$, we can write \[
|D_{k,i}(z, G_1, \eta_0, \rho_n) - D_{k,i}(z,G_2,\eta_0,\rho_n)| \le C_\H r_{kn} \delta. 
\]
for some $r_{kn}$ and $C_\H > 1$. Fix some $\kappa > 0$ and let $\delta
= 2 C_\H \kappa$.
Then \[
\delta \log(1/\delta) = 2 C_\H  \kappa \log(1/\kappa) - 2 C_\H  \kappa \log (2 C_\H)
\]
For all sufficiently small $\kappa$ such that $\log(1/\kappa) > 2 \log (2C_\H)$, the above
is
bounded above by $C_\H \kappa \log (1/ \kappa)$. 
 This immediately shows that \begin{align*}
 \log N( r_{kn} \delta \log(1/\delta), \mathcal P(\R), d_{k,\infty, M}) &\le \log N
(r_{kn} C_\H \kappa \log(1/\kappa),  \mathcal P(\R), d_{k,\infty, M}) \\
&\le \log N
(\kappa \log(1/\kappa), 
\mathcal P(\R), d^
{(q)}_{\infty, M}) \\ 
&\leH \log(1/\kappa)^2 \max \pr{\frac{M}{\sqrt{\log(1/\kappa)}}, 1} \\
&\leH \log(1/\delta)^2 \max\pr{\frac{M}{\sqrt{\log(1/\delta)}}, 1}.
 \end{align*}
 This holds for all sufficiently small $\delta$. Bounds for larger $\delta$ can be
 absorbed into the constant. Plugging in $r_{kn}$ to the left-hand side concludes the
 result. 
\end{proof}

\subsection{Auxiliary lemmas}

\lemmalb*
\begin{proof}
For (1), observe that since $\hats,s_0$ are bounded away from 0 and $\infty$ under 
\cref{as:holder}, $|\hat Z_i|\vee 1 \leH (1 + \Delta_n) M_n +
\Delta_n \lesssim
(1+\Delta_n) M_n$. Hence by \cref{as:Delta_M_rate}, $|\hat
Z_i| \vee 1
\leH M_n$.

For (2), we note by Theorem 5 in \citet{jiang2020general}, \[
f_{\hat G_n, \hat\nu_i}(\hat Z_i) \ge \frac{1}{n^3 \hat\nu_i}
\] thanks to the choice $\kappa_n$ in \cref{as:npmle}. That is, \[
\int \varphi\pr{\frac{\hat Z_i - \tau}{\hat\nu_i}} \,\hat G_n(d\tau) \ge \frac{1}{n^3}.
\]
Now, note that \[
\frac{\hat Z_i - \tau}{\hat\nu_i} = 
\frac{Z_i - \tau}{\nu_i} + \frac{m_{0i} -
\hatm_{i}}{\sigma_i} + \frac{1}{\sigma_i}(\hats_i - s_{0i})\tau \equiv \frac{Z_i - \tau}
{\nu_i} +
\xi(\tau) \numberthis \label{eq:decompose_fake_z}
\]
where $|\xi(\tau)| \leH \Delta_n M_n$ over the support of $\tau$ under $\hat G_n$,  under
our assumptions.

Then, for all $Z_i$, since $|Z_i| \le M_n$ by assumption, \begin{align*}
\varphi\pr{
    \frac{\hat Z_i - \tau}{\hat\nu_i}
} &= \varphi\pr{ \frac{ Z_i - \tau}{\nu_i}} \exp\pr{
    -\frac{1}{2}\xi^2(\tau) - \xi(\tau)\frac{Z_i - \tau}{\nu_i}
}
\\
&\le \varphi\pr{ \frac{ Z_i - \tau}{\nu_i}}  \exp\pr{
    C_\Hyperparams \Delta_n M_n \absauto{\frac{ Z_i - \tau}{\nu_i}}
} \tag{$C_\Hyperparams$ is defined by optimizing over $|\xi(\tau)| \leH \Delta_n M_n$}
\\
&\le \varphi\pr{ \frac{Z_i - \tau}{\nu_i}} \exp\pr{
    C_{\Hyperparams} \Delta_n M_n^2
}. \tag{$\absauto{\frac{ Z_i - \tau}{\nu_i}} \leH M_n$}
\end{align*}

Therefore, \[
\int \varphi\pr{\frac{Z_i - \tau}{\nu_i}} \hat G_n(d\tau) \ge \frac{1}{n^3} e^
{-C_\Hyperparams \Delta_n M_n^2}.
\]
Dividing by $\nu_i$ on both sides finishes the proof of (2).
Claim (3) is immediate by calculating $\log(1/\rho_n) = \pr{3\log n + C_\Hyperparams M_n^2
\Delta_n} \vee \log(e\sqrt{2\pi}) \leH \log n$ and applying \cref{as:Delta_M_rate}(1) to
obtain that $\Delta_n
M_n^2 \leH 1$.
\end{proof}

\begin{lemma}[Lemma 2, \citet{jiang2020general}]
\label{thm:jianglemma2}
For all $x \in \R$ and all $\rho \in (0, 1/\sqrt{2\pi e})$,
\[
\abs[\bigg]{\frac{\nu^2f'_{H,\nu}(x)}{(\rho / \nu) \vee f_{H, \nu}(x)}} \le \nu
\invphi(\rho).
\]
Moreover, for all $x \in \R $ and all $\rho \in
(0, e^{-1} /
\sqrt{2\pi}),$
\[
\abs[\bigg]{\pr{\frac{
\nu^2 f''_{H, \nu}(x)}{f_{H, \nu}(z)} + 1} \pr{\frac{\nu f_{H, \nu}(x)}{(\nu f_{
G, \nu}(x)) \vee \rho}}} \le \invphi^2(\rho),
\]
where we recall $\invphi(\rho) = \sqrt{\log \frac{1}{2\pi \rho^2}}$ from 
\eqref{eq:invphidef}.
\end{lemma}

\begin{proof}
    The first claim is immediate from Lemma 2 in \citet{jiang2020general}.
    The second claim follows from parts of the proof, which we reproduce here. Lemma 1 in 
    \citet{jiang2020general}
    shows that \[
0 \le \frac{
\nu^2 f''_{H, \nu}(x)}{f_{H, \nu}(z)} + 1 \le {\log \frac{1}{2\pi \nu^2 f_
{H,\nu}(z)^2}} = {\invphi^2(\nu f_{H, \nu}(z))}.
    \]

We study two cases separetely depending on whether the truncation binds:
\begin{enumerate}[wide]
  \item $\nu f_
{H,\nu}(x) \le \rho < e^{-1}/\sqrt{2\pi}$: Observe that $t \log \frac{1}{2\pi t^2}$ is
increasing over $t\in (0, e^{-1} (2\pi)^
    {-1/2})$. Hence, \[
\pr{\frac{
\nu^2 f''_{H, \nu}(x)}{f_{H, \nu}(z)} + 1} \nu f_{H, \nu}(x)
\le \nu f_{H, \nu} \log \frac{1}{2\pi \nu^2 f_{H,\nu}(z)^2}
\le \rho \log \frac{1}{2\pi
\rho^2} = \rho \varphi_+^2(\rho).
    \]
    Dividing by $(\nu f) \vee \rho = \rho$ confirms the bound for $\nu f < \rho$.
    \item 
    $\nu f_
{H,\nu}(x) > \rho$:  Since $\log \frac{1}{2\pi t^2}$ is decreasing in $t$, we have
that \[
\abs[\bigg]{\pr{\frac{
\nu^2 f''_{H, \nu}(x)}{f_{H, \nu}(z)} + 1} \pr{\frac{\nu f_{H, \nu}(x)}{(\nu f_{
G, \nu}(x)) \vee \rho}}} =  {\frac{
\nu^2 f''_{H, \nu}(x)}{f_{H, \nu}(z)} + 1} \le \invphi^2(\nu f_{H,\nu}) \le
\log \frac{1}{2\pi \rho^2} = \varphi_+^2(\rho). \qedhere
    \]
\end{enumerate}
\end{proof}

\begin{lemma}[\citet{zhang1997empirical}, p.186]
\label{lemma:chebyshev}

Let $f$ be a density and let $\sigma(f)$ be the standard deviation of the corresponding
distribution, assumed to be finite. Then, for any $M, t > 0$,
\[
\int_{-\infty}^{\infty} \one (f(z) \le t) f(z) \,dz \le \frac{\sigma(f)^2}{M^2} + 2Mt.
\]
In particular, choosing $M = t^{-1/3}\sigma(f)^{2/3}$ gives \[
\int_{-\infty}^{\infty} \one (f(z) \le t) f(z) \,dz \le 3 t^{2/3}\sigma^{2/3}.
\]
\end{lemma}

\begin{proof}
Since the value of the integral does not change if we shift $f(z)$ to $f(z-c)$, it is
without loss of generality to assume that $\E_f[Z] = 0$.
\begin{align*}
\int_{-\infty}^{\infty} \one (f(z) \le t) f(z) \,dz &\le \int_{-\infty}^{\infty} \one (f(z)
\le t, |z| < M) f(z) \,dz  + \int_{-\infty}^{\infty} \one (f(z) \le t, |z| > M) f(z)
\,dz
\\
&\le \int_{-M}^M t\, dz + \P(|Z| > M)
\\
&\le 2Mt + \frac{\sigma^2(f)}{M^2}. \tag{Chebyshev's inequality}
\end{align*}
\end{proof}

\begin{lemma}
\label{lemma:qbound} Recall that $Q_i(z, \eta, G) = \int (z-\tau)\tau \varphi\pr{\frac
 {z-\tau}{\nu_i(\eta)}} \frac{1}{\nu_i(\eta)} \, G(d\tau)$ in \eqref{eq:dpsids}. Then,
 for any $G, z$ and $\rho_n
\in (0,
e^{-1}/\sqrt{2\pi})$,
\begin{equation}
    \abs[\bigg]{ \frac{Q_i(z, \eta_0, G)}{f_{G, \nu_i}(z) \vee (\rho_n / \nu_i)} }
    \le
    \invphi(\rho_n) \nu_i \pr{|z| + \nu_i \invphi(\rho_n)}.
\label{eq:qbound}
\end{equation}
Under the choice \eqref{eq:rhodef} and on the event $\bar Z_n \le M_n$ such that 
\cref{as:Delta_M_rate} holds, \[
\abs[\bigg]{ \frac{Q_i(z, \eta_0, G)}{f_{G, \nu_i}(z) \vee (\rho_n / \nu_i)} } \leH M_n
\sqrt{\log n}.
\]
\end{lemma}
\begin{proof}
We can write \[
Q_i(z, \eta_0, G) = f_{G, \nu_i}(z) \br{z\PE_{G,
\nu_i}[(z-\tau) \mid z]  -\PE_{ G_n,
\nu_i} [(z-\tau)^2 \mid z]
}.
\]
From \cref{thm:jianglemma2}, \[
\absauto{\frac{f_{G, \nu_i}(z)}{f_{G, \nu_i}(z) \vee (\rho_n / \nu_i)}
\PE_{G, \nu_i}[(z-\tau) \mid z]} \le \nu_i \invphi(\rho_n)
\]
and \[
\frac{f_{ G, \nu_i }(z)}{f_{ G, \nu_i}(z) \vee (\rho_n / \nu_i)}
\PE_{ G, \nu_i} [(z-\tau)^2 \mid z]
= \nu_i^2 \pr{\frac{\nu_i^2 f_{i, G}''}{f_{i,
G}} + 1} \frac{f_{G, \nu_i}(z)}{f_{G, \nu_i}(z) \vee
(\rho_n / \nu_i)} \le \nu_i^2 \invphi^2(\rho_n).
\]
Therefore, \[
\abs[\bigg]{ \frac{Q_i(z, \eta_0, G)}{f_{G, \nu_i}(z) \vee (\rho_n /
\nu_i)} } \le \invphi(\rho_n) \nu_i \pr{ |z| + \nu_i \invphi(\rho_n)}. \qedhere
\]
\end{proof}

\begin{lemma}
\label{lemma:secondderivatives}
    Under the assumptions in \cref{thm:lb} and \cref{as:holder}, suppose $\tilde \eta_i$
    lies on the line segment between $\eta_0$ and $\hateta_i$ and define $\tilde
    \nu_i, \tilde m_i, \tilde s_i, \tilde Z_i$ accordingly. Then, the
    second
    derivatives \eqref{eq:d2psidm2}, \eqref{eq:cross_psi_deriv},
    \eqref{eq:d2psids2}, evaluated at $\tilde\eta_i, \hat G_n, \tilde Z_i$,
    satisfy
    \begin{align*}
    |\eqref{eq:d2psidm2}| \leH \log n  \quad 
    |\eqref{eq:cross_psi_deriv}| \leH M_n \log n \quad 
    |\eqref{eq:d2psids2}| \leH M_n^2 \log n.
    \end{align*}
\end{lemma}

\begin{proof}
    First, we show that \[
    |\log(f_{\hat G_n, \tilde \nu_i}(\tilde Z_i) \tilde\nu_i)| \leH \log n.
    \numberthis \label{eq:log_denom}
    \]
    Observe that we can write $
    \hat Z_i = \frac{\tilde s_i \tilde Z_i + \tilde m_i - \hatm_i}{\hats_i}
    $
    where $\norm{\tilde s- \hats}_\infty \le \Delta_n$ and $\norm{\tilde m
    - \hatm}_\infty \le \Delta_n$. This shows that $|\tilde Z_i| \leH M_n$ under the
      assumptions since $\hat s > s_\ell$. Having verified that $|\tilde Z_i| \leH M_n$,
      note that by the same argument in
    \eqref{eq:decompose_fake_z} in \cref{thm:lb},
    we have
    that, since both $\bar Z_n$ is bounded under our assumptions and $\tau$ is bounded
    under $\hat G_n$,
    \[
    \varphi\pr{\frac{\hat Z_i - \tau}{\hat \nu_i}} \le \varphi
    \pr{\frac{\tilde Z_i - \tau}{\tilde \nu_i}} e^{C_\Hyperparams
    \Delta_n M_n^2}
    \implies
    \tilde \nu_i f_{\hat G_n, \tilde \nu_i}(\tilde Z_i) \ge \frac{1}{n^3}e^
    {-C_{\Hyperparams} \Delta_n M_n^2}.
    \]
    This shows \eqref{eq:log_denom}.

    Now, observe that \begin{align*}
    \PE_{\hat G_n, \tilde \nu}[(\tau - Z)^2 \mid \tilde Z_i] \leH \log
    \pr{\frac{1}{\tilde \nu_i f_{\hat G_n, \tilde \nu_i}(\tilde Z_i)}} \leH \log n \\
    \PE_{\hat G_n, \tilde \nu}[|\tau - Z| \mid \tilde Z_i] \leH \sqrt{\log
    \pr{\frac{1}{\tilde \nu_i f_{\hat G_n, \tilde \nu_i}(\tilde Z_i)}}}
    \leH \sqrt{\log n}
    \end{align*}
    by \cref{thm:jianglemma2}, since we can choose $\rho = \tilde
    \nu_i f_{\hat G_n, \tilde \nu_i}(\tilde Z_i) \minwith \frac{1}{\sqrt{2\pi} e}$.
    Similarly, by
    \cref{lemma:qbound}, and plugging in $\rho = \tilde
    \nu_i f_{\hat G_n, \tilde \nu_i}(\tilde Z_i) \minwith \frac{1}{\sqrt{2\pi} e}$, \[
    \absauto{\PE_{\hat G_n, \tilde \nu}[(\tau - Z)Z \mid \tilde Z_i]} \leH
    \sqrt{\log n} |\tilde Z_i| + \log n \leH M_n \sqrt{\log n} .
    \]
    Observe that, since $|\tau| \leH M_n$ under the support of $\hat G_n$, \[
    \absauto{\PE_{\hat G_n, \tilde \nu_i}[(\tau-Z)^2 \tau \mid \tilde Z_i]} \leH M_n
    \PE_{\hat G_n, \tilde \nu_i}[(\tau-Z)^2 \mid \tilde Z_i] \leH M_n \log n.
    \]
    Similarly, \[
    \PE_{\hat G_n, \tilde \nu_i}[(Z-\tau)^2 \tau^2 \mid \tilde Z_i] \leH
    M_n^2 \log n \quad
    \PE_{\hat G_n, \tilde \nu_i}[\tau^2 \mid \tilde Z_i] \leH M_n^2.
    \]
    Plugging these intermediate results into \eqref{eq:d2psidm2}, \eqref{eq:cross_psi_deriv},
    \eqref{eq:d2psids2} proves the claim.
\end{proof}

\begin{lemma}
\label{lemma:tail_bound}
Suppose $Z$ has simultaneous moment control $\E[|Z|^p]^{1/p} \le A p^{1/\alpha}$. Then \[
\P(|Z| > M) \le \exp\pr{-C_{A,\alpha} M^\alpha}.
\]
As a corollary, suppose $ Z \sim f_{G_0, \nu_i}(\cdot)$ and $G_0$ obeys \cref{as:moments},
then \[\P(|Z| > M) \le \exp\pr{-C_{A_0, \alpha, \nu_u} M^\alpha}.\]
\end{lemma}
\begin{proof}
Observe that
\begin{align*}
\P(|Z| > M) = \P(|Z|^p > M^p) \le \br{\frac{Ap^{1/\alpha}}{M}}^p \tag{Markov's
inequality}.
\end{align*}
Choose $p = (M/(eA))^{\alpha}$ such that \[
\br{\frac{Ap^{1/\alpha}}{M}}^p = \exp\pr{-p} = \exp\pr{- \pr{\frac{1}{eA}}^{\alpha} M^
{\alpha}}. \qedhere
\]
\end{proof}

\begin{lemma}
\label{lemma:union_bound}
Let $E$ be some event and assume that \[
\P(E, A > a) \le p_1 \quad \P(E, B>b) \le p_2
\]
Then $\P(E, A+B>a+b) \le p_1 + p_2$
\end{lemma}
\begin{proof}
Note that $A+B > a+b$ implies that one of $A>a$ and $B>b$ occurs. Hence \[
\P(E, A+B>a+b) \le \P(\br{E, A>a} \cup \br{E, B>b}) \le p_1 + p_2
\]
by union bound.
\end{proof}

\begin{lemma}
\label{lemma:posterior_moments}
    Let $\tau \sim G_0$ where $G_0$ satisfies \cref{as:moments}. Let $Z
    \mid \tau \sim \Norm(\tau, \nu^2)$. Then the posterior moment is
    bounded by a power of $|z|$: \[
    \E[|\tau|^p \mid Z=z] \lesssim_{p,\alpha, A_0} (|z| \vee 1)^p.
    \]

\end{lemma}

\begin{proof}
Let $M = |z| \vee 2$. We write \[
    \E[|\tau|^p \mid Z=z] = \frac{1}{f_{G_0, \nu}(z)} \int |\tau|^p
    \varphi\pr{\frac{z-\tau}{\nu}} \frac{1}{\nu} G_0(d\tau).
    \]
    Note that we can decompose based on $|\tau| > 3M$: 
    \begin{align*}
    \int |\tau|^p
    \varphi\pr{\frac{z-\tau}{\nu}} \frac{1}{\nu} G_0(d\tau) &\le (3M)^pf_{G_0, \nu}(z) +
    \int \one(|\tau| > 3M) |\tau|^p \varphi\pr{\frac{z-\tau}{\nu}}\frac{1}
    {\nu} G_0(d\tau)
    \\
    &\le (3M)^pf_{G_0, \nu}(z) + \int_{|\tau | > 3M} |\tau|^p G_0(d\tau) \cdot  \frac{1}
    {\nu}\varphi\pr{|2M|/\nu} \tag{$|z-\tau| \ge 2M$ when $|\tau| > 3M$}
    \end{align*}
    Also note that, since $|z| \le M$, \begin{align*}
    f_{G_0, \nu}(z) = \int \varphi\pr{\frac{z-\tau}{\nu}}\frac{1}
    {\nu} G_0(d\tau) 
\ge \int_{-M}^M \varphi\pr{\frac{z-\tau}{\nu}}\frac{1}
    {\nu} G_0(d\tau)
    \ge \frac{1}{\nu} \varphi\pr{|2M| / \nu} G_0([-M, M])
    \tag{$|z-\tau|
        \le 2M$ if $\tau \in [-M, M]$}
    \end{align*}
    Hence, \[
    \E[|\tau|^p \mid Z=z] \le (3M)^p + \frac{\int |\tau|^p
    G_0(d\tau)}{G_0([-M,M])}.
    \]
    Since $G_0$ is mean zero and variance 1, by Chebyshev's inequality, $G_0([-M, M]) \ge
    G_0([-2,2]) \ge 3/4$. Hence $
    \E[|\tau|^p \mid Z=z] \lesssim_{p, \alpha, A_0} M^p \lesssim_{p, \alpha, A_0} (|z| \vee 1)^p
    $,
    since we have bounded $p$\th{} moments by \cref{as:moments}.
\end{proof}

\section{A large-deviation inequality for the average Hellinger distance}
\label{sec:hellinger}

\begin{theorem}
\label{thm:large-deviation}
For some $n \ge 7$, let $\tau_1, \ldots, \tau_n \mid (\nu_1^2,\ldots, \nu_n^2)
\iid G_0$
where $G_0$
satisfies \cref{as:moments}.
Let $\nu_u = \max_i \nu_i$ and $\nu_\ell = \min_i \nu_i$.
Assume $Z_i \mid \tau_i, \nu_i^2 \sim \Norm(\tau_i,
\nu_i^2)$.
Fix positive sequences $\gamma_n, \lambda_n \to 0$ with $\gamma_n, \lambda_n \le 1$ and
constant $\epsilon > 0$. Fix some positive constant $C^*$.
Consider the set of distributions that approximately maximize the likelihood
\[
A(\gamma_n, \lambda_n) = \br{
    H \in \mathcal P(\R) : \sub_n(H) \le C^* \pr{\gamma_n^2 + \bar h(\fs{H}, \fs{G_0}) \lambda_n}
}.
\]
Also consider the set of distributions that are far from $G_0$ in $\bar h$: \[
B(t, \lambda_n, \epsilon) = \br{
    H  \in \mathcal P(\R) : \barh(\fs{H}, \fs{G_0}) \ge t B \lambda_n^{1-\epsilon}
}
\]
with some constant $B$ to be chosen.
Assume that for some $C_\lambda$, \[
\lambda_n^2 \ge \pr{\frac{C_\lambda}{n} (\log n)^{1+\frac{\alpha+2}{2 \alpha}} }\vee \gamma_n^2.
\numberthis
\label{eq:rate_lambda_gamma}
\] Then the probability that $A \cap B$ is nonempty is bounded for $t > 1$: There exists
a choice of $B$ that depends only on $\nu_\ell, \nu_u, C^*, C_\lambda$ such  that \[
\P\bk{
    A(\gamma_n, \lambda_n) \cap B(t, \lambda_n, \epsilon)  \neq \emptyset
} \le (\log_2(1/\epsilon) + 1) n^{-t^2}.
\numberthis
\label{eq:nonemptyintersection}
\]
\end{theorem}

\corhellinger*
\begin{proof}
Let $\gamma = \frac{2 + \alpha}{2 \alpha} + \beta$. We first note that, for
$\varepsilon_n$
in
\eqref{eq:specific_rate_sub}, the choices
\[\lambda_n = n^{-p/(2p+1)} (\log n)^{\frac{2 + \alpha}{2 \alpha} + \beta} \minwith 1
=
\gamma_n\] satisfy \eqref{eq:rate_lambda_gamma}. Note that the choices of $\lambda_n,
\gamma_n$ are such that $\varepsilon_n \le C_\H (\lambda_n \barh + \gamma_n^2)$.

The event $\br{A_n, \barh(\fs{\hat G_n}, \fs{G_0}) > t \delta_n}$ is a subset of the union
of \[ E_1 = \br{A_n, \sub_n(\hat G_n) > C_\Hyperparams^*
\varepsilon_n}
\text{ and }
E_2 = \br{
    A_n, \sub_n(\hat G_n) \le C_\Hyperparams^*
    \varepsilon_n, \bar h(\fs{\hat G_n}, \fs{G_0}) > t  n^{-p/(2p+1)} (\log
    n)^{\gamma}
}.
\]
Thus $\P\bk{ A_n, \barh(\fs{\hat G_n}, \fs{G_0}) > t \delta_n } \le \P(E_1) +
    \P (E_2)$.  
    \Cref{cor:suboptimality} implies that $\P(E_1) \le 9/n$.

Now, note that
\begin{align*}
\P(E_2) \le \P\bk{
  A_n,  \sub_n(\hat G_n) \le C_\Hyperparams^* C_\Hyperparams (\lambda_n \barh +
  \gamma_n^2),  \bar h(\fs{\hat G_n}, \fs{G_0})  \ge t\lambda_n
}.
\end{align*}
Observe that, for $\epsilon = 1/\log(n)$ \begin{align*}
t \lambda_n^{1-\epsilon} &= t  \bk{n^{-\frac{p}{2p+1}(1-\epsilon)} (\log n)^{\gamma
(1-\epsilon)} \minwith 1}  \\
&= t \pr{n^{-\frac{p}{2p+1}} (\log n)^{\gamma} \bk{n^{\frac{\epsilon p}{2p+1}} (\log n)^
{-\gamma \epsilon}} \minwith 1} \\ 
&= t \pr{n^{-\frac{p}{2p+1}} (\log n)^{\gamma} \bk{e^{\frac{p}{2p+1}} (\log n)^
{-\gamma \epsilon}} \minwith 1} \\
&\le C_{p,\gamma} t \lambda_n \tag{$e^{\frac{p}{2p+1}} (\log n)^
{-\gamma \epsilon}$ is bounded by a constant}
\end{align*}
Thus, by \cref{thm:large-deviation}, for all sufficiently large $t$, \begin{align*}
\P(E_2) &\le \P\bk{
  \sub_n(\hat G_n) \le C_\Hyperparams^* C_\Hyperparams (\lambda_n \barh(\fs{\hat G_n}, \fs{G_0}) +
  \gamma_n^2), \barh(\fs{\hat G_n}, \fs{G_0}) \ge \frac{t}{C_{p,\gamma}} \lambda_n^{1-\epsilon}
  } \\
  &\le \P\bk{
  A(\gamma_n, \lambda_n) \cap B\pr{\frac{t}{BC_{p,\lambda}}, \lambda_n, \epsilon} \neq
  \emptyset
  } \le \pr{\log_2(\log n) + 1} n^{-t^2/C_\H}
\end{align*}
We can pick $t = B_\H$ sufficiently large such that $n^{-t^2/C_\H} \le 1/n$ and \[
\P\bk{ A_n, \barh(\fs{\hat G_n}, \fs{G_0}) > t \delta_n } \le \P(E_1) +
    \P (E_2) \le \pr{\frac{\log \log n}{\log 2} + 10} \frac{1}{n}. \qedhere
\]
\end{proof}

\subsection{Proof of \cref{thm:large-deviation}}

\subsubsection{Decompose $B(t, \lambda_n, \epsilon)$}

We decompose $
B(t, \lambda_n, \epsilon) \subset \bigcup_{k=1}^{K} B_k(t,\lambda_n)
$
where, for some constant $B > 1$ to be chosen, \[
B_k = \br{
    H : \barh\pr{\fs{H}, \fs{G_0}} \in \left(t B\lambda_n^{1-2^{-k}}, t B\lambda_n^{1-2^
    {-k+1}}\right]
}.
\]
The relation $B(t, \lambda_n, \epsilon)\subset \bigcup_k B_k$ holds if we take $K
= \lceil |\log_2(1/\epsilon)| \rceil$. 

In the remainder, we will bound \[
\P(A(\gamma_n, \lambda_n) \cap B_k(t, \lambda_n) \neq \emptyset) \le n^{-t^2}
\]
which becomes the bound \eqref{eq:nonemptyintersection} by a union bound. This argument
follows the argument for Theorem 7 in \citet{soloff2021multivariate} 
and Theorem 4 in \citet{jiang2020general}. For $k \in [K]$, define $\mu_{n,k} = B
\lambda_n^ {1-2^ {-k+1}}$ such that $B_k = \br{ H : \barh\pr{\fs{H}, \fs{G_0}} \in \left(t
\mu_{n,k+1}, t \mu_{n,k}\right] }.$ To that end, fix some $k$.

\subsubsection{Construct a net for the set of densities $f_G$}
Fix a positive constant $M$ and define the pseudonorm \[
\norm{G}_{\infty, M} = \max_{i\in [n]} \sup_{y \in [-M,M]} f_{G,\nu_i}(y).
\]
Note that $\norm{G_1-G_2}_{\infty, M} \rateeq_{\Hyperparams} d_{\infty, M}^{(0)}(G_1,
G_2)$
defined in
\cref{prop:covering_for_moments}. Fix $\omega = \frac{1}{n^2} > 0$ and consider an
$\omega$-net for $\mathcal P(\R)$ under $\norm{\cdot}_{\infty, M}$. Let
$N = N(\omega, \mathcal P(\R),
\norm{\cdot}_{\infty, M})$ and the $\omega$-net consists of the distributions $H_1,\ldots,
H_N$.
For each $j \in [N]$, let $H_{k,j}$ be a distribution, if it exists, with  \[
\norm{H_{k,j} - H_{j}}_{\infty, M} \le \omega \quad \bar h(\fs{H_{k,j}}, \fs{G_0}) \ge t
\mu_ {n,k+1} \numberthis \label{eq:Hkj_def}
\]
and let $J_k$ collect the indices $j$ for which $H_{j,k}$ exists.

\subsubsection{Project to the net and upper bound the likelihood}
Fix a distribution $H \in B_k(t, \lambda_n)$. There exists some member of the covering,
$H_j$, such that $ \norm{H -
H_j}_{\infty, M} \le \omega$. Moreover, $H$ serves as a witness that $H_{k,j}$
exists, with $
\norm{H - H_{k,j}}_{\infty, M} \le 2 \omega.
$

We can construct an upper bound for $f_{H, \nu_i}(z)$ via \[
f_{H, \nu_i}(z) \le \begin{cases}
    f_{H_{k,j}, \nu_i}(z) + 2 \omega & |z| \le M \\
    \frac{1}{\sqrt{2\pi} \nu_i} & |z| > M.
\end{cases}
\]
Define $v(z) = \omega \one(|z| \le M) + \frac{\omega M^2}{z^2} \one(|z| > M).$ Observe
that
\begin{align*}
f_{H,\nu_i}(z) \le
\begin{cases}
f_{H_{k,j}, \nu_i}(z) + 2v(z)   & |z| \le M\\
 \frac{f_{H_{k,j}, \nu_i}(z) + 2v(z)}{\sqrt{2\pi} \nu_i v(z)} & |z| > M.
\end{cases}
\end{align*}
Hence, the likelihood ratio between $H$ and $G_0$ is upper bounded:
\begin{align*}
\prod_{i=1}^n \frac{f_{H, \nu_i}(Z_i)}{ f_{G_0, \nu_i}(Z_i)} &\le
\prod_{i=1}^n \frac{f_{H_{k,j}, \nu_i}(Z_i) + 2v(Z_i)}{f_{G_0, \nu_i}(Z_i)} \prod_{i : |Z_i|
> M} \frac{1}{\sqrt{2\pi} \nu_i v(Z_i)} \\&\le \pr{\max_{j \in J_k}\prod_{i=1}^n \frac{f_
{H_
{k,j}, \nu_i}(Z_i) + 2v(Z_i)}{f_{G_0, \nu_i}(Z_i)}} \prod_{i : |Z_i|
> M} \frac{1}{\sqrt{2\pi} \nu_i v(Z_i)}
\end{align*}
If $H \in A(t, \gamma_n, \lambda_n)$, then the likelihood ratio is also lower bounded: \begin{align*}
\prod_{i=1}^n \frac{f_{H, \nu_i}(Z_i)}{ f_{G_0, \nu_i}(Z_i)} &\ge \exp\pr{-n C^*
(\gamma_n^2 + \barh\pr{\fs{H}, \fs{G_0}} \lambda_n)} \\& \ge \exp\pr{-n t C^*
(t\gamma_n^2 + \barh\pr{\fs{H}, \fs{G_0}} \lambda_n)} \tag{$t>1$} \\ 
&\ge \exp\pr{-n C^* (t^2 \gamma_n^2 + t \bar h \lambda_n)}\\
&\ge \exp\pr{-n C^* (t^2 \gamma_n^2 + t^2 \mu_{n,k} \lambda_n)}.
\end{align*}
Hence, \begin{align*}
&\P\bk{
    A(t, \gamma_n, \lambda_n) \cap B_k(t, \lambda_n) \neq \emptyset
} \\ &\le \P\Bigg\{
    \pr{\max_{j \in J_k}\prod_{i=1}^n \frac{f_{H_
{k,j}, \nu_i}(Z_i) + 2v(Z_i)}{f_{G_0, \nu_i}(Z_i)}} \prod_{i : |Z_i|
> M} \frac{1}{\sqrt{2\pi} \nu_i v(Z_i)} \ge \exp\pr{-nt^2 C^*(\gamma_n^2 +
\mu_{n,k} \lambda_n) }
\Bigg\}
\\
&\le
\P\bk{
    \max_{j\in J_k} \prod_{i=1}^n \frac{f_{H_{k,j}, \nu_i} + 2v(Z_i)}{f_{G_0, \nu_i}
    (Z_i)} \ge e^{-n t^2 a C^*  (\gamma_n^2 + \mu_{n,k} \lambda_n)}
} \numberthis \label{eq:term1}
\\
&\quad\quad+ \P\bk{
    \prod_{i:
|Z_i| > M} \frac{1}{\sqrt{2\pi} \nu_i v(Z_i) } \ge e^{n t^2 (a - 1) C^* (\gamma_n^2 +
 \mu_{n,k} \lambda_n)}
} \numberthis \label{eq:term2}
\end{align*}
The second
inequality follows from choosing some $a > 1$ and applying union bound.

\subsubsection{Bounding \eqref{eq:term1}} We consider bounding the first term
\eqref{eq:term1} now:
\begin{align*}
\eqref{eq:term1}
&\le \sum_{j\in J_k} \P\bk{\prod_{i=1}^n \frac{f_{H_{k,j}, \nu_i} + 2v(Z_i)}{f_{G_0,
\nu_i}
    (Z_i)} \ge e^{-n a t^2 C^* (\gamma_n^2 + \mu_{n,k} \lambda_n)}} \tag{Union bound}
\\
&\le \sum_{j\in J_k} \E\bk{
    \prod_{i=1}^n \sqrt{\frac{f_{H_{k,j}, \nu_i}(Z_i) + 2v(Z_i)}{f_{G_0, \nu_i}
    (Z_i)}}
    } e^{n a t^2 C^* (\gamma_n^2 + \mu_{n,k} \lambda_n)/2} \tag{Take square root
    of both sides, then apply Markov's inequality} \\
    &= \sum_{j\in J_k} e^{n a t^2 C^* (\gamma_n^2 + \mu_{n,k} \lambda_n)/2} \prod_
    {i=1}^n \E\bk{
    \sqrt{\frac{f_{H_{k,j}, \nu_i}(Z_i) + 2v(Z_i)}{f_{G_0, \nu_i}
    (Z_i)}}
    } \numberthis \label{eq:term1_inter}
\end{align*}
where the last step \eqref{eq:term1_inter} is by independence over $i$.
Note that \begin{align*}
\E\bk{
    \sqrt{\frac{f_{H_{k,j}, \nu_i}(Z_i) + 2v(Z_i)}{f_{G_0, \nu_i}
    (Z_i)}}
    } &= \int_{-\infty}^{\infty} \sqrt{ f_{H_{k,j}, \nu_i}(x) + 2v(x)} \sqrt{f_{G_0,
    \nu_i}
    (x)} \,dx
    \\
    &\le 1-h^2(f_{H_{k,j}, \nu_i}, f_{G_0, \nu_i}) + \int_{-\infty}^\infty \underbrace{
    \sqrt{2v (x)
        f_{G_0,\nu_i}(x)}}_{\sqrt{2v(x)/f_{G_0, \nu_i}} \cdot f_{G_0, \nu_i}} \,dx \tag{$
        \sqrt{a+b} \le \sqrt{a} + \sqrt{b}$}
    \\
    & \le  1-h^2(f_{H_{k,j}, \nu_i}, f_{G_0,\nu_i}) + \pr{2\int_{-\infty}^\infty v
    (x)\,dx}^{1/2} \tag{Jensen's inequality}
    \\
    &=  1-h^2(f_{H_{k,j}, \nu_i}, f_{G_0,\nu_i}) + \sqrt{8M   \omega}. 
    \tag{Direct integration}
\end{align*}

Also note that, for $t_i > 0$, we have \[
\prod_{i=1}^n t_i = \exp \pr{\sum_{i=1}^n \log t_i} \le \exp\pr {\sum_{i=1}^n (t_i - 1)}.
\]
Thus, \[
\prod_{i=1}^n \E
\bk{
    \sqrt{\frac{f_{H_{k,j}, \nu_i} + 2v(Z_i)}{f_{G_0, \nu_i}
    (Z_i)}}
    } \le \exp\bk{
    - n\barh^2(\fs{H_{k,j}}, \fs{G_0}) + n\sqrt{8M\omega}
    }.
\]

Thus, we can further bound \eqref{eq:term1_inter}: \begin{align*}
\eqref{eq:term1} &\le \eqref{eq:term1_inter} = \sum_{j\in J_k} e^{n \alpha t^2
(\gamma_n^2 +
\mu_{n,k}
\lambda_n)/2} \prod_{i=1}^n \E
\bk{
    \sqrt{\frac{f_{H_{k,j}, \nu_i} + 2v(Z_i)}{f_{G_0, \nu_i}
    (Z_i)}}
    }
\\
    &\le  \sum_{j\in J_k} \exp\br{
    \frac{na t^2C^*}{2} (\gamma_n^2 + \mu_{n,k}\lambda_n) - n \barh^2(
    \fs{H_{k,j}}, \fs{G_0} )+ n \sqrt{8M\omega}
    } \\
    & \le \sum_{j\in J_k} \exp\br{
    \frac{na t^2C^*}{2} (\gamma_n^2 + \mu_{n,k}\lambda_n) - n t^2 \mu_{n,k+1}^2 + n
    \sqrt{8M\omega}
    } \tag{$\barh^2(
    \fs{H_{k,j}}, \fs{G_0}) \ge t \mu_{n,k+1}$ by \eqref{eq:Hkj_def}} \\
    &\le  \exp\br{
    \frac{na t^2C^*}{2} (\gamma_n^2 + \mu_{n,k}\lambda_n) - n t^2 \mu_{n,k+1}^2 + n
    \sqrt{8M\omega} + \log N
    } \tag{$|J_k| \le N$} \\
    & \le \exp\br{
    \frac{na t^2C^*}{2} (\gamma_n^2 + \mu_{n,k}\lambda_n) - n t^2 \mu_{n,k+1}^2 + n
    \sqrt{8M\omega} + C \abs{\log \omega}^2 \max\pr{\frac{M}{\sqrt{|\log \omega|}}, 1}
    } \tag{\cref{prop:covering_for_moments}, $q=0$} \\
    &= \exp\br{
    \frac{na t^2C^*}{2} (\gamma_n^2 + \mu_{n,k}\lambda_n) - n t^2 \mu_{n,k+1}^2 +
    \sqrt{8M} + C (\log n)^2 \max\pr{\frac{M}{\sqrt{\log n}}, 1}
    }  \tag{Recall that $\omega = \frac{1}{n^2}$}.
\end{align*}

\subsubsection{Bounding \eqref{eq:term2}}

We now consider bounding the second term \eqref{eq:term2}. By Markov's inequality again
(taking $x \mapsto x^{1/(2\log n)}$ on both sides
),
we can choose to bound \[
\eqref{eq:term2} \le \E\bk{
    \prod_{i=1}^n \pr{\frac{1}{(2\pi \nu_i^2)^{1/4}} \frac{Z_i}{M \sqrt{\omega}}
}^{\frac{1}{\log n} \one(|Z_i| > M)}
} \exp\pr{-\frac{n (a - 1) t^2 C^* (\gamma_n^2 + \mu_{n,k} \lambda_n)}{2\log n}}
\]
instead. Define \[
a_i = \frac{1}{(2\pi \nu_i^2)^{1/4} M \sqrt{\omega}} \le  \frac{C_{\nu_\ell}n}{M} \quad
\lambda = \frac{1}{\log n}
\]
Apply \cref{lemma:boundsecond} to obtain the following. Note that to do so, we require $ M
\ge \nu_u \sqrt{8\log n}$ and $p \ge \frac{1}{\log n}. $
\begin{align*}
\log \E\bk{
    \prod_{i=1}^n \pr{\frac{1}{(2\pi \nu_i^2)^{1/4}} \frac{Z_i}{M \sqrt{\omega}}
}^{\frac{1}{\log n} \one(|Z_i| > M)}}
&= \log \E\bk{\prod_i (a_i Z_i)^{ \lambda \one(|Z_i| \ge M) }
}
\\
&\lesssim_{\nu_u} \sum_{i=1}^n (a_i M)^{\lambda} \pr{\frac{1}{Mn} + \frac{2^p \mu_p^p
(G_0)}{M^p}} 
\tag{\cref{lemma:boundsecond}}
\\
&\le \sum_{i=1}^n (C_{\nu_\ell} n)^{\frac{1}{\log n}} \pr{\frac{1}{Mn} + \frac{2^p \mu_p^p
(G_0)}{M^p}}
\\
&\lesssim_{\nu_u, \nu_\ell} \frac{1}{M} +  n \frac{2^p \mu_p^p
(G_0)}{M^p}
\end{align*}
As a result,
\begin{align*}
\log [\eqref{eq:term2}] & \le C_{\nu_u, \nu_\ell} \pr{\frac{1}{M} + \frac{2^p n \mu_p^p
(G_0)}{M^p}} -  \frac{n (a-1) }{2\log n} t^2 C^* \pr{\gamma_n^2 + B \lambda_n^{2(1-2^
{-k})}}. \numberthis \label{eq:second_bound_log}
\end{align*}

To conclude, note that by \cref{as:moments}, $\mu_p^p(G_0) \le A_0^p p^{p/\alpha}$. Let $M = 2e A_0 (c_m \log
n)^{1/\alpha}$ and $p = (M / (2eA_0))^{1/\alpha}$ so that \[
2^p \mu_p^p(G_0) / M^p \le \exp\pr{-c_m \log n}
\]
We choose $c_m \ge 2$ sufficiently large such that $M = 2e A_0 (c_m \log n)^{1/\alpha} >
\nu_u
\sqrt{8 \log n} \vee 1$ and $p \ge 1$ for all $n > 2$ to ensure that our application of
\cref{lemma:boundsecond} is correct. We also choose $a = 1.5$.

Plugging in these choices, we can verify that, via \eqref{eq:rate_lambda_gamma},
\begin{align*}
\log [\eqref{eq:term2}] &\le t^2 \bk{
    2 C_{\nu_u, \nu_\ell} - \frac{C^* B C_\lambda}{4} \pr{\log n}
} \\ 
\log [\eqref{eq:term1}] &\le -t^2 (\log n)^{1+\frac{2+\alpha}{2\alpha}}\bk{ C_\lambda \pr{-\frac{3}{4}C^* -
\frac{3}
{4}C^*B + B^2} - C}
\end{align*}
There exists a sufficiently large choice of $B$ such that $\log [\eqref{eq:term2}] \le
-t^2 \log n - \log 2$ and $\log [\eqref{eq:term1}] \le -t^2 \log n - \log 2$.  Thus, we
obtain that
$
\eqref{eq:term1} + \eqref{eq:term2} \le n^{-t^2}.
$
This concludes the proof.

\subsection{Auxiliary lemmas}

\begin{lemma}[Lemma 5, \citet{jiang2020general}]
\label{lemma:boundsecond}

Suppose $Z_i \mid \tau_i \sim \Norm(\tau_i, \nu_i^2)$ where $\tau_i \mid \nu_i^2 \sim G_0$
independently across $i$. Let $0 < \nu_u, \nu_\ell <\infty$ be the upper and lower bounds
for $\nu_i$. Then,
for all constants $M > 0,\lambda > 0,a_i > 0, p \in \N$
such that $M \ge
\nu_u \sqrt{8\log n}$,
$\lambda \in (0, p\minwith 1)$, and $a_1,\ldots, a_n > 0$:\[
\E\br{\prod_i |a_i Z_i|^{\lambda \one(|Z_i|\ge M)}} \le \exp\br{
    \sum_{i=1}^n (a_i M)^{\lambda} \pr{\frac{4\nu_u}{Mn\sqrt{2\pi}} + \pr{\frac{2\mu_p
    (G_0)}
    {M}}^p}
},
\]
where $\mu_p^p(G_0) = \int |\tau|^p G_0(d\tau)$.
\end{lemma}

\newpage

\part{Additional theoretical results}

\section{Estimating $\eta_0$ by local linear regression}
\label{sec:nuisance_estimation}

This section details how we estimate $\eta_0 = (m_0(\cdot), s_0(\cdot))$ by local linear
regression in \cref{sec:empirical}. It also outlines a detailed procedure and verifies
that this procedure satisfies the conditions we require for the conditional moment
estimation, when the true $\eta_0$ belong to a H\"older class of order $p=2$:
$m_0(\sigma), s_0(\sigma) \in C_{A_1}^2([\sigl,
\sigu])$. 

In our empirical application, we estimate $m_0, s_0$ by nonparametrically regressing $Y_i$
on $x_i \equiv \log_{10}(\sigma_i)$.\footnote{Correspondingly, let $\sigma(x) = 10^x.$}
Our procedure takes the following steps, which simply use kernel-based nonparametric
regression procedures implemented by \citet{calonico2019nprobust}
to estimate $m_0$ and $s_0$ and truncate the estimated $s_0$ below at some data-driven
point. Nonparametric regression is frequently applied to visualize data and to estimate
causal effects in regression discontinuity settings. 
\begin{enumerate}[label=(E-\arabic*)]
  \item \label{item:e1} Use the default procedure \citet{calonico2019nprobust} to
  estimate
  local linear regression of $Y_i$ on $x_i$ (Epanechnikov kernel, IMSE direct plug-in
  bandwidth). The resulting estimated conditional mean is
  $\hatm(\cdot)$. 
  \item Let $\hat R_i^2 = (Y_i - \hatm(x_i))^2$. Use the above local linear regression
   procedure again to estimate the conditional mean of $\hat R_i^2$ on $x_i$, and let
   $\hat v(x)$ be the estimated conditional mean. Let \[
\tilde s^2(\sigma_i) = \hat v(x_i) - \sigma_i^2. 
  \]
  \item Since $\hat v(x_i)$ is a linear smoother, it can be written as \[
\hat v(x) = \sum_{i=1}^n \ell_i(x) \hat R_i^2. 
  \] 
  for some weights $\ell_i(x)$.
  Let an estimate of the effective sample size be \[
p_n = \frac{1}{n}\sum_{i=1}^n \frac{1}{\sum_{j=1}^n \ell_i^2(x_j)}.
  \]
  \item \label{item:e4} Let the estimated conditional variance be \[
\hat s^2(\sigma_i) = \tilde s^2(\sigma_i) \vee \frac{2}{p_n+2} \pr{\hat v(x_i) \vee
\min_{j=1,\ldots, n} \sigma_j^2}
  \]
  where the additional truncation by $\min_{i=1,\ldots, n} \sigma_i^2$ deals with the
  unlikely scenario that $\hat v(x_i)$ is negative. Note that, in theory, the population
  analogue $v (x_i) = \E
  [R_i^2 \mid x_i] = s_0^2(\sigma_i) + \sigma_i^2 \ge \sigma_i^2$. See \cref{rmk:truncate}
  for a heuristic rationale of the above truncation rule. 
\end{enumerate}

The rest of the section analyzes the theoretical properties of a similar procedure for
analyzing $m_0(\cdot), s_0(\cdot)$ and connects them to the requirements in 
\cref{as:holder}. The product of this analysis is \cref{thm:llr_main_regret}, which
 verifies the same regret bound as in \cref{cor:maintext}, where we estimate $m_0, s_0$
 with the procedure we outline below.

There are a few minor inconveniences of the
above
procedure that we
strengthen below:
\begin{itemize}
  \item We would like to control for the fact that the bandwidths $\hat h_n$ for the local
  linear
  regression is data-driven. However, to establish uniform behavior in $\hat h_n$, we
  would like to restrict it to satisfy the optimal convergence rate almost surely: 
  For some $C > 0$, \[
C^{-1} n^{-1/5} \le \hat h_n \le Cn^{-1/5} \text{ almost surely.}
  \]
  \item We would like to ensure that the estimated functions $\hatm, \hats$ are H\"older
   continuous almost surely. Since $m_0(x), s_0(x)$ are H\"older continuous,\footnote
   {Since $\log (\cdot)$ is a smooth transformation on strictly positive compact sets,
   H\"older smoothness conditions for $(m_0, s_0)$ translate to the same conditions on $
   (\E[Y\mid x], \var(Y\mid x) - \sigma^2(x))$, with potentially different constants.
   Moreover, scaling and translating $x_i$ linearly do not affect our technical results.
   As a result, we assume, without essential loss of generality, $x_i \in[0,1]$. We abuse
   and recycle notation to write $m_0(x) = \E[Y_i \mid x_i = x], s_0(x) =
\var(\theta_i \mid x_i = x)$. We also note that $m_0(x), s_0(x) \in C_{A_3}^2([0,1])$ for
some $A_3 \leH A_1$.} we do not
  significantly
  incur estimation error if
  we project to H\"older continuous functions. 
\end{itemize}

We enforce these properties in the below procedure that we analyze. We anticipate the
projection steps to be unnecessary in practice and hold with high probability in theory.
Precisely, we add the steps
\cref{item:llr_project_bandwidth,item:hprojections,item:holderprojectionm,item:holderprojections}
to the procedure in \cref{item:e1}--\cref{item:e4}. We also make the dependence on the
selected bandwidths explicit:

\begin{enumerate}[label=(LLR-\arabic*)]
  \item \label{item:llr1} Fix some kernel $K(\cdot)$. Use the direct plug-in procedure of
  \citet{calonico2019nprobust} to estimate a
  bandwidth $\hat h_{n,m}$.
  \item \label{item:llr_project_bandwidth} For some $C_h > 1$, project $\hat h_{n,m}$ to
  some interval $[C_h^{-1} n^{-1/5},
  C_h
  n^{-1/5}]$ so as to
  enforce that it converges at the optimal rate:\footnote{We use the $\gets$ notation to
  reassign a variable so that we can reduce notation clutter.} \[
\hat h_{n,m} \gets (\hat h_{n,m} \maxwith C_h^{-1} n^{-1/5}) \minwith C_h
  n^{-1/5}.
  \]

  \item \label{item:hprojectionm} Using $\hat h_{n,m}$, estimate $m_0$ with the local
  linear regression estimator
  $\hatm_\raw$ under kernel $K(\cdot)$ and bandwidth $\hat h_{n,m}$.

  \item \label{item:holderprojectionm} Project the resulting estimator $\hatm$ to the
  H\"older class $C_{A_3}^2(
  [0,1])$: \[
\hatm \in \argmin_{m \in C_{A_3}^2(
  [0,1])} \norm{m - \hatm_\raw}_\infty.
  \]
  We obtain $\hatm$ through this procedure.

  \item Form estimated squared residuals $\hat R_i^2 = (Y_i - \hatm(x_i))^2$.

  \item Repeat \cref{item:llr1} on data $(\hat R_i^2, x_i)$ to obtain a bandwidth $\hat h_
  {n,s}$.

  \item \label{item:hprojections} Repeat \cref{item:llr_project_bandwidth} to project
  $\hat h_{n,s}$.
  \item Using $\hat h_{n,s}$, estimate $v(x) = \E[R_i^2 \mid X=x]$ with the local linear
  regression estimator $\hat v$ under kernel $K(\cdot)$.

  \item Since $\hat v$ is a local linear regression estimator, it can be written as a
  linear smoother $
\hat v(x) = \sum_{i=1}^n \ell_i(x; \hat h_{n,s}) \hat R_i^2.
  $
  Let an estimate of the effective sample size be \[p_n = \frac{1}{n} \sum_{i=1}^n
  \frac{1}{\sum_{j=1}^n \ell_i^2(x_j, \hat h_
  {n,s})} \numberthis \label{eq:effective_sample_size}.\]
  \item Truncate the estimated conditional standard deviation: \[
\hat s_\raw(x) = \sqrt{\hat v(x) - \sigma^2(x)} \vee \sqrt{\frac{2}{p_n+2} \hat v(x)}.
\numberthis
\label{eq:truncation_rule}
  \]
  \item \label{item:holderprojections} Finally, project the resulting estimate to the
  H\"older class as in
  \cref{item:holderprojectionm}: \[
\hat s(x) \in \argmin_{\substack{s\in C_{A_3}^2(
[0,1]) \\ s^2
(\cdot) \ge \frac{2}{p_n+2} \min_i \sigma_i^2}} \norm{s - \hats_\raw}_\infty.\]
\end{enumerate}

To ensure we always have a positive estimate of $s_0$, we truncate at a particular point
\eqref{eq:truncation_rule}. This truncation rule is a heuristic (and improper) application
of results from the literature on estimating non-centrality parameters. We digress and
discuss the truncation rule in the next remark.

\begin{rmksq}[The truncation rule in \eqref{eq:truncation_rule}]
\label{rmk:truncate}
The truncation rule in  \eqref{eq:truncation_rule} is an ad hoc adjustment without affecting asymptotic
performance.\footnote{Indeed, since we already assumed that the true conditional variance
$s_0(x) > s_{\ell}$, we can truncate by any vanishing sequence. Given any
vanishing sequence, eventually it is lower than $s_\ell$, and eventually $|\hats
- s_0|$ is small enough for the truncation to not bind. This is, in some sense,
  silly, since finite sample performance is likely affected if we truncate by, say,
  $\frac{1}{\log \log n}$, reflected in a large constant in the corresponding rate
  expression. Our following argument assumes that the truncation of order $O(n^{-4/5})$.
  Doing so is likely to achieve a smaller constant in the rate expression, despite not
  mattering asymptotically. } It is based on a literature on the estimation of non-central
  $\chi^2$ parameters \citep{kubokawa1993estimation}. Specifically, let $U_i \iid \Norm
  (\lambda_i, 1)$ and let $V = \sum_{i=1}^p U_i^2$ be a noncentral $\chi^2$ random
  variable with $p$ degrees of freedom and noncentrality parameter $\lambda = \sum_{i=1}^p
  \lambda_i^2$. The UMVUE for $\lambda$ is $V - p,$ which is dominated by its positive
  part $(V-p)_+$. \citet{kubokawa1993estimation} derive a class of estimators of the form
  $V - \phi(V; p)$ that dominate $(V-p)_+$ in squared error risk. An estimator in this
  class is $(V - p) \vee \frac{2}{p+2} V$.\footnote{Though, since neither $(V-p)_+$ nor
  $(V-p) \vee \frac{2}{p+2} V$ is differentiable in $V$, they are not
  admissible.}

  This setting is loosely connected to ours. Suppose $m_0$ is known, and we were using
  a Nadaraya--Watson estimator with uniform kernel. Then, for a given evaluation point
  $x_0$, we would be averaging nearby $R_i^2$'s. Each $R_i$ is conditionally Gaussian,
  $R_i \mid (\theta_i, \sigma_i) \sim \Norm(\theta_i - m_0
  (\sigma_i), \sigma_i^2)$ with approximately equal variance $\sigma_i^2 \approx
  \sigma(x_0)^2$. If there happens to be $p_0$ $R_i^2$'s that we are averaging, the
  Nadaraya--Watson estimator is of the form \[
\hat v(x_0) = \frac{\sigma(x_0)^2}{p_0} \sum_{i=1}^p \pr{\frac{R_i}{\sigma(x_0)}}^2
  \]
  Conditional on $\sigma_i^2, \theta_i$, the quantity $\sum_{i=1}^p \pr{\frac{R_i}{\sigma(x_0)}}^2$
  is (approximately) noncentral $\chi^2$ with $p$ degrees of freedom and noncentrality
  parameter \[
\lambda = \sum_{i=1}^{p_0} \pr{\frac{\theta_i - m_0(x_i)}{\sigma(x_0)}}^2
  \]
  Therefore, correspondingly, applying the truncation rule from
  \citet{kubokawa1993estimation}, an estimator for the sample variance of $\theta_i$, $
\frac{1}{p_0} \sum_{i=1}^{p_0} (\theta_i - m_0(x_i))^2,
  $ is \[
\pr{\hat v(x_0) - \sigma^2(x_0)} \vee \frac{2}{p_0+2} \hat v(x_0).
  \]

  Here, we apply this truncation rule (improperly) to the case where $\hat v(x_0)$ is a
  weighted average of the squared residuals, with potentially negative weights due to
  higher-order polynomials in the local polynomial regression. To do so, we would need to
  plug
  in an analogue of $p_0$. We note that when independent random variables $V_i$ have unit
  variance, the weighted average
has variance equal to the squared length of the weights  \[
\var\pr{\sum_i \ell_i(x) V_i} = \sum_{i=1}^n \ell_i^2(x).
  \]
  Since a simple average has variance equal to $1/n$, we can take $\pr{\sum_{i=1}^n
  \ell_i^2(x)}^{-1}$ to be an effective sample size. Our rule simply takes the average
  effective sample size over evaluation points in \eqref{eq:effective_sample_size} and use
  it as a candidate for $p$.
\end{rmksq}

The goal in this section is to control the following probability as a function of $t > 0$
\[
\P\pr{\norm{\hat\eta - \eta_0}_\infty > C_\Hyperparams t n^{-2/5} (\log n)^\beta}
\]
for some constants $\beta, C_\Hyperparams$ to be chosen. Since we treat $x_{1},\ldots,
x_n$ as fixed (fixed design), we shall do so placing some assumptions on sequences of the
design points $x_{1:n}$ as a function of $n$. These assumptions are mild and satisfied
when the design points are equally spaced. They are also satisfied with high probability
when the design points are drawn from a well-behaved density $f(\cdot)$.

Before doing so, we introduce some notation on the local linear regression estimator. Note
that, by translating and scaling if necessary, it is without essential loss of generality
to assume $x_i$ take values in $[0,1]$. Let $h_n$ denote some (possibly data-driven) choice of bandwidth.
Let $u(x) = [1, x]'$ and let $B_{nx} = B_{nx}(h_n) = \frac{1}{n h_n} \sum_{i=1}^n K\pr{
\frac{x_i - x}{h_n}}
u\pr{\frac{x_i - x}{h_n}} u\pr{\frac{x_i - x}{h_n}}'$. Then, it is easy to see that the
local linear regression weights can be written in terms of $B_{nx}$ and $u(\cdot)$: \[
s_n \equiv nh_n \quad \ell_i(x) = \ell_i(x, h_n) \equiv \frac{1}{s_n} u(0)'B_{nx}^{-1} u
\pr{\frac{x_i - x}
{h_n}} K\pr{\frac{x_i - x}{h_n}}.
\]

We shall maintain the following assumptions on the design points. The following
assumptions introduce constants $(C_h, n_0, \lambda_0, a_0, K_0, K(\cdot), c, C, C_K,
V_K)$ which we shall take as primitives like those in $\mathcal H$. The symbols $\lesssim,
\gtrsim,
\rateeq$ are relative to these constants, and we will not keep track of exact
dependencies through subscripts.

\begin{as}
\label{as:bandwidth}
For some constant $C_h > 1$, the data-driven bandwidth $h_n$ is almost surely contained in
the set $H_n \equiv [C_h^ {-1} n^{-1/5}
\maxwith
\frac{1}{2n}, C_h n^{-1/5}]$.
\end{as}

\Cref{as:bandwidth} is automatically satisfied by the projection steps
\cref{item:llr_project_bandwidth,item:hprojections}.

\begin{as}
\label{as:tsybakov_LP}
The sequence of design points $(x_i : i =1, \ldots, n)$ satisfy:
\begin{enumerate}
  \item There exists a real number $\lambda_0 > 0$ and integer $n_0 > 0$ such that, for
  all $n \ge n_0$, any $x \in[0,1]$, and any $\tilde h
  \in [C_h^{-1} n^{-1/5} \maxwith
\frac{1}{2n}, C_h n^{-1/5}]$,  the
  smallest eigenvalue $\lambda_{\min}(B_{nx}(\tilde h)) \ge \lambda_0$.
  \item There exists a real number $a_0 > 0$ such that for any interval $I \subset [0,1]$
  and all $n \ge 1$, \[
\frac{1}{n} \sum_{i=1}^n \one(x_i \in I) \le a_0 \pr{\lambda(I) \vee \frac{1}{n}}
  \]
  where $\lambda(I)$ is the Lebesgue measure of $I$.
  \item The kernel $K$ is supported on $[-1,1]$ and uniformly bounded by some
  positive constant $K_0$.
  \item There exists $c, C, n_0 > 0$ such that for all $n > n_0$, the choice of $p_n$ in
  \eqref{eq:effective_sample_size} satisfies $c  n^{4/5}
  \le p_n(\tilde h)
  \le C n^
  {4/5}$ for all $\tilde h \in [C_h^{-1} n^
  {-1/5} \maxwith
\frac{1}{2n}, C_h n^{-1/5}]$.
\end{enumerate}
\end{as}

\cref{as:tsybakov_LP}(1--3) is nearly the same as Assumption (LP) in
\citet{tsybakov2008introduction}. The only difference is that \cref{as:tsybakov_LP}(1)
requires the lower bound $\lambda_0$ to hold uniformly over a range of bandwidth choices,
relative to LP-1 in \citet{tsybakov2008introduction}, which requires $\lambda_0$ to hold
for some deterministic sequence $h_n$.
This is a mild strengthening of LP-1: Note that if $x_i$ are drawn from a
Lipschitz-continuous, everywhere-positive density $f(x)$, then for $h \to 0, nh \to
\infty$, \[B_{nx}(h) \approx \int
K(t)u(t)u(t)' f(x) \,dt \succeq \int K(t) u(t)u(t)' \,dt \pr{\min_ {x
\in
[0,1]} f(x)}
\]
where $\succ$ denotes the positive-definite matrix order. Thus the minimum eigenvalue of
$B_{nx}(h)$ should be positive irrespective of $x$ and $h$. See, also, Lemma 1.5
in \citet{tsybakov2008introduction}.

\cref{as:tsybakov_LP}(2)--(3) are the same as (LP-2)--(LP-3) in
\citet{tsybakov2008introduction}. (2) expects that the design points are sufficiently
spread out, and (3) is satisfied by, say, the Epanechnikov kernel.

Lastly, (4) expects that the average effective sample size is about $s_n = nh_n \rateeq n^
{-4/5}.$
Again, heuristically, if $x_i$ are drawn from a Lipschitz and everywhere-positive density
$f(x)$, then \[
\sum_{i=1}^n \ell_i^2(x_j) \approx n  \frac{1}{s_n^2} h_n \cdot \int (u(0)'B_{n,x_j}^{-1}
u (t) K (t))^2 f(x_j)\,dt = \frac{1}{s_n} \int (u(0)'B_{n,x_j}^{-1} u (t) K (t))^2
f(x_j)\,dt.
\]
Hence the mean reciprocal $p_n$ is of order $s_n$. We also remark that
\cref{as:tsybakov_LP} is satisfied by regular design points $x_i = i/n$.

\begin{as}
\label{as:polynomial_cover}
The kernel satisfies the following VC subgraph-type conditions. Let \[\mathcal F_k =
\br{y\mapsto \pr{\frac{y-x}{h}}^{k-1} K\pr{\frac{y-x}{h}} : x \in [0,1], h\in H_n}\]
for $k = 1,2$. For any finitely supported
measure $Q$, \[
N(\epsilon, \mathcal F_k, L_2(Q)) \le C_K (1/\epsilon)^{V_K}
\]
for $C_K, V_K$ that do not depend on $Q$.
\end{as}

\Cref{as:polynomial_cover} is satisfied for a wide range of kernels, e.g. the Epanechnikov
kernel. By Lemma 7.22 in \citet{sen2018gentle}, reproduced as \cref{lemma:sen2018} below,
so long as the function $t\mapsto t^{k-1} K(t)$ is bounded (assumed in
\cref{as:tsybakov_LP}(3)) and of bounded variation (satisfied by any absolutely continuous
kernel function), the covering number conditions hold by exploiting the finite VC
dimension of subgraphs of these functions.

We now state and prove the main results in this section.
The key to these arguments is \cref{prop:llr_large_dev} on the bias and variance of local
linear regression estimators. \Cref{prop:llr_large_dev} is uniform in both the evaluation point $x$
and the bandwidth $h$, as long as the latter converges at the optimal rate.

\begin{theorem}
\label{thm:kernel_rates}
  Suppose the conditional distribution $\theta_i \mid \sigma_i$ and the design points
  $\sigma_{1:n}$ satisfy   \cref{as:variance_bounds,as:moments,as:tsybakov_LP}. Moreover,
  suppose $m_0, s_0$ satisfies \cref{as:holder}(1) with $p=2$. Suppose the kernel
  $K(\cdot)$ satisfies \cref{as:polynomial_cover}. Let $\hatm, \hats$ denote the
  estimators computed by \cref{item:llr1} through \cref{item:holderprojections}. Then,
  there exists some $n_0 > 0$ such that 
   \begin{enumerate}
   \item $\P\pr{\hatm, \hats \in C_{A_3}^2([0,1])} = 1$
     \item For some $C$ depending only on the parameters in the assumptions, for all $n
     \ge 7$ and $t > 1$, \[
\P\pr{
  \max\pr{\norm{\hatm - m_0}_\infty, \norm{\hats- s_0}_\infty} \ge C t n^{-\frac{2}{5}} 
  (\log n)^
  {1+2/\alpha}
} \le \frac{1}{n^{10} t^2}. \numberthis \label{eq:largedeviationkernel}
     \]
     \item For some $c > 0$ depending only on the parameters in the assumptions, for all
     $n
     \ge n_0$, \[
\P\pr{\frac{c}{n} \le \hats} = 1.
     \]
   \end{enumerate}
\end{theorem}

\begin{proof}
The first claim is true automatically by the projection to the H\"older space. 

The third
claim is true for all $n > n_0$ automatically by \cref{item:holderprojections}, since $p_n
\ge c n^ {4/5}$
and $n^{-4/5} \gtrsim n^
{-1}$. For $n \le n_0$, note that $\sum_i \ell_i(x, h) = \sum_i \ell_i(x,h) \underbrace{u(
(x_i-x)/h)'u
(0) }_{1} = \norm{u(0)}^2 = 1$
for all $h, x$. Hence \[
\frac{1}{n}\sum_i \ell_i^2(x, h) \ge \pr{\frac{1}{n}\sum_i \ell_i(x,h)}^2 \implies \sum_i
\ell_i^2(x, h) \ge \frac{1}{n} \implies p_n \le n
\]
regardless of $h$. As a result, the truncation point for $\hat s^2$ is at least of order
$\frac{1}{n}$. This is sufficient for $\hat s \ge c/n$. 

Now, we show the second claim. Since we assume that $m_0, s_0$ lies in the H\"older space
with $s_0 > s_{0\ell}$, then
projection to the H\"older space (and truncation by $2/(2+p_n) \min_i \sigma_i^2)$ worsens
performance by at most a factor of two for all sufficiently large $n$. The projection to
the H\"older space ensures that $\norm{\hateta - \eta_0}_\infty$ is bounded a.s. for all
$n$, so that
we can remove ``for all sufficiently large $n$'' at the cost of enlarging a constant so
as to accommodate the first finitely many values of $n$.
As a result, it suffices to show that \[
\P\pr{
  \max\pr{\norm{\hatm_\raw - m_0}_\infty, \norm{\hats_\raw - s_0}_\infty} > Ct n^{-2/5}
  (\log n)^\beta
} \le \frac{1}{n^{10} t^2}
\]
for some $C$ and $\beta = 1+2/\alpha$.

Let $Y_i = m_0(x_i) + \xi_i$ where $\xi_i = \theta_i - m_0(x_i) + (Y_i - \theta_i)$. Note
that we have simultaneous moment control for $\xi_i$: \[
\max_i \E[|\xi_i|^p]^{1/p} \lesssim p^{1/\alpha}
\]
where $\alpha$ is the constant in \cref{as:moments}. Therefore, we can apply
\cref{prop:llr_large_dev} to obtain \[
\P\pr{\norm{\hatm_\raw - m_0}_\infty > C t n^{-2/5} (\log n)^{1 + 1/\alpha}} \le \frac{1}
{2 n^{10} t^2}
\]
for the local linear regression estimator $\hatm_\raw$.

The same argument to control $\norm{\hats_\raw - s_0}_\infty$ is more involved. First
observe that \[
|\hats^2_\raw - s_0^2| = |\hats_\raw - s_0|(\hats_\raw + s_0) \ge s_{0\ell} |\hats_\raw -
s_0|.
\]
Also observe that for a positive $f_0$, \[
|\hat f \vee g - f_0| \le |\hat f - f_0| \vee |g|.
\]
As a result, it suffices to control the upper bound in \begin{align*}
\norm{\hats_\raw -
s_0}_\infty &\le \frac{1}{s_{0\ell} } \pr{\norm{\hat v - v_0}_\infty \vee \pr{\frac{2}
{2+p_n}\hat v}} \tag{$v_0(x) \equiv \var(Y_i \mid x_i=x)$}
\\
& \lesssim \norm{\hat v - v_0}_\infty \vee \frac{\norm{\hat v -
v_0}_\infty
+ \norm{v_0}_\infty}
{2+n^{4/5}} \tag{\cref{as:tsybakov_LP}}\\
& \lesssim \norm{\hat v - v_0}_\infty \numberthis \label{eq:dominate_s_with_v}
\end{align*}

Now, observe that $\hat R_i^2 = R_i^2 + (m_0 - \hatm)^2 - 2 (m_0 - \hatm) \xi_i$. Hence,
\begin{align*}
|\hat v(x) - v_0(x)| &\le \absauto{\sum_{i=1}^n \ell_i(x, \hat h_{n, s}) R_i^2 - v_0(x)} + \br{
\norm{m_0
- \hatm}_\infty^2  + 2 \norm{m_0 - \hat m}_\infty
\pr{\max_{i\in[n]} |\xi_i|}} \sum_{i=1}^n |\ell_i(x, \hat h_{n, s})| \\
& \le \absauto{\sum_{i=1}^n \ell_i(x, \hat h_{n, s}) R_i^2 - v_0(x)} + C\br{
\norm{m_0
- \hatm}_\infty^2  + 2 \norm{m_0 - \hat m}_\infty
\pr{\max_{i\in[n]} |\xi_i|}}. \numberthis
\label{eq:decomp_s_error}
\end{align*}
By Lemma 1.3 in \citet{tsybakov2008introduction}, the term $\sum_{i=1}^n |\ell_i(x, \hat h_
{n, s})|$ is bounded uniformly in $h$ and $x$ by a constant. Note that \[
\tilde \xi_i \equiv R_i^2 - v_0(x_i)
\]
has simultaneous moment control with a different parameter ($\tilde \alpha = \alpha / 2$):
\[
\max_{i} (\E|\tilde \xi_i|^p)^{1/p} \lesssim p^{2/\alpha}.
\]
Thus, applying \cref{prop:llr_large_dev} and taking care to plug in $\tilde \xi, \tilde
\alpha$, we can bound the first term in \eqref{eq:decomp_s_error} \[
\P\pr{
  \normauto{\sum_{i=1}^n \ell_i(x, \hat h_{n, s}) R_i^2 - v_0(x)}_\infty \ge C t n^{-2/5}
  (\log n)^{1 + 2/\alpha}
}\le \frac{1}{4 n^{10} t^2}.
\]

Note that by an application of \cref{lemma:tail_bound_max}, for any $a, b> 0$, we have
that \[
\P\pr{
  \max_{i} |\xi_i| > C(a,b)t (\log n)^{1/\alpha}
} < an^{-b} e^{-t^2}
\]
As a result, the second term in \eqref{eq:decomp_s_error} admits \[
\P\pr{
\norm{m_0
- \hatm}_\infty^2  + 2 \norm{m_0 - \hat m}_\infty
\pr{\max_{i\in[n]} |\xi_i|} > C t n^{-2/5} (\log n)^{1+2/\alpha}
} \le \frac{1}{4 n^{10}t^2}
\]
Finally, putting these bounds together, we have that \[
\P\pr{
  \norm{\hat v - v_0}_\infty > C tn^{-2/5} (\log n)^{1+2/\alpha}
} \le \frac{1}{2n^{10}t^2},
\]
where the same bound (with a different constant) holds for $\hats_\raw$ by
\eqref{eq:dominate_s_with_v}.

Combining the bounds for $\hatm$ and $\hats$, we obtain \eqref{eq:largedeviationkernel}.
This concludes the proof.
\end{proof}

\begin{theorem}
\label{thm:llr_main_regret}
Under the assumptions of \cref{thm:kernel_rates}, let $\hateta = (\hatm, \hats)$ denote
   estimators computed by \cref{item:llr1} through \cref{item:holderprojections}. Then, \[
\E\bk{\reg(\hat G_n, \hateta)} \lesssim n^{-2/5} (\log n)^{1+2/\alpha}.
   \]
\end{theorem}

\begin{proof}
Recall the event $A_n$ in \eqref{eq:barzn_def} for $\Delta_n = C_1 n^{-2/5} (\log
n)^\beta$ and $M_n = C_2 (\log n)^{1/\alpha}$, where $C_1, C_2$ are to be chosen and
$\beta = 1+2/\alpha$. Define $\tilde A_n = A_n \cap \br{s_{0\ell}/2 \le \hat s \le 2s_
{0u}}$.
Decompose \[
\E\bk{\reg(\hat G_n, \hateta)} = \E\bk{\reg(\hat G_n, \hateta) \one(\tilde A_n)} + \E\bk{\reg
(\hat G_n, \hateta)\one(\tilde A_n^\comp)}.
\]

Note that, for all sufficiently large $n > N$, such that $N$ depends only on $C_1, \beta,
s_\ell, s_u$, the event $A_n$ already implies $\br{s_{0\ell}/2 \le \hat s \le 2s_{0u}}$
and hence
$A_n = \tilde A_n$. Thus, by
\cref{thm:kernel_rates}, for all sufficiently large $n$,  on the event $A_n$,
statements analogous to \cref{as:holder}(2--4) hold for the
estimator $\hateta$. As a result, we may apply \cref{thm:regret_on_An}, \emph{mutatis
mutandis}, to obtain that \[
\E\bk{\reg(\hat G_n, \hateta) \one(\tilde A_n)} \lesssim n^{-4/5} (\log n)^{
\frac{2+\alpha}
{\alpha} + 3 + 2\beta}
\]
for all sufficiently large choices of $C_1, C_2$.

To control $\E\bk{\reg
(\hat G_n, \hateta)\one(\tilde A_n^\comp)}$, we observe that under
\cref{lemma:distance_bound_hard} and \cref{thm:kernel_rates}(1 and 3), we have that almost
surely on $A_n^\comp$, \[
\reg(\hat G_n, \hateta) \lesssim n^{2} \bar Z_n^2.
\]
Hence, by Cauchy--Schwarz as in \cref{lemma:reg_on_anc}, \[
\E\bk{\reg
(\hat G_n, \hateta)\one(\tilde A_n^\comp)} \lesssim \P(\tilde A_n^{\comp})^{1/2} n^4
(\log n
)^{2/\alpha},
\]
where we apply \cref{lemma:tail_bound_max} to bound $\E[\bar Z_n^4]$. This bound holds for
all $n \ge 7$. 

For all sufficiently large $n > N$, \[
\P(A_n^\comp) = \P(\tilde A_n^\comp) \le \P(\bar Z_n > M_n) + \P(\norm{\hateta -
\eta_0}_\infty > \Delta_n).
\]
Sufficiently large $C_1, C_2$ can be chosen such that the right-hand side is bounded by
$n^{-10}$. To wit, we can apply \cref{thm:kernel_rates} to bound $\norm{\hateta -
\eta_0}_\infty$. We can apply \cref{lemma:tail_bound_max} to bound $\P(\bar Z_n > M_n)$.
As a result, we would obtain \[
\E\bk{\reg(\hat G_n, \hateta) \one(\tilde A_n^\comp)} \lesssim \frac{1}{n}(\log n)^
{2/\alpha}
\]
for all sufficiently large $n$.

Since $\E[\reg(\hat G_n, \hateta)] \lesssim n^{4} (\log n)^{2/\alpha}$ is finite for all
$n$, at the cost of enlarging the implicit constant, we have the result of the theorem
holding for all $n$.
\end{proof}

\subsection{Auxiliary lemmas}

\begin{prop}
\label{prop:llr_large_dev}
Consider the local linear regression of data $Y_i = f_0(x_i) + \xi_i$ on the design points
$x_i$, for $i = 1,\ldots, n$. Suppose $f_0$ belongs to a H\"older class of order two:
$f_0 \in C_L^2([0,1])$ for some $L > 0$. Suppose that the design points satisfy
\cref{as:tsybakov_LP} and the (possibly data-driven) bandwidths $h_n$ satisfy
\cref{as:bandwidth}. Assume the kernel additionally satisfies \cref{as:polynomial_cover}.

Assume that the residuals $\xi_i$ are mean zero, and there exists a constant $A_\xi > 0,
\alpha > 0$ such that \[
\max_{i=1,\ldots, n} (\E[|\xi_i|^p])^{1/p} \le A_\xi p^{1/\alpha}
\]
for all $p \ge 2$. Let $\ell_i(x,h)$ be the weights corresponding to local linear
regression, and define the bias part $b(x, h_n) = \pr{\sum_{i=1}^n \ell_i(x, h_n) f_0
(x_i)} - f_0(x_i)$
and the stochastic part $v(x, h) =
\sum_{i=1}^n \ell_i(x,h) \xi_i$. Recall that $H_n$ is the interval for $h_n$ in \cref{as:bandwidth}.
Then:
\begin{enumerate}
  \item The bias term is of order $n^{-2/5}$: \[
\sup_{x \in [0,1], h \in H_n} |b(x,h)| \lesssim n^{-2/5}.
  \]
  \item The variance term admits the following large-deviation inequality: For any
  $a, b > 0$, there exists
  a constant $C(a,b)$, which may additionally depend on the constants in the assumptions,
  such that for all $n > 1$ and $t \ge
  1$ \[
\P\pr{
  \sup_{x\in [0,1], h \in H_n} |v(x, h)| > C(a,b) \cdot t \cdot (\log n)^{1 + 1/\alpha} n^
  {-2/5}
} \le a n^{-b} \frac{1}{t^2}.
\]
\item As a result, let $\hat f(\cdot) = b(\cdot, h_n) + v(\cdot, h_n) + f_0(\cdot)$, we
have that for
any $a,b > 0$, there exists a constant $C(a,b)$ such that for all $n>1$ and $t \ge 1$, \[
\P\pr{
  \norm{\hat f - f_0}_{\infty} > C(a,b) t (\log n)^{1+1/\alpha} n^{-2/5}
} \le a n^{-b} \frac{1}{t^2}.
\]
\end{enumerate}
\end{prop}

\begin{proof}
Note that (3) follows immediately from (1) and (2) since the bounds in (1) and (2) are
uniform over all $h \in H_n$. We now verify (1) and (2).

\begin{enumerate}[wide]
  \item This claim follows immediately from the bound for $b(x_0)$ in Proposition 1.13 in
  \citet{tsybakov2008introduction}. The argument in \citet{tsybakov2008introduction} shows
  that \[
\sup_{x \in [0,1]}|b(x, h_n)| \le C h_n^2,
  \]
  which is uniformly bounded by $ C n^{-2/5}$ by \cref{as:bandwidth}. Hence \[
\sup_{x \in [0,1], h \in H_n} |b(x,h)| \lesssim n^{-2/5}.
  \]
  \item Let $M$ be a truncation point to be defined. Let \[\xi_{i, <M} = \xi_i \one(|\xi_i|
  \le
  M) - \E[\xi_i \one(|\xi_i|\le
  M)]
\quad \xi_{i, >M} = \xi_i \one(|\xi_i|
  >
  M) - \E[\xi_i \one(|\xi_i |>
  M)]
  \]
  be truncated and demeaned variables. Note that \[
\xi_i = \xi_{i, <M} + \xi_{i, >M}.
  \]

  First, let $V_{1n}(x, h_n) = \sum_{i=1}^n \ell_i(x, h_n) \xi_{i, >M}$. Note that by
  Cauchy--Schwarz, uniformly over $x, h_n$, \begin{align*}
  V_{1n}^2 &\le \sum_{i=1}^n \ell_i(x, h_n)^2 \sum_{i=1}^n \xi_{i, >M}^2 \\
  &\lesssim \frac{1}{h_n^2} \frac{1}{n} \sum_{i=1}^n \xi_{i, >M}^2 \tag{Lemma 1.3(i) in
  \citet{tsybakov2008introduction} shows that $|\ell_i(x,h_n)| \le \frac{C}{nh_n}$} \\
  &\lesssim n^{2/5} \frac{1}{n} \sum_{i=1}^n \xi_{i, >M}^2
  \end{align*}
  Now, for some $C$ related to the implicit constant in the above display, \[
\P\pr{\sup_{x \in [0,1], h_n \in H_n} V_{1n}^2(x,h_n) > C t^2} \le \P\pr{
  \frac{1}{n} \sum_{i=1}^n \xi_{i, >M}^2 > t^2 n^{-2/5}
} \le \frac{\max_{i} \E\xi_{i, >M}^2}{ t^2} n^{2/5} \tag{Markov's inequality}. \]
We note that by Cauchy--Schwarz, \[
\E[\xi_{i, >M}^2] \le \sqrt{\E[\xi_i^4]} \sqrt{\P(|\xi_i| > M)} \lesssim \sqrt{\P
(|\xi_i| > M)} \le \exp\pr{-c M^\alpha} \tag{\cref{lemma:tail_bound}}
\]
where $c$ depends on $A_\xi$.
Hence, for a potentially different constant $C$, \[
\P\pr{\sup_{x \in [0,1], h_n \in H_n} |V_{1n}(x,h_n)| > Ct} \le \exp\pr{-c M^\alpha - 2
\log t +
\frac{2}{5} \log n }. \numberthis \label{eq:llr_large_tail_term}
\]

Next, consider the process \begin{align*}
V_{2n}(x, h_n) &= \sum_{i=1}^n \ell_i(x,h_n) \xi_{i, <M} \\&= \frac{1}{nh_n}\sum_{i=1}^n
\underbrace{u
(0)'
B_{nx}^{-1}
\colvecb{2}{1}{0}}_{A_1(x, h_n)} K\pr{\frac{x_i - x}{h_n}} \xi_{i, <M} \\
&\quad + \frac{1}{nh_n}\sum_{i=1}^n \underbrace{u
(0)'
B_{nx}^{-1}
\colvecb{2}{0}{1}}_{A_2(x, h_n)} K\pr{\frac{x_i - x}{h_n}}\pr{\frac{x_i - x}{h_n}} \xi_{i,
<M} \\
&\equiv \frac{A_{1}(x, h_n)}{h_n} \frac{1}{n}\sum_{i=1}^n K\pr{\frac{x_i - x}{h_n}} \xi_
{i, <M} + \frac{A_{2}(x, h_n)}{h_n} \frac{1}{n}\sum_{i=1}^n K\pr{\frac{x_i - x}{h_n}} \pr{\frac{x_i - x}{h_n}} \xi_
{i, <M}.
\end{align*}
Note that, by \cref{as:tsybakov_LP}(1), uniformly over $x \in [0,1]$ and $h_n \in H_n$,
\[
|A_k(x, h_n)| \le \norm{u(0)' B_{nx}^{-1}} \le \frac{1}{\lambda_0}.
\]
By triangle inequality, \begin{align*}
V_{2n}(x, h_n) &\lesssim \frac{1}{h_n} \absauto{\frac{1}{n}\sum_{i=1}^n K\pr{
\frac{x_i - x}{h_n}} \xi_
{i, <M}} + \frac{1}{h_n} \absauto{\frac{1}{n}\sum_{i=1}^n K\pr{
\frac{x_i - x}{h_n}} \pr{\frac{x_i - x}{h_n}} \xi_
{i, <M}} \\
&\equiv \frac{1}{\sqrt{n} h_n} V_{2n, 1}(x, h_n) + \frac{1}{\sqrt{n} h_n} V_{2n, 2}(x, h_n).
\end{align*}
We will control the $\psi_2$-norm of the left-hand side. Note that it suffices to control
the $\psi_2$-norm of both terms on the right-hand side: \[
\normauto{\sup_{x \in [0,1], h_n \in H_n} |V_{2n}(x, h_n)|}_{\psi_2} \lesssim \frac{1}{
\sqrt{n} h_n}
\max_{k = 1,2}\pr{
  \normauto{\sup_{x \in [0,1], h_n \in H_n} |V_{2n, k}(x, h_n)|}_{\psi_2}
}.
\]
The above display follows from replacing the sum with two times the maximum and Lemma
2.2.2 in \citet{vaart1996weak}.

We will do so by applying \cref{lemma:dudley}. The analogue of $f$ in \cref{lemma:dudley}
is \[
t\mapsto f(t; x,h) = \pr{\frac{t-x}{h}}^{k-1} K\pr{\frac{t-x}{h}}
\]
for $V_{2n, k}, k = 1,2$. Naturally, the analogues of $\mathcal F$ is \[
\mathcal F_k = \br{t\mapsto f(t; x,h) : x \in [0,1], h\in H_n} \cup \{t \mapsto 0\}.
\]
Note that \[
f(t; x,h) \le \one(|t-x| \le h) K_0
\]
and thus the diameter of $\mathcal F_k$ is at most \[
\sup_{A\subset [0,1] : \lambda(A) \le 4C_h n^{-1/5}} K_0 \sqrt{\frac{1}{n}\sum_{i=1}^n
\one(x_i \in A)} \lesssim n^{-1/10}
\]
by \cref{as:tsybakov_LP}(2). Therefore, by \cref{as:polynomial_cover}, we apply
\cref{lemma:dudley} and obtain that for $k = 1,2$ \[
\normauto{\sup_{x\in [0,1], h \in H_n} |V_{2n,k}(x,h)|}_{\psi_2} \lesssim M n^{-1/10}
\sqrt{\log n}.
\]
Finally, this argument shows that \[
\normauto{\sup_{x \in [0,1], h \in H_n} |V_{2n}(x, h)|}_{\psi_2} \lesssim \frac{1}
{\sqrt{n } h_n n^{1/10}} M \sqrt{\log n} \lesssim n^{-2/5} M \sqrt{\log n} \numberthis
\label{eq:llr_bbd_term}.
\]

Putting things together, we can choose $M = (c_m \log n)^{1/\alpha}$ for sufficiently
large $c_m$ so that by \eqref{eq:llr_large_tail_term}, \[
\P\pr{\sup_{x \in [0,1], h \in H_n} |V_{1n}(x,h)| > C  t n^{-2/5}} \le \frac{a}{2} n^{-b}
\frac{1}
{t^2},
\]
where $c_m$ depends on $a,b$.
The bound \eqref{eq:llr_bbd_term} in turns shows that \[
\P\pr{
  \sup_{x \in [0,1], h_n \in H_n} |V_{2n}(x,h_n)| > C(a,b) t (\log n)^{\frac{2 + \alpha}
  {2
  \alpha}} n^{-2/5}
} \le 2 e^{-t^2}
\]
Taking $t = \sqrt{b\log n + \log(a/4)} s$ gives \[
\P\pr{
  \sup_{x \in [0,1], h_n \in H_n} |V_{2n}(x,h_n)| > C(a,b) s (\log n)^{1 + 1/\alpha} n^
  {-2/5}
  e^
  {-s^2}
} \le \frac{a}{2} n^{-b} e^{-s^2} < \frac{a}{2} n^{-b} \frac{1}{s^2}
\]
for all $s > 1$.

Therefore, combining the two bounds, \[
\P\pr{
  \sup_{x\in [0,1], h_n \in H_n} |v(x, h_x)| > C(a,b) t (\log n)^{1 + 1/\alpha} n^{-2/5}
} \le a n^{-b} \frac{1}{t^2}.
\]
\end{enumerate}
\end{proof}

\begin{lemma}
\label{lemma:dudley}
Suppose $\xi_i$ are bounded by $M \ge 1$ and mean zero. Consider the process \[
V_n(f) = \frac{1}{\sqrt{n}}\sum_{i=1}^n f(x_i) \xi_{i}
\]
over a class of real-valued functions $f \in \mathcal F$ and evaluation points
$x_1,\ldots, x_n \in [0,1]$.
Define the seminorm $\norm{\cdot}_n$
relative to $x_1,\ldots, x_n$ by \[
\norm{f}_n^2 = \frac{1}{n} \sum_{i=1}^n f(x_i)^2.
\]
Suppose $0 \in \mathcal F$ and
$\mathcal F$ has
polynomial covering numbers: \[N(\epsilon,
\mathcal F, \norm{\cdot}_n) \le C (1/\epsilon)^
{V} \quad \epsilon \in [0,1]\]
where $C, V > 0$ depend solely on $\mathcal F$.
Then \[
\normauto{\sup_{f\in \mathcal F} |V_n(f)|}_{\psi_2} \lesssim M \mathrm{diam}(\mathcal F)
\sqrt{\log (1/\mathrm{diam}(\mathcal F))},
\]
where $\mathrm{diam}(\mathcal F) = \sup_{f_1, f_2 \in \mathcal F} \norm{f_1 - f_2}_n$.
\end{lemma}

\begin{proof}
The process $V_n(f)$ has subgaussian increments with respect to $\norm{\cdot}_n$: \[
\norm{V_n(f_1) - V_n(f_2)}_{\psi_2} \lesssim M \norm{f_1 - f_2}_n.
\]
Hence, by Dudley's chaining argument (e.g. Corollary 2.2.5 in \citet{vaart1996weak}), for
some fixed $f_0 \in \mathcal F$, \[
\normauto{\sup_{f} V_n(f)}_{\psi_2} \le \norm{V_n(f_0)}_{\psi_2} + CM \int_0^{
\mathrm{diam}(\mathcal F)} \sqrt{\log N(\delta, \mathcal F, \norm{\cdot}_n)} \,d \delta.
\]
Note that (i) the metric entropy
integral is bounded by $C \mathrm{diam}(\mathcal F)
\sqrt{\log (1/\mathrm{diam}(\mathcal F))}$, and (ii) for a fixed $f_0$, $\norm{V_n(f_0)}_
{\psi_2} \lesssim \norm{f_0}_n M \le \mathrm{diam}(\mathcal F) M$ since $0 \in \mathcal
F$. Therefore, \[
\normauto{\sup_{f} V_n(f)}_{\psi_2} \lesssim M \mathrm{diam}(\mathcal F)
\sqrt{\log (1/\mathrm{diam}(\mathcal F))}.
\]
\end{proof}

\begin{lemma}[Lemma 7.22(ii) in \citet{sen2018gentle}]
\label{lemma:sen2018}
Let $q(\cdot)$ be a real-valued function of bounded variation on
$\R$. The covering number of $\mathcal F = \br{x \mapsto q(ax + b) : (a,b) \in \R}$
satisfies \[
N(\epsilon, \mathcal F, L_2(Q)) \le K_1 \epsilon^{-V_1}
\]
for some $K_1$ and $V_1$ and for a constant envelope.
\end{lemma}

\section{Auxiliary lemmas for \cref{thm:minimaxlower,thm:worstcaserisk}}

\begin{lemma}
\label{lemma:minimax}
In the proof of \cref{thm:minimaxlower}, suppose $Y_i \sim \Norm(m_0(\sigma_i), s_0^2 +
\sigma_i^2)$, then 
\[\inf_{\hatm} \sup_{m_0}  \E\bk{\frac{1}{n} \sum_{i=1}^n (\hat
m(\sigma_i) - m_0 (\sigma_i))^2}\gtrsim_\H n^{-2p/(2p+1)},\] where the infimum is over all
estimators of $m_0$ from $(Y_i, \sigma_i)_{i=1}^n$ and the supremum is over the H\"older
space $m_0 \in C_{A_1}^p([\sigl, \sigu])$.
\end{lemma}

\begin{proof} 

First, note that learning $m_0$ from $(Y_i, \sigma_i)$ is a nonparametric
regression problem with heteroskedastic variances. This problem is more difficult than
a corresponding problem with homoskedastic variances $\sigl^2 + s_0^2$, since we may
represent \[
Y_i = \theta_i + \sigl W_i + (\sigma_i^2 - \sigl^2)^{1/2} U_i
\]
for independent Gaussians $W_i, U_i \sim \Norm(0,1)$. Let $V_i = \theta_i + \sigl W_i$.
Note that we can do no worse for estimating $m_0$ with $(V_i, \sigma_i)$ than with $(Y_i,
\sigma_i)$, and estimating $m_0$ from $(V_i, \sigma_i)$ is a homoskedastic regression
problem, where $V_i \sim \Norm(m_0(\sigma_i), \sigl^2 + s_0^2)$. It remains to show that
the minimax rate for estimating $m_0$ on the grid points $\sigma_ {1:n}$ from $(V_i,
\sigma_i)$ is $n^{-2p/(2p+1)}$.

Since we simply have a nonparametric regression problem,
we may translate and rescale so that the design points $\sigma_ {1:n}$ are equally spaced
in $[0, 1]$ ($\sigma_i = i/n$) and the variance of $V_i$ is 1---potentially changing the
constant $A_1$ for
the H\"older smoothness condition. Corollary 2.3 in \citet{tsybakov2008introduction} shows
a lower bound for integrated MSE: \[
\inf_{\tilde m} \sup_{m_0} \E\bk{\int_{0}^1 (\tilde m
(x) - m_0(x))^2\,dx} \gtrsim_\Hyperparams n^{-\frac{2p}{2p+1}}
\]
where the infimum is over all (randomized) estimators using $(V_i, \sigma_i)$. It thus
suffices to connect the MSE objective over the fixed design points
$\sigma_1,\ldots,\sigma_n$ to the integrated
MSE.

Observe that for any $\hatm(\sigma_1),\ldots,
\hatm(\sigma_n)$, we can define the piecewise constant function $\tilde m : [0,1] \to
\R$ such
 that it is equal to $\hatm(\sigma_i)$ on $[(i-1)/n, i/n)$.  Now, note that \begin{align*}
\int_{0}^1 (\tilde m(x) - m_0(x))^2 \,dx &= \sum_{i=1}^n \int_{[(i-1)/n, i/n]} (\tilde
m (x) - m_0(x))^2 \,dx
\\
&\le 2\sum_{i=1}^n \int_{[(i-1)/n, i/n]} (\tilde m(x) - m_0(\sigma_i))^2 + (m_0(\sigma_i)
- m_0(x))^2 \,dx \tag{$(a+b)^2 \le 2a^2 + 2b^2$} \\
&\le 2\sum_{i=1}^n \bk{\frac{1}{n} (\hat m_i - m_0(\sigma_i))^2  + \frac{L^2}{n^3}} \\
&=  \frac{2}{n} \sum_{i=1}^n (\hat m_i - m_0(\sigma_i))^2 + \frac{2L^2}{n^2}.
\end{align*}
The third line follows by observing that $m_0(\cdot)$ is Lipschitz for some constant $L$
that depends solely on $p, A_1$, since $p \ge 1$ in \cref{as:holder}. Therefore,
\begin{align*}
\inf_{\hatm} \sup_{m_0}  \E\bk{\frac{1}{n} \sum_{i=1}^n (\hat m(\sigma_i) - m_0
(\sigma_i))^2} &\ge  \frac{1}{2}  \inf_{\tilde m} \sup_{m_0} \br{\E\bk{\int_{0}^1 (\tilde m
(x) - m_0(x))^2\,dx} - \frac{2L^2}{n^2}} \\
&\gtrsim_\Hyperparams n^{-\frac{2p}{2p+1}}.\qedhere
\end{align*}
\end{proof}

\begin{lemma}
\label{lemma:optimal_bayes_linear}
Assume that $
\theta_i \mid \sigma_i $ has mean $m_0(\sigma)$ and variance $s_0^2(\sigma)$, without
 assuming \eqref{eq:location_scale}. Consider the decision rule 
\eqref{eq:gaussian_cond_b} and denote it by $\delta^*$.
 Then \[
\delta^*(Y_i, \sigma_i) \in \argmin_{\delta(Y_i, \sigma_i) \in L} \E\bk{
  (\delta(Y_i, \sigma_i) - \theta_i )^2 \mid \sigma_i
}
\]
where $
L = \br{
  \delta(Y_i, \sigma_i) = a(\sigma_i) + b(\sigma_i) Y_i : a(\cdot), b(\cdot) \text{
  measurable}
}.
$
\end{lemma}
\begin{proof}
For a given $a(\cdot), b(\cdot)$, we can compute by bias-variance decomposition, \[
\E[(\delta(Y_i, \sigma_i) - \theta_i)^2 \mid \sigma_i] = (a(\sigma_i) + b(\sigma_i)m
(\sigma_i) - m
(\sigma_i))^2 + b^2(\sigma_0) s_0^2 + b^2(\sigma_0) \sigma_i^2 + \sigma_i^2 - 2b(\sigma_i)
s_0^2(\sigma_i).
\]
Minimizing the above expression for $a(\sigma_i), b(\sigma_i)$ yields \[
b(\sigma_i) = \frac{s_0^2(\sigma_i)}{s_0^2(\sigma_i) + \sigma_i^2} \quad a(\sigma_i) = 
(1-b(\sigma_i)) m(\sigma_i).
\]
This corresponds exactly to the decision rule $\delta^*$. 
\end{proof}

\begin{lemma}
\label{lemma:minimax_close_gauss}
Consider the setup of \cref{thm:worstcaserisk}. The minimax risk is achieved by the
decision rule \eqref{eq:gaussian_cond_b}. That is, let $\delta_i^*$ denote the decision
rule
\eqref{eq:gaussian_cond_b}. Then, \[
\sup_{P_0 \in \mathcal P(m_0, s_0)} \E_{P_0}\bk{
  \frac{1}{n} \sum_{i=1}^n (\delta_i^* - \theta_i)^2
}  = \inf_{\delta_{1:n}} \sup_{P_0 \in \mathcal P(m_0, s_0)} \E_{P_0}\bk{
  \frac{1}{n} \sum_{i=1}^n (\delta_i - \theta_i)^2
} 
\]
where the infimum is taken over all (randomized) decision rules $\delta_i(Y_{1:n},
\sigma_{1:n})$, with knowledge of $m_0, s_0$. 
\end{lemma} 

\begin{proof}
The $\ge$ direction is immediate. We consider the $\le$ direction. 
Note that \[
\E_{P_0}\bk{
  \frac{1}{n} \sum_{i=1}^n (\delta_i^* - \theta_i)^2
}  = \frac{1}{n} \sum_{i=1}^n \frac{s_0^2(\sigma_i) \sigma_i^2 }{s_0^2(\sigma_i) +
\sigma_i^2}
\]
for all $P_0 \in \mathcal P(m_0, s_0)$, regardless of $P_0$. Thus, 
\[
\sup_{P_0 \in \mathcal P(m_0, s_0)} \E_{P_0}\bk{
  \frac{1}{n} \sum_{i=1}^n (\delta_i^* - \theta_i)^2
} = \frac{1}{n} \sum_{i=1}^n \frac{s_0^2(\sigma_i) \sigma_i^2 }{s_0^2(\sigma_i) +
\sigma_i^2}.
\]
Suppose $P_0$ denotes the distribution where $\theta_i \mid \sigma_i \sim \Norm(m_0
(\sigma_i), s_0^2(\sigma_i))$. Then, for any decision rule $\delta_i$, \[
\E_{P_0}\bk{\frac{1}{n} \sum_{i=1}^n (\delta_i - \theta_i)^2} \ge\frac{1}{n} \sum_{i=1}^n \frac{s_0^2(\sigma_i) \sigma_i^2 }{s_0^2(\sigma_i) +
\sigma_i^2}.
\]
This is because the right-hand side is the Bayes risk under the Gaussian model. As a
result, \[
\inf_{\delta_{1:n}} \sup_{P_0 \in \mathcal P(m_0, s_0)} \E_{P_0}\bk{
  \frac{1}{n} \sum_{i=1}^n (\delta_i - \theta_i)^2
} \ge 
\frac{1}{n} \sum_{i=1}^n \frac{s_0^2(\sigma_i) \sigma_i^2 }{s_0^2(\sigma_i) +
\sigma_i^2}.
\]
This concludes the proof.
\end{proof}

\begin{lemma}
\label{lemma:indepgauss_not_bounded}
Given $s_0(\cdot), m_0(\cdot)$, let \[
s_0^2 = \frac{1}{n} \sum_{i=1}^n (s_0^2(\sigma_i) + (m_0(\sigma_i) - m_0)^2)
\text{ and }
m_0 = \frac{1}{n} \sum_{i=1}^n m_0(\sigma_i).
\]
Fix $C > 0$, there exists choices of $s_0(\cdot) > 0, \sigma_i, m_0(\cdot), P_0 \in
\mathcal
P(m_0, s_0)$ such that $\max_i
s_0^2
(\sigma_i) / \sigma_i^2 < C$ but \[
\frac{\E_{P_0}\bk{
  \frac{1}{n} \sum_{i=1}^n (\hat\theta_i - \theta_i)^2
}}{
  \E_{P_0}\bk{
  \frac{1}{n} \sum_{i=1}^n (\check\theta_i - \theta_i)^2
}
}
\]
is arbitrarily large, where $\hat\theta_i$ are the \indepgauss{} posterior means and
$\check\theta_i$ are the \closegauss{} posterior means: \begin{align*}
\hat\theta_i &= \frac{\sigma_i^2}{s_0^2 + \sigma_i^2} m_0 + \frac{s_0^2}{s_0^2 +
\sigma_i^2} Y_i \\
\check\theta_i &= \frac{\sigma_i^2}{s_0^2(\sigma_i) + \sigma_i^2} m_0(\sigma_i) + 
\frac{s_0^2 (\sigma_i) }{s_0^2(\sigma_i)
+
\sigma_i^2} Y_i.
\end{align*}
\end{lemma}
\begin{proof}
Choose $\sigma_i = 1 + i/n$. Choose a constant $s_0(\sigma_i) = \epsilon>0$ and some
non-constant $m_0
(\sigma_i)$ normalized so that $m_0 = 0$. Thus $s_0^2 \ge \frac{1}{n} \sum_{i=1}^n m_0
(\sigma_i)^2 \equiv K$. Thus, \[
\var(\hat\theta_i) \ge \pr{\frac{K}{K+\sigma_i^2}}^2 (\sigma_i^2 + s_0^2(\sigma^2)) > c >
0
\] for some $c >0$ for all $\epsilon > 0$. Therefore, the numerator is bounded below by
 $c$. The denominator converges to zero as $\epsilon \to 0$. With this choice, $\max_i
 s_0^2(\sigma_i)/\sigma_i^2 < \epsilon$ is eventually smaller than any positive $C$. Thus
 the ratio is arbitrarily large as we take $\epsilon \to 0$. 
\end{proof}

\section{Maximum posterior discrepancy of priors satisfying moment constraints}
\label{asec:max_gauss}

This section contains the main result of \citet{chen2023bayesrisk} (arXiv:2303.08653) and
supersedes that paper.\footnote{\label{fn:error}I am grateful to Isaiah Andrews, Xiao-Li
Meng, Natesh
Pillai, Neil Shephard, and Elie Tamer for their comments. The previous version of the
paper (arXiv:2303.08653v1) claimed that $R (G_1, \sigma; G_0)$ is uniformly bounded over
all $G_0, G_1, \sigma>0$, subjected to the constraints on the first two moments of
$G_0,G_1$. Regrettably, it contained a critical error that rendered its proof incorrect.
In particular, in that version, the display before (A1) on p.4 is incorrect: Posterior
means of mixture priors are mixtures of posterior means under each mixing component, but
the mixing weights are posterior probabilities assigned to each mixing component; thus,
the mixing weights depend on the data rather than being fixed.

This section partially restores that result. \Cref{thm:max_risk}
shows that the maximum Bayes risk under $G_0$ is uniformly bounded over all $G_0, G_1,
\sigma^2$ where $G_1$ satisfies an additional tail condition \eqref{eq:g1_cond}. The
bound we obtain depends on the tail condition, and thus \cref{thm:max_risk} is insufficient
for the result claimed in arXiv:2303.08653v1.} The
notation and setup is entirely self-contained.

Consider an observation $X$ whose likelihood is $X \mid \theta \sim \Norm(\theta,
\sigma^2)$ for some known
$\sigma^2$. There are two priors for $\theta$, denoted by $G_0$ and $G_1$. Suppose both
priors have zero mean and have finite variances bounded by $V > 0$. Consider the decision
problem of estimating $\theta$ under squared error, with $L(a,\theta) = (a-\theta)^2$. For
the Bayesian with prior $G_1$, the Bayes decision rule is the posterior mean $\E_{G_1}
[\theta \mid X]$ under the prior $G_1$. This decision rule attains Bayes risk under the
prior $G_0$\[
R(G_1, \sigma; G_0) \equiv \E_{\theta \sim G_0}\bk{
  (\E_{G_1}[\theta \mid X] - \theta)^2
}. 
\] We can think of $R(G_1, \sigma; G_0)$ as a measure of decision quality under
 disagreement. It measures the quality of $G_1$'s decision from $G_0$'s point of view.
 When $G_1 \neq G_0$, how large can $R (G_1,
\sigma; G_0)$ be?

Since \[ R(G_1, \sigma; G_0) \le 2\pr{
\E_{G_0}[\E_{G_1}[\theta \mid X]^2] + \E_{G_0}[\theta^2] } \le 2V + 2\E_{G_0}[\E_{G_1}
 [\theta \mid X]^2],\] it thus suffices to bound $\E_{G_0}[\E_{G_1}[\theta \mid X]^2]$
 modulo constants. That is,  it suffices to bound the 2-norm $\norm{\E_{G_1}
 [\theta \mid X]}^2$ under the law $X \sim \Norm(0,\sigma^2)
\star G_0$. 

The rest of the section shows that this quantity is uniformly bounded over all $G_0,
G_1, \sigma^2 > 0$. Specifically, \cref{lemma:bound_small_sigma} shows that for all $G_0,
G_1$ that are mean zero and have variance bounded by $V$, $\E_{G_0}[\E_{G_1}[\theta
\mid X]^2]$ is bounded by a constant that depends only on $(V, \sigma^2)$. This bound is
 large when $\sigma^2$ is large. To improve this bound, \cref{thm:max_risk} then shows
 that, if $G_1$ additionally satisfies some conditions on its tail behavior, $\E_ {G_0}
 [\E_{G_1}[\theta
\mid X]^2]$ is bounded by a constant that depends only on $V$ and the tail condition---and
does not depend on $\sigma$.

\begin{lemma}
\label{lemma:bound_small_sigma}
Suppose $G_0, G_1$ have mean zero and variances bounded by $V$, then
\[
\E_{G_0}[\E_{G_1}[\theta \mid X]^2] \le 6V + 4\sigma^2
\]
uniformly over $G_0, G_1, \sigma^2$. 
\end{lemma}
\begin{proof}
Let $f_{G,\sigma}(x) = \int f_X(x \mid \theta) \,G(d\theta)$.
\citet{jiang2020general} (Lemma 1) shows that \[
\pr{\frac{f'_{G,\sigma}(x)}{f_{G,\sigma}(x)}}^2 \le \frac{1}{\sigma^2} \log \pr{\frac{1}
{2\pi \sigma^2 f^2_{G,\sigma}(x)}}.
\]
Plugging in the bound \eqref{eq:jensen_denom} in \cref{lemma:jensen_denom}, we have that
for all $X$, \[
\pr{\E_{G_1}[\theta \mid X] - X}^2 = \pr{\sigma^2\frac{f'_{G_1,\sigma}(x)}{f_
{G_1,\sigma} (x)}}^2 \le \sigma^4 \frac{1}
{\sigma^2} \frac{X^2 + V}
{\sigma^2} = X^2 + V,
\]
where the first equality is due to Tweedie's formula. 

Now, note that \[
(\E_{G_1}[\theta \mid X])^2 \le 2 \pr{
  (\E_{G_1}[\theta \mid X] - X)^2 + X^2
} \tag{$(a+b)^2 \le 2(a^2 + b^2)$}.
\]
Hence, \[
\E_{G_0}\bk{
  \E_{G_1}[\theta \mid X])^2
} \le 2 \E_{G_0}[2X^2+V] \le 2(2(V+\sigma^2) + V) = 6V + 4\sigma^2. \qedhere
\]
\end{proof}

To show a more powerful bound, we require a stronger condition on the tails of $G_1$ and
derive bounds that are independent of $\sigma$ but are dependent on the tail conditions.
In particular, assume
\[ \max\pr{1-G_1(s), G_1(-s) }\le C_{G_1} s^{-k} \numberthis\label{eq:g1_cond}
\]
for some $k > 2$ and $C_{G_1} > 0$, for all $s > 0$. We will also assume that $\E
_{G_1}[\theta^2 \mid X]$ exists almost surely. Note that if $\E_{G_1} |\theta|^
{2+\epsilon} < m$, then $k$ can be taken to be $2+\epsilon$ and $C_{G_1}$ can be taken to
be $m$ by Markov's inequality. In the rest of the proof, we let $C_t < \infty$ denote a
positive constant that depends only on $t$. Occurrences of
$C_t$ might correspond to different constant values.

\begin{theorem}
\label{thm:max_risk}
Suppose $k > 2.$ There exists a constant $Q < \infty$ that depends solely on $(C_ {G_1},
k, V)$ such that, uniformly for all $(G_0, G_1)$ and $\sigma \in \R$, where (i) $G_0,G_1$
have mean zero and variance bounded by $V$ and (ii) $G_1$ satisfies
\eqref{eq:g1_cond} with $(C_{G_1},
k)$, \[
\E_{G_0}\bk{\E_{G_1}[\theta\mid X]^2} \le Q.
\]
\end{theorem}

\begin{proof}
Assume that $ \sigma^2 \ge 1 $. For all $\sigma^2 < 1$, we can apply 
\cref{lemma:bound_small_sigma}
so that $\E_{G_0}\bk{\E_{G_1}[\theta\mid X]^2} \le 6V + 4$. 

Observe that \begin{align*}
\E_{G_0}\bk{\E_{G_1} [\theta \mid X]^2} &\le \E_{G_0} \bk{\E_{G_1}[\theta^2 \mid X]}  \tag{Jensen's
inequality} \\
&= \E_{G_0}\bk{\int_{0}^\infty P_{G_1}(\theta^2 > t \mid X) \,dt} \\
&= 2\E_{G_0}\bk{\int_0^\infty s P_{G_1}(|\theta| > s \mid X)\,ds \tag{Change of variable $s
=
\sqrt{t}$}}
\\
&= 2\E_{G_0}\bk{\int_0^\infty s P_{G_1}(\theta > s \mid X)\,ds + \int_0^\infty s P_{G_1}
(-\theta > -s
\mid X)\,ds}.
\end{align*}
Therefore, it suffices to bound the first term, since the second term follows by a
symmetric argument. We do so in the remainder of the proof. Here, $\E_{G_1}[\theta^2
\mid X]$ exists since $(\theta^2, X)$ is integrable. 

Writing out the first term as an integral:
\begin{align*} &\E_{G_0}
\bk{\int_0^\infty s P_ {G_1} (\theta > s \mid
X)\,ds} \\ &= \int_ {\mu=-\infty}^\infty\int_{x=-\infty}^\infty \int_{s=0}^\infty  s
P_{G_1}[\theta > s
\mid
X=x]\, ds f_X (x\mid \mu)
\,dx\,G_0(d\mu)\\
&= \int_
{\mu=-\infty}^\infty  \int_{s=0}^\infty s \int_{x=-\infty}^\infty P_{G_1}[\theta > s
\mid
X=x]  f_X (x\mid \mu) \,dx \,ds \,G_0(d\mu). \tag{Fubini's theorem}
\end{align*}

The outer integral in $\mu$ can be decomposed into $|\mu| \le \sigma$ and $|\mu| >
\sigma$: \begin{align}
&\E_{G_0}\bk{\int_0^\infty s P_ {G_1} (\theta > s \mid
X)\,ds} \nonumber\\
&= \int_
{|\mu| > \sigma}  \int_{s=0}^\infty s \int_{x=-\infty}^\infty P_{G_1}[\theta > s
\mid
X=x]  f_X (x\mid \mu) \,dx \,ds \,G_0(d\mu)\label{eq:small_sigma}
\\
&\quad+ \int_
{|\mu| < \sigma} \int_{s=0}^\infty s \int_{x=-\infty}^\infty P_{G_1}[\theta > s
\mid
X=x]  f_X (x\mid \mu) \,dx \,ds \,G_0(d\mu) \label{eq:large_sigma}
\end{align}

First, we consider \eqref{eq:small_sigma}. Decompose the integral in $x$ further along $x
\le s/2$ and $x > s/2$: 
\begin{align}
\eqref{eq:small_sigma} &= \int_{|\mu| > \sigma} \int_0^\infty s \int_{s/2}^{\infty}P_{G_1}
(\theta > s
\mid X=x) f_X(x\mid \mu) \,dx \,ds \,G_0(d\mu)\label{eq:term3}
\\&\quad + \int_{|\mu| > \sigma} \int_0^\infty s \int_{-\infty}^{s/2} P_{G_1}
(\theta > s
\mid X=x) f_X(x\mid \mu) \,dx \,ds\,G_0(d\mu) \label{eq:term2_}.
\end{align}

For large $\mu$ and large $x$ \eqref{eq:term3}, we have that
\begin{align*}
\eqref{eq:term3} &\le \int_{|\mu| > \sigma} \int_0^\infty s \int_{s/2}^{\infty} f_X(x\mid
\mu) \,dx ds G_0(d\mu) \tag{$P_{G_1}
(\theta > s
\mid X=x) \le 1$}\\
&= \int_{|\mu| > \sigma} \int_0^\infty s P(X > s/2 \mid \mu)\,ds \,G_0(d\mu) \\ 
&\le C\int_{|\mu| > \sigma} \underbrace{\E[X^2 \mid \mu]}_{\mu^2 + \sigma^2 \le
2\mu^2}\, G_0(d\mu)  \le C \int
2\mu^2
\,G_0(d\mu) \le C_{V}.\tag{$\int 2sP(X > s \mid \mu)\,ds \le \E[X^2 \mid \mu]$}
\end{align*}

For large $\mu$ and small $x$ \eqref{eq:term2_}, note that for $x < s/2< s$, by 
\cref{lemma:jensen_denom}
\begin{align*}
P_{G_1}(\theta > s \mid X=x) &\le C_{V} e^{x^2/(2\sigma^2)} e^{-\frac{1}{2\sigma^2}(x-s)^2}
(1-G_1(s)). \tag{$f_X(x\mid \theta) \le \frac{1}{\sqrt{2\pi} \sigma} e^{-\frac{1}
{2\sigma^2} (x-s)^2}$}
\end{align*}
Now, integrating the above display with respect to $f_X(x\mid \mu)\,dx$ yields \[
 \int_{-\infty}^{s/2}P_{G_1}
(\theta > s
\mid X=x) f_X(x\mid \mu) \,dx \le C_{V}(1-G_1(s)) \cdot \frac{\sigma^2}{s} \le 
C_{V}(1-G_1(s)) \frac{\mu^2} {s} \tag{$|\mu| > \sigma$ for \eqref{eq:small_sigma}}
\]
Finally, integrating it again with respect to $s$ yields \[
\int_0^\infty s \times C_{V}(1-G_1(s)) \frac{\mu^2} {s} \,ds = C_V \mu^2 \int_{0}^\infty
1-G_1(s)\,ds \le C_V \mu^2 \E_{G_1}[|\theta|] \le C_V \mu^2.
\]
Therefore, \[
\eqref{eq:term2_} \le  C_V \E_{G_0}\mu^2 \le C_V.
\]
This shows that \eqref{eq:small_sigma} is uniformly bounded.

Moving on to \eqref{eq:large_sigma}, we first decompose the integral on $s$ into $s \le K$
and $s > K$, for some $K \ge e$ to be chosen: \[
\eqref{eq:large_sigma} \le \underbrace{\int_{|\mu| < \sigma} \int_0^K s \,ds \,G_0(d\mu)}_
{\le K^{2}/2} + \int_
{|\mu| < \sigma} \int_K^\infty s \int_{-\infty}^\infty P_{G_1}(\theta > s \mid X=x) f_X
(x\mid \mu) \,dx\,G_0(d\mu) \numberthis \label{eq:large_sigma_bound}
\] Thus we focus on the piece where $s > K$. Fix \[u = C \sigma \sqrt{\log (s)}\] for some
 $C
\ge 2$ to be chosen. On $s > K$,
$u/\sigma > 2$ and thus $\frac{u}{\sigma} - 1 > \frac{u}{2\sigma}$.
Observe that by \cref{lemma:jensen_denom} and the fact that $\sigma > 1$, \[
P_{G_1}(\theta > s \mid X=x) \le C_V \exp\pr{\frac{x^2}{2\sigma^2}} (1-G(s)). \numberthis
\label{eq:survival_bound}
\]
Therefore, \begin{align*}
&\int_{-\infty}^\infty P_{G_1}(\theta > s \mid X=x) f_X
(x\mid \mu)\,dx \\
&\le \int_{|x| \le u} P_{G_1}(\theta > s \mid X=x) f_X
(x\mid \mu)\,dx + \P(|X|>u \mid \mu) \tag{$P_{G_1}(\theta > s \mid X=x) \le 1$}\\
&\le C_V e^{-\mu^2/(2\sigma^2)}(1-G(s)) \bk{
  \frac{\sinh\pr{\frac{\mu}{\sigma} \frac{u}{\sigma}}}{\mu/\sigma}
} + 2\underbrace{\bar\Phi\pr{\frac{u}{\sigma} - \frac{|\mu|}{\sigma}}}_{\le \bar\Phi
\pr{u/\sigma - 1} \le \bar\Phi\pr{\frac{u}{2 \sigma}}} \tag{$\bar\Phi = 1-\Phi$ is the
complementary Gaussian CDF}\\
&\le C_V(1-G(s)) \frac{\sinh\pr{\frac{\mu}{\sigma} \frac{u}{\sigma}}}{\mu/\sigma} +
2\bar\Phi\pr{\frac{u}{2 \sigma}} \numberthis \label{eq:intermediate_large_sigma}
\end{align*}
where the second inequality follows from directly integrating the upper bound 
\eqref{eq:survival_bound} within $|x| \le u$. Now, observe that, for $|c| < 1$ and $t >
0$, \begin{align*}
t \frac{\sinh(ct)}{ct} &\le t \frac{\sinh(|c|t)}{|c|t} \tag{$\sinh(x)/x$ is an even
function}  \\
&\le t \frac{\sinh(t)}{t} \tag{$\sinh(x)/x$ is an increasing function on $x > 0$} \\
&\le \frac{1}{2}e^{t} \tag{$\sinh(x) = (e^x-e^{-x})/2\le \frac{1}{2}e^x$ for $x > 0$}.
\end{align*}
Therefore, \begin{align*}
\eqref{eq:intermediate_large_sigma} &\le C_V (1-G(s)) \exp\pr{\frac{C}{\sqrt{\log s}} \log
s} + 2 \bar\Phi
(C \sqrt{\log s}) \\
&\le C_V (1-G(s)) \exp\pr{\frac{C}{\sqrt{\log s}} \log
s} + \exp\pr{-\frac{C^2}{2} \log s} \tag{\cref{lemma:millsratio}}.
\end{align*}
Choose $C = k$ and $ K = \exp\pr{1 \vee \frac{(2C)^2}{(k-2)^2}}$. This yields that, for
$s > K$, \[
\frac{C}{\sqrt{\log s}}  \le \frac{k-2}{2} \quad \frac{C^2}{2} = \frac{k^2}{2}.
\]
Hence, integrating with respect to $s$: \begin{align*}
&\int_K^\infty s \int_{-\infty}^\infty P_{G_1}(\theta > s \mid X=x) f_X
(x\mid \mu)\,dx\,ds \\
&\le \int_K^\infty s 
\pr{C_V (1-G(s)) \exp\pr{ \frac{k-2}{2} \log
s} + \exp\pr{-\frac{k^2}{2} \log s}} \,ds \\
&\le C_{V} C_{G_1} \int_K^\infty s^{1-k+\frac{k-2}{2}}\,ds + \int_{K}^\infty s^{-k^2/2 +1}
\,ds \\
&\le C_V C_{G_1} C_k + C_k. \tag{$1-k+(k-2)/2 < -1$ and $-k^2/2+1 < -1$}
\end{align*}
as both integrals converge and depend only on $k > 2$. Returning to 
\eqref{eq:large_sigma_bound}, this shows that \eqref{eq:large_sigma} is uniformly bounded
with a constant that depends only on $V, C_{G_1}, k$. This concludes the proof.
\end{proof}

\begin{lemma}
\label{lemma:jensen_denom}
Suppose $G_1$ has mean zero and variance bounded by $V$. Let \[
f_{G_1, \sigma}(x) \equiv \int f_X(x\mid \theta) \, G_1(d\theta).
\] Then, \[
f_{G_1,\sigma}(x) \ge \frac{1}{\sqrt{2\pi}\sigma} \exp\pr{-\frac{x^2 + V}
{2\sigma^2}}
\text{, or }
\frac{1}{f_{G_1, \sigma}(x)} \le \sqrt{2\pi} \sigma \exp\pr{\frac{x^2 + V}{2\sigma^2}}.
\numberthis \label{eq:jensen_denom}
\]
\end{lemma}

\begin{proof}
Observe that, by Jensen's inequality, \[
f_{G_1, \sigma}(x) \equiv \int f_X(x\mid \theta) \, G_1(d\theta) \ge \exp \int \log f_X
(x\mid
\theta) \,G_1(d\theta).
\]
We compute \[
\log f_X(x \mid \theta) = \log\frac{1}{\sqrt{2\pi} \sigma} - \frac{1}{2\sigma^2} 
(x-\theta)^2. 
\]
Note that $\E_{\theta \sim G_1}[(x-\theta)^2] = x^2 - 2x \E_{G_1}[\theta] + \E_
{G_1}\theta^2 \le x^2 + V$.
Thus \eqref{eq:jensen_denom} follows. \end{proof}

\begin{lemma}
\label{lemma:millsratio}
For all $x \ge 0$, $\bar\Phi(x) \le \frac{1}{2} e^{-x^2/2}$. 
\end{lemma}
\begin{proof}
Note that $\bar\Phi(0) = \frac{1}{2}$ and thus the bound holds with equality at $x=0$.
Differentiate, \[
\bar \Phi'(x) = -\varphi(x) \quad \frac{d}{dx} \frac{1}{2} e^{-x^2/2} = -\frac{x}{2} e^
{-x^2/2} = (-x \sqrt{\pi/2}) \varphi(x)
\]
For $x \in [0, \sqrt{2/\pi}]$, \[
\frac{d}{dx}\bar\Phi(x) \le \frac{d}{dx} \frac{1}{2} e^{-x^2/2} \implies \bar\Phi(x) \le 
\frac{1}{2} e^{-x^2/2}. 
\]
Note that since Mill's ratio is bounded by $1/x$, we have that for all $x > 0$ \[
\bar\Phi(x) \le \varphi(x)/x.
\]
Take $x > \sqrt{2/\pi}$, we have that\[
\bar \Phi(x) \le \varphi(x) \sqrt{\frac{\pi}{2}} = \frac{1}{2} e^{-x^2/2}.
\]
Hence the inequality holds for all $ x \ge 0 $.
\end{proof}

\end{appendices}

\printbibliography
\end{refsection}

\end{document}